\documentclass[reprint,10pt,superscriptaddress]{revtex4-2}%reprint
\usepackage{braket}
\usepackage{mathtools}
\usepackage{adjustbox}
\usepackage[utf8]{inputenc}
\usepackage{lineno}
\usepackage{csquotes}
\usepackage[dvipsnames]{xcolor}
\usepackage{color}
\usepackage{tikz}
\usepackage{pgfplots}
\pgfplotsset{compat=1.17}
\usepackage{graphics}
\usepackage{float}
\usetikzlibrary{shapes.geometric,backgrounds,positioning,shapes.geometric,decorations.markings,decorations.pathreplacing,arrows,knots,hobby,angles,quotes}
\usepackage{amsmath}
\usepackage{amssymb}
\usepackage{amsthm}
\usepackage{braket}
\usepackage{bbm}
\usepackage{theoremref}
\usepackage{pgfplots}

\usepackage{caption} 
\usepackage{subcaption}

\usepackage[draft]{todonotes} 
\DeclareMathOperator{\tr}{tr}

\usepackage{hyperref}
\hypersetup{
	breaklinks=true,   % splits links across lines
	colorlinks=true,   % displays links as colored text instead of blocks
	pdfusetitle=true,  % \title and \author values into pdf metadata
}
\usepackage{breakurl}

\usepackage{cleveref}
\crefname{section}{Section}{Sections}
\Crefname{section}{Section}{Sections}
\crefname{equation}{}{}
\Crefname{equation}{}{}
\crefname{figure}{Figure}{Figures}
\Crefname{figure}{Figure}{Figures}
\crefname{appendix}{Appendix}{Appendices}
\Crefname{appendix}{Appendix}{Appendices}
\usepackage[T1]{fontenc}
\definecolor{color1}{RGB}{94,60,153}
\definecolor{color2}{RGB}{230,97,1}
\definecolor{color3}{RGB}{253,184,99}

\newtheorem{lem}{Lemma}

\newtheorem{theorem}{Theorem}[section]
\newtheorem{proposition}{Proposition}[section]

\newtheorem{claim}{Claim}

\newtheorem{remark}{Remark}[section]
\newtheorem{corollary}{Corollary}[section]

\theoremstyle{definition}
\newtheorem{example}{Example}[section]

\newcommand{\la}{\langle}
\newcommand{\ra}{\rangle}

\newcommand{\eps}{\varepsilon}

\newcommand{\floor}[1]{\lfloor #1 \rfloor}
\newcommand{\kbb}[2]{|#1\rangle\!\langle#2|} %ketbra

\def\ve{\varepsilon}

\def\bea{\begin{eqnarray}}
\def\eea{\end{eqnarray}}
\def\nn{\nonumber}
\def\beq{\begin{equation}}
\def\eeq{\end{equation}}
\def\ba{\begin{eqnarray}}
\def\ea{\end{eqnarray}}
\def\be{\ba\displaystyle}
\def\ee{\ea}

\usepackage{titlesec}
\newcommand{\h}{\mathcal{H}}
\newcommand{\T}{\mathcal{T}}
\newcommand{\U}{\mathcal{U}}
\newcommand{\D}{\mathcal{D}}
\newcommand{\B}{\mathcal{B}}
\newcommand{\N}{\mathcal{N}}
\newcommand{\M}{\mathcal{M}}
\newcommand{\id}{\mathbb{I}}
\newcommand{\E}{\mathcal{E}}

\newcommand{\kb}[1]{|#1\rangle\!\langle#1|} %ketbra
\newcommand{\bE}{\mathbbm{E}}

\newcommand{\1}{\mathbbm{1}}

\begin{document}
%\linenumbers

\title{Learning Quantum Processes with Memory - Quantum Recurrent Neural Networks} 
\author{Dmytro Bondarenko}
\affiliation{Institut f\"ur Theoretische Physik, Leibniz Universit\"at Hannover, Appelstraße 2, 30167 Hannover, Germany}
\affiliation{Stewart Blusson Quantum Matter Institute, University of British Columbia, 2355 East Mall, V6T 1Z4 Vancouver, BC, Canada, dimbond@live.com}
\author{Robert Salzmann}
\affiliation{Institut f\"ur Theoretische Physik, Leibniz Universit\"at Hannover, Appelstraße 2, 30167 Hannover, Germany}
\affiliation{Department of Applied Mathematics and Theoretical Physics, University of Cambridge, Cambridge CB3 0WA, United Kingdom, rals.salzmann@web.de}
\author{Viktoria-S. Schmiesing}
\affiliation{Institut f\"ur Theoretische Physik, Leibniz Universit\"at Hannover, Appelstraße 2, 30167 Hannover, Germany, viktoria.schmiesing@itp.uni-hannover.de}
\maketitle

\textbf{Recurrent neural networks play an important role in both research and industry. With the advent of quantum machine learning, the quantisation of recurrent neural networks has become recently relevant. We propose fully quantum recurrent neural networks, based on dissipative quantum neural networks, capable of learning general causal quantum automata. A quantum training algorithm is proposed and classical simulations for the case of product outputs with the fidelity as cost function are carried out. We thereby demonstrate the potential of these algorithms to learn complex quantum processes with memory in terms of the exemplary delay channel, the time evolution of quantum states governed by a time-dependent Hamiltonian, and high- and low-frequency noise mitigation. Numerical simulations indicate that our quantum recurrent neural networks exhibit a striking ability to generalise from small training sets.}\\[1ex]

\noindent

\textbf{{\large Introduction}}\\

Machine learning (ML) has made remarkable progress during the last decades~\cite{Goodfellow2016, Nielsen2015, Jordan2015}. Central among the many key successes of ML are \emph{recurrent neural networks} (RNNs): These networks are especially suited to the learning of sequential data generated in a wide variety of natural settings, with use cases including time evolution~\cite{Medsker2001}, speech recognition~\cite{Mikolov2010}, the prediction of electric power demand~\cite{Costa1999}, and machine translation~\cite{Kalchbrenner2013}. Various RNN architectures have been proposed, ranging from fundamental RNNs to the more complicated, such as LSTMs~\cite{Hochreiter1997}, and gated recurrent units (GRUs)~\cite{Cho2014, Chung2014}. However, challenges remain as ML algorithms are typically computationally expensive to implement. 

Quantum computing promises new techniques to bolster the faltering Moore's law~\cite{Prati2017}. This has led, in the context of ML, to the advent of 
\emph{quantum machine learning} (QML)~\cite{Biamonte2017}. A rough taxonomy of QML algorithms may be provided by categorising them according to whether they exploit either quantum processing (CQ), quantum data (QC), or both (QQ)~\cite{ABG06, Wikipedia2018}. CQ ML is built on quantum algorithms to learn classical data, mainly aimed at speeding up classical learning algorithms~\cite{Aimeur2013, Paparo2014, Schuld2014, Wiebe2016}. QC ML utilises classical machine learning for quantum problems, such as quantum metrology~\cite{Lovett2013}, simulating quantum many-body systems~\cite{Carleo2017}, or adaptive quantum computation~\cite{Tiersch2015}. QQ ML directly targets the development of quantum learning algorithms to analyse quantum data~\cite{DB18,SC02,G08,SCMB12,DTB16,MSW17,Alvarez2017,Amin2018,Du2020,SMMCB19}.

Many QML architectures would benefit from the ability to repeatedly query a (quantum) memory storage via a recurrent network structure. We say the system has memory if the output at time $t$ depends on the inputs at times $\tau \leq t$. 

\begin{figure}[t]
	\begin{tikzpicture}[scale=0.6]
		\draw (0,0.75) node[circle,minimum height=0.3cm,minimum width=.3cm,draw] (U000){};
		\draw (0,0) node[circle,minimum height=0.3cm,minimum width=.3cm,draw] (U001){};
		\draw (1,1.125) node[circle,minimum height=0.3cm,minimum width=.3cm,draw] (U010){};
		\draw (1,0.375) node[circle,minimum height=0.3cm,minimum width=.3cm,draw] (U011){};
		\draw (1,-0.375) node[circle,minimum height=0.3cm,minimum width=.3cm,draw] (U012){};
		\draw (2,0.75) node[circle,minimum height=0.3cm,minimum width=.3cm,draw] (U020){};
		\draw (2,0) node[circle,minimum height=0.3cm,minimum width=.3cm,draw] (U021){};
		\draw (U000.north east)--(U010.west);
		\draw (U001.north east)--(U010.south west);
		\draw (U000.east)--(U011.north west);
		\draw (U001.east)--(U011.south west);
		\draw (U000.south east)--(U012.north west);
		\draw (U001.south east)--(U012.west);
		\draw (U010.east)--(U020.north west);
		\draw (U011.north east)--(U020.west);
		\draw (U012.north east)--(U020.south west);
		\draw (U010.south east)--(U021.north west);
		\draw (U011.south east)--(U021.west);
		\draw (U012.east)--(U021.south west);
		\draw[->,orange,thick] (-0.75,0)--(U001.west);
		\draw[->,blue,thick] (-0.75,0.75)--(U000.west);
		\draw[->,purple,thick] (U021.east)--(10.75,0);		
		\draw (3,1.5) node[circle,minimum height=0.3cm,minimum width=.3cm,draw] (U100){};
		\draw (3,0.75) node[circle,minimum height=0.3cm,minimum width=.3cm,draw] (U101){};
		\draw (4,1.875) node[circle,minimum height=0.3cm,minimum width=.3cm,draw] (U110){};
		\draw (4,1.125) node[circle,minimum height=0.3cm,minimum width=.3cm,draw] (U111){};
		\draw (4,0.375) node[circle,minimum height=0.3cm,minimum width=.3cm,draw] (U112){};
		\draw (5,1.5) node[circle,minimum height=0.3cm,minimum width=.3cm,draw] (U120){};
		\draw (5,0.75) node[circle,minimum height=0.3cm,minimum width=.3cm,draw] (U121){};
		\draw (U100.north east)--(U110.west);
		\draw (U101.north east)--(U110.south west);
		\draw (U100.east)--(U111.north west);
		\draw (U101.east)--(U111.south west);
		\draw (U100.south east)--(U112.north west);
		\draw (U101.south east)--(U112.west);
		\draw (U110.east)--(U120.north west);
		\draw (U111.north east)--(U120.west);
		\draw (U112.north east)--(U120.south west);
		\draw (U110.south east)--(U121.north west);
		\draw (U111.south east)--(U121.west);
		\draw (U112.east)--(U121.south west);
		\draw[->,orange,thick] (U020.east)--(U101.west);
		\draw[->,blue,thick] (-0.75,1.5)--(U100.west);
		\draw[->,purple,thick] (U121.east)--(10.75,0.75);
		\draw[->,orange,thick] (U120.east)--(5.75,1.5);
		\draw (8,3.75) node[circle,minimum height=0.3cm,minimum width=.3cm,draw] (Un00){};
		\draw (8,3) node[circle,minimum height=0.3cm,minimum width=.3cm,draw] (Un01){};
		\draw (9,4.125) node[circle,minimum height=0.3cm,minimum width=.3cm,draw] (Un10){};
		\draw (9,3.375) node[circle,minimum height=0.3cm,minimum width=.3cm,draw] (Un11){};
		\draw (9,2.625) node[circle,minimum height=0.3cm,minimum width=.3cm,draw] (Un12){};
		\draw (10,3.75) node[circle,minimum height=0.3cm,minimum width=.3cm,draw] (Un20){};
		\draw (10,3) node[circle,minimum height=0.3cm,minimum width=.3cm,draw] (Un21){};
		\draw (Un00.north east)--(Un10.west);
		\draw (Un01.north east)--(Un10.south west);
		\draw (Un00.east)--(Un11.north west);
		\draw (Un01.east)--(Un11.south west);
		\draw (Un00.south east)--(Un12.north west);
		\draw (Un01.south east)--(Un12.west);
		\draw (Un10.east)--(Un20.north west);
		\draw (Un11.north east)--(Un20.west);
		\draw (Un12.north east)--(Un20.south west);
		\draw (Un10.south east)--(Un21.north west);
		\draw (Un11.south east)--(Un21.west);
		\draw (Un12.east)--(Un21.south west);
		\draw[->,orange,thick] (7.25,3)--(Un01.west);
		\draw[->,blue,thick] (-0.75,3.75)--(Un00.west);
		\draw[->,purple,thick] (Un21.east)--(10.75,3);
		\draw[->,orange,thick] (Un20.east)--(10.75,3.75);
		\draw[dotted, thick, purple] (10.5,1)--(10.5,2.75);
		\draw[dotted, thick, blue] (-0.25,1.75)--(-0.25,3.5);
		\draw[dotted, thick, orange] (6,1.75)--(7,2.75);
		\draw (-0.75,-0.25)--(-0.75,4)--(-1.5,4)--(-1.5,-0.25)--(-0.75,-0.25);
		\draw (-1.125,1.875) node[](){\(\rho^\text{IN}\)};
		\draw[decorate, decoration={brace}, yshift=2ex]  (11,3.75) -- node[right=0.4ex] {\(\rho^\text{OUT}\)}  (11,-0.75);
	\end{tikzpicture}
\caption{\textbf{A general Quantum Recurrent Neural Network.} A QRNN is an iteration over the memory system (orange) of a feed-forward QNN with \(L\) hidden layers (\(L=1\) in Figure) where the total input is split into an input (blue) and the memory and the total output is split into an output (purple) and the memory. For a concrete number of iterations we can input any state \(\rho^\text{IN}\) on \(\h^\text{m}\otimes\h^{\text{in}}\otimes \dots\otimes\h^{\text{in}}\) and get out a state \(\rho^\text{OUT}\) on \(\h^{\text{out}}\otimes \dots\otimes\h^{\text{out}}\otimes \h^\text{m}\).}
\label{QRNN}
\end{figure}
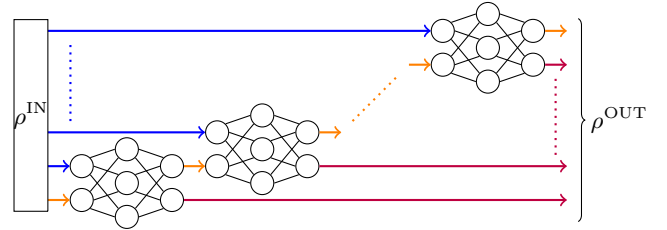
In the CQ context, there are attempts to use quantum computers, quantum channels~\cite{cptp_language_cats} and states~\cite{Information_retreival-Hilbert_space, Q_language_IR_Benjio, QNLP_embedding} for tasks where classical RNNs are often used, such as natural language processing~\cite{QNLP_cat_2016, QLSTM_translation, Qubits_speak, QNLP_music}. In this context, there are already several quantisations of RNNs~\cite{BehrmannQRNN2000, CirsteahandwritingQRNN2018, ElkenawyQDRNN2021, ZakQRNN1999, CQRNN_Qadvantage}, usually exploiting a substitution of the feed-forward NNs in RNNs, LSTMs or GRUs with some form of QNN. Typically such QNNs encode classical information into a quantum state, apply a variational circuit~\cite{takakiQRNN2020, BauschQRNN2020, ChoiQGRNN2021,chenquantumLSTM2020, QLSTM_translation, ChenQREDNN2020} or a tunable evolution~\cite{DawesQRNN1992, BeheraQRNN2004, BeheraQRNN2006, GandhiQRNNEEG2014, GandhiQRNNEMG2013, LuQRNNQKD2018, CQRNN_Qadvantage}, and then carry out a measurement, followed by further classical processing. A similar approach can be pursued with non-RNN variational classes used for machine learning with memory, such as matrix-product operators~\cite{DrD2021, QTN_ML_2019}. 
Non-universal CQ RNNs that exhibit contextuality were proven capable of better scaling of memory resources than their entirely classical counterparts~\cite{CQRNN_Qadvantage}. Amazingly, these networks were also able to outperform widely-used classical architectures of similar size in a natural language translation task~\cite{CQRNN_Qadvantage}. This result suggests a wide industrial appeal of QRNNs when fault-tolerant quantum computers become available.

Unfortunately, with state-of-the-art classical natural language processing ML models training hundreds of billions of parameters and employing highly connected sophisticated architectures~\cite{GPT3, Megatron-NLG}, it is unclear if CQ RNNs can be competitive in the current noisy intermediate-scale quantum (NISQ) era.

Unlike for CQ learning, there is no classical alternative to QQ ML. If the task calls for direct handling of quantum data, QQ ML can immediately improve the state-of-the-art even if applied to modestly-sized systems.
Many tasks require quantum data spanning simulation, metrology, error correction and denoising, and communication~\cite{DrD2021}. It is known that entangled states are essential to represent many-body correlations in materials, able to enhance variance scaling of an $n$-shot measurement from $\sim 1/n$ to $\sim 1/n^2$, allow for unprecedentedly secure key distribution, and can significantly speed up specific algorithms if properly shielded from noise. For complex systems in hard-to-characterise environments, data is often easier to procure than exact solutions. Thus, QQ ML is a natural approach to searching for good entangled states for many applications.

In the QQ scenario, some exciting tasks also require {\emph{quantum memory}. {We call the memory quantum if the time dependence is facilitated by a subsystem that is able to store entangled states.} One of the simplest examples where quantum memory enlarges the capabilities of quantum protocols is when an extra state does not change during a particular process yet facilitates it. This is precisely what a catalyst does --- a widely applied concept in chemistry~\cite{MaselCatalyst2001} and quantum information processing~\cite{QCatalist_PRL1999}. There are also many practically relevant tasks where it is desirable to add (quantum or classical) memory that can change during the process.} These {QQ} tasks include modelling the time-dependent evolution of a quantum system in order to steer it in a desired direction. Further application areas arise in chemistry and molecular dynamics~\cite{TDObservables_RKrems_2022}, quantum control{~\cite{DongQuantumControl2010, DAlessandroQuantumControl2021, Non-Markovian_Steady_States}}, scheduling for quantum annealers~\cite{QAnnealing_speedup_schedule}, and the design of classical and quantum adaptive filters \cite{brandenburg1999mp3, Adaptive_filter_theory_book, DSP_book, DrD2021}.

{However, the QQ scenario is targeted by comparatively fewer architectures. One example of such a fully quantum RNN was presented in~\cite{killoranContinuousvariableQuantumNeural2019} for systems of continuous variables. In~\cite{verdon2019QGNN}, a graph quantum RNN was used to represent a Trotterised Hamiltonian evolution. For all but the first and last steps, the memory in~\cite{verdon2019QGNN} is treated as the entire input and output. Instead, we want to focus on qudits and having in- and outputs in every time step.
}

{To build a fully quantum recurrent neural network (QRNN), we must pick a suitable neuron. Just like classical RNNs, neurons should not modify subsystems not connected to it, and big enough networks should be universal~\cite{RNN_Turing1, RNN_Turing886, RNN_Turing3, RNN_Turing4}. We say that a QRNN architecture is universal if it can implement any causal quantum automaton, i.e. any quantum process where the output at time \(t\) only depends on the input at times \(t^\prime \leq t\)~\cite{Kretschmann2005}. Implementing a dissipative quantum channel might be desirable even if the data is pure, as the ability to forget is quintessential for learning long-term dependencies~\cite{LSTM_forget} and robustness to data corruption~\cite{RNN_grad_problems_3}. These requirements are conveniently satisfied by dissipative QNNs from~\cite{beerTrainingDeepQuantum2020}, which can represent any quantum channel.} These QNNs are then connected as depicted in Fig.~\ref{QRNN} to build the QRNN. We apply our QRNN to learn sequential quantum data generated in various contexts, including the delay channel, the time evolution of a state governed by a time-dependent Hamiltonian, and high- and low-frequency noise mitigation. Our classical simulations demonstrate a striking capability of the QRNN architecture to generalise from small training sets. If desired, our QRNNs can also be restricted to classical processing or classical data.
\noindent\\[1ex]
\textbf{{\large Results}}\\[1ex]
\textbf{The network architecture.}
A feed-forward (dissipative) quantum neural network (QNN) for qudits as proposed in~\cite{beerTrainingDeepQuantum2020} is built from connected quantum perceptrons arranged in \(L+2\) layers, i.e.\(\, L\) hidden layers, one input and one output layer, with \(m_l\) perceptrons in layer \(l\). To the perceptron in the \(l^\text{th}\) layer at place \(j\) we assign a qubit and a unitary \(U_j^l\). The unitary acts on the corresponding perceptron qubit and the qubits in the \((l-1)^{\text{th}}\) layer connected to it. The unitaries of layer \(l\) can be pooled together to create a layer unitary \(U^l=U^l_{m_l}\dots U^l_1\), and these can again be pooled together to yield the network unitary \(\mathcal{U}=U^{L+1}\dots U^1\). For a total input state \(\rho^0\), the total output is 
\begin{align}
	\label{eq:a}
	\rho^{L+1}=& \mathcal{N}_{\U}\left( \rho^{0} \right)\\ =&\mathrm{tr}_{0,...,L}(\mathcal{U}(\rho^{0}\otimes \ket{0\dots 0}_{1,...,L+1}\bra{0\dots 0})\mathcal{U}^\dagger),
\end{align}
which also can be written as \(\mathcal{E}^{L+1}\left(\dots \mathcal{E}^2 \left(\mathcal{E}^1\left(\rho^{0}\right)\right)\right)\) with the layer-to-layer channels \(\mathcal{E}^l\) being defined by \(\mathcal{E}^l(X^{l-1})=\mathrm{tr}_{l-1}(U^l(X^{l-1}\otimes \ket{0\dots 0}_{l}\bra{0\dots 0})U^{l^\dagger})\). This definition of a perceptron is naturally connected to the open-system representation of a quantum channel ~\cite{wolf2012quantum, DrD2021, DrBeer2022}.

The defining feature of a recurrent network is that some perceptron outputs are used again as part of the input in the same or previous layers. Without loss of generality, we exploit a subsystem of the total output as part of the input, which is often called the \textit{memory}. Hence, the total input Hilbert space \(\h^\text{intotal}\) is split into an input Hilbert space \(\h^\text{in}\) and a memory Hilbert space \(\h^\text{m}\), and accordingly the total output Hilbert space \(\h^\text{outtotal}\) is split into an output Hilbert space \(\h^\text{out}\) and \(\h^\text{m}\). The memory part of the total output of a QNN is then input into the memory part of the total input of the QNN in the next iteration step.

We index the Hilbert spaces, states, channels and unitaries with the number of the iteration.
In general, we can allow an arbitrary input state \(\rho^\text{IN}\) on \(\h^\text{IN}_{N}=\h_0^\text{m}\otimes \h_1^\text{in}\otimes \dots \otimes \h_{N}^\text{in}\) giving rise to the output state
\begin{align}
	\rho^\text{OUT} &:=\left(\left(\id^\text{out}_{1,...,N-1}\otimes \N_{\U}\right) \circ \dots \circ  \left(\N_{\U}\otimes \id^\text{in}_{2,...,N}\right)\right)\left( \rho^\text{IN} \right)\nonumber\\
	&:=\left( \N_{\U_{N}}\circ \dots \circ  \N_{\U_1}\right)\left( \rho^\text{IN} \right)
\end{align}
on \(\h^\text{OUT}_{N}=\h_1^\text{out}\otimes \dots \otimes \h_{N}^\text{out} \otimes \h_{N}^\text{m} \). See Fig.~\ref{QRNN} for an illustration.
It is a direct consequence of the quantum-circuit structure of our QRNNs that they can be used to describe all causal quantum automata with finite in- and output systems. To argue this firstly note that the strategy above directly gives rise to concatenated quantum channels with memory, as described in, e.g.~\cite{Kretschmann2005, Rybar2009, Rybar2015}, only that the QNN with channel \(\N_{\U}\) in the above description is replaced with an arbitrary channel. Secondly, quantum channels with memory can be used to describe general causal quantum automata~\cite{Kretschmann2005} where the memory system is finite in the case of finite in- and output systems~\cite{Eggeling2002}. Finally, a QNN, as described above, can, if it is large enough, represent any quantum channel on a finite input or output system~\cite{beerTrainingDeepQuantum2020}. Together, this means that our QRNNs can represent any causal quantum automaton if the underlying QNN is large enough.\\
\noindent
\textbf{The training process.}
Now that we have an architecture for our QRNN, we can specify the learning task. To quantise the classical scenario, one might first replace each classical input and output with a quantum state. However, in the quantum mechanical setting, the different in- and outputs could also be entangled. In general, the training set is given by \(S=\left\{ S_\alpha \right\}_{\alpha=1}^{M}=  \left\{ \left( \rho^\text{IN}_\alpha, \sigma^\text{OUT}_\alpha \right) \right\}_{{\alpha}=1}^{M} \in \bigtimes_{{{\alpha}}=1}^{M} \left(\D (\h^\text{IN}_{N_\alpha})\times \D (\h^\text{OUT}_{N_\alpha} \right)) \) for \(M\) runs of a QRNN where we iterate over the underlying QNN \(N_{\alpha}\) times, and \(\D (\h)\) denotes the set of density operators on a Hilbert space \(\h\).  Note that the training set can contain states of varying dimensionality. Nevertheless, for the quantum algorithm to work, we need to be able to request multiple copies of every in- and output state.

We focus on the case of product in- and output states here. That is, the input state is given by \(\rho^\text{IN}_\alpha=\rho_{0 \, \alpha}^\text{m}\otimes \rho_{1\, \alpha}^\text{in}\otimes \dots \otimes \rho_{N_{\alpha} \, \alpha}^\text{in}\) and the output state by \(\sigma_\alpha^\text{OUT}=\sigma_{1 \, \alpha}^\text{out}\otimes \dots \sigma_{N_\alpha \, \alpha}^\text{out}\otimes \sigma_{N_{{\alpha}} \, \alpha}^\text{m}\). In this case, we often write \(S_\alpha=\left(\rho_{0 \, \alpha}^\text{m},\left(\rho_{1 \, \alpha}^\text{in},\rho_{1 \, \alpha}^\text{out}\right),\dots,  \left(\rho_{N_{{\alpha}} \, \alpha}^\text{in},\rho_{N_{{\alpha}} \, \alpha}^\text{out}\right) \right)\). The case of entangled inputs is left to a future publication.
To evaluate the performance of our QRNN for learning training data, i.e., to evaluate how close the network output $\rho_{\alpha}^{\text{OUT}}$ for the input $\rho_\alpha^\text{IN}$ is to the correct output $\sigma_\alpha^\text{OUT}$, we need a cost function. If $\sigma_\alpha^\text{OUT}$ is a product state, we can either compare $\rho_{\alpha}^{\text{OUT}}$ with $\sigma_\alpha^\text{OUT}$ locally or globally. 
The local case corresponds to trying to locally learn the different states \(\sigma_{x \, \alpha}^\text{out}\). We compare the states \(\sigma_{x \, \alpha}^\text{out}\) with the reduced density operators on the \(x^\text{th}\) output in the \(\alpha ^\text{th}\) run, i.e.
\(\rho_{x \, \alpha}^\text{out}=\tilde{\mathrm{tr}}_{\text{out}}^{x\, {{\alpha}}}  \left(\rho_\alpha^\text{OUT}\right) =\tilde{\mathrm{tr}}_{\text{out}}^{x\, {{\alpha}}} \left( \M_{\U}\left(\rho_\alpha^\text{IN}\right)\right) , \)
where \(\tilde{\mathrm{tr}}\) traces out everything but the indicated layer. Operationally, the fidelity is the essentially unique measure of closeness for pure quantum states. Therefore, we use it to define our cost function for pure output states. Hence, as a direct generalisation of the risk function considered in training classical RNNs, for \(\sigma_{x \, \alpha}^\text{out}=\ket{\phi_{x \, \alpha}^\text{out}} \bra{\phi_{x \, \alpha}^\text{out}}\) the \textit{local cost} is given by
\begin{equation}
	C_\text{local,pure}=\frac{1}{M}\sum_{\alpha=1}^{M}{\frac{1}{N_{\alpha}}}\sum_{x=1}^{N_{{\alpha}}}\bra{\phi_{x \, \alpha}^\text{out}} \rho_{x \, \alpha}^\text{out} \ket{\phi_{x \, \alpha}^\text{out}}.
\end{equation}
We ignore the last memory state as it is unimportant to us in which state the memory is at the end. For mixed states the fidelity is given by \(F(\rho,\sigma)=\left(\mathrm{tr}\sqrt{\rho^{1/2}\sigma \rho^{1/2}}\right)^2\). {A polynomial-time algorithm exists to estimate mixed state fidelity~\cite{MixedFidelityAlg}}. Nevertheless, the Hilbert-Schmidt distance is much easier to estimate and minimise in classical simulation or quantum implementation than the mixed state fidelity~\cite{MixedFidelityAlg, Estimator_QAE}. Hence, we use the square of the Hilbert-Schmidt distance {$\mathrm{tr}(\rho -\sigma)^2$} instead if \(\sigma_{x\, \alpha}^\text{out}\) is not pure in general.

The global case corresponds to trying to learn the overall product state. We compare the global output of the training set with the global output of the QRNN. For pure outputs, we use the fidelity between the global output state of the QRNN and the correct one. Again, the memory state at the end is not important to us. As a state is a product state, if the partial trace over a part of the system is a pure state, we can also neglect the state of the memory at the end and choose
\begin{align}
	C_\text{global, pure}=& \frac{1}{M}\sum_{\alpha=1}^M \left( \bra{\phi_{1 \, \alpha}^\text{out}}\otimes \dots\otimes \bra{\phi_{N_{{\alpha}} \, \alpha}^\text{out}} \right) \nonumber\\ &\mathrm{tr}_\text{m}^N \left(\rho_\alpha^\text{OUT}\right)
	 \left( \ket{\phi_{1 \, \alpha}^\text{out}}\otimes \dots\otimes \ket{\phi_{N_{{\alpha}} \, \alpha}^\text{out}} \right)
\end{align} 
as the cost function. We call this the \textit{global cost}.
\begin{figure*}[t]
	\begin{subfigure}{0.32\textwidth}
		\begin{tikzpicture}
			\node at (-1,3.6) {\textbf{a}};
			\begin{axis}[
				width=\linewidth,  
				xlabel= number of training pairs, 
				ylabel= Cost,
				xmin=0,
				xmax=11,
				ymin=0.45,
				ymax=1.05,
				legend columns=2, 
				legend style={at={(0.5,-0.3)},anchor=north}, 
				]
				\addplot+[only marks,color=red,mark=triangle,mark size=2.9pt]  table[x=N,y=Cwomround,col sep=comma] {Generalisationbehaviour.csv};
				\addlegendentry{\begin{tikzpicture}[scale=0.1, baseline]
						\node(1) [circle,draw,inner sep=0pt,minimum size=3.5pt] at (-1,0.15) {};
						\node(4) [circle,draw,inner sep=0pt,minimum size=3.5pt] at (0,0.15) {};
						\draw (1)--(4);
					\end{tikzpicture} training} 
				\addplot+[only marks,color=blue,mark=triangle,mark size=2.9pt]   table[x=N,y=C1mround,col sep=comma] {Generalisationbehaviour.csv};
				\addlegendentry{\begin{tikzpicture}[xscale=0.1,yscale=0.175, baseline]
						\node(1) [circle,draw,inner sep=0pt,minimum size=3.5pt] at (-1,0) {};
						\node(2) [circle,draw,inner sep=0pt,minimum size=3.5 pt] at (-1,0.3) {};
						\node(4) [circle,draw,inner sep=0pt,minimum size=3.5pt] at (0,0) {};
						\node(5) [circle,draw,inner sep=0pt,minimum size=3.5pt] at (0,0.3) {};
						\draw (1)--(4);
						\draw (2)--(4);
						\draw (1)--(5);
						\draw (2)--(5);
						\draw (5) -- (0.3,0.3);
						\draw[out=0,in=0] (0.3,0.3) to (0.3,0.5);
						\draw (0.3,0.5) to (-1.3,0.5);
						\draw[out=180,in=180] (-1.3,0.5) to (-1.3,0.3);
						\draw (-1.3,0.3) -- (2);
					\end{tikzpicture} training}
				\addplot+[only marks,color=red,mark=square,mark size=2pt]   table[x=N,y=Cwomtestround,col sep=comma] {Generalisationbehaviour.csv}; 
				\addlegendentry{\begin{tikzpicture}[scale=0.1, baseline]
						\node(1) [circle,draw,inner sep=0pt,minimum size=3.5pt] at (-1,0.15) {};
						\node(4) [circle,draw,inner sep=0pt,minimum size=3.5pt] at (0,0.15) {};
						\draw (1)--(4);
					\end{tikzpicture} testing}
				\addplot+[only marks,color=blue,mark=square,mark size=2pt]   table[x=N,y=C1mtestround,col sep=comma] {Generalisationbehaviour.csv};
				\addlegendentry{\begin{tikzpicture}[xscale=0.1,yscale=0.175, baseline]
						\node(1) [circle,draw,inner sep=0pt,minimum size=3.5pt] at (-1,0) {};
						\node(2) [circle,draw,inner sep=0pt,minimum size=3.5 pt] at (-1,0.3) {};
						\node(4) [circle,draw,inner sep=0pt,minimum size=3.5pt] at (0,0) {};
						\node(5) [circle,draw,inner sep=0pt,minimum size=3.5pt] at (0,0.3) {};
						\draw (1)--(4);
						\draw (2)--(4);
						\draw (1)--(5);
						\draw (2)--(5);
						\draw (5) -- (0.3,0.3);
						\draw[out=0,in=0] (0.3,0.3) to (0.3,0.5);
						\draw (0.3,0.5) to (-1.3,0.5);
						\draw[out=180,in=180] (-1.3,0.5) to (-1.3,0.3);
						\draw (-1.3,0.3) -- (2);
					\end{tikzpicture} testing}
				\addplot [domain=0:11, samples=10, color=gray,]{0.5};
				\addplot [domain=0:11, samples=10, color=gray,]{1};
			\end{axis}
		\end{tikzpicture}
	\end{subfigure}
	\hfill
	\begin{subfigure}{0.32\textwidth}
		\begin{tikzpicture}[scale=0.65]
			\node at (-1,6.6) {\textbf{b}};
			\draw[] (3,3) circle (2.675);
			\draw[->] (3,3)--(3,6) node[above]{\(tr(\rho \sigma^z)\)};
			\draw[->] (3,3)--(6,3) node[right]{\(tr(\rho \sigma^y)\)};
			\draw[white] (-1,-1)--(6,-1);
			\node[blue](4) [circle,draw,inner sep=0pt,minimum size=12pt] at (3.094,0.328) {};
			\node[blue](5) [circle,draw,inner sep=0pt,minimum size=12pt] at (0.325,3.005) {};
			\node[blue](6) [circle,draw,inner sep=0pt,minimum size=12pt] at (3.013,5.675) {};
			\node[thick](1) [star,draw,inner sep=0pt,minimum size=12pt] at (3,5.675) {\(1\)};
			\node[thick](2) [star,draw,inner sep=0pt,minimum size=12pt] at (0.325,3) {\(2\)};
			\node[thick](3) [star,draw,inner sep=0pt,minimum size=12pt] at (3,0.325) {\(3\)};
			\node[red](4) [circle,draw,inner sep=0pt,minimum size=12pt] at (0.526,1.994) {\(1\)};
			\node[red](5) [circle,draw,inner sep=0pt,minimum size=12pt] at (0.389,2.609) {\(3\)};
			\node[red](6) [circle,draw,inner sep=0pt,minimum size=12pt] at (0.469,2.35) {\(2\)};
			\draw[->] (2.8,5.77) arc (95:174:2.85);
			\draw[->] (0.23,2.8) arc (185:264:2.85);
			\draw[->] (0.12,3.25) arc (175:93.5:2.87);
			\draw[->] (2.75,0.12) arc (265:183.5:2.87);
			
			\node[thick](7) [star,draw,inner sep=0pt,minimum size=12pt] at (2,-0.75){};
			\node (8) at (4,-0.75){training data};
			
			\node[blue](10) [circle,draw,inner sep=0pt,minimum size=12pt] at (1.1,-1.7) {};
			\node(11) [circle,draw,inner sep=0pt,minimum size=3.5pt] at (1.7,-1.85) {};
			\node(12) [circle,draw,inner sep=0pt,minimum size=3.5 pt] at (1.7,-1.55) {};
			\node(14) [circle,draw,inner sep=0pt,minimum size=3.5pt] at (2.45,-1.85) {};
			\node(15) [circle,draw,inner sep=0pt,minimum size=3.5pt] at (2.45,-1.55) {};
			\node(16) [circle,draw,inner sep=0pt,minimum size=3.5pt] at (3.2,-1.85) {};
			\node(17) [circle,draw,inner sep=0pt,minimum size=3.5pt] at (3.2,-1.55) {};
			\draw (11)--(14);
			\draw (12)--(14);
			\draw (11)--(15);
			\draw (12)--(15);
			\draw (16)--(14);
			\draw (17)--(14);
			\draw (16)--(15);
			\draw (17)--(15);
			\draw (17) -- (3.3,-1.55);
			\draw[out=0,in=0] (3.3,-1.55) to (3.3,-1.35);
			\draw (3.3,-1.35) to (1.6,-1.35);
			\draw[out=180,in=180] (1.6,-1.35) to (1.6,-1.55);
			\draw (1.6,-1.55) -- (12);
			
			\node[red](20) [circle,draw,inner sep=0pt,minimum size=12pt] at (4.1,-1.7) {};
			\node(21) [circle,draw,inner sep=0pt,minimum size=3.5pt] at (4.7,-1.7) {};
			\node(24) [circle,draw,inner sep=0pt,minimum size=3.5pt] at (5.45,-1.85) {};
			\node(25) [circle,draw,inner sep=0pt,minimum size=3.5pt] at (5.45,-1.55) {};
			\node(26) [circle,draw,inner sep=0pt,minimum size=3.5pt] at (6.2,-1.7) {};
			\draw (21)--(24);
			\draw (21)--(25);
			\draw (26)--(24);
			\draw (26)--(25);
			
			\draw (0.5,-2.2)--(6.5,-2.2)--(6.5,-0.25)--(0.5,-0.25)--(0.5,-2.2);
		\end{tikzpicture}
	\end{subfigure}
	\hfill
	\begin{subfigure}{0.32\textwidth}
		\begin{tikzpicture}
			\node at (-1,3.5) {\textbf{c}};
			\begin{axis}[
				width=\linewidth,
				xlabel= Sequence length, 
				ylabel= Cost,
				xmin=9,
				xmax=101,
				ymin=0.45,
				ymax=0.75,
				legend columns=3, 
				legend style={at={(0.4,-0.3)},anchor=north}, 
				]
				\addplot[thick,color=red,mark=None]  table[x=x,y=i1m0round,col sep=comma] {fidelity_DriftAndSmooth.csv}; 
				\addlegendentry{\begin{tikzpicture}[scale=0.1, baseline]
						\node(1) [circle,draw,inner sep=0pt,minimum size=3.5pt] at (-1,0.15) {};
						\node(4) [circle,draw,inner sep=0pt,minimum size=3.5pt] at (0,0.15) {};
						\draw (1)--(4);
				\end{tikzpicture}}
				\addplot[thick,color=magenta,mark=None]  table[x=x,y=i1m0h2round,col sep=comma] {fidelity_DriftAndSmooth.csv}; 
				\addlegendentry{\begin{tikzpicture}[xscale=0.1,yscale=0.175, baseline]
						\node(1) [circle,draw,inner sep=0pt,minimum size=3.5pt] at (-1,0.15) {};
						\node(4) [circle,draw,inner sep=0pt,minimum size=3.5pt] at (0,0) {};
						\node(5) [circle,draw,inner sep=0pt,minimum size=3.5pt] at (0,0.3) {};
						\node(6) [circle,draw,inner sep=0pt,minimum size=3.5pt] at (1,0.15) {};
						\draw (1)--(4);
						\draw (1)--(5);
						\draw (6)--(4);
						\draw (6)--(5);
				\end{tikzpicture}}
				\addplot[thick,color=orange,mark=None]  table[x=x,y=i1m0h2h2round,col sep=comma] {fidelity_DriftAndSmooth.csv}; 
				\addlegendentry{\begin{tikzpicture}[xscale=0.1,yscale=0.175, baseline]
						\node(1) [circle,draw,inner sep=0pt,minimum size=3.5pt] at (-1,0.15) {};
						\node(2) [circle,draw,inner sep=0pt,minimum size=3.5pt] at (0,0) {};
						\node(3) [circle,draw,inner sep=0pt,minimum size=3.5pt] at (0,0.3) {};
						\node(4) [circle,draw,inner sep=0pt,minimum size=3.5pt] at (1,0) {};
						\node(5) [circle,draw,inner sep=0pt,minimum size=3.5pt] at (1,0.3) {};
						\node(6) [circle,draw,inner sep=0pt,minimum size=3.5pt] at (2,0.15) {};
						\draw (1)--(2);
						\draw (1)--(3);
						\draw (2)--(4);
						\draw (2)--(5);
						\draw (3)--(4);
						\draw (3)--(5);
						\draw (6)--(4);
						\draw (6)--(5);
				\end{tikzpicture}}
				\addplot[thick,color=blue,mark=None]  table[x=x,y=i1m1round,col sep=comma] {fidelity_DriftAndSmooth.csv}; 
				\addlegendentry{\begin{tikzpicture}[xscale=0.1,yscale=0.175, baseline]
						\node(1) [circle,draw,inner sep=0pt,minimum size=3.5pt] at (-1,0) {};
						\node(2) [circle,draw,inner sep=0pt,minimum size=3.5 pt] at (-1,0.3) {};
						\node(4) [circle,draw,inner sep=0pt,minimum size=3.5pt] at (0,0) {};
						\node(5) [circle,draw,inner sep=0pt,minimum size=3.5pt] at (0,0.3) {};
						\draw (1)--(4);
						\draw (2)--(4);
						\draw (1)--(5);
						\draw (2)--(5);
						\draw (5) -- (0.3,0.3);
						\draw[out=0,in=0] (0.3,0.3) to (0.3,0.5);
						\draw (0.3,0.5) to (-1.3,0.5);
						\draw[out=180,in=180] (-1.3,0.5) to (-1.3,0.3);
						\draw (-1.3,0.3) -- (2);
				\end{tikzpicture}}
				\addplot[ thick,color=cyan,mark=None]  table[x=x,y=i1m1h2round,col sep=comma] {fidelity_DriftAndSmooth.csv}; 
				\addlegendentry{\begin{tikzpicture}[xscale=0.1,yscale=0.175, baseline]
						\node(1) [circle,draw,inner sep=0pt,minimum size=3.5pt] at (-1,0) {};
						\node(2) [circle,draw,inner sep=0pt,minimum size=3.5pt] at (-1,0.3) {};
						\node(4) [circle,draw,inner sep=0pt,minimum size=3.5pt] at (0,0) {};
						\node(5) [circle,draw,inner sep=0pt,minimum size=3.5pt] at (0,0.3) {};
						\node(6) [circle,draw,inner sep=0pt,minimum size=3.5pt] at (1,0) {};
						\node(7) [circle,draw,inner sep=0pt,minimum size=3.5pt] at (1,0.3) {};
						\draw (1)--(4);
						\draw (2)--(4);
						\draw (1)--(5);
						\draw (2)--(5);
						\draw (6)--(4);
						\draw (7)--(4);
						\draw (6)--(5);
						\draw (7)--(5);
						\draw (7) -- (1.3,0.3);
						\draw[out=0,in=0] (1.3,0.3) to (1.3,0.5);
						\draw (1.3,0.5) to (-1.3,0.5);
						\draw[out=180,in=180] (-1.3,0.5) to (-1.3,0.3);
						\draw (-1.3,0.3) -- (2);
				\end{tikzpicture} }
				\addplot[thick,color=ForestGreen,mark=None]  table[x=x,y=i1m2round,col sep=comma] {fidelity_DriftAndSmooth.csv}; 
				\addlegendentry{\begin{tikzpicture}[xscale=0.1,yscale=0.175, baseline]
						\node(1) [circle,draw,inner sep=0pt,minimum size=3.5pt] at (-1,-0.15) {};
						\node(2) [circle,draw,inner sep=0pt,minimum size=3.5pt] at (-1,0.15) {};
						\node(3) [circle,draw,inner sep=0pt,minimum size=3.5pt] at (-1,0.45) {};
						\node(4) [circle,draw,inner sep=0pt,minimum size=3.5pt] at (0,-0.15) {};
						\node(5) [circle,draw,inner sep=0pt,minimum size=3.5pt] at (0,0.15) {};
						\node(6) [circle,draw,inner sep=0pt,minimum size=3.5pt] at (0,0.45) {};
						\draw (1)--(4);
						\draw (2)--(4);
						\draw (3)--(4);
						\draw (1)--(5);
						\draw (2)--(5);
						\draw (3)--(5);
						\draw (1)--(6);
						\draw (2)--(6);
						\draw (3)--(6);
						\draw (6) -- (0.3,0.45);
						\draw[out=0,in=0] (0.3,0.45) to (0.3,0.65);
						\draw (0.3,0.65) to (-1.3,0.65);
						\draw[out=180,in=180] (-1.3,0.65) to (-1.3,0.45);
						\draw (-1.3,0.45) -- (3);
						\draw (5) -- (0.35,0.15);
						\draw[out=0,in=0] (0.35,0.15) to (0.35,0.75);
						\draw (0.35,0.75) to (-1.35,0.75);
						\draw[out=180,in=180] (-1.35,0.75) to (-1.35,0.15);
						\draw (-1.35,0.15) -- (2);
				\end{tikzpicture}}
				\addplot [domain=9:101, samples=10, color=gray,]{0.5};
			\end{axis}
		\end{tikzpicture}
	\end{subfigure}
	\caption[]{\textbf{Comparison of QRNN and QNN.}  Panel \textbf{(a)} shows the generalisation behaviour of a
		\begin{tikzpicture}[xscale=0.4,yscale=0.6, baseline]
			\node(1) [circle,draw,inner sep=0pt,minimum size=3.5pt] at (-1,0) {};
			\node(2) [circle,draw,inner sep=0pt,minimum size=3.5 pt] at (-1,0.3) {};
			\node(4) [circle,draw,inner sep=0pt,minimum size=3.5pt] at (0,0) {};
			\node(5) [circle,draw,inner sep=0pt,minimum size=3.5pt] at (0,0.3) {};
			\draw (1)--(4);
			\draw (2)--(4);
			\draw (1)--(5);
			\draw (2)--(5);
			\draw (5) -- (0.3,0.3);
			\draw[out=0,in=0] (0.3,0.3) to (0.3,0.5);
			\draw (0.3,0.5) to (-1.3,0.5);
			\draw[out=180,in=180] (-1.3,0.5) to (-1.3,0.3);
			\draw (-1.3,0.3) -- (2);
		\end{tikzpicture} QRNN with one memory qubit (blue) and a \begin{tikzpicture}[scale=0.5, baseline]
		\node(1) [circle,draw,inner sep=0pt,minimum size=3.5pt] at (-1,0.15) {};
		\node(4) [circle,draw,inner sep=0pt,minimum size=3.5pt] at (0,0.15) {};
		\draw (1)--(4);
	\end{tikzpicture} feed-forward QNN (red) for the delay by one channel. We trained the QNNs for $1000$ rounds with a learning rate of $\eta\epsilon=0.06$ and \(M=1\) sequence with $N = 1, 2, …, 10$ training pairs and evaluated the cost function on the training set (triangles) and for a set of $30$ test pairs (squares) afterwards,  we then averaged this over $50$ rounds. Panel \textbf{(b)} shows the $y,z$-plane of the Bloch sphere. A \begin{tikzpicture}[xscale=0.4,yscale=0.6, baseline]
	\node(1) [circle,draw,inner sep=0pt,minimum size=3.5pt] at (-1,0) {};
	\node(2) [circle,draw,inner sep=0pt,minimum size=3.5pt] at (-1,0.3) {};
	\node(4) [circle,draw,inner sep=0pt,minimum size=3.5pt] at (0,0) {};
	\node(5) [circle,draw,inner sep=0pt,minimum size=3.5pt] at (0,0.3) {};
	\node(6) [circle,draw,inner sep=0pt,minimum size=3.5pt] at (1,0) {};
	\node(7) [circle,draw,inner sep=0pt,minimum size=3.5pt] at (1,0.3) {};
	\draw (1)--(4);
	\draw (2)--(4);
	\draw (1)--(5);
	\draw (2)--(5);
	\draw (6)--(4);
	\draw (7)--(4);
	\draw (6)--(5);
	\draw (7)--(5);
	\draw (7) -- (1.3,0.3);
	\draw[out=0,in=0] (1.3,0.3) to (1.3,0.5);
	\draw (1.3,0.5) to (-1.3,0.5);
	\draw[out=180,in=180] (-1.3,0.5) to (-1.3,0.3);
	\draw (-1.3,0.3) -- (2);
\end{tikzpicture} QRNN with one memory qubit (blue) and a \begin{tikzpicture}[xscale=0.4,yscale=0.6, baseline]
\node(1) [circle,draw,inner sep=0pt,minimum size=3.5pt] at (-1,0.15) {};
\node(4) [circle,draw,inner sep=0pt,minimum size=3.5pt] at (0,0) {};
\node(5) [circle,draw,inner sep=0pt,minimum size=3.5pt] at (0,0.3) {};
\node(6) [circle,draw,inner sep=0pt,minimum size=3.5pt] at (1,0.15) {};
\draw (1)--(4);
\draw (1)--(5);
\draw (6)--(4);
\draw (6)--(5);
\end{tikzpicture} feed-forward QNN (red) are trained over $1000$ training rounds with a learning rate of $\eta\epsilon=0.05$ to learn the time evolution, where the state goes from $\ket{0}$ (black point 1) over $\ket{\phi^-}$ (black star 2) to \(\ket{1}\) (black star 3), over $\ket{\phi^-}$ (black star 2) back to $\ket{0}$ (black star 3) and so on. To train the networks, we used an \(M=1\) sequence with $N=10$ training pairs and predicted the time evolution for $N_\text{test}=20$ points. The blue circles show the prediction for the QRNN and the red of the feed-forward QNN. In the latter case, the evolution approximately goes from 1 to 2 to 1 to 3 to 1 to 2 to 1 to 3 to 1, etc. {Panel \textbf{(c)} shows performance of various networks for the low- and high-frequency noise mitigation task, see Eq.~\ref{eq:filter}. The QRNNs are trained using $M=25$ sequences of $N_{\alpha } =25$ pairs each and are tested on sequences with length from $10$ to $100$ in the increments of $2$. We used $150$ training rounds with the Nadam~\cite{Nadam} optimiser; hyperparameters are learning rate---$555.(5)$, momentum---$0.9$, RMSprop---$0.999$.}}
\label{fig:resultsmainpaper}
\end{figure*}
Thus the network can be trained with several cost functions. The cost function for pure outputs varies between $0$ (worst) and $1$ (best), and the one for mixed outputs between \(4\) (worst) and \(0\) (best). We can use the cost functions in different cases. Both have advantages and drawbacks: While the global cost is more accurate when it comes to high costs (this is because \(C_\text{global, pure}\leq C_\text{local, pure}\), as shown in Appendix \ref{comparisoncosts}) the local cost function is easier to train with, and fewer qubits and gates are needed to calculate it on a quantum computer. Hence, although using the global cost will also result in learning the correlations between the individual input and output systems, using the local cost is easier to realise on present and next-generation quantum information processing devices. 

As in the classical case, we obtain the training algorithm by maximising/minimising the derivative of the cost function. The algorithms for classical simulations of the learning process can be found in Appendices \ref{localcostpurestates}, \ref{localcostmixedstates}, \ref{globalcostpurestates} and \ref{manysequences} for the different costs. Mainly the one for the local cost with pure output states (Appendix \ref{localcostpurestates}) is used to obtain numerical results. The updated unitaries in this algorithm can be calculated in a layer-by-layer manner. Hence, only unitaries on \(m_l+m_{l+1}\) and not all qubits have to be applied in one go, i.e., the size of the matrices used to simulate the learning process only scales with the width and not the depth of the network. This reduces the requirements of the simulation and allows us to classically simulate the training of modest-sized QRNNs.

The quantum training algorithm may be found in Appendix \ref{quantumalgorithm}.\\
\noindent
\textbf{Simulation of learning tasks.}
For simplicity, we focus on pure qubit outputs and training with the local cost in the following. Extensive numerical results may be found in Appendix \ref{numericalresults}. The exponential growth of Hilbert space dimension with the width of the network restricts us to QRNNs with small widths in our classical simulations. We classically simulated the learning process for different tasks while we used total input and output spaces up to \(4\) qubits. We illustrate the results in terms of three exemplary tasks.

The first learning task is focused on the delay, SWAP or shift channel, which often serves as the first example in papers on quantum channels with memory, see, e.g.~\cite{Rybar2009, Rybar2015, Kretschmann2005}. This channel is defined as follows: The \(x^\text{th}\) output equals the \((x-1)^\text{th}\) input, i.e. \(\rho_x^\text{in}=\ket{\phi_x^\text{in}} \bra{\phi_x^\text{in}}\) and  \(\ket{\phi_x^\text{out}} = \ket{\phi_{x-1}^\text{in}} \). We thus studied the ability of a QRNN to learn a training set of form \(S=\left(\ket{\phi_0^\text{m}},\left( \ket{\phi_1^\text{in}} ,\ket{\phi_0^\text{m}}\right),\left( \ket{\phi_2^\text{in}}, \ket{\phi_1^\text{in}} \right),\dots,  \left(\ket{\phi_N^\text{in}},\ket{\phi_{N-1}^\text{in}}\right) \right)\) with varying size \(N\) and to generalise to (bigger) test sets of the same form with different input states. We chose the input states at random with respect to the Haar measure. The generalisation behaviour is shown in Figure \ref{fig:resultsmainpaper} (a). The QRNN learns the SWAP channel perfectly and generalises well, even for only eight training pairs. A feed-forward QNN, by contrast, learns worse for every training pair added and does not generalise at all. This is what we expected as a feed-forward QNN can only guess the \(x^\text{th}\) output since it only has access to the \(x^\text{th}\) input and not the \((x-1)^\text{th}\).

The second task we consider is the learning of the time evolution in a time interval \(\left( T_2, T_3\right]\) from a time evolution of a state \(\ket{\psi (t) }\) in a time interval \(t\in\left[T_1, T_2\right]\), where \(T_1< T_2 < T_3\). Classically this is a prototypical example where RNNs are often more efficient than feed-forward NNs~\cite{Connor1991,vanLint2002, Medsker2001}. To carry out the quantum task, we exploited the training set \( S = ( \ket{\psi(T_1)} , ( \ket{\psi(T_1+\tau)},  \ket{\psi(T_1+2\tau)}  ), \dots , ( \ket{\psi(T_2-\tau)} , \allowbreak \ket{\psi(T_2)} ) )\) to train our network. We then use \( \ket{\psi(T_2)} \) as input of the network after \( \ket{\psi(T_2-\tau)} \). The goal is for the output state to be close to  \( \ket{\psi(T_2+\tau)} \). This state is then used as input again and so forth. Then the time predictions are compared to the correct time evolution via the fidelity. As an example we started in state \(\ket{0}\), then went over \(\ket{\phi^-}=\frac{1}{\sqrt{2}}\ket{0}-\frac{i}{\sqrt{2}}\ket{1}\) to \(\ket{1}\), back to \(\ket{0}\) over \(\ket{\phi^-}\) and so on. This is shown in Figure \ref{fig:resultsmainpaper} (b). As one can see, the training and prediction process works better with QRNNs than with feed-forward QNNs, as the blue circles for QRNNs are nearly coincident with the black stars (target states) and the red ones for feed-forward QNNs are completely different.

{Learning how to filter low- and high-frequency noise in quantum data is the third task we considered. Consider an interferometric experiment that encodes information of interest into a relative phase $\phi$ in the state $| \phi \rangle \equiv \frac{| 0 \rangle + e^{i \phi} | 1 \rangle}{\sqrt{2}}$. Suppose that one desires a cheaper, less complex, or more transportable version of said experiment. Naturally, the relative phase in the output states will become noisier. We can circumvent this by designing an appropriate post-processing device trained with the supervision of an expensive state-of-the-art machine. Moreover, we might want to use denoised states for further quantum computation. In this case, QNNs are a capable platform for learning denoisers~\cite{DiPol2020, achache2020denoising, DrD2021}. Shot-to-shot denoisers can have a feed-forward structure. However, operating in the frequency domain, e.g. to combat drifts, requires memory~\cite{DrD2021}.

We train QRNNs on the data set $S = \{ S_\alpha \} = \left\{|0\rangle, \left(|\phi^\alpha_t \rangle, |\psi^\alpha_t\rangle \right)_{t=1}^{T_\alpha} \right\}_{\alpha=1}^{N}$. We construct a simple example of $\phi^\alpha_t$ that require bandwidth filtering. Consider noise consisting of two components. The low-frequency noise can be cancelled by subtracting a drift with speed $v$. The high-frequency noise is cancelled by exponential smoothing, i.e. taking a convex combination with previous correct outcomes. That is for all \(\alpha\)
\begin{align} \label{eq:filter}
    \psi^\alpha_0 &= \phi^\alpha_0 \nonumber\\    
    \psi^\alpha_t &= w \cdot (\phi^\alpha_t - v \cdot t) + (1-w) \cdot \psi^\alpha_{t-1}, \ t>1.
\end{align}
This task can not be solved with fidelity one due to the no-cloning theorem. Nevertheless, QRNNs with sufficient memory can significantly improve the quality of the outputs, see Figure~\ref{fig:resultsmainpaper}(c).

For an in-depth discussion and more numerical results, see Appendix~\ref{numericalresults}.
}\\
\noindent\\[1ex]
\textbf{{\large Discussion}}\\[1ex]

In this paper, we have proposed a natural {quantisation} of recurrent neural networks. These are universal, i.e. they can be used to simulate general causal quantum automata. Our approach is scalable -- only a few qubits must be stored to train the network. The number of qubits needed only scales with the width and not the network's depth. One can view our QRNNs as trainable matrix-product quantum channels. The experimental requirements for implementing QRNNs are comparable with those demonstrated in producing matrix-product states and density matrices~\cite{QMPDO_experiment_ion, QMPS_2020barrat, QMPS_2021foss, QMPS_2021haghshenas, QMPS_PhaseTransition_2019smith}. Indeed, in~\cite{QMPDO_experiment_ion}, the authors have experimentally implemented a quantum channel with memory for 6 time steps on an ion-trap quantum computer. QRNNs are also suitable for some existing measurement-based photonic devices. These devices were used to demonstrate up to 100 sequential operations on a single logical qubit~\cite{MBQC_exp_100op} and generate more than 1 million mode 1d cluster states~\cite{MBQC_exp_1000000modes}. As such, while suited even for current devices, QRNNs can leverage the rapid progress of quantum hardware~\cite{IBM_roadmap_2022}.

We demonstrated that QRNNs can outperform feed-forward QNNs with respect to several exemplary learning tasks. For this, we have designed some benchmarking tasks and classically simulated the training of QRNNs. We started with simple yet paradigmatic examples, such as learning a delay channel or cancelling a linear drift. We then transitioned to relevant tasks for applications, such as learning time-dependent dynamics and bandwidth filters.

Many intriguing questions remain in the study of QRNNs: We have concentrated on the case of product in- and outputs; are there interesting tasks involving data entangled at different points in time (or space), and how well can QRNNs cope with such data? Theoretical results show the benefit of quantum memory for learning algorithms~\cite{QExp_Memory_Advantage_Proof}, so how large is this set of tasks? When is quantum memory necessary, and when can it be substituted by a classical memory~\cite{CQRNN_Qadvantage}? Are QRNNs particularly suited for matrix-product states and density matrices, and is there a deeper connection to CC ML tensor network methods~\cite{TN_ML, TN_QML}? If so, can QRNNs be used for state preparation and to extend the notion of the parent Hamiltonian~\cite{MPSTheo,ParentHThesis,MPSUncle,ParentH2,ParentPerturb} and Lindbladian~\cite{parentL,MPDOparentT,Dissipation_parentH}? {Can one use QRNNs for the discovery of catalysts?} Does recurrency influence robustness to noise of QNNs? 
\\

\noindent

\begin{acknowledgments}
Helpful correspondence and discussions with Tobias J. Osborne, Kerstin Beer, Polina Feldmann, Thomas Cope, Gabriel Müller, Taras Kucherenko, Terry Farrelly, Yihui Quek, Ramona Wolf, Luis Mantilla Calderon, Robert Rau{\ss}endorf are gratefully acknowledged. This work was funded by the Deutsche Forschungsgemeinschaft (German Research Foundation) under Germany's Excellence Strategy (EXC-2123, QuantumFrontiers Grant No. 390837967), through CRC 1227 (DQ-mat) within project No. A06, by the Quantum Valley Lower Saxony
(QVLS) through the Volkswagen Foundation, Canada First Research Excellence Fund, Quantum Materials and Future Technologies Program, and 
the Cambridge Commonwealth, European and International Trust.

%\textbf{Author contributions:} This project was conceived of, and initiated in, discussions of ... and .... The QRNN architecture was formulated by ... and ....  Operational considerations were investigated by .... Classical numerical implementations and investigations were developed by ... and .... The quantum implementation was developed by .... All authors contributed to writing the manuscript. ...
%
%\textbf{Competing interests:} Authors declare no competing interests.
\textbf{Data availability:} We obtained all results using Python and Matlab. The code is available at \url{https://github.com/qigitphannover/DeepQuantumNeuralNetworks/tree/master/QRNN}.

\end{acknowledgments}

\onecolumngrid
\newpage
\appendix
\section{Comparison of the different Costs}
\label{comparisoncosts}
Generally, the global cost will be smaller or equal to the average cost, which is the subject of the following lemma:
\begin{lem}
	For all global input states \(\rho^\text{IN}\) with corresponding output \(\rho^\text{OUT}\) being a pure product state, i.e. \(\rho^{OUT}=  \kb{\Phi^\text{OUT}}=\bigotimes_{x=1}^N \kb{\phi_x^\text{out}}\in\h^{\text{out}^{\otimes n}}\), and network unitaries \(\U\) we have
	\begin{equation}
		\label{eq:costcomparison}
		C_\text{global,pure}\leq C_\text{local,pure} .
	\end{equation}
\end{lem}
\begin{proof}
	The fidelity of two quantum states, \(F(\rho,\sigma)=\left(\mathrm{tr}\sqrt{\rho^{1/2}\sigma \rho^{1/2}}\right)^2\), satisfies the data-processing inequality (see e.g.~\cite[Lemma 9.2.1, Exercise 9.2.8]{Wilde2019}), i.e. for all quantum channels \(\E\) we have
	\begin{equation}
		\label{eq:datafidelity}
		F\left( \E(\rho),\E(\sigma)\right)\geq F\left(\rho,\sigma\right).
	\end{equation}
	Equation \eqref{eq:costcomparison} then simply follows by using \eqref{eq:datafidelity} for the quantum channel being the partial trace:
	\begin{align*}
		C_\text{local,pure}&=\frac{1}{N} \sum_{x=1}^{N}\bra{\phi_x^\text{out}} \rho_x^\text{out} \ket{\phi_x^\text{out}}= \frac{1}{N} \sum_{x=1}^{N} F\left(  \kb{\phi_x^\text{out}}, \tilde{\mathrm{tr}}_{\text{out}}^x  \left(\rho^\text{OUT}\right) \right)\\
		&=\frac{1}{N} \sum_{x=1}^{N} F\left( \tilde{\mathrm{tr}}_{\text{out}}^x  \left(\kb{\Phi^\text{Out}}\right), \tilde{\mathrm{tr}}_{\text{out}}^x   \left(\mathrm{tr}_\text{m}^N \left(\rho^\text{OUT}\right)\right) \right) \\
		&\geq \frac{1}{N} \sum_{x=1}^{N} F\left( \kb{\Phi^\text{Out}}, \mathrm{tr}_\text{m}^N \left(\rho^\text{OUT}\right)\right) = F\left( \kb{\Phi^\text{Out}}, \mathrm{tr}_\text{m}^N \left(\rho^\text{OUT}\right)\right) =C_\text{global,pure}.
	\end{align*}
\end{proof}

\section{Optimising the local Cost with pure output states classically}
\label{localcostpurestates}
We now want to derive the algorithm to maximise the local cost with pure outputs. In a first step, we only consider \(M=1\) run of the QRNN in Appendices \ref{localcostpurestates} to \ref{globalcostpurestates}, and only consider many runs with \(M > 1\) in Appendix \ref{manysequences}. Hence, we can ignore \(\alpha\) for now. It does not matter for the classical simulation if we trace out everything but the \(x^\text{th}\) output at the end or if we calculate the reduced density operators \(\rho_x^\text{out}\) and \(\rho_x^\text{m}\) after each iteration step and use \(\rho_x^\text{m}\otimes \rho_{x+1}^\text{in}\) as total input in the next iteration step. The state \(\rho_x^\text{m}\) is then given by
\begin{align*}
	\rho^\text{m}_x=&\mathrm{tr}_{{\text{in,}1:L\text{,out}}}^x(\mathrm{tr}_{\text{m}}^{x-1}(\mathcal{U}_x(\rho_x^{\text{in}}\otimes \rho_{x-1}^{\text{m}} \otimes \ket{0\dots 0}_{ {1:L+1}}^x\bra{0\dots 0})\mathcal{U}_x^\dagger))\\
	=&\mathrm{tr}_{{\text{in,}1:L\text{,out}}}^x(\mathrm{tr}_{\text{m}}^{x-1}(\mathcal{U}_x(\rho_x^{\text{in}}\otimes \mathrm{tr}_{{\text{in,}1:L\text{,out}}}^{x-1}(\mathrm{tr}_{\text{m}}^{x-2}(\mathcal{U}_{x-1}(\rho_{x-1}^{\text{in}}\otimes \rho_{x-2}^{\text{m}} \otimes \ket{0\dots 0}_{ {1:L+1}}^{x-1}\bra{0\dots 0})\mathcal{U}_{x-1}^\dagger))\\ 
	&\otimes \ket{0\dots 0}_{ {1:L+1}}^x\bra{0\dots 0})\mathcal{U}_x^\dagger)).
\end{align*}
For \(A\in \B(\h_A)\) and \(B\in \B(\h_A\otimes \h_B)\) for some Hilbert spaces \(\h_A\) and \(\h_B\) it is
\begin{align*}
	\mathrm{tr}_B((A\otimes \id_B)B)&%=\sum_k (\id_A\otimes\bra{k})(A\otimes \id_B)B(\id_A\otimes\ket{k}) =\sum_k A(\id_A\otimes\bra{k})B(\id_A\otimes\ket{k})\\
	%&= A\sum_k(\id_A\otimes\bra{k})B(\id_A\otimes\ket{k})
	=A\mathrm{tr}_B(B).
\end{align*}
Hence, as \(\U_{x}=\U_{x}\otimes \id^\text{m}_{x-2}\otimes \id^\text{in, hid, out}_{x-1}\) it follows that
\begin{align*}
	\rho^\text{m}_x=&\mathrm{tr}_{{\text{in,}1:L\text{,out}}}^{x,x-1}(\mathrm{tr}_{\text{m}}^{x-1,x-2}(\mathcal{U}_x(\rho_x^{\text{in}}\otimes (\mathcal{U}_{x-1}(\rho_{x-1}^{\text{in}}\otimes \rho_{x-2}^{\text{m}} \otimes \ket{0\dots 0}_{ {1:L+1}}^{x-1}\bra{0\dots 0})\mathcal{U}_{x-1}^\dagger))\\ 
	&\otimes \ket{0\dots 0}_{ {1:L+1}}^x\bra{0\dots 0})\mathcal{U}_x^\dagger)).
\end{align*}
Also, \(\U_{x-1}=\U_{x-1}\otimes  \id^\text{in, hid, out,m}_{x}\) which leads to
\begin{align}
	\rho^\text{m}_x=&\mathrm{tr}_{{\text{in,}1:L\text{,out}}}^{x,x-1}(\mathrm{tr}_{\text{m}}^{x-1,x-2}(\mathcal{U}_x\mathcal{U}_{x-1}(\rho_x^{\text{in}}\otimes \rho_{x-1}^{\text{in}}\otimes \rho_{x-2}^{\text{m}} \otimes \ket{0\dots 0}_{ {1:L+1}}^{x-1,x}\bra{0\dots 0})\mathcal{U}_{x-1}^\dagger\mathcal{U}_x^\dagger))\nonumber\\
	=& \dots\nonumber\\
	=&\mathrm{tr}_{{\text{in,}1:L\text{,out}}}^{x:1}(\mathrm{tr}_{\text{m}}^{x-1:0}(\mathcal{U}_x\dots\mathcal{U}_{1}(\bigotimes_{y=x}^1\rho_y^{\text{in}}\otimes \rho_{0}^{\text{m}} \otimes \ket{0\dots 0}_{ {1:L+1}}^{1:x}\bra{0\dots 0})\mathcal{U}_{1}^\dagger\dots\mathcal{U}_x^\dagger)), \label{eq:d}
\end{align}
where \(i:j\) is a short notation for \(i,\dots,j\).\\
The loss function is given by 
\begin{align}
	C=&\frac{1}{N}\sum_{x=1}^N\bra{\phi_x^\text{out}}\rho_x^\text{out}\ket{\phi_x^\text{out}} \nonumber\\
	=&\frac{1}{N}\sum_{x=1}^N\bra{\phi_x^\text{out}}\mathrm{tr}_{\text{in,}1:L\text{,m}}^x(\mathrm{tr}_{\text{m}}^{x-1}(\mathcal{U}_x(\rho_x^{\text{in}}\otimes \rho_{x-1}^{\text{m}} \otimes \ket{0\dots 0}_{ {1:L+1}}^x\bra{0\dots 0})\mathcal{U}_x^\dagger))\ket{\phi_x^\text{out}} \nonumber\\	=&\frac{1}{N}\sum_{x=1}^N \mathrm{tr}\left((\id_x^{\text{in,}1:L\text{,m}}\otimes \id_{x-1}^\text{m}\otimes\ket{\phi_x^\text{out}}\bra{\phi_x^\text{out}})\mathcal{U}_x(\rho_x^{\text{in}}\otimes \rho_{x-1}^{\text{m}} \otimes \ket{0\dots 0}_{ {1:L+1}}^x\bra{0\dots 0})\mathcal{U}_x^\dagger\right) \label{eq:c} 
\end{align}
and we write \(C=C_\text{local,pure}\) here for simplicity. In order to train the QRNN, we have to find the perceptron unitaries \(U_j^l\) that maximize the cost function. To do so, we use the same ansatz as in \cite{beerTrainingDeepQuantum2020} and the unitaries are updated in each step according to
\[U_j^l\mapsto U_j^{l \prime}= e^{i\epsilon K_j^l}U_j^l\]
where the matrices \(K_j^l\) are hermitian and depend on the matrices \(U_j^l\). We now have to determine the matrices \(K_j^l\) that make the biggest positive change in the cost function. To do so, we have to maximise the derivative of the cost function at \(U_j^l\) when changing the matrices as presented above
\[\delta C= \lim_{\epsilon\to 0}\frac{\Delta C}{\epsilon}= \lim_{\epsilon\to 0}\frac{C^\prime- C}{\epsilon}. \]
In the expression for the cost function in equation \(\eqref{eq:c}\) the unitaries \(U_j^l\), that are the object of change, appear directly in the equation and indirectly in the state \(\rho_{x-1}^{\text{m}}\) that also depends on the unitaries (for \(x>1\)). So we first look at the change in \(\rho_{x-1}^{\text{m}}\) after a step of size \(\epsilon\). As described in equation \(\eqref{eq:d}\) it is
\begin{align*}
	\rho^{\text{m}}_x=&\mathrm{tr}_{{\text{in,}1:L\text{,out}}}^{x:1}(\mathrm{tr}_{\text{m}}^{x-1:0}(U_{m_{L+1}\, x}^{L+1}\dots U_{1\,x}^1\dots U_{j\, z}^l\dots U_{m_{L+1}\, 1}^{L+1} \dots U^1_{1\,1}\\
	&(\bigotimes_{y=x}^1\rho_y^{\text{in}}\otimes \rho_{0}^{\text{m}} \otimes \ket{0\dots 0}_{ 1:L+1}^{1:x}\bra{0\dots 0})U_{1\,1}^{1^\dagger} \dots U_{m_{L+1}\, x}^{L+1^\dagger})),
\end{align*}
hence
\begin{align*}
	\rho^{\text{m}\prime}_x=&\mathrm{tr}_{{\text{in,}1:L\text{,out}}}^{x:1}(\mathrm{tr}_{\text{m}}^{x-1:0}(e^{i\epsilon K_{m_{L+1}}^{L+1}}U_{m_{L+1}\, x}^{L+1}\dots e^{i\epsilon K_{1}^{1}} U^1_{1\,1}(\bigotimes_{y=x}^1\rho_y^{\text{in}}\otimes \rho_{0}^{\text{m}} \\
	&\otimes \ket{0\dots 0}_{ 1:L+1}^{1:x}\bra{0\dots 0}) U_{1\,1}^{1^\dagger}e^{-i\epsilon K_{1}^{1}} \dots U_{m_{L+1}\, x}^{L+1^\dagger}e^{-i\epsilon K_{m_{L+1}}^{L+1}})).
\end{align*}
Using the Taylor expansion of the exponential function it follows that
\begin{align*}
	\rho^{\text{m}\prime}_x=&\mathrm{tr}_{{\text{in,}1:L\text{,out}}}^{x:1}(\mathrm{tr}_{\text{m}}^{x-1:0}(U_{m_{L+1}\, x}^{L+1}\dots  U^1_{1\,1}(\bigotimes_{y=x}^1\rho_y^{\text{in}}\otimes \rho_{0}^{\text{m}} \otimes \ket{0\dots 0}_{ 1:L+1}^{1:x}\bra{0\dots 0}) U_{1\,1}^{1^\dagger} \dots U_{m_{L+1}\, x}^{L+1^\dagger}))\\
	&+i\epsilon \sum_{z=1}^{x}\sum_{l=1}^{L+1}\sum_{j=1}^{m_l} \mathrm{tr}_{1:L+1}^{x:1}\left(\mathrm{tr}_{\text{m}}^{x-1:0} \left( \mathcal{U}_x\dots \mathcal{U}_{z+1} U_{m_{L+1}\, z}^{L+1} \dots U_{j+1\, z}^{l} K_{j\, z}^l U_{j\, z}^{l} \dots U_{1\,z}^1\mathcal{U}_{z-1}\dots \mathcal{U}_1 \right.\right.\\
	&\cdot  (\bigotimes_{y=x}^1\rho_y^{\text{in}}\otimes \rho_{0}^{\text{m}} \otimes \ket{0\dots 0}_{ 1:L+1}^{1:x}\bra{0\dots 0}) \U_1^\dagger \dots \U_{z-1}^\dagger U_{1\, z}^{1^\dagger }\dots U_{j\, z}^{l^\dagger }U_{j+1\, z}^{l^\dagger }\dots U_{m_{L+1}\, z}^{L+1^\dagger } \U_{z+1}^\dagger\dots \U_{x}^\dagger\\
	&-\mathcal{U}_x\dots \mathcal{U}_{z+1} U_{m_{L+1}\, z}^{L+1} \dots U_{j+1\, z}^{l}  U_{j\, z}^{l} \dots U_{1\,z}^1\mathcal{U}_{z-1}\dots \mathcal{U}_1 (\bigotimes_{y=x}^1\rho_y^{\text{in}}\otimes \rho_{0}^{\text{m}} \otimes \ket{0\dots 0}_{ {1:L+1}}^{1:x}\bra{0\dots 0}) \\
	&\cdot  \left.\left.\U_1^\dagger \dots \U_{z-1}^\dagger U_{1\, z}^{1^\dagger }\dots U_{j\, z}^{l^\dagger }K_{j\, z}^l U_{j+1\, z}^{l^\dagger }\dots U_{m_{L+1}\, z}^{L+1^\dagger } \U_{z+1}^\dagger\dots \U_{x}^\dagger \right)\right)+\mathcal{O}(\epsilon^2).
\end{align*}
Hence,
\begin{align}
	\Delta \rho_x^m=&\rho^{\text{m}\prime}_x-\rho^\text{m}_x\nonumber \\
	=&i\epsilon \sum_{z=1}^{x}\sum_{l=1}^{L+1}\sum_{j=1}^{m_l} \mathrm{tr}_{{\text{in,}1:L\text{,out}}}^{x:1}\left(\mathrm{tr}_{\text{m}}^{x-1:0} \left( \mathcal{U}_x\dots \mathcal{U}_{z+1} U_{m_{L+1}\, z}^{L+1} \dots U_{j+1\, z}^{l}  \right.\right.\nonumber\\
	&\cdot \left[ K_{j\, z}^l, U_{j\, z}^{l} \dots U_{1\,z}^1\mathcal{U}_{z-1}\dots \mathcal{U}_1 (\bigotimes_{y=x}^1\rho_y^{\text{in}}\otimes \rho_{0}^{\text{m}} \otimes \ket{0\dots 0}_{ {1:L+1}}^{1:x}\bra{0\dots 0}) \U_1^\dagger \dots \U_{z-1}^\dagger U_{1\, z}^{1^\dagger }\dots U_{j\, z}^{l^\dagger }\right]\nonumber\\
	&\cdot\left.\left.U_{j+1\, z}^{l^\dagger }\dots U_{m_{L+1}\, z}^{L+1^\dagger } \U_{z+1}^\dagger\dots \U_{x}^\dagger\right)\right)+\mathcal{O}(\epsilon^2). \nonumber
\end{align}
As \(\U_x...U^1_{1\, z}\) acts as the identity on  \(\h_{z-1:1}^{\text{in,}1:L\text{,out}}\) and \(\h_{z-2:0}^\text{m}\), we can write
\begin{align}
	\Delta \rho_x^m=&i\epsilon \sum_{z=1}^{x}\sum_{l=1}^{L+1}\sum_{j=1}^{m_l} \mathrm{tr}_{{\text{in,}1:L\text{,out}}}^{x:z}\Big(\mathrm{tr}_{\text{m}}^{x-1:z-1} \Big( \mathcal{U}_x\dots \mathcal{U}_{z+1} U_{m_{L+1}\, z}^{L+1} \dots U_{j+1\, z}^{l}\Big[ K_{j\, z}^l, U_{j\, z}^{l} \dots U_{1\,z}^1  \nonumber\\
	&\cdot \mathrm{tr}_{{\text{in,}1:L\text{,out}}}^{z-1:1}\left(\mathrm{tr}_{\text{m}}^{z-2:0} \left( \mathcal{U}_{z-1}\dots \mathcal{U}_1 \left(\bigotimes_{y=x}^1\rho_y^{\text{in}}\otimes \rho_{0}^{\text{m}} \otimes \ket{0\dots 0}_{ {1:L+1}}^{1:x}\bra{0\dots 0}\right) \U_1^\dagger \dots \U_{z-1}^\dagger\right)\right)\nonumber\\
	&\cdot U_{1\, z}^{1^\dagger }\dots U_{j\, z}^{l^\dagger }\Big]U_{j+1\, z}^{l^\dagger }\dots U_{m_{L+1}\, z}^{L+1^\dagger } \U_{z+1}^\dagger\dots \U_{x}^\dagger\Big)\Big)+\mathcal{O}(\epsilon^2). \nonumber\\
	=&i\epsilon \sum_{z=1}^{x}\sum_{l=1}^{L+1}\sum_{j=1}^{m_l} \mathrm{tr}_{{\text{in,}1:L\text{,out}}}^{x:z}\Bigg(\mathrm{tr}_{\text{m}}^{x-1:z-1} \Bigg( \mathcal{U}_x\dots \mathcal{U}_{z+1} U_{m_{L+1}\, z}^{L+1} \dots U_{j+1\, z}^{l}\Bigg[ K_{j\, z}^l, U_{j\, z}^{l} \dots U_{1\,z}^1   \left( \bigotimes_{y=x}^z\rho_y^{\text{in}} \right. \nonumber\\
	&\otimes \mathrm{tr}_{{\text{in,}1:L\text{,out}}}^{z-1:1}  \left(\mathrm{tr}_{\text{m}}^{z-2:0} \left( \mathcal{U}_{z-1}\dots \mathcal{U}_1 \left(\bigotimes_{y=z-1}^1\rho_y^{\text{in}}\otimes \rho_{0}^{\text{m}} \otimes \ket{0\dots 0}_{ {1:L+1}}^{1:z-1}\bra{0\dots 0}\right) \U_1^\dagger \dots \U_{z-1}^\dagger\right)\right)  \nonumber\\
	&\otimes \ket{0\dots 0}_{ {1:L+1}}^{z:x}\bra{0\dots 0}\Bigg) U_{1\, z}^{1^\dagger }\dots U_{j\, z}^{l^\dagger }\Bigg]U_{j+1\, z}^{l^\dagger }\dots U_{m_{L+1}\, z}^{L+1^\dagger } \U_{z+1}^\dagger\dots \U_{x}^\dagger\Bigg)\Bigg)+\mathcal{O}(\epsilon^2). \nonumber
\end{align}
Actually, it is
\[\rho_{z-1}^\text{m}= \mathrm{tr}_{{\text{in,}1:L\text{,out}}}^{z-1:1}  \left(\mathrm{tr}_{\text{m}}^{z-2:0} \left( \mathcal{U}_{z-1}\dots \mathcal{U}_1 \left(\bigotimes_{y=z-1}^1\rho_y^{\text{in}}\otimes \rho_{0}^{\text{m}} \otimes \ket{0\dots 0}_{ {1:L+1}}^{1:z-1}\bra{0\dots 0}\right) \U_1^\dagger \dots \U_{z-1}^\dagger\right)\right), \]
hence
\begin{align}
	\Delta \rho_x^m=&i\epsilon \sum_{z=1}^{x}\sum_{l=1}^{L+1}\sum_{j=1}^{m_l} \mathrm{tr}_{{\text{in,}1:L\text{,out}}}^{x:z}\Bigg(\mathrm{tr}_{\text{m}}^{x-1:z-1} \Bigg( \mathcal{U}_x\dots \mathcal{U}_{z+1} U_{m_{L+1}\, z}^{L+1} \dots U_{j+1\, z}^{l}\Bigg[ K_{j\, z}^l, U_{j\, z}^{l} \dots U_{1\,z}^1   \left( \bigotimes_{y=x}^z\rho_y^{\text{in}} \right. \nonumber\\
	&\otimes \rho_{z-1}^\text{m}  \otimes \ket{0\dots 0}_{ {1:L+1}}^{z:x}\bra{0\dots 0}\Bigg) U_{1\, z}^{1^\dagger }\dots U_{j\, z}^{l^\dagger }\Bigg]U_{j+1\, z}^{l^\dagger }\dots U_{m_{L+1}\, z}^{L+1^\dagger } \U_{z+1}^\dagger\dots \U_{x}^\dagger\Bigg)\Bigg)+\mathcal{O}(\epsilon^2). \label{eq:e}
\end{align}
By keeping in mind, that the the unitaries \(U_j^l\) appear directly in the cost function (equation \(\eqref{eq:c}\)) and indirectly in the state \(\rho_{x-1}^{\text{m}}\), we analogously get 
\begin{align*}
	\Delta C=& C^\prime-C\\
	=& \frac{1}{N}\sum_{x=1}^{N}\Bigg[ i\epsilon\sum_{l=1}^{L+1}\sum_{j=1}^{m_l} \mathrm{tr} \left( \left( \id_x^{{\text{in,}1:L\text{,m}}} \otimes \id_{x-1}^\text{m} \otimes \ket{\phi_x^\text{out}}\bra{\phi_x^\text{out}}\right)U_{m_{L+1}\,x}^{L+1}\dots U_{j+1\,x}^{l}\left[ K_{j\,x}^l,\right.\right.\\
	&\left. \left. U_{j\,x}^l\dots U_{1\,x}^1 \left( \rho_x^\text{in}\otimes \rho_{x-1}^\text{m} \otimes \ket{0\dots 0}_\text{1:L+1}^x \bra{0\dots 0}\right)U_{1\, x}^{1^\dagger}\dots U_{j\, x}^{l^\dagger}\right] U_{j+1\, x}^{l^\dagger} \dots U_{m_{L+1}\, x}^{L+1^\dagger} \right)\\
	&+\mathrm{tr}\left((\id_x^{\text{in,}1:L\text{,m}}\otimes \id_{x-1}^\text{m}\otimes \ket{\phi_x^\text{out}}\bra{\phi_x^\text{out}})\mathcal{U}_x(\rho_x^{\text{in}}\otimes \Delta \rho_{x-1}^{\text{m}} \otimes \ket{0\dots 0}_{ 1:L+1}^x\bra{0\dots 0})\mathcal{U}_x^\dagger\right) \Bigg]+\mathcal{O}(\epsilon^2).
\end{align*}
With equation \(\eqref{eq:e}\) it follows that
\begin{align*}
	\Delta C=& \frac{i\epsilon}{N}\sum_{x=1}^{N}\Bigg[ \sum_{l=1}^{L+1}\sum_{j=1}^{m_l} \mathrm{tr} \left( \left( \id_x^{{\text{in,}1:L\text{,m}}} \otimes \id_{x-1}^\text{m} \otimes \ket{\phi_x^\text{out}}\bra{\phi_x^\text{out}}\right)U_{m_{L+1}\,x}^{L+1}\dots U_{j+1\,x}^{l}\left[ K_{j\,x}^l,\right.\right.\\
	&\left. \left. U_{j\,x}^l\dots U_{1\,x}^1 \left( \rho_x^\text{in}\otimes \rho_{x-1}^\text{m} \otimes \ket{0\dots 0}_\text{1:L+1}^x \bra{0\dots 0}\right)U_{1\, x}^{1^\dagger}\dots U_{j\, x}^{l^\dagger}\right] U_{j+1\, x}^{l^\dagger} \dots U_{m_{L+1}\, x}^{L+1^\dagger} \right)\\
	&+\mathrm{tr}\left((\id_x^{\text{in,}1:L\text{,m}}\otimes \id_{x-1}^\text{m}\otimes \ket{\phi_x^\text{out}}\bra{\phi_x^\text{out}})\mathcal{U}_x\Bigg(\rho_x^{\text{in}}\otimes \Bigg(\sum_{z=1}^{x-1}\sum_{l=1}^{L+1}\sum_{j=1}^{m_l} \mathrm{tr}_{{\text{in,}1:L\text{,out}}}^{x-1:z}\Bigg(\mathrm{tr}_{\text{m}}^{x-2:z-1} \Bigg( \mathcal{U}_{x-1}\dots \mathcal{U}_{z+1}\right. \nonumber\\
	&\cdot U_{m_{L+1}\, z}^{L+1} \dots U_{j+1\, z}^{l}\Bigg[ K_{j\, z}^l, U_{j\, z}^{l} \dots U_{1\,z}^1   \left( \bigotimes_{y=x-1}^z\rho_y^{\text{in}} \otimes \rho_{z-1}^\text{m}  \otimes \ket{0\dots 0}_{ 1:L+1}^{z:x-1}\bra{0\dots 0}\right) U_{1\, z}^{1^\dagger }\dots U_{j\, z}^{l^\dagger }\Bigg]\\
	&\cdot U_{j+1\, z}^{l^\dagger }\dots U_{m_{L+1}\, z}^{L+1^\dagger } \U_{z+1}^\dagger\dots \U_{x-1}^\dagger\Bigg)\Bigg)\Bigg) \otimes \ket{0\dots 0}_{ 1:L+1}^x\bra{0\dots 0}\Bigg)\mathcal{U}_x^\dagger \Bigg) \Bigg]+\mathcal{O}(\epsilon^2).
\end{align*}
Hence, we get up to first order
\begin{align*}
	\delta C=& \lim_{\epsilon\to 0}\frac{\Delta C}{\epsilon}\\
	=& \frac{i}{N}\sum_{x=1}^{N}\Bigg[ \sum_{l=1}^{L+1}\sum_{j=1}^{m_l} \mathrm{tr} \left( \left( \id_x^{{\text{in,}1:L\text{,m}}} \otimes \id_{x-1}^\text{m} \otimes \ket{\phi_x^\text{out}}\bra{\phi_x^\text{out}}\right)U_{m_{L+1}\,x}^{L+1}\dots U_{j+1\,x}^{l}\left[ K_{j\,x}^l,\right.\right.\\
	&\left. \left. U_{j\,x}^l\dots U_{1\,x}^1 \left( \rho_x^\text{in}\otimes \rho_{x-1}^\text{m} \otimes \ket{0\dots 0}_\text{1:L+1}^x \bra{0\dots 0}\right)U_{1\, x}^{1^\dagger}\dots U_{j\, x}^{l^\dagger}\right] U_{j+1\, x}^{l^\dagger} \dots U_{m_{L+1}\, x}^{L+1^\dagger} \right)\\
	&+\sum_{z=1}^{x-1}\sum_{l=1}^{L+1}\sum_{j=1}^{m_l}\mathrm{tr}\left((\id_{x:z}^{\text{in,}1:L\text{,m}}\otimes \id_{z-1}^\text{m}\otimes \id_{x-1:z}^\text{out}\otimes\ket{\phi_x^\text{out}}\bra{\phi_x^\text{out}})\mathcal{U}_x\Bigg(\rho_x^{\text{in}}\otimes \Bigg( \mathcal{U}_{x-1}\dots \mathcal{U}_{z+1}\right. \nonumber\\
	&\cdot U_{m_{L+1}\, z}^{L+1} \dots U_{j+1\, z}^{l}\Bigg[ K_{j\, z}^l, U_{j\, z}^{l} \dots U_{1\,z}^1   \left( \bigotimes_{y=x-1}^z\rho_y^{\text{in}} \otimes \rho_{z-1}^\text{m}  \otimes \ket{0\dots 0}_{ 1:L+1}^{z:x-1}\bra{0\dots 0}\right) U_{1\, z}^{1^\dagger }\dots U_{j\, z}^{l^\dagger }\Bigg]\\
	&\cdot U_{j+1\, z}^{l^\dagger }\dots U_{m_{L+1}\, z}^{L+1^\dagger } \U_{z+1}^\dagger\dots \U_{x-1}^\dagger\Bigg)\Bigg)\Bigg) \otimes \ket{0\dots 0}_{ 1:L+1}^x\bra{0\dots 0}\Bigg)\mathcal{U}_x^\dagger \Bigg) \Bigg]\\
	=& \frac{i}{N}\sum_{x=1}^{N}\sum_{l=1}^{L+1}\sum_{j=1}^{m_l}\Bigg[  \mathrm{tr} \left( \left( \id_x^{{\text{in,}1:L\text{,m}}} \otimes \id_{x-1}^\text{m} \otimes \ket{\phi_x^\text{out}}\bra{\phi_x^\text{out}}\right)U_{m_{L+1}\,x}^{L+1}\dots U_{j+1\,x}^{l}\left[ K_{j\,x}^l,\right.\right.\\
	&\left. \left. U_{j\,x}^l\dots U_{1\,x}^1 \left( \rho_x^\text{in}\otimes \rho_{x-1}^\text{m} \otimes \ket{0\dots 0}_\text{1:L+1}^x \bra{0\dots 0}\right)U_{1\, x}^{1^\dagger}\dots U_{j\, x}^{l^\dagger}\right] U_{j+1\, x}^{l^\dagger} \dots U_{m_{L+1}\, x}^{L+1^\dagger} \right)\\
	&+\sum_{z=1}^{x-1}\mathrm{tr}\Bigg((\id_{x:z}^{\text{in,}1:L\text{,m}}\otimes \id_{z-1}^\text{m}\otimes \id_{x-1:z}^\text{out}\otimes\ket{\phi_x^\text{out}}\bra{\phi_x^\text{out}})\mathcal{U}_x\dots \mathcal{U}_{z+1}U_{m_{L+1}\, z}^{L+1}\dots U_{j+1\, z}^{l} \nonumber\\
	&\cdot  \Bigg[ K_{j\, z}^l, U_{j\, z}^{l} \dots U_{1\,z}^1   \left( \bigotimes_{y=x}^z\rho_y^{\text{in}} \otimes \rho_{z-1}^\text{m}  \otimes \ket{0\dots 0}_{ 1:L+1}^{z:x}\bra{0\dots 0}\right) U_{1\, z}^{1^\dagger }\dots U_{j\, z}^{l^\dagger }\Bigg]\\
	&\cdot U_{j+1\, z}^{l^\dagger }\dots U_{m_{L+1}\, z}^{L+1^\dagger } \U_{z+1}^\dagger\dots \mathcal{U}_x^\dagger \Bigg) \Bigg].
\end{align*}
By using the cyclic rule of trace we get
\begin{align*}
	\delta C
	=& \frac{i}{N}\sum_{x=1}^{N}\sum_{l=1}^{L+1}\sum_{j=1}^{m_l}\Bigg[  \mathrm{tr} \left( U_{x}^{l+1^\dagger}\dots U_{x}^{L+1^\dagger} \left( \id_x^{{\text{in,}1:L\text{,m}}} \otimes \id_{x-1}^\text{m} \otimes \ket{\phi_x^\text{out}}\bra{\phi_x^\text{out}}\right) U_{x}^{L+1}\dots U_{x}^{l+1}   \right.\\
	&\left. U_{m_{l}\,x}^{l}\dots U_{j+1\,x}^{l} \left[ K_{j\,x}^l, U_{j\,x}^l\dots U_{1\,x}^1 \left( \rho_x^\text{in}\otimes \rho_{x-1}^\text{m} \otimes \ket{0\dots 0}_\text{1:L+1}^x \bra{0\dots 0}\right)U_{1\, x}^{1^\dagger}\dots U_{j\, x}^{l^\dagger}\right] U_{j+1\, x}^{l^\dagger} \dots U_{m_{l}\, x}^{l^\dagger} \right)\\
	&+\sum_{z=1}^{x-1}\mathrm{tr}\Bigg(U_{z}^{l+1^\dagger}\dots U_{z}^{L+1^\dagger}\U_{z+1}^\dagger\dots \mathcal{U}_x^\dagger(\id_{x:z}^{\text{in,}1:L\text{,m}}\otimes \id_{z-1}^\text{m}\otimes \id_{x-1:z}^\text{out}\otimes\ket{\phi_x^\text{out}}\bra{\phi_x^\text{out}})\mathcal{U}_x\dots \mathcal{U}_{z+1} U_{z}^{L+1}\dots U_{z}^{l+1} \nonumber\\
	&\cdot U_{m_{l}\, z}^{l}\dots U_{j+1\, z}^{l} \Bigg[ K_{j\, z}^l, U_{j\, z}^{l} \dots U_{1\,z}^1   \left( \bigotimes_{y=x}^z\rho_y^{\text{in}} \otimes \rho_{z-1}^\text{m}  \otimes \ket{0\dots 0}_{ 1:L+1}^{z:x}\bra{0\dots 0}\right) U_{1\, z}^{1^\dagger }\dots U_{j\, z}^{l^\dagger }\Bigg]U_{j+1\, z}^{l^\dagger }\dots U_{m_{l}\, z}^{l^\dagger }  \Bigg) \Bigg].
\end{align*}
With
\begin{align}
	A=&U_{x}^{l+1^\dagger}\dots U_{x}^{L+1^\dagger} \left( \id_x^{l:L} \otimes \id_{x}^{\text{m}} \otimes \ket{\phi_x^\text{out}}\bra{\phi_x^\text{out}}\right) U_{x}^{L+1}\dots U_{x}^{l+1}\label{eq:f} \\
	B_x=&U_{m_{l}\,x}^{l}\dots U_{j+1\,x}^{l} \left[ K_{j\,x}^l, U_{j\,x}^l\dots U_{1\,x}^1 \left( \rho_x^\text{in}\otimes \rho_{x-1}^\text{m} \otimes \ket{0\dots 0}_\text{1:l}^x \bra{0\dots 0}\right)U_{1\, x}^{1^\dagger}\dots U_{j\, x}^{l^\dagger}\right] U_{j+1\, x}^{l^\dagger} \dots U_{m_{l}\, x}^{l^\dagger}\label{eq:g}\\
	C=&U_{z}^{l+1^\dagger}\dots U_{z}^{L+1^\dagger}\U_{z+1}^\dagger\dots \mathcal{U}_x^\dagger(\id_{x:z+1}^{\text{in,}1:L\text{,m}}\otimes \id_{z}^\text{l:L+1}\otimes \id_{x-1:z+1}^\text{out}\otimes\ket{\phi_x^\text{out}}\bra{\phi_x^\text{out}})\mathcal{U}_x\dots \mathcal{U}_{z+1} U_{z}^{L+1}\dots U_{z}^{l+1}\label{eq:h}
\end{align}
it is 
\begin{align}
	\delta C
	=& \frac{i}{N}\sum_{x=1}^{N}\sum_{l=1}^{L+1}\sum_{j=1}^{m_l}\Bigg[  \mathrm{tr} \left( \left(\id_x^{0:l-1}\otimes A\right)  \left(B_x\otimes \ket{0\dots 0}_{l+1:L+1}^x\bra{0\dots 0}\right) \right)\nonumber\\
	&+\sum_{z=1}^{x-1}\mathrm{tr}\left( \left(\id_z^{0:l-1}\otimes C\right)  \left( \bigotimes_{y=x}^{z+1}\rho_y^\text{in}\otimes  B_z\otimes \ket{0\dots 0}_{l+1:L+1}^z\bra{0\dots 0}\otimes \ket{0\dots 0}_{1:L+1}^{z+1:x}\bra{0\dots 0} \right) \right)\Bigg]\nonumber\\
	=& \frac{i}{N}\sum_{x=1}^{N}\sum_{l=1}^{L+1}\sum_{j=1}^{m_l}\Bigg[  \mathrm{tr} \left( A  \left(\mathrm{tr}_{0:l-1}^x(B_x)\otimes \ket{0\dots 0}_{l+1:L+1}^x\bra{0\dots 0}\right) \right)\nonumber\\
	&+\sum_{z=1}^{x-1}\mathrm{tr}\left( C  \left( \bigotimes_{y=x}^{z+1}\rho_y^\text{in}\otimes \mathrm{tr}_{0:l-1}^z( B_z)\otimes \ket{0\dots 0}_{l+1:L+1}^z\bra{0\dots 0}\otimes \ket{0\dots 0}_{1:L+1}^{z+1:x}\bra{0\dots 0} \right) \right)\Bigg]\nonumber\\
	=& \frac{i}{N}\sum_{x=1}^{N}\sum_{l=1}^{L+1}\sum_{j=1}^{m_l}\Bigg[  \mathrm{tr} \left( A  \left(\mathrm{tr}_{0:l-1}^x(B_x)\otimes \id^{l+1:L+1}_x\right) \left(\id_{x}^l\otimes\ket{0\dots 0}_{l+1:L+1}^x\bra{0\dots 0}\right)\right) \nonumber\\
	&+\sum_{z=1}^{x-1}\mathrm{tr}\Bigg( C \left(\id^\text{in}_{x:z+1}\otimes  \mathrm{tr}_{0:l-1}^z( B_z)\otimes \id_z^{l+1:L+1}\otimes \id_{z+1:x}^{1:L+1}\right)  \nonumber \\
	&\cdot\left( \bigotimes_{y=x}^{z+1}\rho_y^\text{in}\otimes\id_z^l\otimes \ket{0\dots 0}_{l+1:L+1}^z\bra{0\dots 0}\otimes \ket{0\dots 0}_{1:L+1}^{z+1:x}\bra{0\dots 0} \right) \Bigg)\Bigg]\nonumber\\
	=& \frac{i}{N}\sum_{x=1}^{N}\sum_{l=1}^{L+1}\sum_{j=1}^{m_l}\Bigg[  \mathrm{tr} \left(\left(\id_{x}^l\otimes\ket{0\dots 0}_{l+1:L+1}^x\bra{0\dots 0}\right) A  \left(\mathrm{tr}_{0:l-1}^x(B_x)\otimes \id^{l+1:L+1}_x\right) \right) \nonumber\\
	&+\sum_{z=1}^{x-1}\mathrm{tr}\Bigg( \left( \bigotimes_{y=x}^{z+1}\rho_y^\text{in}\otimes\id_z^l\otimes \ket{0\dots 0}_{l+1:L+1}^z\bra{0\dots 0}\otimes \ket{0\dots 0}_{1:L+1}^{z+1:x}\bra{0\dots 0} \right) C   \nonumber \\
	&\cdot \left(\id^\text{in}_{x:z+1}\otimes  \mathrm{tr}_{0:l-1}^z( B_z)\otimes \id_z^{l+1:L+1}\otimes \id_{z+1:x}^{1:L+1}\right)\Bigg)\Bigg]\nonumber
\end{align}
which leads to
\begin{align}
	\delta C
	=& \frac{i}{N}\sum_{x=1}^{N}\sum_{l=1}^{L+1}\sum_{j=1}^{m_l}\Bigg[  \mathrm{tr} \left(\mathrm{tr}_{l+1:L+1}^x\left(\left(\id_{x}^l\otimes\ket{0\dots 0}_{l+1:L+1}^x\bra{0\dots 0}\right) A \right) \mathrm{tr}_{0:l-1}^x(B_x) \right) \nonumber\\
	&+\sum_{z=1}^{x-1}\mathrm{tr}\Bigg(\mathrm{tr}_\text{in}^{x:z+1} \Bigg(\mathrm{tr}^z_{l+1:L+1}\Bigg(\mathrm{tr}^{z+1:x}_{1:L+1}\Bigg(\Bigg( \bigotimes_{y=x}^{z+1}\rho_y^\text{in}\otimes\id_z^l\otimes \ket{0\dots 0}_{l+1:L+1}^z\bra{0\dots 0}\nonumber \\
	&\otimes \ket{0\dots 0}_{1:L+1}^{z+1:x}\bra{0\dots 0} \Bigg) C  \Bigg)\Bigg)\Bigg)  \mathrm{tr}_{0:l-1}^z( B_z)\Bigg)\Bigg].
\end{align}
These partial traces now have to be expressed in terms of layer-to-layer channels. With \eqref{eq:g} it follows that
\begin{align*}
	\mathrm{tr}_{0:l-1}^x&(B_x) \\
	=&\mathrm{tr}_{0:l-1}^x \Big(U_{m_{l}\,x}^{l}\dots U_{j+1\,x}^{l} \Big[ K_{j\,x}^l, U_{j\,x}^l\dots U_{1\,x}^1 \Big( \rho_x^\text{in}\otimes \rho_{x-1}^\text{m} \otimes \ket{0\dots 0}_\text{1:l}^x \bra{0\dots 0}\Big)U_{1\, x}^{1^\dagger}\dots U_{j\, x}^{l^\dagger}\Big] U_{j+1\, x}^{l^\dagger} \dots U_{m_{l}\, x}^{l^\dagger}\Big)\\
	=&\mathrm{tr}_{l-1}^x \Big(U_{m_{l}\,x}^{l}\dots U_{j+1\,x}^{l} \Big[ K_{j\,x}^l, U_{j\,x}^l\dots U_{1\,x}^l \Big( \mathrm{tr}_{0:l-2}^x \Big(U_{x}^{l-1}\dots  U_{x}^1 \Big( \rho_x^\text{in}\otimes \rho_{x-1}^\text{m} \otimes \ket{0\dots 0}_\text{1:l-1}^x \bra{0\dots 0}\Big)U_{x}^{1^\dagger} \dots U_{ x}^{l-1^\dagger}\Big)\\
	&\otimes \ket{0\dots 0}_l^x \bra{0\dots 0} \Big) U_{1\, x}^{l^\dagger} \dots U_{j\, x}^{l^\dagger}\Big] U_{j+1\, x}^{l^\dagger} \dots U_{m_{l}\, x}^{l^\dagger}\Big)
\end{align*}
With the layer to layer channel
\begin{align}
	\mathcal{E}^l:B(\h^{l-1})&\to  \B(\h^{l})\\
	X^{l-1}&\mapsto\mathrm{tr}_{l-1}\left(U^l\left(X^{l-1}\otimes \ket{0\dots 0}_l\bra{0\dots 0}\right)U^{l^\dagger}\right)\nonumber
\end{align}
it is
\begin{align}
	\rho_x^{l-1}&=\mathcal{E}^{l-1}\left(\dots \mathcal{E}^1\left( \rho_x^\text{in}\otimes \rho_{x-1}^\text{m} \right)\dots\right)
\end{align}
as shown in \cite{beerTrainingDeepQuantum2020},  hence
\begin{align*}
	\mathrm{tr}&_{0:l-2}^x \Big(U_{x}^{l-1}\dots  U_{x}^1 \Big( \rho_x^\text{in}\otimes \rho_{x-1}^\text{m} \otimes \ket{0\dots 0}_\text{1:l-1}^x \bra{0\dots 0}\Big)U_{x}^{1^\dagger} \dots U_{ x}^{l-1^\dagger}\Big)\\
	&= \mathrm{tr}_{1:l-2}^x \Big(U_{x}^{l-1}\dots U_{x}^1 \Big(\mathrm{tr}_{0}^x \Big( U_{x}^1 \Big( \rho_x^\text{in}\otimes \rho_{x-1}^\text{m} \otimes \ket{0\dots 0}_\text{1}^x \bra{0\dots 0}\Big)U_{x}^{1^\dagger}\Big) \otimes \ket{0\dots 0}_\text{2:l-1}^x \bra{0\dots 0}\Big) U_{x}^{2^\dagger }\dots U_{ x}^{l-1^\dagger}\Big)\\
	&= \mathrm{tr}_{1:l-2}^x \Big(U_{x}^{l-1}\dots U_{x}^1 \Big(\mathcal{E}^1\Big(\rho_x^\text{in}\otimes \rho_{x-1}^\text{m}\Big) \otimes \ket{0\dots 0}_\text{2:l-1}^x \bra{0\dots 0}\Big) U_{x}^{2^\dagger }\dots U_{ x}^{l-1^\dagger}\Big)\\
	&=\dots\\
	&=\mathcal{E}^{l-1}\left(\dots \mathcal{E}^1\left( \rho_x^\text{in}\otimes \rho_{x-1}^\text{m} \right)\dots\right)\\
	&=\rho_x^{l-1}.
\end{align*}
Thus we get
\begin{align}
	\mathrm{tr}_{0:l-1}^x&(B_x) \nonumber\\
	=&\mathrm{tr}_{l-1}^x \Big(U_{m_{l}\,x}^{l}\dots U_{j+1\,x}^{l} \Big[ K_{j\,x}^l, U_{j\,x}^l\dots U_{1\,x}^l \Big( \rho_x^{l-1}\otimes \ket{0\dots 0}_l^x \bra{0\dots 0} \Big) U_{1\, x}^{l^\dagger} \dots U_{j\, x}^{l^\dagger}\Big] U_{j+1\, x}^{l^\dagger} \dots U_{m_{l}\, x}^{l^\dagger}\Big).
\end{align}
Now let 
\[\sigma_x^l:=\mathrm{tr}_{l+1:L+1}^x\left(\left(\id_{x}^l\otimes\ket{0\dots 0}_{l+1:L+1}^x\bra{0\dots 0}\right) A \right).\]
With equation \(\eqref{eq:f}\) it follows that
\begin{align*}
	\sigma_x^l=&\mathrm{tr}_{l+1:L+1}^x\left(\left(\id_{x}^l\otimes\ket{0\dots 0}_{l+1:L+1}^x\bra{0\dots 0}\right) U_{x}^{l+1^\dagger}\dots U_{x}^{L+1^\dagger} \left( \id_x^{l:L} \otimes \id_{x}^{\text{m}} \otimes \ket{\phi_x^\text{out}}\bra{\phi_x^\text{out}}\right) U_{x}^{L+1}\dots U_{x}^{l+1} \right)\\
	=&\mathrm{tr}_{l+1:L+1}^x\Big(\left(\id_{x}^l\otimes\ket{0\dots 0}_{l+1:L}^x\bra{0\dots 0} \otimes \id_x^{L+1}\right)\left(\id_{x}^{l:L}\otimes\ket{0\dots 0}_{L+1}^x\bra{0\dots 0}\right) \left(U_{x}^{l+1^\dagger}\dots U_{x}^{L^\dagger}\otimes \id_x^{L+1}\right)\\
	&\cdot \left(\id_x^{l:L-1}\otimes  U_{x}^{L+1^\dagger}\right) \left( \id_x^{l:L} \otimes \id_{x}^{\text{m}} \otimes \ket{\phi_x^\text{out}}\bra{\phi_x^\text{out}}\right) \left(\id_x^{l:L-1}\otimes  U_{x}^{L+1}\right) \left( U_{x}^{L}\dots U_{x}^{l+1}\otimes \id_x^{L+1}\right)\Big)
	\\
	=&\mathrm{tr}_{l+1:L+1}^x\Big(\left(\id_{x}^l\otimes\ket{0\dots 0}_{l+1:L}^x\bra{0\dots 0} \otimes \id_x^{L+1}\right)\left(U_{x}^{l+1^\dagger}\dots U_{x}^{L^\dagger}\otimes \id_x^{L+1}\right)\left(\id_{x}^{l:L}\otimes\ket{0\dots 0}_{L+1}^x\bra{0\dots 0}\right) \\
	&\cdot \left(\id_x^{l:L-1}\otimes  U_{x}^{L+1^\dagger}\right) \left( \id_x^{l:L} \otimes \id_{x}^{\text{m}} \otimes \ket{\phi_x^\text{out}}\bra{\phi_x^\text{out}}\right) \left(\id_x^{l:L-1}\otimes  U_{x}^{L+1}\right) \left( U_{x}^{L}\dots U_{x}^{l+1}\otimes \id_x^{L+1}\right)\Big)\\
	=&\mathrm{tr}_{l+1:L}^x\Big(\left(\id_{x}^l\otimes\ket{0\dots 0}_{l+1:L}^x\bra{0\dots 0}\right)U_{x}^{l+1^\dagger}\dots U_{x}^{L^\dagger}\Big(\id_x^{l:L-1}\otimes \mathrm{tr}_{L+1}^x\Big(\left(\id_{x}^{L}\otimes\ket{0\dots 0}_{L+1}^x\bra{0\dots 0}\right) \\
	&\cdot   U_{x}^{L+1^\dagger} \left( \id_x^{L} \otimes \id_{x}^{\text{m}} \otimes \ket{\phi_x^\text{out}}\bra{\phi_x^\text{out}}\right)  U_{x}^{L+1}\Big)\Big)  U_{x}^{L}\dots U_{x}^{l+1}\Big).
\end{align*}
By defining the channel
\[\mathcal{F}^l(X^l)=\mathrm{tr}_l\left( \left(\id^{l-1}\otimes \ket{0\dots 0}_l\bra{0\dots 0}\right)U^{l^\dagger}\left(\id^{l-1}\otimes X^l\right)U^l\right),\]
that is the adjoint channel of \(\mathcal{E}^l\) as shown in \cite{beerTrainingDeepQuantum2020}, actually, we get
\begin{align}
	\sigma_x^l=&\mathrm{tr}_{l+1:L}^x\Big(\left(\id_{x}^l\otimes\ket{0\dots 0}_{l+1:L}^x\bra{0\dots 0}\right)U_{x}^{l+1^\dagger}\dots U_{x}^{L^\dagger}\Big(\id_x^{l:L-1}\otimes \mathcal{F}^{L+1}\left(\id_{x}^{\text{m}} \otimes \ket{\phi_x^\text{out}}\bra{\phi_x^\text{out}}\right)\Big) U_{x}^{L}\dots U_{x}^{l+1}\Big)\nonumber \\
	=&\dots\nonumber\\
	=&\mathrm{tr}_{l+1}^x\Big(\left(\id_{x}^l\otimes\ket{0\dots 0}_{l+1}^x\bra{0\dots 0}\right)U_{x}^{l+1^\dagger}\Big(\id_x^{l}\otimes \mathcal{F}^{l+2}\left(\dots\mathcal{F}^{L+1}\left(\id_{x}^{\text{m}} \otimes \ket{\phi_x^\text{out}}\bra{\phi_x^\text{out}}\right)\dots\right)\Big)  U_{x}^{l+1}\Big)\nonumber \\
	=& \mathcal{F}^{l+1}\left(\dots\mathcal{F}^{L+1}\left(\id_{x}^{\text{m}} \otimes \ket{\phi_x^\text{out}}\bra{\phi_x^\text{out}}\right)\dots\right).
\end{align}
Now, the only partial trace left is 
\begin{align*}
	\omega_{zx}^l:=&\mathrm{tr}_\text{in}^{x:z+1} \Bigg(\mathrm{tr}^z_{l+1:L+1}\Bigg(\mathrm{tr}^{x:z+1}_{1:L+1}\Bigg(\Bigg( \bigotimes_{y=x}^{z+1}\rho_y^\text{in}\otimes\id_z^l\otimes \ket{0\dots 0}_{l+1:L+1}^z\bra{0\dots 0}\nonumber \otimes \ket{0\dots 0}_{1:L+1}^{z+1:x}\bra{0\dots 0} \Bigg) C  \Bigg)\Bigg)\Bigg)
\end{align*}
With equation \(\eqref{eq:h}\) it follows that
\begin{align*}
	\omega_{zx}^l=&\mathrm{tr}^z_{l+1:L+1}\Bigg(\mathrm{tr}^{x:z+1}_{\text{in,}1:L+1}\Bigg(\Bigg( \bigotimes_{y=x}^{z+1}\rho_y^\text{in}\otimes\id_z^l\otimes \ket{0\dots 0}_{l+1:L+1}^z\bra{0\dots 0}\nonumber \otimes \ket{0\dots 0}_{1:L+1}^{z+1:x}\bra{0\dots 0} \Bigg)\\
	&\cdot U_{z}^{l+1^\dagger}\dots U_{z}^{L+1^\dagger}\U_{z+1}^\dagger\dots \mathcal{U}_x^\dagger(\id_{x:z+1}^{\text{in,}1:L\text{,m}}\otimes \id_{z}^\text{l:L+1}\otimes \id_{x-1:z+1}^\text{out}\otimes\ket{\phi_x^\text{out}}\bra{\phi_x^\text{out}})\mathcal{U}_x\dots \mathcal{U}_{z+1} U_{z}^{L+1}\dots U_{z}^{l+1}  
	\Bigg)\Bigg)\\
	=&\mathrm{tr}^z_{l+1:L+1}\Bigg(\mathrm{tr}^{x:z+1}_{\text{in,}1:L+1}\Bigg(\Bigg(\id_x^\text{in}\otimes  \bigotimes_{y=x-1}^{z+1}\rho_y^\text{in}\otimes\id_z^l\otimes \ket{0\dots 0}_{l+1:L+1}^z\bra{0\dots 0}\nonumber \otimes \ket{0\dots 0}_{1:L+1}^{z+1:x-1}\bra{0\dots 0}\otimes \id_x^{1:L+1} \Bigg)\\
	&\cdot \Big(\rho_x^\text{in}\otimes  \id_{z+1:x-1}^\text{in}\otimes \id^{l:L+1}_z \otimes \id^{1:L+1}_{z+1:x-1}\otimes \ket{0\dots 0}_{1:L+1}^{x}\bra{0\dots 0} \Big)\Big(U_{z}^{l+1^\dagger}\dots U_{z}^{L+1^\dagger}\U_{z+1}^\dagger\dots \mathcal{U}_{x-1}^\dagger\otimes \id_x^{1:L+1}\Big)\\
	&\cdot \Big(\id_z^{l:L+1}\otimes\id_{z+1:x-2}^{\text{in},1:L+1}\otimes \id_{x-1}^{\text{in},1:L,\text{out}}\otimes \mathcal{U}_x^\dagger\Big)\Big(\id_{x:z+1}^{\text{in,}1:L\text{,m}}\otimes \id_{z}^\text{l:L+1}\otimes \id_{x-1:z+1}^\text{out}\otimes\ket{\phi_x^\text{out}}\bra{\phi_x^\text{out}}\Big)\\
	&\cdot \Big(\id_z^{l:L+1}\otimes\id_{z+1:x-2}^{\text{in},1:L+1}\otimes \id_{x-1}^{\text{in},1:L,\text{out}}\otimes \mathcal{U}_x\Big)\Big(\mathcal{U}_{x-1}\dots \mathcal{U}_{z+1} U_{z}^{L+1}\dots U_{z}^{l+1}\otimes \id_x^{1:L+1} \Big) 
	\Bigg)\Bigg)\\
	=&\mathrm{tr}^z_{l+1:L+1}\Bigg(\mathrm{tr}^{x:z+1}_{\text{in,}1:L+1}\Bigg(\Bigg(\id_x^\text{in}\otimes  \bigotimes_{y=x-1}^{z+1}\rho_y^\text{in}\otimes\id_z^l\otimes \ket{0\dots 0}_{l+1:L+1}^z\bra{0\dots 0}\nonumber \otimes \ket{0\dots 0}_{1:L+1}^{z+1:x-1}\bra{0\dots 0}\otimes \id_x^{1:L+1} \Bigg)\\
	&\cdot \Big(U_{z}^{l+1^\dagger}\dots U_{z}^{L+1^\dagger}\U_{z+1}^\dagger\dots \mathcal{U}_{x-1}^\dagger\otimes \id_x^{1:L+1}\Big)\Big(\rho_x^\text{in}\otimes  \id_{z+1:x-1}^\text{in}\otimes \id^{l:L+1}_z \otimes \id^{1:L+1}_{z+1:x-1}\otimes \ket{0\dots 0}_{1:L+1}^{x}\bra{0\dots 0} \Big)\\
	&\cdot \Big(\id_z^{l:L+1}\otimes\id_{z+1:x-2}^{\text{in},1:L+1}\otimes \id_{x-1}^{\text{in},1:L,\text{out}}\otimes \mathcal{U}_x^\dagger\Big)\Big(\id_{x:z+1}^{\text{in,}1:L\text{,m}}\otimes \id_{z}^\text{l:L+1}\otimes \id_{x-1:z+1}^\text{out}\otimes\ket{\phi_x^\text{out}}\bra{\phi_x^\text{out}}\Big)\\
	&\cdot \Big(\id_z^{l:L+1}\otimes\id_{z+1:x-2}^{\text{in},1:L+1}\otimes \id_{x-1}^{\text{in},1:L,\text{out}}\otimes \mathcal{U}_x\Big)\Big(\mathcal{U}_{x-1}\dots \mathcal{U}_{z+1} U_{z}^{L+1}\dots U_{z}^{l+1}\otimes \id_x^{1:L+1} \Big) 
	\Bigg)\Bigg)\\
	=&\mathrm{tr}^z_{l+1:L+1}\Bigg(\mathrm{tr}^{x-1:z+1}_{\text{in,}1:L+1}\Bigg(\Bigg( \bigotimes_{y=x-1}^{z+1}\rho_y^\text{in}\otimes\id_z^l\otimes \ket{0\dots 0}_{l+1:L+1}^z\bra{0\dots 0}\nonumber \otimes \ket{0\dots 0}_{1:L+1}^{z+1:x-1}\bra{0\dots 0} \Bigg)\\
	&\cdot U_{z}^{l+1^\dagger}\dots U_{z}^{L+1^\dagger}\U_{z+1}^\dagger\dots \mathcal{U}_{x-1}^\dagger\Big(\id_z^{l:L+1}\otimes\id_{z+1:x-2}^{\text{in},1:L+1}\otimes \id_{x-1}^{\text{in},1:L,\text{out}}\otimes\mathrm{tr}^{x}_{\text{in,}1:L+1}\Big(\Big(\rho_x^\text{in}\otimes \id^{\text{m}}_{x-1}\otimes \ket{0\dots 0}_{1:L+1}^{x}\bra{0\dots 0} \Big)\\
	&\cdot \mathcal{U}_x^\dagger\Big(\id_{x-1}^\text{m}\otimes \id_{x}^{\text{in,}1:L\text{,m}}\otimes\ket{\phi_x^\text{out}}\bra{\phi_x^\text{out}}\Big)\mathcal{U}_x\Big)\Big)\mathcal{U}_{x-1}\dots \mathcal{U}_{z+1} U_{z}^{L+1}\dots U_{z}^{l+1}\Bigg)\Bigg).
\end{align*}
With
\begin{align}
	\tilde{\mathcal{S}}_x:\B(\h_x^\text{out})&\to \B(\h_{x-1}^\text{m})\\
	X_x^\text{out}&\mapsto\mathrm{tr}^{x}_{\text{in,}1:L+1}\Big(\Big(\rho_x^\text{in}\otimes \id^{\text{m}}_{x-1}\otimes \ket{0\dots 0}_{1:L+1}^{x}\bra{0\dots 0} \Big) \mathcal{U}_x^\dagger\Big(\id_{x-1}^\text{m}\otimes \id_{x}^{\text{in,}1:L\text{,m}}\otimes X_x^\text{out}\Big)\mathcal{U}_x\Big)\nonumber
\end{align}
it is
\begin{align*}
	\omega_{zx}^l=&\mathrm{tr}^z_{l+1:L+1}\Bigg(\mathrm{tr}^{x-1:z+1}_{\text{in,}1:L+1}\Bigg(\Bigg( \bigotimes_{y=x-1}^{z+1}\rho_y^\text{in}\otimes\id_z^l\otimes \ket{0\dots 0}_{l+1:L+1}^z\bra{0\dots 0}\nonumber \otimes \ket{0\dots 0}_{1:L+1}^{z+1:x-1}\bra{0\dots 0} \Bigg) U_{z}^{l+1^\dagger}\dots U_{z}^{L+1^\dagger}\\
	&\cdot\U_{z+1}^\dagger\dots \mathcal{U}_{x-1}^\dagger\Big(\id_z^{l:L+1}\otimes\id_{z+1:x-2}^{\text{in},1:L+1}\otimes \id_{x-1}^{\text{in},1:L,\text{out}}\otimes\tilde{\mathcal{S}}_x\big(\ket{\phi_x^\text{out}}\bra{\phi_x^\text{out}}\big)\Big)\mathcal{U}_{x-1}\dots \mathcal{U}_{z+1} U_{z}^{L+1}\dots U_{z}^{l+1}\Bigg)\Bigg)\\
	=&\mathrm{tr}^z_{l+1:L+1}\Bigg(\mathrm{tr}^{x-1:z+1}_{\text{in,}1:L+1}\Bigg(\Bigg( \bigotimes_{y=x-1}^{z+1}\rho_y^\text{in}\otimes\id_z^l\otimes \ket{0\dots 0}_{l+1:L+1}^z\bra{0\dots 0}\nonumber \otimes \ket{0\dots 0}_{1:L+1}^{z+1:x-1}\bra{0\dots 0} \Bigg) U_{z}^{l+1^\dagger}\dots U_{z}^{L+1^\dagger}\\
	&\cdot\U_{z+1}^\dagger\dots \mathcal{U}_{x-1}^\dagger\Big(\id_z^{l:L+1}\otimes\id_{z+1:x-2}^{\text{in},1:L+1}\otimes \id_{x-1}^{\text{in},1:L,\text{out}}\otimes\tilde{\mathcal{S}}_x\big(\ket{\phi_x^\text{out}}\bra{\phi_x^\text{out}}\big)\Big)\mathcal{U}_{x-1}\dots \mathcal{U}_{z+1} U_{z}^{L+1}\dots U_{z}^{l+1}\Bigg)\Bigg)\\
	=&\mathrm{tr}^z_{l+1:L+1}\Bigg(\mathrm{tr}^{x-1:z+1}_{\text{in,}1:L+1}\Bigg(\Bigg(\id_{x-1}^\text{in}\otimes  \bigotimes_{y=x-2}^{z+1}\rho_y^\text{in}\otimes\id_z^l\otimes \ket{0\dots 0}_{l+1:L+1}^z\bra{0\dots 0}\nonumber \otimes \ket{0\dots 0}_{1:L+1}^{z+1:x-2}\bra{0\dots 0}\otimes \id_{x-1}^{1:L+1} \Bigg)\\
	&\cdot \Big(\rho_{x-1}^\text{in}\otimes  \id_{z+1:x-2}^\text{in}\otimes\id_z^{l:L+1} \otimes \id^{1:L+1}_{z+1:x-2}\otimes \ket{0\dots 0}_{1:L+1}^{x-1}\bra{0\dots 0} \Big)\Big( U_{z}^{l+1^\dagger}\dots U_{z}^{L+1^\dagger}\U_{z+1}^\dagger\dots \mathcal{U}_{x-2}^\dagger\otimes \id_{x-1}^{\text{in},1:L+1}\Big)\\
	&\cdot \Big(\id_z^{l:L+1}\otimes \id_{z+1:x-3}^{\text{in},1:L+1}\otimes \id_{x-2}^{\text{in},1:L,\text{out}}\otimes \mathcal{U}_{x-1}^\dagger\Big)\Big(\id_z^{l:L+1}\otimes\id_{z+1:x-2}^{\text{in},1:L+1}\otimes \id_{x-1}^{\text{in},1:L,\text{out}}\otimes\tilde{\mathcal{S}}_x\big(\ket{\phi_x^\text{out}}\bra{\phi_x^\text{out}}\big)\Big)\\
	&\cdot\Big(\id_z^{l:L+1}\otimes \id_{z+1:x-3}^{\text{in},1:L+1}\otimes \id_{x-2}^{\text{in},1:L,\text{out}}\otimes \mathcal{U}_{x-1}\Big)\Big(\mathcal{U}_{x-2}\dots \mathcal{U}_{z+1} U_{z}^{L+1}\dots U_{z}^{l+1}\otimes \id_{x-1}^{\text{in},1:L+1}\Big)\Bigg)\Bigg)\\
	=&\mathrm{tr}^z_{l+1:L+1}\Bigg(\mathrm{tr}^{x-1:z+1}_{\text{in,}1:L+1}\Bigg(\Bigg(\id_{x-1}^\text{in}\otimes  \bigotimes_{y=x-2}^{z+1}\rho_y^\text{in}\otimes\id_z^l\otimes \ket{0\dots 0}_{l+1:L+1}^z\bra{0\dots 0}\nonumber \otimes \ket{0\dots 0}_{1:L+1}^{z+1:x-2}\bra{0\dots 0}\otimes \id_{x-1}^{1:L+1} \Bigg)\\
	&\cdot \Big( U_{z}^{l+1^\dagger}\dots U_{z}^{L+1^\dagger}\U_{z+1}^\dagger\dots \mathcal{U}_{x-2}^\dagger\otimes \id_{x-1}^{\text{in},1:L+1}\Big)\Big(\rho_{x-1}^\text{in}\otimes  \id_{z+1:x-2}^\text{in}\otimes\id_z^{l:L+1} \otimes \id^{1:L+1}_{z+1:x-2}\otimes \ket{0\dots 0}_{1:L+1}^{x-1}\bra{0\dots 0} \Big)\\
	&\cdot \Big(\id_z^{l:L+1}\otimes \id_{z+1:x-3}^{\text{in},1:L+1}\otimes \id_{x-2}^{\text{in},1:L,\text{out}}\otimes \mathcal{U}_{x-1}^\dagger\Big)\Big(\id_z^{l:L+1}\otimes\id_{z+1:x-2}^{\text{in},1:L+1}\otimes \id_{x-1}^{\text{in},1:L,\text{out}}\otimes\tilde{\mathcal{S}}_x\big(\ket{\phi_x^\text{out}}\bra{\phi_x^\text{out}}\big)\Big)\\
	&\cdot\Big(\id_z^{l:L+1}\otimes \id_{z+1:x-3}^{\text{in},1:L+1}\otimes \id_{x-2}^{\text{in},1:L,\text{out}}\otimes \mathcal{U}_{x-1}\Big)\Big(\mathcal{U}_{x-2}\dots \mathcal{U}_{z+1} U_{z}^{L+1}\dots U_{z}^{l+1}\otimes \id_{x-1}^{\text{in},1:L+1}\Big)\Bigg)\Bigg)\\
	=&\mathrm{tr}^z_{l+1:L+1}\Bigg(\mathrm{tr}^{x-2:z+1}_{\text{in,}1:L+1}\Bigg(\Bigg( \bigotimes_{y=x-2}^{z+1}\rho_y^\text{in}\otimes\id_z^l\otimes \ket{0\dots 0}_{l+1:L+1}^z\bra{0\dots 0}\nonumber \otimes \ket{0\dots 0}_{1:L+1}^{z+1:x-2}\bra{0\dots 0} \Bigg)U_{z}^{l+1^\dagger}\dots U_{z}^{L+1^\dagger}\\
	&\cdot  \U_{z+1}^\dagger\dots \mathcal{U}_{x-2}^\dagger\Big(\id_z^{l:L+1}\otimes \id_{z+1:x-3}^{\text{in},1:L+1}\otimes \id_{x-2}^{\text{in},1:L,\text{out}}\otimes\mathrm{tr}^{x-1}_{\text{in,}1:L+1}\Big(\Big(\rho_{x-1}^\text{in} \otimes \id^{\text{m}}_{x-2}\otimes \ket{0\dots 0}_{1:L+1}^{x-1}\bra{0\dots 0} \Big)\\
	&\cdot \mathcal{U}_{x-1}^\dagger\Big(\id_{x-2}^{\text{m}}\otimes \id_{x-1}^{\text{in},1:L,\text{out}}\otimes\tilde{\mathcal{S}}_x\big(\ket{\phi_x^\text{out}}\bra{\phi_x^\text{out}}\big)\Big)\mathcal{U}_{x-1}\Big)\Big)\mathcal{U}_{x-2}\dots \mathcal{U}_{z+1} U_{z}^{L+1}\dots U_{z}^{l+1} \Bigg)\Bigg)
\end{align*}
By defining 
\begin{align}
	\mathcal{S}_x:\B(\h_x^\text{m})&\to \B(\h_{x-1}^\text{m})\\
	X_x^\text{m}&\mapsto\mathrm{tr}^{x}_{\text{in,}1:L+1}\Big(\Big(\rho_x^\text{in}\otimes \id^{\text{m}}_{x-1}\otimes \ket{0\dots 0}_{1:L+1}^{x}\bra{0\dots 0} \Big) \mathcal{U}_x^\dagger\Big(\id_{x-1}^\text{m}\otimes \id_{x}^{\text{in,}1:L\text{,out}}\otimes X_x^\text{m}\Big)\mathcal{U}_x\Big)\nonumber
\end{align}
we get
\begin{align}
	\omega_{zx}^l=&\mathrm{tr}^z_{l+1:L+1}\Bigg(\mathrm{tr}^{x-2:z+1}_{\text{in,}1:L+1}\Bigg(\Bigg( \bigotimes_{y=x-2}^{z+1}\rho_y^\text{in}\otimes\id_z^l\otimes \ket{0\dots 0}_{l+1:L+1}^z\bra{0\dots 0}\nonumber \otimes \ket{0\dots 0}_{1:L+1}^{z+1:x-2}\bra{0\dots 0} \Bigg)U_{z}^{l+1^\dagger}\dots U_{z}^{L+1^\dagger}\\
	&\cdot  \U_{z+1}^\dagger\dots \mathcal{U}_{x-2}^\dagger\Big(\id_z^{l:L+1}\otimes \id_{z+1:x-3}^{\text{in},1:L+1}\otimes \id_{x-2}^{\text{in},1:L,\text{out}}\otimes \mathcal{S}_{x-1}\big(\tilde{\mathcal{S}}_x\big(\ket{\phi_x^\text{out}}\bra{\phi_x^\text{out}}\big)\big)\Big)\mathcal{U}_{x-2}\dots \mathcal{U}_{z+1} U_{z}^{L+1}\dots U_{z}^{l+1} \Bigg)\Bigg)\nonumber\\
	=&\dots \nonumber\\
	=&\mathrm{tr}^z_{l+1:L+1}\Big(\Big( \id_z^l\otimes \ket{0\dots 0}_{l+1:L+1}^z\bra{0\dots 0}\Big) U_{z}^{l+1^\dagger}\dots U_{z}^{L+1^\dagger} \Big(\id_z^{l:L,\text{out}}\otimes \mathcal{S}_{z+1}\big(\dots \mathcal{S}_{x-1}\big(\tilde{\mathcal{S}}_x\big(\ket{\phi_x^\text{out}}\bra{\phi_x^\text{out}}\big)\big)\dots\big)\Big)\nonumber \\
	&\cdot U_{z}^{L+1}\dots U_{z}^{l+1} \Big)\nonumber
\end{align}
Analogously to the calculation for \(\sigma_x^l\) this is just
\begin{align}
	\omega_{zx}^l=\mathcal{F}_z^{l+1}\big(\dots \mathcal{F}_z^{L+1}\big(\mathcal{S}_{z+1}\big(\dots \mathcal{S}_{x-1}\big(\tilde{\mathcal{S}}_x\big(\ket{\phi_x^\text{out}}\bra{\phi_x^\text{out}}\big)\big)\dots\big)\big)\dots\big).
\end{align}
The channels \(\mathcal{S}\) and \(\tilde{\mathcal{S}}\) also have to be written in terms of layer-to-layer channels. It is
\begin{align}
	\mathcal{S}_x(X_x^\text{m})=&\mathrm{tr}^{x}_{\text{in,}1:L+1}\Big(\Big(\rho_x^\text{in}\otimes \id^{\text{m}}_{x-1}\otimes \ket{0\dots 0}_{1:L+1}^{x}\bra{0\dots 0} \Big) \mathcal{U}_x^\dagger\Big(\id_{x-1}^\text{m}\otimes \id_{x}^{\text{in,}1:L\text{,out}}\otimes X_x^\text{m}\Big)\mathcal{U}_x\Big)\nonumber\\
	=&\mathrm{tr}^{x}_{\text{in,}1:L+1}\Big(\Big(\rho_x^\text{in}\otimes \id^{\text{m}}_{x-1}\otimes \ket{0\dots 0}_{1:L+1}^{x}\bra{0\dots 0} \Big) U_x^{1^\dagger}\dots U_x^{L+1^\dagger}\Big(\id_{x-1}^\text{m}\otimes \id_{x}^{\text{in,}1:L\text{,out}}\otimes X_x^\text{m}\Big)U_x^{L+1}\dots U_x^1\Big)\nonumber\\
	=&\mathrm{tr}^{x}_{\text{in,}1:L+1}\Big(\Big(\rho_x^\text{in}\otimes \id^{\text{m}}_{x-1}\otimes \ket{0\dots 0}_{1:L}^{x}\bra{0\dots 0} \otimes \id_x^{L+1} \Big) \Big(\id_x^{0:L} \otimes \ket{0\dots 0}_{L+1}^{x}\bra{0\dots 0} \Big)\nonumber\\
	&\cdot \Big(U_x^{1^\dagger}\dots U_x^{L^\dagger}\otimes \id_x^{L+1}\Big)\Big(\id_x^{0:L-1}\otimes U_x^{L+1^\dagger}\Big)\Big(\id_{x}^{0:L,\text{out}}\otimes X_x^\text{m}\Big) \Big(\id_x^{0:L-1}\otimes U_x^{L+1}\Big) \Big(U_x^{L}\dots U_x^1\otimes \id_x^{L+1}\Big)\Big)\nonumber\\
	=&\mathrm{tr}^{x}_{\text{in,}1:L+1}\Big(\Big(\rho_x^\text{in}\otimes \id^{\text{m}}_{x-1}\otimes \ket{0\dots 0}_{1:L}^{x}\bra{0\dots 0} \otimes \id_x^{L+1} \Big)\Big(U_x^{1^\dagger}\dots U_x^{L^\dagger}\otimes \id_x^{L+1}\Big) \nonumber\\
	&\cdot \Big(\id_x^{0:L} \otimes \ket{0\dots 0}_{L+1}^{x}\bra{0\dots 0} \Big) \Big(\id_x^{0:L-1}\otimes U_x^{L+1^\dagger}\Big)\Big(\id_{x}^{0:L,\text{out}}\otimes X_x^\text{m}\Big)\nonumber\\
	&\cdot \Big(\id_x^{0:L-1}\otimes U_x^{L+1}\Big) \Big(U_x^{L}\dots U_x^1\otimes \id_x^{L+1}\Big)\Big)\nonumber\\
	=&\mathrm{tr}^{x}_{\text{in,}1:L}\Big(\Big(\rho_x^\text{in}\otimes \id^{\text{m}}_{x-1}\otimes \ket{0\dots 0}_{1:L}^{x}\bra{0\dots 0}  \Big)U_x^{1^\dagger}\dots U_x^{L^\dagger}  \Big(\id_x^{0:L-1}\otimes\mathrm{tr}^{x}_{L+1}\Big(\Big(\id_x^{L} \otimes \ket{0\dots 0}_{L+1}^{x}\bra{0\dots 0} \Big)\nonumber\\
	&\cdot  U_x^{L+1^\dagger}\Big(\id_{x}^{L,\text{out}}\otimes X_x^\text{m}\Big) U_x^{L+1}\Big)\Big) U_x^{L}\dots U_x^1\Big)\nonumber\\
	=&\mathrm{tr}^{x}_{\text{in,}1:L}\Big(\Big(\rho_x^\text{in}\otimes \id^{\text{m}}_{x-1}\otimes \ket{0\dots 0}_{1:L}^{x}\bra{0\dots 0}  \Big)U_x^{1^\dagger}\dots U_x^{L^\dagger}  \Big(\id_x^{0:L-1}\otimes \mathcal{F}^{L+1}\Big(\id_x^\text{out}\otimes X_x^\text{m}\Big)\Big) U_x^{L}\dots U_x^1\Big)\nonumber\\
	=&\dots\nonumber\\
	=&\mathrm{tr}^{x}_{\text{in,}1}\Big(\Big(\rho_x^\text{in}\otimes \id^{\text{m}}_{x-1}\otimes \ket{0\dots 0}_{1}^{x}\bra{0\dots 0}  \Big)U_x^{1^\dagger} \Big(\id_x^{0}\otimes \mathcal{F}^{2} \Big(\dots\mathcal{F}^{L+1}\Big(\id_x^\text{out}\otimes X_x^\text{m}\Big)\Big)\Big)  U_x^1\Big)\nonumber\\
	=& \tilde{\mathcal{F}}^1_x\Big( \mathcal{F}^{2} \Big(\dots\mathcal{F}^{L+1}\Big(\id_x^\text{out}\otimes X_x^\text{m}\Big)\Big)\Big) ,
\end{align}
where 
\begin{align}
	\tilde{\mathcal{F}}^1_x:\B(\h_x^1)&\to \B(\h_{x-1}^\text{m})\\
	X_x^1&\mapsto \mathrm{tr}^{x}_{\text{in,}1}\Big(\Big(\rho_x^\text{in}\otimes \id^{\text{m}}_{x-1}\otimes \ket{0\dots 0}_{1}^{x}\bra{0\dots 0}  \Big)U_x^{1^\dagger} \Big(\id_x^{0}\otimes X_x^1\Big)  U_x^1\Big).\nonumber
\end{align}
Analogously, we get
\begin{equation}
	\tilde{\mathcal{S}}_x(X_x^\text{out})=\tilde{\mathcal{F}}^1_x\Big( \mathcal{F}^{2} \Big(\dots\mathcal{F}^{L+1}\Big(\id_x^\text{m}\otimes X_x^\text{out}\Big)\Big)\Big).
\end{equation}
In fact, \(\tilde{\mathcal{S}}_x\) is the adjoint channel to the channel
\begin{align*}
	\mathcal{T}_x:\B(\h^\text{m})&\to \B(\h^\text{out})\\
	X_{x-1}^\text{m}&\mapsto \mathrm{tr}^x_{0:L,\text{m}} \left( \U_{x}\left(  \rho_x^\text{in}\otimes X^\text{m}_{x-1}\otimes \ket{0\dots 0}_{1:L+1}\bra{0\dots 0} \right) \U_{x}\right)
\end{align*}
as it is
\begin{align*}
	\mathrm{tr}(X_x^{\text{out}^\dagger}\mathcal{T}_x(X_{x-1}^\text{m}))&=\mathrm{tr}\left(X_x^{\text{out}^\dagger}\mathrm{tr}^x_{0:L,\text{m}} \left( \U_{x}\left(  \rho_x^\text{in}\otimes X^\text{m}_{x-1}\otimes \ket{0\dots 0}_{1:L+1}\bra{0\dots 0} \right) \U_{x}^\dagger\right)\right)\\
	&=\sum_{k}\mathrm{tr}\left(X_x^{\text{out}^\dagger}  \left(\bra{k}\otimes \id^\text{out}\right) \U_{x}\left( \rho_x^\text{in}\otimes X^\text{m}_{x-1}\otimes \ket{0\dots 0}_{1:L+1}\bra{0\dots 0} \right) \U_{x}^\dagger\left(\ket{k}\otimes \id^\text{out}\right)\right)\\
	&=\sum_{k}\mathrm{tr}\left(\U_{x}^\dagger\left(\ket{k}\otimes \id^\text{out}\right)X_x^{\text{out}^\dagger}  \left(\bra{k}\otimes \id^\text{out}\right) \U_{x}\left( \rho_x^\text{in}\otimes X^\text{m}_{x-1}\otimes \ket{0\dots 0}_{1:L+1}\bra{0\dots 0} \right) \right)\\
	&=\mathrm{tr}\left(\U_{x}^\dagger\left(\id_x^{0:L,\text{m}}\otimes X_x^{\text{out}^\dagger}\right)  \U_{x}\left( \rho_x^\text{in}\otimes X^\text{m}_{x-1}\otimes \ket{0\dots 0}_{1:L+1}\bra{0\dots 0} \right) \right)\\
	&=\mathrm{tr}\left(\U_{x}^\dagger\left(\id_x^{0:L,\text{m}}\otimes X_x^{\text{out}^\dagger}\right)  \U_{x}\left(\id_x^{\text{in},1:L+1}\otimes X^\text{m}_{x-1}\right) \left( \rho_x^\text{in}\otimes \id^\text{m}_{x-1}\otimes \ket{0\dots 0}_{1:L+1}\bra{0\dots 0} \right)\right)\\
	&=\mathrm{tr}\left(\left( \rho_x^\text{in}\otimes \id^\text{m}_{x-1}\otimes \ket{0\dots 0}_{1:L+1}\bra{0\dots 0} \right)\U_{x}^\dagger\left(\id_x^{0:L,\text{m}}\otimes X_x^{\text{out}^\dagger}\right)  \U_{x}\left(\id_x^{\text{in},1:L+1}\otimes X^\text{m}_{x-1}\right) \right)\\
	&=\mathrm{tr}\left(\mathrm{tr}^x_{\text{in},1:L+1}\left(\left( \rho_x^\text{in}\otimes \id^\text{m}_{x-1}\otimes \ket{0\dots 0}_{1:L+1}\bra{0\dots 0} \right)\U_{x}^\dagger\left(\id_x^{0:L,\text{m}}\otimes X_x^{\text{out}^\dagger}\right)  \U_{x}\right) X^\text{m}_{x-1} \right)\\
	&=\mathrm{tr}\left(\left(\mathrm{tr}^x_{\text{in},1:L+1}\left(\U_{x}^\dagger\left(\id_x^{0:L,\text{m}}\otimes X_x^{\text{out}}\right)  \U_{x}\left( \rho_x^\text{in}\otimes \id^\text{m}_{x-1}\otimes \ket{0\dots 0}_{1:L+1}\bra{0\dots 0} \right)\right)\right)^\dagger X^\text{m}_{x-1} \right)\\
	&=\mathrm{tr}\left(\left(\mathrm{tr}^x_{\text{in},1:L+1}\left(\left( \rho_x^\text{in}\otimes \id^\text{m}_{x-1}\otimes \ket{0\dots 0}_{1:L+1}\bra{0\dots 0} \right)\U_{x}^\dagger\left(\id_x^{0:L,\text{m}}\otimes X_x^{\text{out}}\right)  \U_{x}\right)\right)^\dagger X^\text{m}_{x-1} \right)\\
	&=\mathrm{tr}\left(\left(\tilde{\mathcal{S}}_x(X_x^\text{out})\right)^\dagger X^\text{m}_{x-1} \right)
\end{align*}
for any orthonormal basis \(\{\ket{k}\}_k\) of \(\h^{0:L,\text{m}}\) and any \(X_x^\text{out}\in \B(\h_x^\text{out}),\ X_{x-1}^\text{m}\in \B(\h_{x-1}^\text{m})\).\\
Together, we get
\begin{align*}
	\delta C
	=& \frac{i}{N}\sum_{x=1}^{N}\sum_{l=1}^{L+1}\sum_{j=1}^{m_l}\Bigg[  \mathrm{tr} \Big(\sigma_x^l \mathrm{tr}_{l-1}^x \Big(U_{m_{l}\,x}^{l}\dots U_{j+1\,x}^{l} \Big[ K_{j\,x}^l, U_{j\,x}^l\dots U_{1\,x}^l \Big( \rho_x^{l-1}\otimes \ket{0\dots 0}_l^x \bra{0\dots 0} \Big) U_{1\, x}^{l^\dagger} \dots U_{j\, x}^{l^\dagger}\Big]\\
	&\cdot U_{j+1\, x}^{l^\dagger} \dots U_{m_{l}\, x}^{l^\dagger}\Big) \Big) +\sum_{z=1}^{x-1}\mathrm{tr}\Bigg(\omega_{zx}^l \mathrm{tr}_{l-1}^z \Big(U_{m_{l}\,z}^{l}\dots U_{j+1\,z}^{l} \Big[ K_{j\,z}^l, U_{j\,z}^l\dots U_{1\,z}^l \Big( \rho_z^{l-1}\otimes \ket{0\dots 0}_l^z \bra{0\dots 0} \Big) U_{1\, z}^{l^\dagger} \dots U_{j\, z}^{l^\dagger}\Big]\\
	&\cdot  U_{j+1\, z}^{l^\dagger} \dots U_{m_{l}\, z}^{l^\dagger}\Big)\Bigg)\Bigg].
\end{align*}
As \(\h_z^\text{m}=\h_x^\text{m}\), \(\h_z^\text{in}=\h_x^\text{in}\), \(\h_z^\text{out}=\h_x^\text{out}\) and \(U_{j\,z}^l=U_{j\,x}^l\) for all \(x,z\in \{1,...,N\}\) and we only operate on \(x\) or \(z\), respectively, we can drop the subscripts \(x\) and \(z\) on the unitaries and traces. By also pulling out the traces \(\mathrm{tr}_{l-1}\) we get
\begin{align*}
	\delta C
	=& \frac{i}{N}\sum_{x=1}^{N}\sum_{l=1}^{L+1}\sum_{j=1}^{m_l}\Bigg[  \mathrm{tr} \Big(\Big(\id^{l-1}\otimes \sigma_x^l\Big) U_{m_{l}}^{l}\dots U_{j+1}^{l} \Big[ K_{j}^l, U_{j}^l\dots U_{1}^l \Big( \rho_x^{l-1}\otimes \ket{0\dots 0}_l\bra{0\dots 0} \Big) U_{1}^{l^\dagger} \dots U_{j}^{l^\dagger}\Big]\\
	&\cdot U_{j+1}^{l^\dagger} \dots U_{m_{l}}^{l^\dagger}\Big)  +\sum_{z=1}^{x-1}\mathrm{tr}\Big(\Big(\id^{l-1}\otimes\omega_{zx}^l\Big) U_{m_{l}}^{l}\dots U_{j+1}^{l} \Big[ K_{j}^l, U_{j}^l\dots U_{1}^l \Big( \rho_z^{l-1}\otimes \ket{0\dots 0}_l \bra{0\dots 0} \Big)\\
	&\cdot  U_{1}^{l^\dagger} \dots U_{j}^{l^\dagger}\Big] U_{j+1}^{l^\dagger} \dots U_{m_{l}}^{l^\dagger}\Big)\Bigg].
\end{align*}
As it is
\begin{align*}
	\mathrm{tr}(A[B,C]D)=\mathrm{tr}(ABCD-ACBD)=\mathrm{tr}(CDAB-DACB)=\mathrm{tr}([C,DA]B)
\end{align*}
for all \(A,B,C,D\in \B(\h)\) for some Hilbert space \(\h\), it is
\begin{align*}
	\delta C
	=& \frac{i}{N}\sum_{x=1}^{N}\sum_{l=1}^{L+1}\sum_{j=1}^{m_l}\Bigg[  \mathrm{tr} \Big( \Big[ U_{j}^l\dots U_{1}^l \Big( \rho_x^{l-1}\otimes \ket{0\dots 0}_l\bra{0\dots 0} \Big) U_{1}^{l^\dagger} \dots U_{j}^{l^\dagger}, U_{j+1}^{l^\dagger} \dots U_{m_{l}}^{l^\dagger} \Big(\id^{l-1}\otimes \sigma_x^l\Big) U_{m_{l}}^{l}\dots U_{j+1}^{l}  \Big]K_{j}^l \Big)\\
	&+\sum_{z=1}^{x-1}\mathrm{tr}\Big(\Big[U_{j}^l\dots U_{1}^l \Big( \rho_z^{l-1}\otimes \ket{0\dots 0}_l \bra{0\dots 0} \Big) U_{1}^{l^\dagger} \dots U_{j}^{l^\dagger},U_{j+1}^{l^\dagger} \dots U_{m_{l}}^{l^\dagger} \Big(\id^{l-1}\otimes\omega_{zx}^l\Big) U_{m_{l}}^{l}\dots U_{j+1}^{l} \Big]K_{j}^l\Big)\Bigg]\\
	=&\frac{i}{N}\sum_{x=1}^{N}\sum_{l=1}^{L+1}\sum_{j=1}^{m_l} \mathrm{tr}\Big(M_{j\,x}^lK_j^l\Big)
\end{align*}
with
\begin{align*}
	M_{j\,x}^l=&  \Big[ U_{j}^l\dots U_{1}^l \Big( \rho_x^{l-1}\otimes \ket{0\dots 0}_l\bra{0\dots 0} \Big) U_{1}^{l^\dagger} \dots U_{j}^{l^\dagger}, U_{j+1}^{l^\dagger} \dots U_{m_{l}}^{l^\dagger} \Big(\id^{l-1}\otimes \sigma_x^l\Big) U_{m_{l}}^{l}\dots U_{j+1}^{l}  \Big]\\
	&+\sum_{z=1}^{x-1}\Big[U_{j}^l\dots U_{1}^l \Big( \rho_z^{l-1}\otimes \ket{0\dots 0}_l \bra{0\dots 0} \Big) U_{1}^{l^\dagger} \dots U_{j}^{l^\dagger},U_{j+1}^{l^\dagger} \dots U_{m_{l}}^{l^\dagger} \Big(\id^{l-1}\otimes\omega_{zx}^l\Big) U_{m_{l}}^{l}\dots U_{j+1}^{l} \Big]
\end{align*}
and
\begin{align*}
	\rho_x^{l-1}&=\mathcal{E}^{l-1}\left(\dots \mathcal{E}^1\left( \rho_x^\text{in}\otimes \rho_{x-1}^\text{m} \right)\dots\right),\\
	\sigma_x^l&= \mathcal{F}^{l+1}\left(\dots\mathcal{F}^{L+1}\left(\id_{x}^{\text{m}} \otimes \ket{\phi_x^\text{out}}\bra{\phi_x^\text{out}}\right)\dots\right),\\
	\omega_{zx}^l&=\mathcal{F}_z^{l+1}\big(\dots \mathcal{F}_z^{L+1}\big(\mathcal{S}_{z+1}\big(\dots \mathcal{S}_{x-1}\big(\tilde{\mathcal{S}}_x\big(\ket{\phi_x^\text{out}}\bra{\phi_x^\text{out}}\big)\big)\dots\big)\big)\dots\big).
\end{align*}
In \cite{beerTrainingDeepQuantum2020} it was shown that the matrices \(K_j^l\) that maximise a derivative of the cost function in the form 
\begin{align*}
	\delta C=\frac{i}{N}\sum_{x=1}^{N}\sum_{l=1}^{L+1}\sum_{j=1}^{m_l} \mathrm{tr}\Big(M_{j\,x}^lK_j^l\Big)
\end{align*}
are the matrices 
\begin{equation}
	K_j^l=\frac{i2^{m_{l-1}}\eta}{N}\sum_{x=1}^N \mathrm{tr}_\text{rest}\Big(M_{j\,x}^l\Big),
\end{equation}
if one demands, that the sum over the matrix elements of \(K_j^l\) should be smaller than some constant. Thereby, \(\mathrm{tr}_\text{rest}\) denotes the trace over all qubits not affected by \(U_j^l\) and \(\eta\) the learning rate. Altogether, this leads to the algorithm in Box~1:\\[1ex] \vspace{20pt}
{\centering
	\fbox{\label{box1}
		\begin{minipage}{.485\textwidth}
			\scriptsize
			\flushleft	
			{\normalsize \textbf{Box 1: Training algorithm local cost with pure output for \(M=1\)}\\\vspace{5pt}}
			\textbf{1. Initialize:}\\
			Choose the perceptron unitaries \(U_j^l\) randomly.\\
			\hspace{2pt}\\
			\textbf{2. Feedforward:} 
			For \(x=1,...,N\) do the following:\\
			\textbf{2.1.} Set
			\(\rho_x^0=\rho_x^\text{in}\otimes\rho_{x-1}^\text{m}\) where \(\rho^\text{m}_0\) is given.\\
			\textbf{2.2.} For \(l=1,...,L+1\) set
			\(\rho_x^l=\mathrm{tr}_{l-1}\left( U^l(\rho_x^{l-1}\otimes \ket{0\dots 0}_{l}\bra{0\dots 0}) U^{l^\dagger} \right).\)\\
			\textbf{2.3.} Set
			\(\rho_x^\text{m}=\mathrm{tr}_\text{out}^x(\rho_x^{L+1})\)
			and 
			\(\rho_x^\text{out}=\mathrm{tr}_\text{m}^x(\rho_x^{L+1}).\)\\
			\hspace{2pt}\\
			\textbf{3. Feedbackward:} For \(x=N,...,1\) do the following:\\
			\textbf{3.1.} Set \(\sigma_x^{L+1}=\id ^\text{m}_x\otimes \ket{\phi_x^{\text{out}}}\bra{\phi_x^{\text{out}}}\).\\
			\textbf{3.2.} For \(l=L,...,1\) set
			\(\sigma_x^l=\mathrm{tr}_{l+1}\left( (\id ^l\otimes \ket{0\dots 0}_{l+1}\bra{0\dots 0}) U^{l+1^\dagger} (\id ^l\otimes \sigma_x^{l+1})  U^{l+1} \right)\).\\
			\textbf{3.3.} Set 
			\(\omega_{xx}^\text{m}=\mathrm{tr}_{1,\text{in}}\left( (\rho_x^\text{in}\otimes \id _{x-1}^m\otimes \ket{0\dots 0}_{1}\bra{0\dots 0}) U^{1^\dagger} (\id _x^0\otimes \sigma_x^{1})  U^{1} \right)\).\\
			\textbf{3.4.} For \(z=x-1,...,1\) do the following steps:\\
			\textbf{3.4.1.} Set \(\omega_{zx}^{L+1}=\omega_{z+1\,x}^\text{m}\otimes \id ^\text{out}\).\\
			\textbf{3.4.2.} For \(l=L,...,1\) set
			\(\omega_{zx}^l=\mathrm{tr}_{l+1}\left( (\id ^l\otimes \ket{0\dots 0}_{l+1}\bra{0\dots 0}) U^{l+1^\dagger}(\id ^l\otimes \omega_{zx}^{l+1})   U^{l+1} \right).\)\\
			\textbf{3.4.3.} Set 
			\(\omega_{zx}^\text{m}=\mathrm{tr}_{1,\text{in}}\left( (\rho_x^\text{in}\otimes \id _{x-1}^m\otimes \ket{0\dots 0}_{1}\bra{0\dots 0}) U^{1^\dagger}  (\id _x^0\otimes \omega_{zx}^{1})  U^{1} \right).\)\\
			\textbf{4. Update the network:}  For \(l=1,...,L+1\) and \(j=1,...,m_l\) do the following: 
			
		\end{minipage}
		\hfill
		\begin{minipage}{.485\textwidth}
			\scriptsize
			\flushleft	
			%\phantom{\textbf{Box 1: Training algorithm}\\}
			
			\textbf{4.1.} For \(x=1,...,N\) set
			\begin{align*}
				M_{j\,x}^l=& \Bigg[\prod_{k=j}^{1} U_k^l\left( \rho_x^{l-1}\otimes\ket{0\dots 0}_l\bra{0\dots 0}\right) \prod_{k=1}^{j} U_k^{l^\dagger},\\
				&\ \prod_{k=j+1}^{m_l} U_k^{l^\dagger} \left(\id ^{l-1}\otimes \sigma_x^l\right)  \prod_{k=m_l}^{j+1} U_k^l \Bigg]\\
				&+ \sum_{z=1}^{x-1} \Bigg[ \prod_{k=j}^{1} U_k^l \left( \rho_z^{l-1}\otimes\ket{0\dots 0}_l\bra{0\dots 0}\right)  \prod_{k=1}^{j} U_k^{l^\dagger} ,\\
				&\ \prod_{k=j+1}^{m_l} U_k^{l^\dagger} \left(\id ^{l-1}\otimes \omega_{zx}^l\right) \prod_{k=m_l}^{j+1} U_k^l \Bigg].
			\end{align*}
			\textbf{4.2.} Set
			\[K_j^l=\frac{i2^{m_{l-1}} \eta}{N}\sum_{x=1}^N \mathrm{tr}_\text{rest}(M_{j\,x}^l),\]
			where \(\mathrm{tr}_\text{rest}\) denotes that the trace is taken over all systems that are unaffected by \(U_j^l\).\\
			\textbf{4.3.} Update each unitary $U_j^l$ according to $U_j^l\rightarrow e^{i\epsilon K_j^l} U_j^l$.\\
			\hspace{2pt}\\
			\textbf{5. Repeat:} Repeat steps 2. to 4. until the cost function reaches its maximum.
		\end{minipage}
	}
}\vspace{10pt}
\section{Optimising the local cost with mixed output states classically}
\label{localcostmixedstates}
If we want to learn from mixed output states, i.e. a training set of the form
\[S=\left(\rho_0^\text{m},\left(\rho_1^\text{in},\sigma_1^\text{out}\right),\dots ,  \left(\rho_N^\text{in},\sigma_N^\text{out}\right)  \right)\]
the first idea would be to use the fidelity for mixed states given by
\[F(\rho,\sigma )=\left(\mathrm{tr}\sqrt{\sqrt{\rho}\sigma \sqrt{\rho}}\right)^2\]
in the optimisation process, but the algorithm is not as easy to get with this, because it is not easy to determine the derivative of this anymore. But we can instead use the Hilbert-Schmidt norm for the cost function such that we get the cost function
\begin{equation}
	C=C_\text{local,mixed}=\frac{1}{N}\sum_{x=1}^N \Vert \rho_x^\text{out}- \sigma_x^\text{out} \Vert _2 ^2= \frac{1}{N} \sum_{x=1}^N \mathrm{tr}\left(\rho_x^\text{out}- \sigma_x^\text{out}\right)^2.
\end{equation}
Because of the cyclic rule of trace, this is just
\[ 	C= \frac{1}{N} \sum_{x=1}^N \mathrm{tr}\left(\rho_x^{\text{out}^2}+ \sigma_x^{\text{out}^2} - 2 \sigma_x^\text{out}\rho_x^\text{out}\right). \]
We now want to get a similar algorithm to the one in Box~1. As before, we start the training by choosing the perceptron unitaries \(U_j^l\) randomly at first and updating the perceptron unitaries in each step according to
\[U_j^l\mapsto U_j^{l\prime}=e^{i\epsilon K_j^l}U_j^l,\]
where the \(K_j^l\) are chosen in a way so that they minimise the derivative of the cost function, as opposed to the case with pure output states as the norm distance has to be minimised.\\
Up to first order in \(\epsilon\) it is because of the cyclic rule of trace
\[ \Delta	C= \frac{1}{N} \sum_{x=1}^N 2\mathrm{tr}\left(\rho_x^{\text{out}}\Delta \rho_x^\text{out}\right) - 2\mathrm{tr}\left( \sigma_x^\text{out}\Delta\rho_x^\text{out}\right) +\mathcal{O}(\epsilon ^2) =\frac{2}{N} \sum_{x=1}^N \mathrm{tr}\left(\left(\rho_x^{\text{out}}- \sigma_x^\text{out}\right)\Delta\rho_x^\text{out}\right)+\mathcal{O}(\epsilon ^2). \]
Like before, it is
\[\rho_x^\text{out}=\mathrm{tr}_{\text{in,}1:L\text{,m}}^x(\mathrm{tr}_{\text{m}}^{x-1}(U^{L+1}_{m_{L+1}} \dots U^{l}_{j} \dots U^{1}_{1} (\rho_x^{\text{in}}\otimes \rho_{x-1}^{\text{m}} \otimes \ket{0\dots 0}_{ {1:L+1}}^x\bra{0\dots 0}) U^{1^\dagger}_{1} \dots U^{l^\dagger}_{j} \dots U^{L+1^\dagger}_{m_{L+1}}))\]
and hence
\begin{align*}
	\Delta  \rho_x^\text{out} 
	=&  i\epsilon\sum_{l=1}^{L+1}\sum_{j=1}^{m_l} \mathrm{tr}_{\text{in,}1:L\text{,m}}^x\Big(\mathrm{tr}_{\text{m}}^{x-1}\Big( U_{m_{L+1}\,x}^{L+1}\dots U_{j+1\,x}^{l}\Big[ K_{j\,x}^l, U_{j\,x}^l\dots U_{1\,x}^1 \big( \rho_x^\text{in}\otimes \rho_{x-1}^\text{m} \\
	& \otimes \ket{0\dots 0}_\text{1:L+1}^x \bra{0\dots 0}\big) U_{1\, x}^{1^\dagger}\dots U_{j\, x}^{l^\dagger}\Big] U_{j+1\, x}^{l^\dagger} \dots U_{m_{L+1}\, x}^{L+1^\dagger} \Big) \\
	&+\mathcal{U}_x(\rho_x^{\text{in}}\otimes \Delta \rho_{x-1}^{\text{m}} \otimes \ket{0\dots 0}_{ 1:L+1}^x\bra{0\dots 0})\mathcal{U}_x^\dagger +\mathcal{O}(\epsilon^2)\Big)\Big) .
\end{align*}
We already know, that
\begin{align*}
	\Delta \rho_x^m=&i\epsilon \sum_{z=1}^{x}\sum_{l=1}^{L+1}\sum_{j=1}^{m_l} \mathrm{tr}_{{\text{in,}1:L\text{,out}}}^{x:z}\Bigg(\mathrm{tr}_{\text{m}}^{x-1:z-1} \Bigg( \mathcal{U}_x\dots \mathcal{U}_{z+1} U_{m_{L+1}\, z}^{L+1} \dots U_{j+1\, z}^{l}\Bigg[ K_{j\, z}^l, U_{j\, z}^{l} \dots U_{1\,z}^1   \left( \bigotimes_{y=x}^z\rho_y^{\text{in}} \right.\\
	&\otimes \rho_{z-1}^\text{m}  \otimes \ket{0\dots 0}_{ {1:L+1}}^{z:x}\bra{0\dots 0}\Bigg) U_{1\, z}^{1^\dagger }\dots U_{j\, z}^{l^\dagger }\Bigg]U_{j+1\, z}^{l^\dagger }\dots U_{m_{L+1}\, z}^{L+1^\dagger } \U_{z+1}^\dagger\dots \U_{x}^\dagger\Bigg)\Bigg)+\mathcal{O}(\epsilon^2),
\end{align*}
hence 
\begin{align*}
	\Delta  \rho_x^\text{out} 
	=&  i\epsilon\sum_{l=1}^{L+1}\sum_{j=1}^{m_l} \Bigg[ \mathrm{tr}_{\text{in,}1:L\text{,m}}^x\Big(\mathrm{tr}_{\text{m}}^{x-1}\Big( U_{m_{L+1}\,x}^{L+1}\dots U_{j+1\,x}^{l}\left[ K_{j\,x}^l, U_{j\,x}^l\dots U_{1\,x}^1 \left( \rho_x^\text{in}\otimes \rho_{x-1}^\text{m} \right. \right.\\
	& \left. \otimes \ket{0\dots 0}_\text{1:L+1}^x \bra{0\dots 0}\right) \left. \left. U_{1\, x}^{1^\dagger}\dots U_{j\, x}^{l^\dagger}\right] U_{j+1\, x}^{l^\dagger} \dots U_{m_{L+1}\, x}^{L+1^\dagger} \right) \Big)\Big) + \sum_{z=1}^x \mathrm{tr}_{\text{in,}1:L\text{,m}}^x\Bigg(\mathrm{tr}_{\text{m}}^{x-1}\Bigg( \mathcal{U}_x\\
	&\cdot\Bigg(\rho_x^{\text{in}}\otimes \mathrm{tr}_{{\text{in,}1:L\text{,out}}}^{x-1:z}\Bigg(\mathrm{tr}_{\text{m}}^{x-2:z-1} \Bigg( \mathcal{U}_{x-1}\dots \mathcal{U}_{z+1} U_{m_{L+1}\, z}^{L+1} \dots U_{j+1\, z}^{l}\Bigg[ K_{j\, z}^l, U_{j\, z}^{l} \dots U_{1\,z}^1  \\
	&\cdot   \left( \bigotimes_{y=x-1}^z\rho_y^{\text{in}} \otimes \rho_{z-1}^\text{m}  \otimes \ket{0\dots 0}_{ {1:L+1}}^{z:x-1}\bra{0\dots 0}\right) U_{1\, z}^{1^\dagger }\dots U_{j\, z}^{l^\dagger }\Bigg]U_{j+1\, z}^{l^\dagger }\dots U_{m_{L+1}\, z}^{L+1^\dagger } \U_{z+1}^\dagger\dots \U_{x-1}^\dagger\Bigg)\Bigg) \\ 
	& \otimes \ket{0\dots 0}_{ 1:L+1}^{x}\bra{0\dots 0}\Bigg)\mathcal{U}_x^\dagger \Bigg) \Bigg) \Bigg] +\mathcal{O}(\epsilon^2) .\\
	=&  i\epsilon\sum_{l=1}^{L+1}\sum_{j=1}^{m_l} \Bigg[ \mathrm{tr}_{\text{in,}1:L\text{,m}}^{x}\Big(\mathrm{tr}_{\text{m}}^{x-1}\Big( U_{m_{L+1}\,x}^{L+1}\dots U_{j+1\,x}^{l}\left[ K_{j\,x}^l, U_{j\,x}^l\dots U_{1\,x}^1 \left( \rho_x^\text{in}\otimes \rho_{x-1}^\text{m} \right. \right.\\
	& \left. \otimes \ket{0\dots 0}_\text{1:L+1}^x \bra{0\dots 0}\right) \left. \left. U_{1\, x}^{1^\dagger}\dots U_{j\, x}^{l^\dagger}\right] U_{j+1\, x}^{l^\dagger} \dots U_{m_{L+1}\, x}^{L+1^\dagger} \right) \Big)\Big) + \sum_{z=1}^x \mathrm{tr}_{\text{in,}1:L\text{,m}}^{x:z}\Bigg(\mathrm{tr}_{\text{m}}^{x-1:z-1} \\
	&\cdot \Bigg( \mathcal{U}_{x}\dots \mathcal{U}_{z+1} U_{m_{L+1}\, z}^{L+1} \dots U_{j+1\, z}^{l}\Bigg[ K_{j\, z}^l, U_{j\, z}^{l} \dots U_{1\,z}^1 \left( \bigotimes_{y=x}^z\rho_y^{\text{in}} \otimes \rho_{z-1}^\text{m}  \otimes \ket{0\dots 0}_{ {1:L+1}}^{z:x}\bra{0\dots 0}\right) \\
	&\cdot    U_{1\, z}^{1^\dagger }\dots U_{j\, z}^{l^\dagger }\Bigg]U_{j+1\, z}^{l^\dagger }\dots U_{m_{L+1}\, z}^{L+1^\dagger } \U_{z+1}^\dagger\dots \U_{x}^\dagger\Bigg)\Bigg)   \Bigg) \Bigg] +\mathcal{O}(\epsilon^2) .
\end{align*}
Therefore, we get up to first order
\begin{align*}
	\delta C
	=& \frac{i}{N}\sum_{x=1}^{N}\sum_{l=1}^{L+1}\sum_{j=1}^{m_l}\Bigg[  \mathrm{tr} \left( \left( \id_x^{{\text{in,}1:L\text{,m}}} \otimes \id_{x-1}^\text{m} \otimes \left(\rho_x^{\text{out}}- \sigma_x^\text{out}\right)\right)U_{m_{L+1}\,x}^{L+1}\dots U_{j+1\,x}^{l}\left[ K_{j\,x}^l,\right.\right.\\
	&\left. \left. U_{j\,x}^l\dots U_{1\,x}^1 \left( \rho_x^\text{in}\otimes \rho_{x-1}^\text{m} \otimes \ket{0\dots 0}_\text{1:L+1}^x \bra{0\dots 0}\right)U_{1\, x}^{1^\dagger}\dots U_{j\, x}^{l^\dagger}\right] U_{j+1\, x}^{l^\dagger} \dots U_{m_{L+1}\, x}^{L+1^\dagger} \right)\\
	&+\sum_{z=1}^{x-1}\mathrm{tr}\Bigg(\left(\id_{x:z}^{\text{in,}1:L\text{,m}}\otimes \id_{z-1}^\text{m}\otimes \id_{x-1:z}^\text{out}\otimes\left(\rho_x^{\text{out}}- \sigma_x^\text{out}\right)\right)\mathcal{U}_x\dots \mathcal{U}_{z+1}U_{m_{L+1}\, z}^{L+1}\dots U_{j+1\, z}^{l} \nonumber\\
	&\cdot  \Bigg[ K_{j\, z}^l, U_{j\, z}^{l} \dots U_{1\,z}^1   \left( \bigotimes_{y=x}^z\rho_y^{\text{in}} \otimes \rho_{z-1}^\text{m}  \otimes \ket{0\dots 0}_{ 1:L+1}^{z:x}\bra{0\dots 0}\right) U_{1\, z}^{1^\dagger }\dots U_{j\, z}^{l^\dagger }\Bigg]\\
	&\cdot U_{j+1\, z}^{l^\dagger }\dots U_{m_{L+1}\, z}^{L+1^\dagger } \U_{z+1}^\dagger\dots \mathcal{U}_x^\dagger \Bigg) \Bigg].
\end{align*}
We can see, that this has the same form as before with pure states except that we have \(\rho_x^{\text{out}}- \sigma_x^\text{out}\) instead of \(\ket{\phi_x^\text{out}}\bra{\phi_x^\text{out}}\). The other difference is, that we have to minimise and not maximise the cost function, hence with
\begin{align*}
	\chi_x^\text{out}=&\sigma_x^\text{out}-\rho_x^\text{out}\\
	\rho_x^{l-1}=&\mathcal{E}^{l-1}\left(\dots \mathcal{E}^1\left( \rho_x^\text{in}\otimes \rho_{x-1}^\text{m} \right)\dots\right)\\
	\chi_x^l=& \mathcal{F}^{l+1}\left(\dots\mathcal{F}^{L+1}\left(\id_{x}^{\text{m}} \otimes \chi_x^\text{out}\right)\dots\right)\\
	\omega_{zx}^l=&\mathcal{F}_z^{l+1}\big(\dots \mathcal{F}_z^{L+1}\big(\mathcal{S}_{z+1}\big(\dots \mathcal{S}_{x-1}\big(\tilde{\mathcal{S}}_x\big(\chi_x^\text{out}\big)\big)\dots\big)\big)\dots\big)\\
	M_{j\,x}^l=&  \Big[ U_{j}^l\dots U_{1}^l \Big( \rho_x^{l-1}\otimes \ket{0\dots 0}_l\bra{0\dots 0} \Big) U_{1}^{l^\dagger} \dots U_{j}^{l^\dagger}, U_{j+1}^{l^\dagger} \dots U_{m_{l}}^{l^\dagger} \Big(\id^{l-1}\otimes \chi_x^l\Big) U_{m_{l}}^{l}\dots U_{j+1}^{l}  \Big]\\
	&+\sum_{z=1}^{x-1}\Big[U_{j}^l\dots U_{1}^l \Big( \rho_z^{l-1}\otimes \ket{0\dots 0}_l \bra{0\dots 0} \Big) U_{1}^{l^\dagger} \dots U_{j}^{l^\dagger},U_{j+1}^{l^\dagger} \dots U_{m_{l}}^{l^\dagger} \Big(\id^{l-1}\otimes\omega_{zx}^l\Big) U_{m_{l}}^{l}\dots U_{j+1}^{l} \Big]
\end{align*}
we get
\begin{align*}
	\delta C
	=&-\frac{2i}{N}\sum_{x=1}^{N}\sum_{l=1}^{L+1}\sum_{j=1}^{m_l} \mathrm{tr}\Big(M_{j\,x}^lK_j^l\Big).
\end{align*}
The matrices \(K_j^l\) that minimize the derivative of the cost function are then given by
\begin{equation}
	K_j^l=\frac{i2^{m_{l-1}+1}\eta}{N}\sum_{x=1}^N \mathrm{tr}_\text{rest}\Big(M_{j\,x}^l\Big).
\end{equation}
Therefore, the optimisation algorithm is the following one.\\[1ex] \vspace{10pt}
{\centering
	\fbox{\label{box2}
		\begin{minipage}{.485\textwidth}
			\scriptsize
			\flushleft	
			{\normalsize \textbf{Box 2: Training algorithm local cost with mixed output for \(M= 1\)}\\\vspace{5pt}}
			\textbf{1. Initialize:}\\
			Choose the perceptron unitaries \(U_j^l\) randomly.\\
			\hspace{2pt}\\
			\textbf{2. Feedforward:} 
			For \(x=1,...,N\) do the following:\\
			\textbf{2.1.} Set
			\(\rho_x^0=\rho_x^\text{in}\otimes\rho_{x-1}^\text{m}\) where \(\rho^\text{m}_0\) is given.\\
			\textbf{2.2.} For \(l=1,...,L+1\) set
			\(\rho_x^l=\mathrm{tr}_{l-1}\left( U^l(\rho_x^{l-1}\otimes \ket{0\dots 0}_{l}\bra{0\dots 0}) U^{l^\dagger} \right).\)\\
			\textbf{2.3.} Set
			\(\rho_x^\text{m}=\mathrm{tr}_\text{out}^x(\rho_x^{L+1})\)
			and 
			\(\rho_x^\text{out}=\mathrm{tr}_\text{m}^x(\rho_x^{L+1}).\)\\
			\hspace{2pt}\\
			\textbf{3. Feedbackward:} For \(x=N,...,1\) do the following:\\
			\textbf{3.1.} Set \(\chi_x^{L+1}=\id ^\text{m}_x\otimes  \left(\sigma_x^\text{out}- \rho_x^{\text{out}}\right)\).\\
			\textbf{3.2.} For \(l=L,...,1\) set
			\(\chi_x^l=\mathrm{tr}_{l+1}\left( (\id ^l\otimes \ket{0\dots 0}_{l+1}\bra{0\dots 0}) U^{l+1^\dagger} (\id ^l\otimes \chi_x^{l+1})  U^{l+1} \right)\).\\
			\textbf{3.3.} Set 
			\(\omega_{xx}^\text{m}=\mathrm{tr}_{1,\text{in}}\left( (\rho_x^\text{in}\otimes \id _{x-1}^m\otimes \ket{0\dots 0}_{1}\bra{0\dots 0}) U^{1^\dagger} (\id _x^0\otimes \chi_x^{1})  U^{1} \right)\).\\
			\textbf{3.4.} For \(z=x-1,...,1\) do the following steps:\\
			\textbf{3.4.1.} Set \(\omega_{zx}^{L+1}=\omega_{z+1\,x}^\text{m}\otimes \id ^\text{out}\).\\
			\textbf{3.4.2.} For \(l=L,...,1\) set
			\(\omega_{zx}^l=\mathrm{tr}_{l+1}\left( (\id ^l\otimes \ket{0\dots 0}_{l+1}\bra{0\dots 0}) U^{l+1^\dagger}(\id ^l\otimes \omega_{zx}^{l+1})   U^{l+1} \right).\)\\
			\textbf{3.4.3.} Set 
			\(\omega_{zx}^\text{m}=\mathrm{tr}_{1,\text{in}}\left( (\rho_x^\text{in}\otimes \id _{x-1}^m\otimes \ket{0\dots 0}_{1}\bra{0\dots 0}) U^{1^\dagger}  (\id _x^0\otimes \omega_{zx}^{1})  U^{1} \right).\)\\
			\textbf{4. Update the network:}  For \(l=1,...,L+1\) and \(j=1,...,m_l\) do the following: 
			
		\end{minipage}
		\hfill
		\begin{minipage}{.485\textwidth}
			\scriptsize
			\flushleft	
			%\phantom{\textbf{Box 1: Training algorithm}\\}
			
			\textbf{4.1.} For \(x=1,...,N\) set
			\begin{align*}
				M_{j\,x}^l=& \Bigg[\prod_{k=j}^{1} U_k^l\left( \rho_x^{l-1}\otimes\ket{0\dots 0}_l\bra{0\dots 0}\right) \prod_{k=1}^{j} U_k^{l^\dagger},\\
				&\ \prod_{k=j+1}^{m_l} U_k^{l^\dagger} \left(\id ^{l-1}\otimes \chi_x^l\right)  \prod_{k=m_l}^{j+1} U_k^l \Bigg]\\
				&+ \sum_{z=1}^{x-1} \Bigg[ \prod_{k=j}^{1} U_k^l \left( \rho_z^{l-1}\otimes\ket{0\dots 0}_l\bra{0\dots 0}\right)  \prod_{k=1}^{j} U_k^{l^\dagger} ,\\
				&\ \prod_{k=j+1}^{m_l} U_k^{l^\dagger} \left(\id ^{l-1}\otimes \omega_{zx}^l\right) \prod_{k=m_l}^{j+1} U_k^l \Bigg].
			\end{align*}
			\textbf{4.2.} Set
			\[K_j^l=\frac{i2^{m_{l-1}} \eta}{N}\sum_{x=1}^N \mathrm{tr}_\text{rest}(M_{j\,x}^l),\]
			where \(\mathrm{tr}_\text{rest}\) denotes that the trace is taken over all systems that are unaffected by \(U_j^l\).\\
			\textbf{4.3.} Update each unitary $U_j^l$ according to $U_j^l\rightarrow e^{i\epsilon K_j^l} U_j^l$.\\
			\hspace{2pt}\\
			\textbf{5. Repeat:} Repeat steps 2. to 4. until the cost function reaches its maximum.
		\end{minipage}
	}
}\vspace{10pt}
Note that the \(\chi_x^l\) and \(\omega_{zx}^l\) are not normed, so no states, anymore.
\section{Optimising the Global Cost with Pure Output classically}
\label{globalcostpurestates}
The global cost for pure output states is given by  
\begin{align*}
	C=C_\text{global,pure}&=\left( \bra{\phi_1^\text{out}}\otimes \dots\otimes \bra{\phi_N^\text{out}} \right) \mathrm{tr}_\text{m}^N \left(\rho^\text{OUT}\right)  \left( \ket{\phi_1^\text{out}}\otimes \dots\otimes \ket{\phi_N^\text{out}} \right)\\
	&=\mathrm{tr}\left( \left( \ket{\phi_1^\text{out}}\bra{\phi_1^\text{out}}\otimes \dots \otimes \ket{\phi_N^\text{out}}\bra{\phi_N^\text{out}} \otimes \id_N^\text{m}   \right) \rho^\text{OUT} \right)\\
	&=\mathrm{tr}\left( \left( \ket{\phi_1^\text{out}}\bra{\phi_1^\text{out}}\otimes \dots \otimes \ket{\phi_N^\text{out}}\bra{\phi_N^\text{out}} \otimes \id_N^\text{m}   \right) \left(\N_{\U_N}\circ \dots \circ \N_{\U_1}\right) \left(\rho_0^\text{m}\otimes \rho_1^\text{in}\otimes \dots \otimes \rho_N^\text{in}\right) \right).
\end{align*}
With \(\sigma_i^\text{out}= \ket{\phi_i^\text{out}}\bra{\phi_i^\text{out}} \),
\[ \mathcal{N}_{\U}\left( \rho^{0} \right) =\mathrm{tr}^{0:L}(\mathcal{U}(\rho^{0}\otimes \ket{0\dots 0}_{1:L+1}\bra{0\dots 0})\mathcal{U}^\dagger) \]
and hence
\begin{align*}
	\left(\N_{\U_N}\circ \dots \circ \N_{\U_1}\right) \left(\rho_0^\text{m}\otimes \rho_1^\text{in}\otimes \dots \otimes \rho_N^\text{in}\right) =\mathrm{tr}^{0:L}_{1:N}(\mathcal{U}_N\dots \U_1( \rho_0^\text{m}\otimes \rho_1^\text{in}\otimes \dots \otimes \rho_N^\text{in} \otimes \ket{0\dots 0}_{1:L+1}^{1:x}\bra{0\dots 0})\mathcal{U}_1^\dagger \dots \U_N^\dagger )
\end{align*}
we get
\begin{align*}
	C=&\mathrm{tr}\left( \left( \id_{1:N}^{0:L} \otimes \sigma_1^\text{out}\otimes \dots \otimes \sigma_N^\text{out} \otimes \id_N^\text{m}   \right) \mathcal{U}_N\dots \U_1( \rho_0^\text{m}\otimes \rho_1^\text{in}\otimes \dots \otimes \rho_N^\text{in} \otimes \ket{0\dots 0}_{1:L+1}^{1:x}\bra{0\dots 0})\mathcal{U}_1^\dagger \dots \U_N^\dagger \right)\\
	=&\mathrm{tr}\Big( \left( \id_{1:N}^{0:L} \otimes \id_1^\text{out} \otimes \sigma_2^\text{out} \otimes \dots \otimes \sigma_N^\text{out} \otimes \id_N^\text{m}   \right)  \left( \id_{1:N}^{0:L} \otimes \sigma_1^\text{out}\otimes \id^\text{out}_{2:N}\otimes \id_N^\text{m}   \right) \left(\U_N\dots \U_2 \otimes \id_{0}^\text{m}\otimes \id_1^{\text{in},1:L,\text{out}}\right) \left(\U_1\otimes \id_{2:N}^{\text{in},1:L+1}\right)\\
	&\cdot \left( \rho_0^\text{m}\otimes \rho_1^\text{in} \ket{0\dots 0}_{1:L+1}^{1}\bra{0\dots 0}\otimes \id_{2:N}^{\text{in},1:L+1}\right)  \left( \id_0^\text{m}\otimes \id_1^{\text{in},1:L+1}\otimes \rho_2^\text{in}\otimes \dots \otimes \rho_N^\text{in} \otimes \ket{0\dots 0}_{1:L+1}^{2:x}\bra{0\dots 0}\right) \\
	&\cdot \left(\U_1^\dagger\otimes \id_{2:N}^{\text{in},1:L+1}\right)  \left(\U_2^\dagger\dots \U_N^\dagger \otimes \id_{0}^\text{m}\otimes \id_1^{\text{in},1:L,\text{out}}\right) \Big)\\
	=& \mathrm{tr}\Big( \left( \id_{1:N}^{0:L} \otimes \id_1^\text{out} \otimes \sigma_2^\text{out} \otimes \dots \otimes \sigma_N^\text{out} \otimes \id_N^\text{m}   \right)   \left(\U_N\dots \U_2 \otimes \id_{0}^\text{m}\otimes \id_1^{\text{in},1:L,\text{out}}\right)  \left( \id_{1:N}^{0:L} \otimes \sigma_1^\text{out}\otimes \id^\text{out}_{2:N}\otimes \id_N^\text{m}   \right) \left(\U_1\otimes \id_{2:N}^{\text{in},1:L+1}\right)  \\
	&\cdot \left( \rho_0^\text{m}\otimes \rho_1^\text{in} \otimes \ket{0\dots 0}_{1:L+1}^{1}\bra{0\dots 0}\otimes \id_{2:N}^{\text{in},1:L+1}\right) \left(\U_1^\dagger\otimes \id_{2:N}^{\text{in},1:L+1}\right)  \\
	&\cdot   \left( \id_0^\text{m}\otimes \id_1^{\text{in},1:L+1}\otimes \rho_2^\text{in}\otimes \dots \otimes \rho_N^\text{in} \otimes \ket{0\dots 0}_{1:L+1}^{2:x}\bra{0\dots 0}\right) \left(\U_2^\dagger\dots \U_N^\dagger \otimes \id_{0}^\text{m}\otimes \id_1^{\text{in},1:L,\text{out}}\right) \Big)\\
	=& \mathrm{tr}\Big( \left( \id_{2:N}^{0:L} \otimes \sigma_2^\text{out}\otimes \dots \otimes \sigma_N^\text{out} \otimes \id_N^\text{m}   \right)   \U_N\dots \U_2   \Big( \mathrm{tr}_0^\text{m}\Big( \mathrm{tr}_1^{\text{in}, 1:L,\text{out}}\Big(\left( \sigma_1^\text{out}\otimes \id_{0,1}^{\text{m}} \otimes \id^{\text{in},1:L}_{1}  \right) \U_1  \\
	&\cdot \left( \rho_0^\text{m}\otimes \rho_1^\text{in} \otimes\ket{0\dots 0}_{1:L+1}^{1}\bra{0\dots 0}\right)   \U_1^\dagger\Big) \Big) \otimes \id_{2:N}^{\text{in},1:L+1}  \Big) \left( \id_1^\text{m}\otimes  \rho_2^\text{in}\otimes \dots \otimes \rho_N^\text{in} \otimes \ket{0\dots 0}_{1:L+1}^{2:x}\bra{0\dots 0}\right) \U_2^\dagger\dots \U_N^\dagger  \Big).
\end{align*}
By setting \(\tau_0^\text{m}=\rho_0^\text{m}\), \(\tau_x^0=\tau_{x-1}^\text{m}\otimes \rho_x^\text{in}\),
\[\tau_x^l=\mathcal{E}^l\left(\tau_x^{l-1}\right)\]
and 
\begin{align}
	\tau_x^\text{m}&=\mathrm{tr}_{x-1}^\text{m}\left(\mathrm{tr}_x^{\text{in},1:L,\text{out}}\left( \left(\sigma_x^\text{out}\otimes \id_{x-1,x}^\text{m}\otimes \id_x^{\text{in},1:L}\right) \U_x \left(\tau_{x-1}^\text{m}\otimes \rho_x^\text{in} \otimes \ket{0...0}_x^{1:L+1}\bra{0...0}\right) \U_x^\dagger\right)\right)\\
	&=\mathrm{tr}_x^\text{out}\left(\left(\sigma_x^\text{out}\otimes \id_x^\text{m}\right)\tau_x^{L+1}\right)
\end{align}
for \(x=1,...,N\) we obtain
\begin{align}
	C=& \mathrm{tr}\Big( \left( \id_{2:N}^{0:L} \otimes \sigma_2^\text{out}\otimes \dots \otimes \sigma_N^\text{out} \otimes \id_N^\text{m}   \right)   \U_N\dots \U_2   \Big( \tau_1^\text{m} \otimes \id_{2:N}^{\text{in},1:L+1}  \Big) \nonumber \\
	&\cdot\left( \id_1^\text{m}\otimes  \rho_2^\text{in}\otimes \dots \otimes \rho_N^\text{in} \otimes \ket{0\dots 0}_{1:L+1}^{2:x}\bra{0\dots 0}\right) \U_2^\dagger\dots \U_N^\dagger  \Big) \nonumber \\
	=& \mathrm{tr}\Big( \left( \id_{2:N}^{0:L} \otimes \sigma_2^\text{out}\otimes \dots \otimes \sigma_N^\text{out} \otimes \id_N^\text{m}   \right)   \U_N\dots \U_2   \left( \tau_1^\text{m}\otimes  \rho_2^\text{in}\otimes \dots \otimes \rho_N^\text{in} \otimes \ket{0\dots 0}_{1:L+1}^{2:x}\bra{0\dots 0}\right) \U_2^\dagger\dots \U_N^\dagger  \Big) \nonumber \\
	=& \dots\nonumber\\
	=&\mathrm{tr}\left(\tau_N^\text{m}\right).
\end{align}
Note that the \(\tau_x^l\) do not have trace 1. By using the same update rule as for the local cost and the same argument we get
\begin{align*}
	\Delta \rho^\text{OUT}=& i \epsilon \sum_{z=1}^{N} \sum_{l=1}^{L+1} \sum_{j=1}^{m_l} \mathrm{tr}_{0:N-1}^{\text{m}}\Big(   \mathrm{tr}_{1:N}^{\text{in},1:L}\Big( \U_N \dots \U_{z+1} U^{L+1}_{m_{L+1}\, z}\dots U^l_{j+1\,z} \Big[ K^l_{j\,z}, U^l_{j\,z} \dots U^1_{1\,z} \U_{z-1}\dots \U_1 \\
	&\left( \rho_0^\text{m}\otimes \rho_1^\text{in}\otimes \dots \otimes\rho_N^\text{in} \otimes \ket{0\dots 0}_{1:N}^{1:L+1}\bra{0\dots 0} \right)\U_1^\dagger \dots \U_{z-1}^\dagger U^{1^\dagger}_{1\, z}\dots U^{l^\dagger}_{j\, z}    \Big] U^{l^\dagger}_{j+1\, z}\dots U^{L+1^\dagger}_{m_{L+1}\, z} \U_{z+1}^\dagger \dots \U_N^\dagger \Big) \Big).
\end{align*}
Therefore we get in the same way as before 
\begin{align*}
	\delta C=& i  \sum_{z=1}^{N} \sum_{l=1}^{L+1} \sum_{j=1}^{m_l} \mathrm{tr}_{0:N}^{\text{m}}\Big(   \mathrm{tr}_{1:N}^{\text{in},1:L,\text{out}}\Big( \U_N \dots \U_{z+1} U^{L+1}_{m_{L+1}\, z}\dots U^l_{j+1\,z} \Big[ K^l_{j\,z}, U^l_{j\,z} \dots U^1_{1\,z} \U_{z-1}\dots \U_1 \\
	&\left( \rho_0^\text{m}\otimes \rho_1^\text{in}\otimes \dots \otimes\rho_N^\text{in} \otimes \ket{0\dots 0}_{1:N}^{1:L+1}\bra{0\dots 0} \right)\U_1^\dagger \dots \U_{z-1}^\dagger U^{1^\dagger}_{1\, z}\dots U^{l^\dagger}_{j\, z}    \Big] U^{l^\dagger}_{j+1\, z}\dots U^{L+1^\dagger}_{m_{L+1}\, z} \U_{z+1}^\dagger \dots \U_N^\dagger \\
	&\left(\sigma_1^\text{out}\otimes\dots\otimes \sigma_N^\text{out} \otimes \id_{0:N}^\text{m}\otimes \id_{1:N}^{\text{in},1:L}\right)\Big) \Big).
\end{align*}
With the cyclic rule of trace it follows that
\begin{align*}
	\delta C=& i  \sum_{z=1}^{N} \sum_{l=1}^{L+1} \sum_{j=1}^{m_l} \mathrm{tr}_{0:N}^{\text{m}}\Big(   \mathrm{tr}_{1:N}^{\text{in},1:L,\text{out}}\Big( U^{l+1^\dagger}_z\dots U^{L+1^\dagger}_z\U_{z+1}^\dagger \dots \U_N^\dagger \left(\sigma_1^\text{out}\otimes\dots \otimes \sigma_N^\text{out} \otimes \id_{0:N}^\text{m}\otimes \id_{1:N}^{\text{in},1:L}\right)\\
	& \U_N \dots \U_{z+1} U^{L+1}_z\dots U^{l+1}_z U^{l}_{m_{l}\, z} \dots U^l_{j+1\,z} \Big[ K^l_{j\,z}, U^l_{j\,z} \dots U^1_{1\,z} \U_{z-1}\dots \U_1 \\
	&\left( \rho_0^\text{m}\otimes \rho_1^\text{in}\otimes \dots \otimes\rho_N^\text{in} \otimes \ket{0\dots 0}_{1:N}^{1:L+1}\bra{0\dots 0} \right)\U_1^\dagger \dots \U_{z-1}^\dagger U^{1^\dagger}_{1\, z}\dots U^{l^\dagger}_{j\, z}    \Big] U^{l^\dagger}_{j+1\, z}\dots U^{l^\dagger}_{m_{l}\, z}  \Big)\Big)\\
	=& i  \sum_{z=1}^{N} \sum_{l=1}^{L+1} \sum_{j=1}^{m_l} \mathrm{tr}_{0:N}^{\text{m}}\Big(   \mathrm{tr}_{1:N}^{\text{in},1:L,\text{out}}\Big( \Big( \sigma_1^\text{out}\otimes\dots \otimes \sigma_{z-1}^\text{out} \otimes U^{l+1^\dagger}_z\dots U^{L+1^\dagger}_z\U_{z+1}^\dagger \dots \U_N^\dagger \Big(\sigma_z^\text{out}\otimes\dots \otimes \sigma_N^\text{out} \\
	&\otimes \id_{z-1:N}^\text{m}\otimes \id_{z:N}^{\text{in},1:L}\Big) \U_N \dots \U_{z+1} U^{L+1}_z\dots U^{l+1}_z  \otimes \id_{0:z-2}^\text{m}\otimes \id_{1:z-1}^{\text{in},1:L} \Big) \Big(U^{l}_{m_{l}\, z} \dots U^l_{j+1\,z} \\
	&\Big[ K^l_{j\,z}, U^l_{j\,z} \dots U^1_{1\,z} \U_{z-1}\dots \U_1 \left( \rho_0^\text{m}\otimes \rho_1^\text{in}\otimes \dots \otimes\rho_z^\text{in} \otimes \ket{0\dots 0}_{1:z}^{1:L+1}\bra{0\dots 0} \right)\U_1^\dagger \dots \U_{z-1}^\dagger U^{1^\dagger}_{1\, z}\dots U^{l^\dagger}_{j\, z}    \Big] \\
	&U^{l^\dagger}_{j+1\, z}\dots U^{l^\dagger}_{m_{l}\, z}\otimes \rho_{z+1}^\text{in}\otimes \dots \otimes\rho_N^\text{in} \otimes \ket{0\dots 0}_{z+1:N}^{1:L+1}\bra{0\dots 0} \Big)  \Big)\Big).
\end{align*}
By writing 
\begin{align}
	A=&U^{l+1^\dagger}_z\dots U^{L+1^\dagger}_z\U_{z+1}^\dagger \dots \U_N^\dagger \Big(\sigma_z^\text{out}\otimes\dots \otimes \sigma_N^\text{out}\otimes \id_{z-1:N}^\text{m}\otimes \id_{z:N}^{\text{in},1:L}\Big) \U_N \dots \U_{z+1} U^{L+1}_z\dots U^{l+1}_z,\\
	B=&U^{l}_{m_{l}\, z} \dots U^l_{j+1\,z} \Big[ K^l_{j\,z}, U^l_{j\,z} \dots U^1_{1\,z} \U_{z-1}\dots \U_1 \left( \rho_0^\text{m}\otimes \rho_1^\text{in}\otimes \dots \otimes\rho_z^\text{in} \otimes \ket{0\dots 0}_{1:z}^{1:L+1}\bra{0\dots 0} \right)\nonumber\\
	&\U_1^\dagger \dots \U_{z-1}^\dagger U^{1^\dagger}_{1\, z}\dots U^{l^\dagger}_{j\, z}    \Big] U^{l^\dagger}_{j+1\, z}\dots U^{l^\dagger}_{m_{l}\, z}
\end{align}
in shorthand we obtain
\begin{align*}
	\delta C	=& i  \sum_{z=1}^{N} \sum_{l=1}^{L+1} \sum_{j=1}^{m_l} \mathrm{tr}_{0:N}^{\text{m}}\Big(   \mathrm{tr}_{1:N}^{\text{in},1:L,\text{out}}\Big( \Big( \sigma_1^\text{out}\otimes\dots \otimes \sigma_{z-1}^\text{out} \otimes A  \otimes \id_{0:z-2}^\text{m}\otimes \id_{1:z-1}^{\text{in},1:L} \Big)\\
	& \Big(B\otimes \rho_{z+1}^\text{in}\otimes \dots \otimes\rho_N^\text{in} \otimes \ket{0\dots 0}_{z+1:N}^{1:L+1}\bra{0\dots 0} \Big)  \Big)\Big)\\
	=& i  \sum_{z=1}^{N} \sum_{l=1}^{L+1} \sum_{j=1}^{m_l} \mathrm{tr}_{0:N}^{\text{m}}\Big(   \mathrm{tr}_{1:N}^{\text{in},1:L,\text{out}}\Big( \Big(  A  \otimes \id_{0:z-2}^\text{m}\otimes \id_{1:z-1}^{\text{in},1:L,\text{out}} \Big) \Big( \sigma_1^\text{out}\otimes\dots \otimes \sigma_{z-1}^\text{out}   \otimes \id_{0:N}^\text{m}\otimes \id_{1:N}^{\text{in},1:L} \otimes \id_{z:N}^\text{out} \Big) \\
	& \Big(B\otimes \id_{z+1:N}^{\text{in},1:L+1}\Big) \Big(\id_{0:z}^\text{m}\otimes \id_{1:z}^{\text{in},1:L,\text{out}}\otimes \rho_{z+1}^\text{in}\otimes \dots \otimes\rho_N^\text{in} \otimes \ket{0\dots 0}_{z+1:N}^{1:L+1}\bra{0\dots 0} \Big) \Big)\Big)\\
	=& i  \sum_{z=1}^{N} \sum_{l=1}^{L+1} \sum_{j=1}^{m_l} \mathrm{tr}_{0:N}^{\text{m}}\Big(   \mathrm{tr}_{1:N}^{\text{in},1:L,\text{out}}\Big( \Big(\id_{0:z}^\text{m}\otimes \id_{1:z}^{\text{in},1:L,\text{out}}\otimes \rho_{z+1}^\text{in}\otimes \dots \otimes\rho_N^\text{in} \otimes \ket{0\dots 0}_{z+1:N}^{1:L+1}\bra{0\dots 0} \Big)  \\
	& \Big(  A  \otimes \id_{0:z-2}^\text{m}\otimes \id_{1:z-1}^{\text{in},1:L,\text{out}} \Big)  \Big( \sigma_1^\text{out}\otimes\dots \otimes \sigma_{z-1}^\text{out}   \otimes \id_{0:N}^\text{m}\otimes \id_{1:N}^{\text{in},1:L} \otimes \id_{z:N}^\text{out} \Big) \Big(B\otimes \id_{z+1:N}^{\text{in},1:L+1}\Big)  \Big)\Big)\\
	=& i  \sum_{z=1}^{N} \sum_{l=1}^{L+1} \sum_{j=1}^{m_l} \mathrm{tr}_{z-1:N}^{\text{m}}\Big(   \mathrm{tr}_{z:N}^{\text{in},1:L,\text{out}}\Big( \Big(\id_{z-1,z}^\text{m}\otimes \id_{z}^{\text{in},1:L,\text{out}}\otimes \rho_{z+1}^\text{in}\otimes \dots \otimes\rho_N^\text{in} \otimes \ket{0\dots 0}_{z+1:N}^{1:L+1}\bra{0\dots 0} \Big)  A\\
	& \Big(\mathrm{tr}_{0:z-2}^\text{m}\Big(\mathrm{tr}_{1:z-1}^{\text{in},1:L,\text{out}}\Big(\Big( \sigma_1^\text{out}\otimes\dots \otimes \sigma_{z-1}^\text{out}   \otimes \id_{0:N}^\text{m}\otimes \id_{1:N}^{\text{in},1:L} \otimes \id_{z:N}^\text{out} \Big) B\Big)\Big) \otimes \id_{z+1:N}^{\text{in},1:L+1} \Big) \Big)\Big).
\end{align*}
By substituting 
\begin{equation}
	E=\mathrm{tr}_{0:z-2}^\text{m}\Big(\mathrm{tr}_{1:z-1}^{\text{in},1:L,\text{out}}\Big(\Big( \sigma_1^\text{out}\otimes\dots \otimes \sigma_{z-1}^\text{out}   \otimes \id_{0:N}^\text{m}\otimes \id_{1:N}^{\text{in},1:L} \otimes \id_{z:N}^\text{out} \Big) B\Big)\Big)
\end{equation}
we get 
\begin{align*}
	\delta C	=& i  \sum_{z=1}^{N} \sum_{l=1}^{L+1} \sum_{j=1}^{m_l} \mathrm{tr}_{z-1:N}^{\text{m}}\Big(   \mathrm{tr}_{z:N}^{\text{in},1:L,\text{out}}\Big( \Big(\id_{z-1,z}^\text{m}\otimes \id_{z}^{\text{in},1:L,\text{out}}\otimes \rho_{z+1}^\text{in}\otimes \dots \otimes\rho_N^\text{in} \otimes \ket{0\dots 0}_{z+1:N}^{1:L+1}\bra{0\dots 0} \Big)  A\\
	& \Big(E \otimes \id_{z+1:N}^{\text{in},1:L+1} \Big) \Big)\Big)\\
	=& i  \sum_{z=1}^{N} \sum_{l=1}^{L+1} \sum_{j=1}^{m_l} \mathrm{tr}_{z-1:z}^{\text{m}}\Big(   \mathrm{tr}_{z}^{\text{in},1:L,\text{out}}\Big( \mathrm{tr}_{z+1:N}^{\text{in},1:L,\text{out,m}}\Big(\Big(\id_{z-1,z}^\text{m}\otimes \id_{z}^{\text{in},1:L,\text{out}}\otimes \rho_{z+1}^\text{in}\otimes \dots \otimes\rho_N^\text{in} \\
	&\otimes \ket{0\dots 0}_{z+1:N}^{1:L+1}\bra{0\dots 0} \Big)  A\Big)E \Big)\Big).
\end{align*}
By further substituting
\begin{equation}
	D=\mathrm{tr}_{z+1:N}^{\text{in},1:L,\text{out,m}}\Big(\Big(\id_{z-1,z}^\text{m}\otimes \id_{z}^{\text{in},1:L,\text{out}}\otimes \rho_{z+1}^\text{in}\otimes \dots \otimes\rho_N^\text{in}\otimes \ket{0\dots 0}_{z+1:N}^{1:L+1}\bra{0\dots 0} \Big)  A\Big)
\end{equation}
we obtain
\begin{align*}
	\delta C	=& i  \sum_{z=1}^{N} \sum_{l=1}^{L+1} \sum_{j=1}^{m_l} \mathrm{tr}_{z-1:z}^{\text{m}}\Big(   \mathrm{tr}_{z}^{\text{in},1:L,\text{out}}\Big( DE \Big)\Big).
\end{align*}
By resubstituting \(B\) in \(E\) we get
\begin{align}
	E=&\mathrm{tr}_{0:z-2}^\text{m}\Big(\mathrm{tr}_{1:z-1}^{\text{in},1:L,\text{out}}\Big(\Big( \sigma_1^\text{out}\otimes\dots \otimes \sigma_{z-1}^\text{out}   \otimes \id_{0:N}^\text{m}\otimes \id_{1:N}^{\text{in},1:L} \otimes \id_{z:N}^\text{out} \Big) U^{l}_{m_{l}\, z} \dots U^l_{j+1\,z} \Big[ K^l_{j\,z}, U^l_{j\,z} \dots U^1_{1\,z} \nonumber\\
	& \U_{z-1}\dots \U_1 \left( \rho_0^\text{m}\otimes \rho_1^\text{in}\otimes \dots \otimes\rho_z^\text{in} \otimes \ket{0\dots 0}_{1:z}^{1:L+1}\bra{0\dots 0} \right)\U_1^\dagger \dots \U_{z-1}^\dagger U^{1^\dagger}_{1\, z}\dots U^{l^\dagger}_{j\, z}    \Big] U^{l^\dagger}_{j+1\, z}\dots U^{l^\dagger}_{m_{l}\, z}  \Big)\Big) \nonumber\\
	=&\mathrm{tr}_{1:z-2}^\text{m}\Big(\mathrm{tr}_{2:z-1}^{\text{in},1:L,\text{out}}\Big(\Big( \sigma_2^\text{out}\otimes\dots \otimes \sigma_{z-1}^\text{out}   \otimes \id_{0:N}^\text{m}\otimes \id_{1:N}^{\text{in},1:L} \otimes \id_{z:N}^\text{out} \Big) U^{l}_{m_{l}\, z} \dots U^l_{j+1\,z} \Big[ K^l_{j\,z}, U^l_{j\,z} \dots U^1_{1\,z} \U_{z-1}\dots \U_2 \nonumber\\
	&  \Big( \mathrm{tr}_{0}^\text{m} \Big( \mathrm{tr}_{1:z-1}^{\text{in},1:L,\text{out}} \Big(\left(\sigma_1^\text{out}\otimes \id_1^{\text{in},1:L}\otimes \id_{0,1}^\text{m}\right)\U_1 \left( \rho_0^\text{m}\otimes \rho_1^\text{in} \otimes \ket{0\dots 0}_{1}^{1:L+1}\bra{0\dots 0} \right)\U_1^\dagger\Big)\Big) \otimes \rho_2^\text{in}\otimes \dots \otimes\rho_z^\text{in}\nonumber \\
	&\otimes \ket{0\dots 0}_{2:z}^{1:L+1}\bra{0\dots 0} \Big)\U_2^\dagger \dots \U_{z-1}^\dagger U^{1^\dagger}_{1\, z}\dots U^{l^\dagger}_{j\, z}    \Big] U^{l^\dagger}_{j+1\, z}\dots U^{l^\dagger}_{m_{l}\, z}  \Big)\Big) \nonumber\\
	=&\mathrm{tr}_{1:z-2}^\text{m}\Big(\mathrm{tr}_{2:z-1}^{\text{in},1:L,\text{out}}\Big(\Big( \sigma_2^\text{out}\otimes\dots \otimes \sigma_{z-1}^\text{out}   \otimes \id_{0:N}^\text{m}\otimes \id_{1:N}^{\text{in},1:L} \otimes \id_{z:N}^\text{out} \Big) U^{l}_{m_{l}\, z} \dots U^l_{j+1\,z} \Big[ K^l_{j\,z}, U^l_{j\,z} \dots U^1_{1\,z} \U_{z-1}\dots \U_2 \nonumber\\
	&  \Big( \tau_1^\text{m} \otimes \rho_2^\text{in}\otimes \dots \otimes\rho_z^\text{in}\nonumber\otimes \ket{0\dots 0}_{2:z}^{1:L+1}\bra{0\dots 0} \Big)\U_2^\dagger \dots \U_{z-1}^\dagger U^{1^\dagger}_{1\, z}\dots U^{l^\dagger}_{j\, z}    \Big] U^{l^\dagger}_{j+1\, z}\dots U^{l^\dagger}_{m_{l}\, z}  \Big)\Big) \nonumber\\
	=& U^{l}_{m_{l}\, z} \dots U^l_{j+1\,z} \Big[ K^l_{j\,z}, U^l_{j\,z} \dots U^1_{1\,z} \Big( \tau_{z-1}^\text{m} \otimes\rho_z^\text{in}\otimes \ket{0\dots 0}_{z}^{1:L+1}\bra{0\dots 0} \Big) U^{1^\dagger}_{1\, z}\dots U^{l^\dagger}_{j\, z}    \Big] U^{l^\dagger}_{j+1\, z}\dots U^{l^\dagger}_{m_{l}\, z} .
\end{align}
Also by resubstituting \(A\) in \(D\) it is
\begin{align}
	D=&\mathrm{tr}_{z+1:N}^{\text{in},1:L,\text{out,m}}\Big(\Big(\id_{z-1,z}^\text{m}\otimes \id_{z}^{\text{in},1:L,\text{out}}\otimes \rho_{z+1}^\text{in}\otimes \dots \otimes\rho_N^\text{in}\otimes \ket{0\dots 0}_{z+1:N}^{1:L+1}\bra{0\dots 0} \Big)   U^{l+1^\dagger}_z\dots U^{L+1^\dagger}_z\nonumber\\
	&\U_{z+1}^\dagger \dots \U_N^\dagger \Big(\sigma_z^\text{out}\otimes\dots \otimes \sigma_N^\text{out}\otimes \id_{z-1:N}^\text{m}\otimes \id_{z:N}^{\text{in},1:L}\Big) \U_N \dots \U_{z+1} U^{L+1}_z\dots U^{l+1}_z \Big) \nonumber \\
	=&\mathrm{tr}_{z+1:N-1}^{\text{in},1:L,\text{out,m}}\Big(\Big(\id_{z-1,z}^\text{m}\otimes \id_{z}^{\text{in},1:L,\text{out}}\otimes \rho_{z+1}^\text{in}\otimes \dots \otimes\rho_{N-1}^\text{in}\otimes \ket{0\dots 0}_{z+1:N-1}^{1:L+1}\bra{0\dots 0} \Big)   U^{l+1^\dagger}_z\dots U^{L+1^\dagger}_z\nonumber\U_{z+1}^\dagger \dots \U_{N-1}^\dagger \\
	&\Big( \sigma_z^\text{out}\otimes\dots \otimes \sigma_{N-1}^\text{out}\otimes \id_{z-1:N-2}^\text{m}\otimes \id_{z:N-1}^{\text{in},1:L} \otimes \mathrm{tr}_{N}^{\text{in},1:L,\text{out,m}} \Big(\Big(\id_{N-1}^\text{m}\otimes \rho_N^\text{in} \otimes \ket{0\dots 0}_{N}^{1:L+1}\bra{0\dots 0} \Big)\U_N^\dagger \nonumber \\
	&\Big(\sigma_N^\text{out}\otimes \id_{N-1,N}^\text{m}\otimes \id_{N}^{\text{in},1:L}\Big) \U_N\Big)\Big) \U_{N-1}\dots \U_{z+1} U^{L+1}_z\dots U^{l+1}_z \Big) \nonumber .
\end{align}
Let \(\chi_N^\text{m}=\id_N^\text{m}\), \(\chi_x^{L+1}=\chi_{x}^\text{m}\otimes \sigma_x^\text{out}\),
\[\chi_x^l=\mathcal{F}^{l+1}\left(\chi_x^{l+1}\right)\]
and 
\begin{align}
	\chi_x^\text{m}=&\mathrm{tr}_{x+1}^{\text{in},1:L+1} \Big(\Big(\id_{x}^\text{m}\otimes \rho_{x+1}^\text{in} \otimes \ket{0\dots 0}_{x+1}^{1:L+1}\bra{0\dots 0} \Big)\U_{x+1}^\dagger\Big( \chi_{x+1}^\text{m}\otimes \sigma_{x+1}^\text{out}\otimes \id_{x}^\text{m}\otimes \id_{x+1}^{\text{in},1:L}\Big) \U_{x+1}\Big)\\
	=&\mathrm{tr}_{x+1}^{\text{in}} \Big(\Big(\id_x^\text{m}\otimes \rho_x^\text{in}\Big)\chi_{y+1}^0\Big)
\end{align}
for \(x=1,...,N\) which, analogously to \(\tau_x^\text{m}\), do not have trace \(1\).
Hence, we get
\begin{align}
	D=&\mathrm{tr}_{z+1:N-1}^{\text{in},1:L,\text{out,m}}\Big(\Big(\id_{z-1,z}^\text{m}\otimes \id_{z}^{\text{in},1:L,\text{out}}\otimes \rho_{z+1}^\text{in}\otimes \dots \otimes\rho_{N-1}^\text{in}\otimes \ket{0\dots 0}_{z+1:N-1}^{1:L+1}\bra{0\dots 0} \Big)   U^{l+1^\dagger}_z\dots U^{L+1^\dagger}_z\nonumber \\
	&\U_{z+1}^\dagger \dots \U_{N-1}^\dagger \Big( \sigma_z^\text{out}\otimes\dots \otimes \sigma_{N-1}^\text{out} \otimes \chi_{N-1}^\text{m} \otimes \id_{z-1:N-2}^\text{m}\otimes \id_{z:N-1}^{\text{in},1:L} \Big) \U_{N-1}\dots \U_{z+1} U^{L+1}_z\dots U^{l+1}_z \Big) \nonumber \\
	=&\dots \nonumber \\
	=&  U^{l+1^\dagger}_z\dots U^{L+1^\dagger}_z\Big(\id_{z}^{0:L}\otimes \chi_z^\text{m}\otimes\sigma_z^\text{out}\Big) U^{L+1}_z\dots U^{l+1}_z.
\end{align}
By resubstituting \(D\) and \(E\) in \(\delta C\), we obtain
\begin{align*}
	\delta C	=& i  \sum_{z=1}^{N} \sum_{l=1}^{L+1} \sum_{j=1}^{m_l} \mathrm{tr}_{z-1:z}^{\text{m}}\Big(   \mathrm{tr}_{z}^{\text{in},1:L,\text{out}}\Big(   U^{l+1^\dagger}_z\dots U^{L+1^\dagger}_z\Big(\id_{z}^{0:L}\otimes \chi_z^\text{m}\otimes\sigma_z^\text{out}\Big) U^{L+1}_z\dots U^{l+1}_z U^{l}_{m_{l}\, z} \dots U^l_{j+1\,z} \\
	&\Big[ K^l_{j\,z}, U^l_{j\,z} \dots U^1_{1\,z} \Big( \tau_{z-1}^\text{m} \otimes\rho_z^\text{in}\otimes \ket{0\dots 0}_{z}^{1:L+1}\bra{0\dots 0} \Big) U^{1^\dagger}_{1\, z}\dots U^{l^\dagger}_{j\, z}    \Big] U^{l^\dagger}_{j+1\, z}\dots U^{l^\dagger}_{m_{l}\, z} \Big)\Big).
\end{align*}
As everything happens in the \(z^\text{th}\) QNN, we can leave the index \(z\) away in the unitaries, identities and traces.
\begin{align*}
	\delta C	=& i  \sum_{z=1}^{N} \sum_{l=1}^{L+1} \sum_{j=1}^{m_l}    \mathrm{tr}^{0:L+1}\Big(   U^{l+1^\dagger}\dots U^{L+1^\dagger}\Big(\id^{0:L}\otimes \chi_z^\text{m}\otimes\sigma_z^\text{out}\Big) U^{L+1}\dots U^{l+1} U^{l}_{m_{l}} \dots U^l_{j+1} \\
	&\Big[ K^l_{j}, U^l_{j} \dots U^1_{1} \Big( \tau_{z-1}^\text{m} \otimes\rho_z^\text{in}\otimes \ket{0\dots 0}^{1:L+1}\bra{0\dots 0} \Big) U^{1^\dagger}_{1}\dots U^{l^\dagger}_{j}    \Big] U^{l^\dagger}_{j+1}\dots U^{l^\dagger}_{m_{l}} \Big)\\
	=& i  \sum_{z=1}^{N} \sum_{l=1}^{L+1} \sum_{j=1}^{m_l}    \mathrm{tr}^{0:L+1}\Big(   \Big(U^{l+1^\dagger}\dots U^{L+1^\dagger}\Big(\id^{l:L}\otimes \chi_z^\text{m}\otimes\sigma_z^\text{out}\Big) U^{L+1}\dots U^{l+1}\otimes \id^{0:l-1} \Big) \Big(U^{l}_{m_{l}} \dots U^l_{j+1} \\
	&\Big[ K^l_{j}, U^l_{j} \dots U^1_{1} \Big( \tau_{z-1}^\text{m} \otimes\rho_z^\text{in}\otimes \ket{0\dots 0}^{1:l}\bra{0\dots 0} \Big) U^{1^\dagger}_{1}\dots U^{l^\dagger}_{j}    \Big] U^{l^\dagger}_{j+1}\dots U^{l^\dagger}_{m_{l}} \otimes \ket{0\dots 0}^{l+1:L+1}\bra{0\dots 0} \Big) \Big)
\end{align*}
By writing
\begin{align}
	G=&U^{l+1^\dagger}\dots U^{L+1^\dagger}\Big(\id^{l:L}\otimes \chi_z^\text{m}\otimes\sigma_z^\text{out}\Big) U^{L+1}\dots U^{l+1},\\
	H=&U^{l}_{m_{l}} \dots U^l_{j+1} \Big[ K^l_{j}, U^l_{j} \dots U^1_{1} \Big( \tau_{z-1}^\text{m} \otimes\rho_z^\text{in}\otimes \ket{0\dots 0}^{1:l}\bra{0\dots 0} \Big) U^{1^\dagger}_{1}\dots U^{l^\dagger}_{j}    \Big] U^{l^\dagger}_{j+1}\dots U^{l^\dagger}_{m_{l}}
\end{align}
in shorthand, we obtain
\begin{align*}
	\delta C	=& i  \sum_{z=1}^{N} \sum_{l=1}^{L+1} \sum_{j=1}^{m_l}    \mathrm{tr}^{0:L+1}\Big(   \Big(G\otimes \id^{0:l-1} \Big) \Big(H \otimes \ket{0\dots 0}^{l+1:L+1}\bra{0\dots 0} \Big) \Big)\\
	=&i  \sum_{z=1}^{N} \sum_{l=1}^{L+1} \sum_{j=1}^{m_l}    \mathrm{tr}^{l:L+1}\Big(   G \Big(\mathrm{tr}^{0:l-1}\left(H\right) \otimes \ket{0\dots 0}^{l+1:L+1}\bra{0\dots 0} \Big) \Big)\\
	=&i  \sum_{z=1}^{N} \sum_{l=1}^{L+1} \sum_{j=1}^{m_l}    \mathrm{tr}^{l:L+1}\Big(   G \Big(\mathrm{tr}^{0:l-1}\left(H\right) \otimes \id^{l+1:L+1}\Big) \Big(\id^l \otimes \ket{0\dots 0}^{l+1:L+1}\bra{0\dots 0} \Big) \Big)\\
	=&i  \sum_{z=1}^{N} \sum_{l=1}^{L+1} \sum_{j=1}^{m_l}    \mathrm{tr}^{l}\Big(  \mathrm{tr}^{l+1:L+1} \Big(\Big(\id^l \otimes \ket{0\dots 0}^{l+1:L+1}\bra{0\dots 0} \Big)  G\Big) \mathrm{tr}^{0:l-1}\left(H\right)  \Big). 
\end{align*}
In the same way as for the local cost, we can do the following steps. It is
\begin{align*}
	\mathrm{tr}^{0:l-1}\left(H\right) =&\mathrm{tr}^{0:l-1}\Big( U^{l}_{m_{l}} \dots U^l_{j+1} \Big[ K^l_{j}, U^l_{j} \dots U^1_{1} \Big( \tau_{z-1}^\text{m} \otimes\rho_z^\text{in}\otimes \ket{0\dots 0}^{1:l}\bra{0\dots 0} \Big) U^{1^\dagger}_{1}\dots U^{l^\dagger}_{j}    \Big] U^{l^\dagger}_{j+1}\dots U^{l^\dagger}_{m_{l}} \Big) \\
	=& \mathrm{tr}^{l-1} \Big( U^{l}_{m_{l}} \dots U^l_{j+1} \Big[ K^l_{j}, U^l_{j} \dots U^l_{1} \Big( \tau_{z}^{l-1} \otimes \ket{0\dots 0}^{l}\bra{0\dots 0} \Big) U^{l^\dagger}_{1}\dots U^{l^\dagger}_{j}    \Big] U^{l^\dagger}_{j+1}\dots U^{l^\dagger}_{m_{l}}\Big)
\end{align*}
and
\begin{align*}
	\mathrm{tr}^{l+1:L+1} \Big(\Big(\id^l \otimes \ket{0\dots 0}^{l+1:L+1}\bra{0\dots 0} \Big)  G\Big)=& \mathrm{tr}^{l+1:L+1} \Big(\Big(\id^l \otimes \ket{0\dots 0}^{l+1:L+1}\bra{0\dots 0} \Big) U^{l+1^\dagger}\dots U^{L+1^\dagger}\\
	&\Big(\id^{l:L}\otimes \chi_z^\text{m}\otimes\sigma_z^\text{out}\Big) U^{L+1}\dots U^{l+1} \Big)\\
	=&\chi_z^l ,
\end{align*}
hence
\begin{align}
	\delta C=&i  \sum_{z=1}^{N} \sum_{l=1}^{L+1} \sum_{j=1}^{m_l}    \mathrm{tr}^{l}\Big( \chi_z^l  \mathrm{tr}^{l-1} \Big( U^{l}_{m_{l}} \dots U^l_{j+1} \Big[ K^l_{j}, U^l_{j} \dots U^l_{1} \Big( \tau_{z}^{l-1} \otimes \ket{0\dots 0}^{l}\bra{0\dots 0} \Big) U^{l^\dagger}_{1}\dots U^{l^\dagger}_{j}    \Big] U^{l^\dagger}_{j+1}\dots U^{l^\dagger}_{m_{l}}\Big)  \Big) \nonumber  \\
	=&i  \sum_{z=1}^{N} \sum_{l=1}^{L+1} \sum_{j=1}^{m_l}    \mathrm{tr}\Big( \Big(\id^{l-1}\otimes\chi_z^l \Big)  U^{l}_{m_{l}} \dots U^l_{j+1} \Big[ K^l_{j}, U^l_{j} \dots U^l_{1} \Big( \tau_{z}^{l-1} \otimes \ket{0\dots 0}^{l}\bra{0\dots 0} \Big) U^{l^\dagger}_{1}\dots U^{l^\dagger}_{j}    \Big] U^{l^\dagger}_{j+1}\dots U^{l^\dagger}_{m_{l}}  \Big) \nonumber  \\
	=&i  \sum_{z=1}^{N} \sum_{l=1}^{L+1} \sum_{j=1}^{m_l}    \mathrm{tr}\Big( \Big[  U^l_{j} \dots U^l_{1} \Big( \tau_{z}^{l-1} \otimes \ket{0\dots 0}^{l}\bra{0\dots 0} \Big) U^{l^\dagger}_{1}\dots U^{l^\dagger}_{j} ,U^{l^\dagger}_{j+1}\dots U^{l^\dagger}_{m_{l}}  \Big(\id^{l-1}\otimes\chi_z^l \Big)  U^{l}_{m_{l}} \dots U^l_{j+1}     \Big]    K^l_{j} \Big) \nonumber  \\
	=& i  \sum_{z=1}^{N} \sum_{l=1}^{L+1} \sum_{j=1}^{m_l}    \mathrm{tr}\Big(M_{j\,z}^l   K^l_{j} \Big) 
\end{align}
with
\begin{equation}
	M_{j\,z}^l = \Big[  U^l_{j} \dots U^l_{1} \Big( \tau_{z}^{l-1} \otimes \ket{0\dots 0}^{l}\bra{0\dots 0} \Big) U^{l^\dagger}_{1}\dots U^{l^\dagger}_{j} ,U^{l^\dagger}_{j+1}\dots U^{l^\dagger}_{m_{l}}  \Big(\id^{l-1}\otimes\chi_z^l \Big)  U^{l}_{m_{l}} \dots U^l_{j+1}     \Big].
\end{equation}
This again leads to 
\begin{equation*}
	K_j^l=i2^{m_{l-1}+1}\eta\sum_{x=1}^N \mathrm{tr}_\text{rest}\Big(M_{j\,x}^l\Big).
\end{equation*}
This matrix is very small, if \(\rho^{OUT}\) and \(\sigma_1^\text{out}\otimes \dots \otimes \sigma_N^\text{out} \otimes \sigma_N^\text{m}\) are far away from each other, and gets bigger and bigger, the more the states are similar. So for a constant learning rate \(\eta\) we would make relatively small changes to the perceptron unitaries at the start and these changes would get bigger with the training. But normally we want this to happen the other way round. Hence, we introduce scaling to the learning rate by setting \(\eta=\frac{\tilde{\eta}}{C}\). This corrensponds to minimizing \(\mathrm{ln}(C)\) which leads to the same result as minimizing \(C\) directly as the logarithm is strictly monotonically increasing. Hence, we choose the update rule
\begin{equation}
	K_j^l=\frac{i2^{m_{l-1}+1}\eta}{C}\sum_{x=1}^N \mathrm{tr}_\text{rest}\Big(M_{j\,x}^l\Big)
\end{equation}
with a constant learning rate \(\eta\) and obtain the algorithm presented below in Box~3.\\[1ex] \vspace{10pt}
{\centering
	\fbox{\label{box3}
		\begin{minipage}{.485\textwidth}
			\scriptsize
			\flushleft	
			{\normalsize \textbf{Box 3: Training algorithm global cost with pure output for \(M=1\)}\\\vspace{5pt}}
			\textbf{1. Initialize:}\\
			Choose the perceptron unitaries \(U_j^l\) randomly.\\
			\hspace{2pt}\\
			\textbf{2. Feedforward:} \\
			\textbf{2.1.} Set \(\tau_0^\text{m}=\rho_0^\text{m}\).\\
			\textbf{2.2.} For \(x=1,...,N\) do the following:\\
			\textbf{2.2.1.} Set
			\(\tau_x^0=\rho_x^\text{in}\otimes\tau_{x-1}^\text{m}\).\\
			\textbf{2.2.2.} For \(l=1,...,L+1\) set
			\(\tau_x^l=\mathrm{tr}_{l-1}\left( U^l(\tau_x^{l-1}\otimes \ket{0\dots 0}_{l}\bra{0\dots 0})U^{l^\dagger} \right).\)\\
			\textbf{2.2.3.} Set
			\(\tau_x^\text{m}=\mathrm{tr}_\text{out}^x\big(\big(\ket{\phi_x^\text{out}}\bra{\phi_x^\text{out}}\otimes \id_x^\text{m}\big)\tau_x^{L+1}\big)\).
			\hspace{2pt}\\
			\textbf{3. Feedbackward:} \\
			\textbf{3.1.} Set \(\chi_N^\text{m}=\id_N^\text{m}\).\\
			\textbf{3.2.} For \(x=N,...,1\) do the following:\\
			\textbf{3.2.1.} Set \(\chi_x^{L+1}=\chi^\text{m}_x\otimes \ket{\phi_x^{\text{out}}}\bra{\phi_x^{\text{out}}}\).\\
			\textbf{3.2.2.}For \(l=L,...,0\) set
			\(\chi_x^l=\mathrm{tr}_{l+1}\left( (\id ^l\otimes \ket{0\dots 0}_{l+1}\bra{0\dots 0}) U^{l+1^\dagger} (\id ^l\otimes \chi_x^{l+1})   U^{l+1} \right)\).\\
			\textbf{3.2.3.} Set 
			\(\chi_{x-1}^\text{m}=\mathrm{tr}_\text{in}\big(\big(\id_{x-1}^\text{m}\otimes \rho_x^\text{in}\big)\chi_x^0\big)\).\\
			\hspace{2pt}\\
			\textbf{4. Update the network:}  For \(l=1,...,L+1\) and \(j=1,...,m_l\) do the following: 
			
		\end{minipage}
		\hfill
		\begin{minipage}{.485\textwidth}
			\scriptsize
			\flushleft	
			%\phantom{\textbf{Box 1: Training algorithm}\\}
			\textbf{4.1.} For \(x=1,...,N\) set
			\begin{align*}
				M_{j\,x}^l=& \Bigg[\prod_{k=j}^{1} U_k^l\left( \tau_x^{l-1}\otimes\ket{0\dots 0}_l\bra{0\dots 0}\right) \prod_{k=1}^{j} U_k^{l^\dagger},\\
				&\ \prod_{k=j+1}^{m_l} U_k^{l^\dagger} \left(\id ^{l-1}\otimes \chi_x^l\right)  \prod_{k=m_l}^{j+1} U_k^l \Bigg].
			\end{align*}
			\textbf{4.2.} Set
			\[K_j^l=\frac{i2^{m_{l-1}+1}\eta}{C}\sum_{x=1}^N \mathrm{tr}_\text{rest}\Big(M_{j\,x}^l\Big),\]
			where \(\mathrm{tr}_\text{rest}\) denotes that the trace is taken over all systems that are unaffected by \(U_j^l\).\\
			\textbf{4.3.} Update each unitary $U_j^l$ according to $U_j^l\rightarrow e^{i\epsilon K_j^l} U_j^l$.\\
			\hspace{2pt}\\
			\textbf{5. Repeat:} Repeat steps 2. to 4. until the cost function reaches its maximum.
		\end{minipage}
	}
}\vspace{10pt}
\section{The training algorithm for many runs (\(M>1\))}
\label{manysequences}
In the same way, we can generalise all of the algorithms to many runs by just adding an average over \(\alpha\).

{\centering
	\fbox{\label{box4}
		\begin{minipage}{.485\textwidth}
			\scriptsize
			\flushleft	
			{\normalsize \textbf{Box 4: Training algorithm local cost with pure output for \(M\geq 1\)}\\\vspace{5pt}}
			\textbf{1. Initialize:}\\
			Choose the perceptron unitaries \(U_j^l\) randomly.\\
			\hspace{2pt}\\
			\textbf{2. Feedforward:} 
			For \(\alpha=1,...,M,\ x=1,...,N\) do the following:\\
			\textbf{2.1.} Set
			\(\rho_{x\, \alpha}^0=\rho_{x\, \alpha}^\text{in}\otimes\rho_{x-1\, \alpha}^\text{m}\) where \(\rho^\text{m}_{0\, \alpha}\) is given.\\
			\textbf{2.2.} For \(l=1,...,L+1\) set
			\(\rho_{x\, \alpha}^l=\mathrm{tr}_{l-1}\left( U^l(\rho_{x\, \alpha}^{l-1}\otimes \ket{0\dots 0}_{l}\bra{0\dots 0}) U^{l^\dagger} \right).\)\\
			\textbf{2.3.} Set
			\(\rho_{x\, \alpha}^\text{m}=\mathrm{tr}_\text{out}^x(\rho_{x\, \alpha}^{L+1})\)
			and 
			\(\rho_x^\text{out}=\mathrm{tr}_\text{m}^x(\rho_{x\, \alpha}^{L+1}).\)\\
			\hspace{2pt}\\
			\textbf{3. Feedbackward:} For \(\alpha=1,...,M,\ x=N,...,1\) do the following:\\
			\textbf{3.1.} Set \(\sigma_{x\, \alpha}^{L+1}=\id ^\text{m}_{x\, \alpha}\otimes \ket{\phi_{x\, \alpha}^{\text{out}}}\bra{\phi_{x\, \alpha}^{\text{out}}}\).\\
			\textbf{3.2.} For \(l=L,...,1\) set
			\(\sigma_{x\, \alpha}^l=\mathrm{tr}_{l+1}\left( (\id ^l\otimes \ket{0\dots 0}_{l+1}\bra{0\dots 0}) U^{l+1^\dagger} (\id ^l\otimes \sigma_{x\, \alpha}^{l+1})  U^{l+1} \right)\).\\
			\textbf{3.3.} Set 
			\(\omega_{xx\, \alpha}^\text{m}=\mathrm{tr}_{1,\text{in}}\left( (\rho_{x\, \alpha}^\text{in}\otimes \id _{x-1}^m\otimes \ket{0\dots 0}_{1}\bra{0\dots 0}) U^{1^\dagger} (\id _x^0\otimes \sigma_{x\, \alpha}^{1})  U^{1} \right)\).\\
			\textbf{3.4.} For \(z=x-1,...,1\) do the following steps:\\
			\textbf{3.4.1.} Set \(\omega_{zx\, \alpha}^{L+1}=\omega_{z+1\,x\, \alpha}^\text{m}\otimes \id ^\text{out}\).\\
			\textbf{3.4.2.} For \(l=L,...,1\) set
			\(\omega_{zx\, \alpha}^l=\mathrm{tr}_{l+1}\left( (\id ^l\otimes \ket{0\dots 0}_{l+1}\bra{0\dots 0}) U^{l+1^\dagger}(\id ^l\otimes \omega_{zx\, \alpha}^{l+1})   U^{l+1} \right).\)\\
			\textbf{3.4.3.} Set 
			\(\omega_{zx\, \alpha}^\text{m}=\mathrm{tr}_{1,\text{in}}\left( (\rho_{x\, \alpha}^\text{in}\otimes \id _{x-1}^m\otimes \ket{0\dots 0}_{1}\bra{0\dots 0}) U^{1^\dagger}  (\id _x^0\otimes \omega_{zx\, \alpha}^{1})  U^{1} \right).\)\\
			\textbf{4. Update the network:}  For \(l=1,...,L+1\) and \(j=1,...,m_l\) do the following: 
			
		\end{minipage}
		\hfill
		\begin{minipage}{.485\textwidth}
			\scriptsize
			\flushleft	
			%\phantom{\textbf{Box 1: Training algorithm}\\}
			
			\textbf{4.1.} For \(\alpha=1,...,M,\ x=1,...,N\) set
			\begin{align*}
				M_{j\,x\,\alpha}^l=& \Bigg[\prod_{k=j}^{1} U_k^l\left( \rho_{x\, \alpha}^{l-1}\otimes\ket{0\dots 0}_l\bra{0\dots 0}\right) \prod_{k=1}^{j} U_k^{l^\dagger},\\
				&\ \prod_{k=j+1}^{m_l} U_k^{l^\dagger} \left(\id ^{l-1}\otimes \sigma_{x\, \alpha}^l\right)  \prod_{k=m_l}^{j+1} U_k^l \Bigg]\\
				&+ \sum_{z=1}^{x-1} \Bigg[ \prod_{k=j}^{1} U_k^l \left( \rho_{z\, \alpha}^{l-1}\otimes\ket{0\dots 0}_l\bra{0\dots 0}\right)  \prod_{k=1}^{j} U_k^{l^\dagger} ,\\
				&\ \prod_{k=j+1}^{m_l} U_k^{l^\dagger} \left(\id ^{l-1}\otimes \omega_{zx\, \alpha}^l\right) \prod_{k=m_l}^{j+1} U_k^l \Bigg].
			\end{align*}
			\textbf{4.2.} Set
			\[K_j^l=\frac{i2^{m_{l-1}} \eta}{NM}\sum_{\alpha=1}^M \sum_{x=1}^N \mathrm{tr}_\text{rest}(M_{j\,x\,\alpha}^l),\]
			where \(\mathrm{tr}_\text{rest}\) denotes that the trace is taken over all systems that are unaffected by \(U_j^l\).\\
			\textbf{4.3.} Update each unitary $U_j^l$ according to $U_j^l\rightarrow e^{i\epsilon K_j^l} U_j^l$.\\
			\hspace{2pt}\\
			\textbf{5. Repeat:} Repeat steps 2. to 4. until the cost function reaches its maximum.
		\end{minipage}
	}
}\vspace{10pt}

{\centering
	\fbox{\label{box5}
		\begin{minipage}{.485\textwidth}
			\scriptsize
			\flushleft	
			{\normalsize \textbf{Box 5: Training algorithm local cost with mixed output for \(M\geq 1\)}\\\vspace{5pt}}
			\textbf{1. Initialize:}\\
			Choose the perceptron unitaries \(U_j^l\) randomly.\\
			\hspace{2pt}\\
			\textbf{2. Feedforward:} 
			For \(\alpha=1,...,M,\ x=1,...,N\) do the following:\\
			\textbf{2.1.} Set
			\(\rho_{x\, \alpha}^0=\rho_{x\, \alpha}^\text{in}\otimes\rho_{x-1\, \alpha}^\text{m}\) where \(\rho^\text{m}_{0\, \alpha}\) is given.\\
			\textbf{2.2.} For \(l=1,...,L+1\) set
			\(\rho_{x\, \alpha}^l=\mathrm{tr}_{l-1}\left( U^l(\rho_{x\, \alpha}^{l-1}\otimes \ket{0\dots 0}_{l}\bra{0\dots 0}) U^{l^\dagger} \right).\)\\
			\textbf{2.3.} Set
			\(\rho_{x\, \alpha}^\text{m}=\mathrm{tr}_\text{out}^x(\rho_{x\, \alpha}^{L+1})\)
			and 
			\(\rho_x^\text{out}=\mathrm{tr}_\text{m}^x(\rho_{x\, \alpha}^{L+1}).\)\\
			\hspace{2pt}\\
			\textbf{3. Feedbackward:} For \(\alpha=1,...,M,\ x=N,...,1\) do the following:\\
			\textbf{3.1.} Set \(\chi_{x\, \alpha}^{L+1}=\id ^\text{m}_{x\, \alpha}\otimes \left(\sigma_{x\, \alpha}^\text{out}- \rho_{x\, \alpha}^{\text{out}}\right)\).\\
			\textbf{3.2.} For \(l=L,...,1\) set
			\(\chi_{x\, \alpha}^l=\mathrm{tr}_{l+1}\left( (\id ^l\otimes \ket{0\dots 0}_{l+1}\bra{0\dots 0}) U^{l+1^\dagger} (\id ^l\otimes \chi_{x\, \alpha}^{l+1})  U^{l+1} \right)\).\\
			\textbf{3.3.} Set 
			\(\omega_{xx\, \alpha}^\text{m}=\mathrm{tr}_{1,\text{in}}\left( (\rho_{x\, \alpha}^\text{in}\otimes \id _{x-1}^m\otimes \ket{0\dots 0}_{1}\bra{0\dots 0}) U^{1^\dagger} (\id _x^0\otimes \chi_{x\, \alpha}^{1})  U^{1} \right)\).\\
			\textbf{3.4.} For \(z=x-1,...,1\) do the following steps:\\
			\textbf{3.4.1.} Set \(\omega_{zx\, \alpha}^{L+1}=\omega_{z+1\,x\, \alpha}^\text{m}\otimes \id ^\text{out}\).\\
			\textbf{3.4.2.} For \(l=L,...,1\) set
			\(\omega_{zx\, \alpha}^l=\mathrm{tr}_{l+1}\left( (\id ^l\otimes \ket{0\dots 0}_{l+1}\bra{0\dots 0}) U^{l+1^\dagger}(\id ^l\otimes \omega_{zx\, \alpha}^{l+1})   U^{l+1} \right).\)\\
			\textbf{3.4.3.} Set 
			\(\omega_{zx\, \alpha}^\text{m}=\mathrm{tr}_{1,\text{in}}\left( (\rho_{x\, \alpha}^\text{in}\otimes \id _{x-1}^m\otimes \ket{0\dots 0}_{1}\bra{0\dots 0}) U^{1^\dagger}  (\id _x^0\otimes \omega_{zx\, \alpha}^{1})  U^{1} \right).\)\\
			\textbf{4. Update the network:}  For \(l=1,...,L+1\) and \(j=1,...,m_l\) do the following: 
			
		\end{minipage}
		\hfill
		\begin{minipage}{.485\textwidth}
			\scriptsize
			\flushleft	
			%\phantom{\textbf{Box 1: Training algorithm}\\}
			
			\textbf{4.1.} For \(\alpha=1,...,M,\ x=1,...,N\) set
			\begin{align*}
				M_{j\,x\,\alpha}^l=& \Bigg[\prod_{k=j}^{1} U_k^l\left( \rho_{x\, \alpha}^{l-1}\otimes\ket{0\dots 0}_l\bra{0\dots 0}\right) \prod_{k=1}^{j} U_k^{l^\dagger},\\
				&\ \prod_{k=j+1}^{m_l} U_k^{l^\dagger} \left(\id ^{l-1}\otimes \sigma_{x\, \alpha}^l\right)  \prod_{k=m_l}^{j+1} U_k^l \Bigg]\\
				&+ \sum_{z=1}^{x-1} \Bigg[ \prod_{k=j}^{1} U_k^l \left( \rho_{z\, \alpha}^{l-1}\otimes\ket{0\dots 0}_l\bra{0\dots 0}\right)  \prod_{k=1}^{j} U_k^{l^\dagger} ,\\
				&\ \prod_{k=j+1}^{m_l} U_k^{l^\dagger} \left(\id ^{l-1}\otimes \omega_{zx\, \alpha}^l\right) \prod_{k=m_l}^{j+1} U_k^l \Bigg].
			\end{align*}
			\textbf{4.2.} Set
			\[K_j^l=\frac{i2^{m_{l-1}} \eta}{NM}\sum_{\alpha=1}^M \sum_{x=1}^N \mathrm{tr}_\text{rest}(M_{j\,x\,\alpha}^l),\]
			where \(\mathrm{tr}_\text{rest}\) denotes that the trace is taken over all systems that are unaffected by \(U_j^l\).\\
			\textbf{4.3.} Update each unitary $U_j^l$ according to $U_j^l\rightarrow e^{i\epsilon K_j^l} U_j^l$.\\
			\hspace{2pt}\\
			\textbf{5. Repeat:} Repeat steps 2. to 4. until the cost function reaches its maximum.
		\end{minipage}
	}
}\vspace{10pt}

{\centering
	\fbox{\label{box6}
		\begin{minipage}{.485\textwidth}
			\scriptsize
			\flushleft	
			{\normalsize \textbf{Box 6: Training algorithm global cost with pure output for \(M\geq 1\)}\\\vspace{5pt}}
			\textbf{1. Initialize:}\\
			Choose the perceptron unitaries \(U_j^l\) randomly.\\
			\hspace{2pt}\\
			\textbf{2. Feedforward:} For \(\alpha=1,...,M\) do the following:\\
			\textbf{2.1.} Set \(\tau_{0\, \alpha}^\text{m}=\rho_{0\, \alpha}^\text{m}\).\\
			\textbf{2.2.} For \(x=1,...,N\) do the following:\\
			\textbf{2.2.1.} Set
			\(\tau_{x\, \alpha}^0=\rho_{x\, \alpha}^\text{in}\otimes\tau_{x-1\, \alpha}^\text{m}\).\\
			\textbf{2.2.2.} For \(l=1,...,L+1\) set
			\(\tau_{x\, \alpha}^l=\mathrm{tr}_{l-1}\left( U^l(\tau_{x\, \alpha}^{l-1}\otimes \ket{0\dots 0}_{l}\bra{0\dots 0})U^{l^\dagger} \right).\)\\
			\textbf{2.2.3.} Set
			\(\tau_{x\, \alpha}^\text{m}=\mathrm{tr}_\text{out}^x\big(\big(\ket{\phi_{x\, \alpha}^\text{out}}\bra{\phi_{x\, \alpha}^\text{out}}\otimes \id_x^\text{m}\big)\tau_{x\, \alpha}^{L+1}\big)\).
			\hspace{2pt}\\
			\textbf{3. Feedbackward:} For \(\alpha=1,...,M\) do the following:\\
			\textbf{3.1.} Set \(\chi_{N\, \alpha}^\text{m}=\id_N^\text{m}\).\\
			\textbf{3.2.} For \(x=N,...,1\) do the following:\\
			\textbf{3.2.1.} Set \(\chi_{x\, \alpha}^{L+1}=\chi^\text{m}_{x\, \alpha}\otimes \ket{\phi_x^{\text{out}}}\bra{\phi_x^{\text{out}}}\).\\
			\textbf{3.2.2.}For \(l=L,...,0\) set
			\(\chi_{x\, \alpha}^l=\mathrm{tr}_{l+1}\left( (\id ^l\otimes \ket{0\dots 0}_{l+1}\bra{0\dots 0}) U^{l+1^\dagger} (\id ^l\otimes \chi_{x\, \alpha}^{l+1})   U^{l+1} \right)\).\\
			\textbf{3.2.3.} Set 
			\(\chi_{x-1\, \alpha}^\text{m}=\mathrm{tr}_\text{in}\big(\big(\id_{x-1}^\text{m}\otimes \rho_{x\, \alpha}^\text{in}\big)\chi_{x\, \alpha}^0\big)\).\\
			\hspace{2pt}\\
			\textbf{4. Update the network:}  For \(l=1,...,L+1\) and \(j=1,...,m_l\) do the following: 
			
		\end{minipage}
		\hfill
		\begin{minipage}{.485\textwidth}
			\scriptsize
			\flushleft	
			%\phantom{\textbf{Box 1: Training algorithm}\\}
			\textbf{4.1.} For \(\alpha=1,...,M,\ x=1,...,N\) set
			\begin{align*}
				M_{j\,x\,\alpha}^l=& \Bigg[\prod_{k=j}^{1} U_k^l\left( \tau_{x\, \alpha}^{l-1}\otimes\ket{0\dots 0}_l\bra{0\dots 0}\right) \prod_{k=1}^{j} U_k^{l^\dagger},\\
				&\ \prod_{k=j+1}^{m_l} U_k^{l^\dagger} \left(\id ^{l-1}\otimes \chi_{x\, \alpha}^l\right)  \prod_{k=m_l}^{j+1} U_k^l \Bigg].
			\end{align*}
			\textbf{4.2.} Set
			\[K_j^l=\frac{i2^{m_{l-1}+1}\eta}{CM}\sum_{\alpha=1}^M\sum_{x=1}^N \mathrm{tr}_\text{rest}\Big(M_{j\,x\,\alpha}^l\Big),\]
			where \(\mathrm{tr}_\text{rest}\) denotes that the trace is taken over all systems that are unaffected by \(U_j^l\).\\
			\textbf{4.3.} Update each unitary $U_j^l$ according to $U_j^l\rightarrow e^{i\epsilon K_j^l} U_j^l$.\\
			\hspace{2pt}\\
			\textbf{5. Repeat:} Repeat steps 2. to 4. until the cost function reaches its maximum.
		\end{minipage}
	}
}\vspace{10pt}
\section{Numerical results}
\label{numericalresults}
In the following examples, we explore the training and testing performance of QRNNs on tasks such as learning a delay channel, a time-evolution with time-dependent Hamiltonian and quantum filters. Note that for pure outputs, we use the same cost function, the fidelity, for training and testing. Unlike in some other situations common in unsupervised or semi-supervised learning, see e.g.~\cite{Goodfellow2016, DiPol2020}, the network performance during training on big enough data is a good indicator of its performance on the training data drawn from the same distribution. It is not necessarily so for small training sets or if the cost function used during training differs from the function used to evaluate the test performance. We use both the fidelity and the Hilbert-Schmidt norm when testing for mixed states.
\subsection{Local Cost for pure output}
Firstly, we present the numerical results for the local cost for pure output as this is the one for which we made most simulations.
\subsubsection{Delay Channel}
The first example is the delay channel where the \(i^\text{th}\) output is given by the \((i-p)^\text{th}\) input for an arbitrary integer \(p\). The learning process, i.e. the cost in dependence of the training step, for \(p=1,2\) is shown in Figure \ref{fig:SWAPs}.
\begin{figure}[htbp]
	\begin{subfigure}{0.49\textwidth}
		\begin{tikzpicture}
			\node at (-1,6.5) {\textbf{a}};
			\begin{axis}[
				width=\linewidth, % Scale the plot to \linewidth
				%grid=major, % Display a grid
				%grid style={dashed,gray!30}, % Set the style
				xlabel= Training step, % Set the labels
				ylabel= Cost,
				xmin=0,
				xmax=1500,
				ymin=0.4,
				ymax=1.05,
				legend columns=2, 
				legend style={at={(0.5,-0.2)},anchor=north},
				%legend style={at={(0.7,0.7)},anchor=north}, % Put the legend below the plot
				%x tick label style={rotate=90,anchor=east} % Display labels sideways
				]
				\addplot[very thick,color=red,mark=None]  table[x=step,y=Cwomround,col sep=comma] {SWAP1m.csv};
				\addlegendentry{\begin{tikzpicture}[scale=0.05, baseline]
						\node(1) [circle,draw,inner sep=0pt,minimum size=3.5pt] at (-1,0.15) {};
						\node(4) [circle,draw,inner sep=0pt,minimum size=3.5pt] at (0,0.15) {};
						\draw (1)--(4);
					\end{tikzpicture} training} 
				\addplot[very thick,color=blue,mark=None]   table[x=step,y=C1mround,col sep=comma] {SWAP1m.csv}; 
				\addlegendentry{\begin{tikzpicture}[xscale=0.05,yscale=0.1, baseline]
						\node(1) [circle,draw,inner sep=0pt,minimum size=3.5pt] at (-1,0) {};
						\node(2) [circle,draw,inner sep=0pt,minimum size=3.5 pt] at (-1,0.3) {};
						\node(4) [circle,draw,inner sep=0pt,minimum size=3.5pt] at (0,0) {};
						\node(5) [circle,draw,inner sep=0pt,minimum size=3.5pt] at (0,0.3) {};
						\draw (1)--(4);
						\draw (2)--(4);
						\draw (1)--(5);
						\draw (2)--(5);
						\draw (5) -- (0.3,0.3);
						\draw[out=0,in=0] (0.3,0.3) to (0.3,0.5);
						\draw (0.3,0.5) to (-1.3,0.5);
						\draw[out=180,in=180] (-1.3,0.5) to (-1.3,0.3);
						\draw (-1.3,0.3) -- (2);
					\end{tikzpicture} training}
				\addplot[very thick,color=red,mark=None,dashed]   table[x=step,y=Cwomtestround,col sep=comma] {SWAP1m.csv}; 
				\addlegendentry{\begin{tikzpicture}[scale=0.05, baseline]
						\node(1) [circle,draw,inner sep=0pt,minimum size=3.5pt] at (-1,0.15) {};
						\node(4) [circle,draw,inner sep=0pt,minimum size=3.5pt] at (0,0.15) {};
						\draw (1)--(4);
					\end{tikzpicture} testing} 
				\addplot[very thick,color=blue,mark=None,dashed]   table[x=step,y=C1mtestround,col sep=comma] {SWAP1m.csv}; 
				\addlegendentry{\begin{tikzpicture}[xscale=0.05,yscale=0.1, baseline]
						\node(1) [circle,draw,inner sep=0pt,minimum size=3.5pt] at (-1,0) {};
						\node(2) [circle,draw,inner sep=0pt,minimum size=3.5 pt] at (-1,0.3) {};
						\node(4) [circle,draw,inner sep=0pt,minimum size=3.5pt] at (0,0) {};
						\node(5) [circle,draw,inner sep=0pt,minimum size=3.5pt] at (0,0.3) {};
						\draw (1)--(4);
						\draw (2)--(4);
						\draw (1)--(5);
						\draw (2)--(5);
						\draw (5) -- (0.3,0.3);
						\draw[out=0,in=0] (0.3,0.3) to (0.3,0.5);
						\draw (0.3,0.5) to (-1.3,0.5);
						\draw[out=180,in=180] (-1.3,0.5) to (-1.3,0.3);
						\draw (-1.3,0.3) -- (2);
					\end{tikzpicture} testing}
				\addplot [domain=0:1500, samples=10, color=gray,]{0.5};
				\addplot [domain=0:1500, samples=10, color=gray,]{1};
				%\legend{training without memory, training with 1 memory qubit, testing without memory, testing with 1 memory qubit}
			\end{axis}
		\end{tikzpicture}
	\end{subfigure}
	\hfill
	\begin{subfigure}{0.49\textwidth}
		\begin{tikzpicture}
			\node at (-1,6.5) {\textbf{b}};
			\begin{axis}[
				width=\linewidth, % Scale the plot to \linewidth
				%grid=major, % Display a grid
				%grid style={dashed,gray!30}, % Set the style
				xlabel= Training step, % Set the labels
				ylabel= Cost,
				xmin=0,
				xmax=1000,
				ymin=0.4,
				ymax=1.05,
				legend columns=3, 
				legend style={at={(0.45,-0.2)},anchor=north},
				%legend style={at={(0.7,0.7)},anchor=north}, % Put the legend below the plot
				%x tick label style={rotate=90,anchor=east} % Display labels sideways
				]
				\addplot[very thick,color=red,mark=None]  table[x=step,y=Cwomround,col sep=comma] {SWAP2m.csv}; 
				\addlegendentry{\begin{tikzpicture}[scale=0.05, baseline]
						\node(1) [circle,draw,inner sep=0pt,minimum size=3.5pt] at (-1,0.15) {};
						\node(4) [circle,draw,inner sep=0pt,minimum size=3.5pt] at (0,0.15) {};
						\draw (1)--(4);
					\end{tikzpicture} training} 
				\addplot[very thick,color=blue,mark=None]   table[x=step,y=C1mround,col sep=comma] {SWAP2m.csv};
				\addlegendentry{\begin{tikzpicture}[xscale=0.05,yscale=0.1, baseline]
						\node(1) [circle,draw,inner sep=0pt,minimum size=3.5pt] at (-1,0) {};
						\node(2) [circle,draw,inner sep=0pt,minimum size=3.5 pt] at (-1,0.3) {};
						\node(4) [circle,draw,inner sep=0pt,minimum size=3.5pt] at (0,0) {};
						\node(5) [circle,draw,inner sep=0pt,minimum size=3.5pt] at (0,0.3) {};
						\draw (1)--(4);
						\draw (2)--(4);
						\draw (1)--(5);
						\draw (2)--(5);
						\draw (5) -- (0.3,0.3);
						\draw[out=0,in=0] (0.3,0.3) to (0.3,0.5);
						\draw (0.3,0.5) to (-1.3,0.5);
						\draw[out=180,in=180] (-1.3,0.5) to (-1.3,0.3);
						\draw (-1.3,0.3) -- (2);
					\end{tikzpicture} training}
				\addplot[very thick,color=ForestGreen,mark=None]   table[x=step,y=C2mround,col sep=comma] {SWAP2m.csv};
				\addlegendentry{\begin{tikzpicture}[xscale=0.05,yscale=0.1, baseline]
						\node(1) [circle,draw,inner sep=0pt,minimum size=3.5pt] at (-1,-0.15) {};
						\node(2) [circle,draw,inner sep=0pt,minimum size=3.5pt] at (-1,0.15) {};
						\node(3) [circle,draw,inner sep=0pt,minimum size=3.5pt] at (-1,0.45) {};
						\node(4) [circle,draw,inner sep=0pt,minimum size=3.5pt] at (0,-0.15) {};
						\node(5) [circle,draw,inner sep=0pt,minimum size=3.5pt] at (0,0.15) {};
						\node(6) [circle,draw,inner sep=0pt,minimum size=3.5pt] at (0,0.45) {};
						\draw (1)--(4);
						\draw (2)--(4);
						\draw (3)--(4);
						\draw (1)--(5);
						\draw (2)--(5);
						\draw (3)--(5);
						\draw (1)--(6);
						\draw (2)--(6);
						\draw (3)--(6);
						\draw (6) -- (0.3,0.45);
						\draw[out=0,in=0] (0.3,0.45) to (0.3,0.65);
						\draw (0.3,0.65) to (-1.3,0.65);
						\draw[out=180,in=180] (-1.3,0.65) to (-1.3,0.45);
						\draw (-1.3,0.45) -- (3);
						\draw (5) -- (0.35,0.15);
						\draw[out=0,in=0] (0.35,0.15) to (0.35,0.75);
						\draw (0.35,0.75) to (-1.35,0.75);
						\draw[out=180,in=180] (-1.35,0.75) to (-1.35,0.15);
						\draw (-1.35,0.15) -- (2);
				\end{tikzpicture} training}
				\addplot[very thick,color=red,mark=None,dashed]   table[x=step,y=Cwomtestround,col sep=comma] {SWAP2m.csv}; 
				\addlegendentry{\begin{tikzpicture}[scale=0.05, baseline]
						\node(1) [circle,draw,inner sep=0pt,minimum size=3.5pt] at (-1,0.15) {};
						\node(4) [circle,draw,inner sep=0pt,minimum size=3.5pt] at (0,0.15) {};
						\draw (1)--(4);
					\end{tikzpicture} testing} 
				\addplot[very thick,color=blue,mark=None,dashed]   table[x=step,y=C1mtestround,col sep=comma] {SWAP2m.csv};
				\addlegendentry{\begin{tikzpicture}[xscale=0.05,yscale=0.1, baseline]
						\node(1) [circle,draw,inner sep=0pt,minimum size=3.5pt] at (-1,0) {};
						\node(2) [circle,draw,inner sep=0pt,minimum size=3.5 pt] at (-1,0.3) {};
						\node(4) [circle,draw,inner sep=0pt,minimum size=3.5pt] at (0,0) {};
						\node(5) [circle,draw,inner sep=0pt,minimum size=3.5pt] at (0,0.3) {};
						\draw (1)--(4);
						\draw (2)--(4);
						\draw (1)--(5);
						\draw (2)--(5);
						\draw (5) -- (0.3,0.3);
						\draw[out=0,in=0] (0.3,0.3) to (0.3,0.5);
						\draw (0.3,0.5) to (-1.3,0.5);
						\draw[out=180,in=180] (-1.3,0.5) to (-1.3,0.3);
						\draw (-1.3,0.3) -- (2);
					\end{tikzpicture} testing}
				\addplot[very thick,color=ForestGreen,mark=None,dashed]   table[x=step,y=C2mtestround,col sep=comma] {SWAP2m.csv}; 
				\addlegendentry{\begin{tikzpicture}[xscale=0.05,yscale=0.1, baseline]
						\node(1) [circle,draw,inner sep=0pt,minimum size=3.5pt] at (-1,-0.15) {};
						\node(2) [circle,draw,inner sep=0pt,minimum size=3.5pt] at (-1,0.15) {};
						\node(3) [circle,draw,inner sep=0pt,minimum size=3.5pt] at (-1,0.45) {};
						\node(4) [circle,draw,inner sep=0pt,minimum size=3.5pt] at (0,-0.15) {};
						\node(5) [circle,draw,inner sep=0pt,minimum size=3.5pt] at (0,0.15) {};
						\node(6) [circle,draw,inner sep=0pt,minimum size=3.5pt] at (0,0.45) {};
						\draw (1)--(4);
						\draw (2)--(4);
						\draw (3)--(4);
						\draw (1)--(5);
						\draw (2)--(5);
						\draw (3)--(5);
						\draw (1)--(6);
						\draw (2)--(6);
						\draw (3)--(6);
						\draw (6) -- (0.3,0.45);
						\draw[out=0,in=0] (0.3,0.45) to (0.3,0.65);
						\draw (0.3,0.65) to (-1.3,0.65);
						\draw[out=180,in=180] (-1.3,0.65) to (-1.3,0.45);
						\draw (-1.3,0.45) -- (3);
						\draw (5) -- (0.35,0.15);
						\draw[out=0,in=0] (0.35,0.15) to (0.35,0.75);
						\draw (0.35,0.75) to (-1.35,0.75);
						\draw[out=180,in=180] (-1.35,0.75) to (-1.35,0.15);
						\draw (-1.35,0.15) -- (2);
				\end{tikzpicture} testing}
				%\legend{training without memory, training with 1 memory qubit, testing without memory, testing with 1 memory qubit}
				\addplot [domain=0:1000, samples=10, color=gray,]{0.5};
				\addplot [domain=0:1000, samples=10, color=gray,]{1};
			\end{axis}
		\end{tikzpicture}
	\end{subfigure}
	\caption[]{\textbf{Comparison of QRNN and QNN for different delay channels.} Both plots show the local cost for pure output states dependent on the training step, while a QNN is trained on a specific training set. Thereby, the colour indicates which QNN is trained (green for a 
	\begin{tikzpicture}[xscale=0.4,yscale=0.6, baseline]
						\node(1) [circle,draw,inner sep=0pt,minimum size=3.5pt] at (-1,-0.15) {};
						\node(2) [circle,draw,inner sep=0pt,minimum size=3.5pt] at (-1,0.15) {};
						\node(3) [circle,draw,inner sep=0pt,minimum size=3.5pt] at (-1,0.45) {};
						\node(4) [circle,draw,inner sep=0pt,minimum size=3.5pt] at (0,-0.15) {};
						\node(5) [circle,draw,inner sep=0pt,minimum size=3.5pt] at (0,0.15) {};
						\node(6) [circle,draw,inner sep=0pt,minimum size=3.5pt] at (0,0.45) {};
						\draw (1)--(4);
						\draw (2)--(4);
						\draw (3)--(4);
						\draw (1)--(5);
						\draw (2)--(5);
						\draw (3)--(5);
						\draw (1)--(6);
						\draw (2)--(6);
						\draw (3)--(6);
						\draw (6) -- (0.3,0.45);
						\draw[out=0,in=0] (0.3,0.45) to (0.3,0.65);
						\draw (0.3,0.65) to (-1.3,0.65);
						\draw[out=180,in=180] (-1.3,0.65) to (-1.3,0.45);
						\draw (-1.3,0.45) -- (3);
						\draw (5) -- (0.35,0.15);
						\draw[out=0,in=0] (0.35,0.15) to (0.35,0.75);
						\draw (0.35,0.75) to (-1.35,0.75);
						\draw[out=180,in=180] (-1.35,0.75) to (-1.35,0.15);
						\draw (-1.35,0.15) -- (2);
				\end{tikzpicture} QRNN with two memory qubits, blue for a
	\begin{tikzpicture}[xscale=0.4,yscale=0.6, baseline]
						\node(1) [circle,draw,inner sep=0pt,minimum size=3.5pt] at (-1,0) {};
						\node(2) [circle,draw,inner sep=0pt,minimum size=3.5 pt] at (-1,0.3) {};
						\node(4) [circle,draw,inner sep=0pt,minimum size=3.5pt] at (0,0) {};
						\node(5) [circle,draw,inner sep=0pt,minimum size=3.5pt] at (0,0.3) {};
						\draw (1)--(4);
						\draw (2)--(4);
						\draw (1)--(5);
						\draw (2)--(5);
						\draw (5) -- (0.3,0.3);
						\draw[out=0,in=0] (0.3,0.3) to (0.3,0.5);
						\draw (0.3,0.5) to (-1.3,0.5);
						\draw[out=180,in=180] (-1.3,0.5) to (-1.3,0.3);
						\draw (-1.3,0.3) -- (2);
					\end{tikzpicture} QRNN with one memory qubit, red for a \begin{tikzpicture}[yscale=0.6,xscale=0.4, baseline]
        \node(1) [circle,draw,inner sep=0pt,minimum size=3.5pt] at (-1,0.15) {};
        \node(4) [circle,draw,inner sep=0pt,minimum size=3.5pt] at (0,0.15) {};
        \draw (1)--(4);
    \end{tikzpicture} feed-forward QNN). In both cases, a QNN without hidden layers and a learningrate \(\epsilon  \eta=0.06\) is used. The drawn-through lines show the training, and the dashed lines the testing. Panel \textbf{(a)} shows the learning and testing of the delay by one channel (where \(\ket{\psi}_x^\text{out}= \ket{\psi}_{x-1}^\text{in} \)) with \(N=30\) training pairs over \(1500\) rounds.   Panel \textbf{(b)} shows the learning and testing of the delay-by-two channel (where \(\ket{\psi}_x^\text{out}= \ket{\psi}_{x-2}^\text{in} \)) with \(N=20\) training pairs over \(1000\) rounds. }
    \label{fig:SWAPs} 
\end{figure}
As you can see in Panel (a), in the case \(p=1\), a QRNN with one memory qubit learns the delay by one channel well and also generalises well to the test set, while the feed-forward QNN does not learn as well and does not generalise at all. In the case \(p=2\), as can be seen in Panel (b), the QRNN with two memory qubits learns the delay-by-two channel well and generalises well to the test set. The QRNN with one memory qubit learns the channel not as well and does not generalise as well, and the feed-forward QNN learns it even less well and does not generalise. As shown in Figure \ref{fig:SWAP1mdifferentarchitectures}, this behaviour has nothing to do with the simple structure of the feed-forward QNN since none of the displayed feed-forward QNNs is doing better than the simplest feed-forward QNN.

\begin{figure}[htbp]
    \centering
	\begin{tikzpicture}
		\begin{axis}[
			width=0.49\linewidth, % Scale the plot to \linewidth
			%grid=major, % Display a grid
			%grid style={dashed,gray!30}, % Set the style
			xlabel= Training step, % Set the labels
			ylabel= Cost,
			xmin=0,
			xmax=1000,
			ymin=0.45,
			ymax=1.05,
			legend columns=5, 
			legend style={at={(0.5,-0.2)},anchor=north},
			%legend style={at={(0.7,0.7)},anchor=north}, % Put the legend below the plot
			%x tick label style={rotate=90,anchor=east} % Display labels sideways
			]
			\addplot[very thick,color=red,mark=None]  table[x=step,y=Cwom1round,col sep=comma] {SWAP1mdifferentarchitectures.csv}; 
			\addlegendentry{\begin{tikzpicture}[yscale=0.1,xscale=0.05, baseline]
        \node(1) [circle,draw,inner sep=0pt,minimum size=3.5pt] at (-1,0.15) {};
        \node(4) [circle,draw,inner sep=0pt,minimum size=3.5pt] at (0,0.15) {};
        \draw (1)--(4);
    \end{tikzpicture} training}
			\addplot[very thick,color=purple,mark=None]  table[x=step,y=Cwom2round,col sep=comma] {SWAP1mdifferentarchitectures.csv}; 
			\addlegendentry{\begin{tikzpicture}[yscale=0.1,xscale=0.05, baseline]
        \node(1) [circle,draw,inner sep=0pt,minimum size=3.5pt] at (-1,0.15) {};
        \node(4) [circle,draw,inner sep=0pt,minimum size=3.5pt] at (0,0) {};
        \node(5) [circle,draw,inner sep=0pt,minimum size=3.5pt] at (0,0.3) {};
        \node(6) [circle,draw,inner sep=0pt,minimum size=3.5pt] at (1,0.15) {};
        \draw (1)--(4);
        \draw (1)--(5);
        \draw (6)--(4);
        \draw (6)--(5);
    \end{tikzpicture} training}
			\addplot[very thick,color=orange,mark=None]  table[x=step,y=Cwom3round,col sep=comma] {SWAP1mdifferentarchitectures.csv};
			\addlegendentry{\begin{tikzpicture}[yscale=0.1,xscale=0.05, baseline]
        \node(1) [circle,draw,inner sep=0pt,minimum size=3.5pt] at (-1,0.15) {};
        \node(3) [circle,draw,inner sep=0pt,minimum size=3.5pt] at (0,0.15) {};
        \node(4) [circle,draw,inner sep=0pt,minimum size=3.5pt] at (0,0.45) {};
        \node(5) [circle,draw,inner sep=0pt,minimum size=3.5pt] at (0,-0.15) {};
        \node(6) [circle,draw,inner sep=0pt,minimum size=3.5pt] at (1,0.15) {};
        \draw (1)--(3);
        \draw (1)--(4);
        \draw (1)--(5);
        \draw (6)--(3);
        \draw (6)--(4);
        \draw (6)--(5);
    \end{tikzpicture} training}
			\addplot[very thick,color=brown,mark=None]  table[x=step,y=Cwom4round,col sep=comma] {SWAP1mdifferentarchitectures.csv};
			\addlegendentry{\begin{tikzpicture}[yscale=0.1,xscale=0.05, baseline]
        \node(1) [circle,draw,inner sep=0pt,minimum size=3.5pt] at (-1,0.15) {};
        \node(4) [circle,draw,inner sep=0pt,minimum size=3.5pt] at (0,0.0) {};
        \node(5) [circle,draw,inner sep=0pt,minimum size=3.5pt] at (0,0.3) {};
        \node(7) [circle,draw,inner sep=0pt,minimum size=3.5pt] at (1,-0.15) {};
        \node(8) [circle,draw,inner sep=0pt,minimum size=3.5pt] at (1,0.15) {};
        \node(9) [circle,draw,inner sep=0pt,minimum size=3.5pt] at (1,0.45) {};
        \node(10) [circle,draw,inner sep=0pt,minimum size=3.5pt] at (2,0) {};
        \node(11) [circle,draw,inner sep=0pt,minimum size=3.5pt] at (2,0.3) {};
        \node(6) [circle,draw,inner sep=0pt,minimum size=3.5pt] at (3,0.15) {};
        \draw (1)--(4);
        \draw (1)--(5);
        \draw (7)--(4);
        \draw (7)--(5);
        \draw (8)--(4);
        \draw (8)--(5);
        \draw (9)--(4);
        \draw (9)--(5);
        \draw (7)--(10);
        \draw (7)--(11);
        \draw (8)--(10);
        \draw (8)--(11);
        \draw (9)--(10);
        \draw (9)--(11);
        \draw (6)--(10);
        \draw (6)--(11);
    \end{tikzpicture} training}
			\addplot[very thick,color=blue,mark=None]   table[x=step,y=C1mround,col sep=comma] {SWAP1mdifferentarchitectures.csv};
			\addlegendentry{\begin{tikzpicture}[xscale=0.05,yscale=0.1, baseline]
						\node(1) [circle,draw,inner sep=0pt,minimum size=3.5pt] at (-1,0) {};
						\node(2) [circle,draw,inner sep=0pt,minimum size=3.5 pt] at (-1,0.3) {};
						\node(4) [circle,draw,inner sep=0pt,minimum size=3.5pt] at (0,0) {};
						\node(5) [circle,draw,inner sep=0pt,minimum size=3.5pt] at (0,0.3) {};
						\draw (1)--(4);
						\draw (2)--(4);
						\draw (1)--(5);
						\draw (2)--(5);
						\draw (5) -- (0.3,0.3);
						\draw[out=0,in=0] (0.3,0.3) to (0.3,0.5);
						\draw (0.3,0.5) to (-1.3,0.5);
						\draw[out=180,in=180] (-1.3,0.5) to (-1.3,0.3);
						\draw (-1.3,0.3) -- (2);
					\end{tikzpicture} training}
			\addplot[very thick,color=red,mark=None,dashed]   table[x=step,y=Cwomtest1round,col sep=comma] {SWAP1mdifferentarchitectures.csv}; 
			\addlegendentry{\begin{tikzpicture}[yscale=0.1,xscale=0.05, baseline]
        \node(1) [circle,draw,inner sep=0pt,minimum size=3.5pt] at (-1,0.15) {};
        \node(4) [circle,draw,inner sep=0pt,minimum size=3.5pt] at (0,0.15) {};
        \draw (1)--(4);
    \end{tikzpicture} testing}
			\addplot[very thick,color=purple,mark=None,dashed]   table[x=step,y=Cwomtest2round,col sep=comma] {SWAP1mdifferentarchitectures.csv};
			\addlegendentry{\begin{tikzpicture}[yscale=0.1,xscale=0.05, baseline]
        \node(1) [circle,draw,inner sep=0pt,minimum size=3.5pt] at (-1,0.15) {};
        \node(4) [circle,draw,inner sep=0pt,minimum size=3.5pt] at (0,0) {};
        \node(5) [circle,draw,inner sep=0pt,minimum size=3.5pt] at (0,0.3) {};
        \node(6) [circle,draw,inner sep=0pt,minimum size=3.5pt] at (1,0.15) {};
        \draw (1)--(4);
        \draw (1)--(5);
        \draw (6)--(4);
        \draw (6)--(5);
    \end{tikzpicture} testing}
			\addplot[very thick,color=orange,mark=None,dashed]   table[x=step,y=Cwomtest3round,col sep=comma] {SWAP1mdifferentarchitectures.csv};
			\addlegendentry{\begin{tikzpicture}[yscale=0.1,xscale=0.05, baseline]
        \node(1) [circle,draw,inner sep=0pt,minimum size=3.5pt] at (-1,0.15) {};
        \node(3) [circle,draw,inner sep=0pt,minimum size=3.5pt] at (0,0.15) {};
        \node(4) [circle,draw,inner sep=0pt,minimum size=3.5pt] at (0,0.45) {};
        \node(5) [circle,draw,inner sep=0pt,minimum size=3.5pt] at (0,-0.15) {};
        \node(6) [circle,draw,inner sep=0pt,minimum size=3.5pt] at (1,0.15) {};
        \draw (1)--(3);
        \draw (1)--(4);
        \draw (1)--(5);
        \draw (6)--(3);
        \draw (6)--(4);
        \draw (6)--(5);
    \end{tikzpicture} testing}
			\addplot[very thick,color=brown,mark=None,dashed]   table[x=step,y=Cwomtest4round,col sep=comma] {SWAP1mdifferentarchitectures.csv};
			\addlegendentry{\begin{tikzpicture}[yscale=0.1,xscale=0.05, baseline]
        \node(1) [circle,draw,inner sep=0pt,minimum size=3.5pt] at (-1,0.15) {};
        \node(4) [circle,draw,inner sep=0pt,minimum size=3.5pt] at (0,0.0) {};
        \node(5) [circle,draw,inner sep=0pt,minimum size=3.5pt] at (0,0.3) {};
        \node(7) [circle,draw,inner sep=0pt,minimum size=3.5pt] at (1,-0.15) {};
        \node(8) [circle,draw,inner sep=0pt,minimum size=3.5pt] at (1,0.15) {};
        \node(9) [circle,draw,inner sep=0pt,minimum size=3.5pt] at (1,0.45) {};
        \node(10) [circle,draw,inner sep=0pt,minimum size=3.5pt] at (2,0) {};
        \node(11) [circle,draw,inner sep=0pt,minimum size=3.5pt] at (2,0.3) {};
        \node(6) [circle,draw,inner sep=0pt,minimum size=3.5pt] at (3,0.15) {};
        \draw (1)--(4);
        \draw (1)--(5);
        \draw (7)--(4);
        \draw (7)--(5);
        \draw (8)--(4);
        \draw (8)--(5);
        \draw (9)--(4);
        \draw (9)--(5);
        \draw (7)--(10);
        \draw (7)--(11);
        \draw (8)--(10);
        \draw (8)--(11);
        \draw (9)--(10);
        \draw (9)--(11);
        \draw (6)--(10);
        \draw (6)--(11);
    \end{tikzpicture} testing}
			\addplot[very thick,color=blue,mark=None,dashed]   table[x=step,y=C1mtestround,col sep=comma] {SWAP1mdifferentarchitectures.csv}; 
			\addlegendentry{\begin{tikzpicture}[xscale=0.05,yscale=0.1, baseline]
						\node(1) [circle,draw,inner sep=0pt,minimum size=3.5pt] at (-1,0) {};
						\node(2) [circle,draw,inner sep=0pt,minimum size=3.5 pt] at (-1,0.3) {};
						\node(4) [circle,draw,inner sep=0pt,minimum size=3.5pt] at (0,0) {};
						\node(5) [circle,draw,inner sep=0pt,minimum size=3.5pt] at (0,0.3) {};
						\draw (1)--(4);
						\draw (2)--(4);
						\draw (1)--(5);
						\draw (2)--(5);
						\draw (5) -- (0.3,0.3);
						\draw[out=0,in=0] (0.3,0.3) to (0.3,0.5);
						\draw (0.3,0.5) to (-1.3,0.5);
						\draw[out=180,in=180] (-1.3,0.5) to (-1.3,0.3);
						\draw (-1.3,0.3) -- (2);
					\end{tikzpicture} testing}
				\addplot [domain=0:1000, samples=10, color=gray,]{0.5};
				\addplot [domain=0:1000, samples=10, color=gray,]{1};
		\end{axis}
	\end{tikzpicture}
	\caption[]{\textbf{Comparison of QRNN and feed-forward QNN with different architectures.}
	The plot shows the local cost for pure output states dependent on the training step while different QNNs are trained over $1000$ rounds for the delay by one channel on $N=20$ training pairs with learningrate $\epsilon  \eta=0.06$. Thereby, the colour indicates which QNN is trained. The drawn-through lines show the training, and the dashed lines the testing. }
	\label{fig:SWAP1mdifferentarchitectures}
\end{figure}
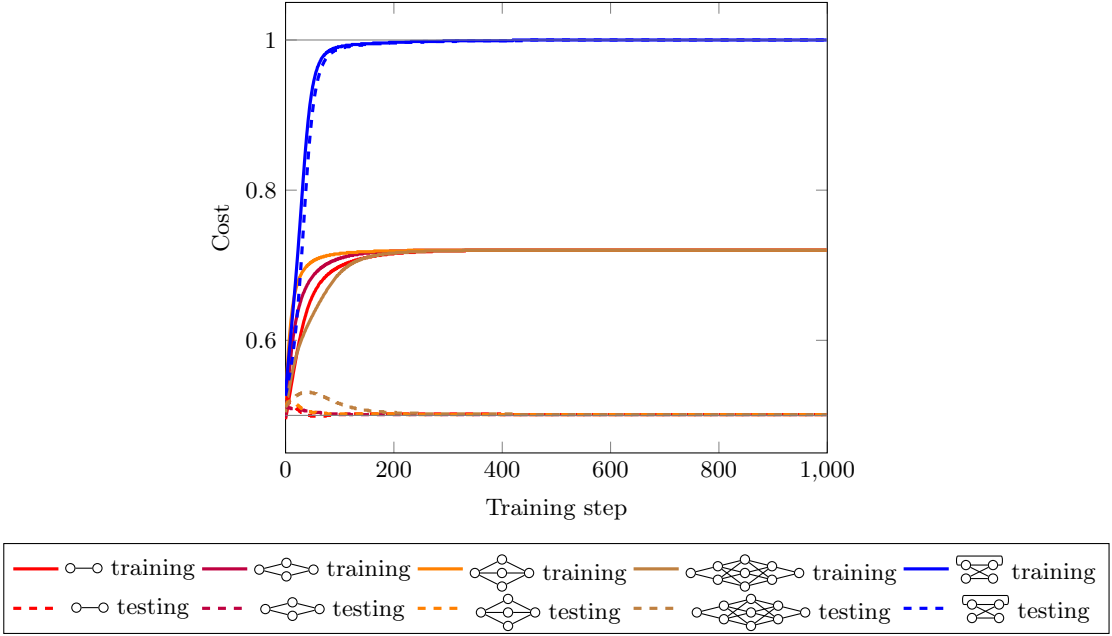

\begin{figure}[htbp]
\centering
		\begin{tikzpicture}
			\begin{axis}[
				width=0.49\linewidth, % Scale the plot to \linewidth
				%grid=major, % Display a grid
				%grid style={dashed,gray!30}, % Set the style
				xlabel= Number of training pairs, % Set the labels
				ylabel= Cost,
				xmin=0,
				xmax=11,
				ymin=0.4,
				ymax=1.05,
			    legend columns=2, 
			    legend style={at={(0.5,-0.2)},anchor=north},
				%legend style={at={(0.7,0.7)},anchor=north}, % Put the legend below the plot
				%x tick label style={rotate=90,anchor=east} % Display labels sideways
				]
				\addplot+[only marks,color=red,mark=triangle,mark size=2.9pt]  table[x=N,y=Cwomround,col sep=comma] {Generalisationbehaviour.csv}; 
				\addlegendentry{\begin{tikzpicture}[scale=0.05, baseline]
						\node(1) [circle,draw,inner sep=0pt,minimum size=3.5pt] at (-1,0.15) {};
						\node(4) [circle,draw,inner sep=0pt,minimum size=3.5pt] at (0,0.15) {};
						\draw (1)--(4);
					\end{tikzpicture} training} 
				\addplot+[only marks,color=blue,mark=triangle,mark size=2.9pt]   table[x=N,y=C1mround,col sep=comma] {Generalisationbehaviour.csv}; 
				\addlegendentry{\begin{tikzpicture}[xscale=0.05,yscale=0.1, baseline]
						\node(1) [circle,draw,inner sep=0pt,minimum size=3.5pt] at (-1,0) {};
						\node(2) [circle,draw,inner sep=0pt,minimum size=3.5 pt] at (-1,0.3) {};
						\node(4) [circle,draw,inner sep=0pt,minimum size=3.5pt] at (0,0) {};
						\node(5) [circle,draw,inner sep=0pt,minimum size=3.5pt] at (0,0.3) {};
						\draw (1)--(4);
						\draw (2)--(4);
						\draw (1)--(5);
						\draw (2)--(5);
						\draw (5) -- (0.3,0.3);
						\draw[out=0,in=0] (0.3,0.3) to (0.3,0.5);
						\draw (0.3,0.5) to (-1.3,0.5);
						\draw[out=180,in=180] (-1.3,0.5) to (-1.3,0.3);
						\draw (-1.3,0.3) -- (2);
					\end{tikzpicture} training}
				\addplot+[only marks,color=red,mark=square,mark size=2pt]   table[x=N,y=Cwomtestround,col sep=comma] {Generalisationbehaviour.csv}; 
				\addlegendentry{\begin{tikzpicture}[scale=0.05, baseline]
						\node(1) [circle,draw,inner sep=0pt,minimum size=3.5pt] at (-1,0.15) {};
						\node(4) [circle,draw,inner sep=0pt,minimum size=3.5pt] at (0,0.15) {};
						\draw (1)--(4);
					\end{tikzpicture} testing} 
				\addplot+[only marks,color=blue,mark=square,mark size=2pt]   table[x=N,y=C1mtestround,col sep=comma] {Generalisationbehaviour.csv}; 
				\addlegendentry{\begin{tikzpicture}[xscale=0.05,yscale=0.1, baseline]
						\node(1) [circle,draw,inner sep=0pt,minimum size=3.5pt] at (-1,0) {};
						\node(2) [circle,draw,inner sep=0pt,minimum size=3.5 pt] at (-1,0.3) {};
						\node(4) [circle,draw,inner sep=0pt,minimum size=3.5pt] at (0,0) {};
						\node(5) [circle,draw,inner sep=0pt,minimum size=3.5pt] at (0,0.3) {};
						\draw (1)--(4);
						\draw (2)--(4);
						\draw (1)--(5);
						\draw (2)--(5);
						\draw (5) -- (0.3,0.3);
						\draw[out=0,in=0] (0.3,0.3) to (0.3,0.5);
						\draw (0.3,0.5) to (-1.3,0.5);
						\draw[out=180,in=180] (-1.3,0.5) to (-1.3,0.3);
						\draw (-1.3,0.3) -- (2);
					\end{tikzpicture} testing}
				\addplot [domain=0:11, samples=10, color=gray,]{0.5};
				\addplot [domain=0:11, samples=10, color=gray,]{1};
			\end{axis}
		\end{tikzpicture}
	\caption[]{\textbf{Generalisation behaviour.} The plot shows the generalisation behaviour of a
	\begin{tikzpicture}[xscale=0.4,yscale=0.6, baseline]
						\node(1) [circle,draw,inner sep=0pt,minimum size=3.5pt] at (-1,0) {};
						\node(2) [circle,draw,inner sep=0pt,minimum size=3.5 pt] at (-1,0.3) {};
						\node(4) [circle,draw,inner sep=0pt,minimum size=3.5pt] at (0,0) {};
						\node(5) [circle,draw,inner sep=0pt,minimum size=3.5pt] at (0,0.3) {};
						\draw (1)--(4);
						\draw (2)--(4);
						\draw (1)--(5);
						\draw (2)--(5);
						\draw (5) -- (0.3,0.3);
						\draw[out=0,in=0] (0.3,0.3) to (0.3,0.5);
						\draw (0.3,0.5) to (-1.3,0.5);
						\draw[out=180,in=180] (-1.3,0.5) to (-1.3,0.3);
						\draw (-1.3,0.3) -- (2);
					\end{tikzpicture}  QRNN with one memory qubit (blue)
     and a \begin{tikzpicture}[yscale=0.6,xscale=0.4, baseline]
        \node(1) [circle,draw,inner sep=0pt,minimum size=3.5pt] at (-1,0.15) {};
        \node(4) [circle,draw,inner sep=0pt,minimum size=3.5pt] at (0,0.15) {};
        \draw (1)--(4);
    \end{tikzpicture} feed-forward QNN (red). We trained the QNNs for $1000$ rounds with a learningrate of $\epsilon  \eta=0.06$ and $N = 1, 2, \dots , 10$ training pairs and evaluated the cost function on the training set (triangles) and for a set of $30$ test pairs (squares) afterwards. This is then averaged over $50$ rounds.}
	\label{fig:generalisationbehaviour} 
\end{figure}
Also, as shown in Figure \ref{fig:generalisationbehaviour}, even \(N=8\) training pairs are sufficient for the simplest QRNN to learn the delay by one channel. The training cost of the QRNN stays at approximately one, and the test cost rises to approximately one with a rising number of training pairs (until \(N=8\)), as expected. The training cost of the feed-forward QNN starts at one as one random state can be learned from another well but then declines with a rising number of training pairs while the test cost stays at about \(0.5\), so the feed-forward QNN does not generalise at all.
\subsubsection{Time evolution}
The other presented task is learning the time evolution of a state governed by a certain Hamiltonian. Thereby, we start with the time evolution \(\ket{\psi (t)}\) of that state, take \(N\) equidistant points out of a first interval \([T_1,T_2]\) and try to predict the time evolution in the next interval \((T_2, T_3]\) where \(T_2=T_1+N\tau\) and \(T_3=T_2+M\tau\). This is done by firstly learning from the training set consisting of the training pairs \((\ket{\psi (t_0+i\tau)},\ket{\psi (t_0+(i+1)\tau)})\) for \(i=1,...,N-1\), then secondly using \(\ket{\psi (t_0+N\tau)}\) as input, hopefully getting something similar to \(\ket{\psi (t_0+(N+1)\tau)}\) as output, thirdly this is then used again as input and so on leading to a prediction for the time evolution in the interval \((T_2, T_3]\). \\
The first example is shown in Figure \ref{fig:sigmaxflip}: we started in state \(\ket{0}\), then went over \(\ket{\phi^-}=\frac{1}{\sqrt{2}}\ket{0}-\frac{i}{\sqrt{2}}\ket{1}\) to \(\ket{1}\), back to \(\ket{0}\) over \(\ket{\phi^-}\) and so on. It can be seen that the QRNN learns and predicts the time evolution a lot better than the feed-forward QNN: The blue circles in Panel (a) are nearly on top of the black ones, and the red ones are completely different. Also, the cost in Panel (b) gets much higher for the QRNN.
\begin{figure}[htbp]
	\begin{subfigure}{0.49\textwidth}
		\begin{tikzpicture}
			\node at (-0.5,6.3) {\textbf{a}};
			\draw[] (3,3) circle (2.675);
			\draw[->] (3,3)--(3,6) node[above]{\(tr(\rho \sigma^z)\)};
			\draw[->] (3,3)--(6,3) node[right]{\(tr(\rho \sigma^y)\)};
			\draw[white] (-1,-1)--(6,-1);
			\node[blue](4) [circle,draw,inner sep=0pt,minimum size=12pt] at (3.094,0.328) {};
			\node[blue](5) [circle,draw,inner sep=0pt,minimum size=12pt] at (0.325,3.005) {};
			\node[blue](6) [circle,draw,inner sep=0pt,minimum size=12pt] at (3.013,5.675) {};
			\node[thick](1) [star,draw,inner sep=0pt,minimum size=12pt] at (3,5.675) {\(1\)};
			\node[thick](2) [star,draw,inner sep=0pt,minimum size=12pt] at (0.325,3) {\(2\)};
			\node[thick](3) [star,draw,inner sep=0pt,minimum size=12pt] at (3,0.325) {\(3\)};
			\node[red](4) [circle,draw,inner sep=0pt,minimum size=12pt] at (0.526,1.994) {\(1\)};
			\node[red](5) [circle,draw,inner sep=0pt,minimum size=12pt] at (0.389,2.609) {\(3\)};
			\node[red](6) [circle,draw,inner sep=0pt,minimum size=12pt] at (0.469,2.35) {\(2\)};
			\draw[->] (2.8,5.77) arc (95:174:2.85);
			\draw[->] (0.23,2.8) arc (185:264:2.85);
			\draw[->] (0.12,3.25) arc (175:93.5:2.87);
			\draw[->] (2.75,0.12) arc (265:183.5:2.87);
			
			\node[thick](7) [star,draw,inner sep=0pt,minimum size=12pt] at (0.9,-0.75){};
			\node (8) at (2.15,-0.75){training data};
			
			\node[blue](10) [circle,draw,inner sep=0pt,minimum size=12pt] at (3.5,-0.75) {};
			\node(11) [circle,draw,inner sep=0pt,minimum size=3.5pt] at (3.9,-0.85) {};
			\node(12) [circle,draw,inner sep=0pt,minimum size=3.5 pt] at (3.9,-0.65) {};
			\node(14) [circle,draw,inner sep=0pt,minimum size=3.5pt] at (4.2,-0.85) {};
			\node(15) [circle,draw,inner sep=0pt,minimum size=3.5pt] at (4.2,-0.65) {};
			\node(16) [circle,draw,inner sep=0pt,minimum size=3.5pt] at (4.5,-0.85) {};
			\node(17) [circle,draw,inner sep=0pt,minimum size=3.5pt] at (4.5,-0.65) {};
			\draw (11)--(14);
			\draw (12)--(14);
			\draw (11)--(15);
			\draw (12)--(15);
			\draw (16)--(14);
			\draw (17)--(14);
			\draw (16)--(15);
			\draw (17)--(15);
			\draw (17) -- (4.3,-0.65);
			\draw[out=0,in=0] (4.6,-0.65) to (4.6,-0.5);
			\draw (4.6,-0.5) to (3.8,-0.5);
			\draw[out=180,in=180] (3.8,-0.5) to (3.8,-0.65);
			\draw (3.8,-0.65) -- (12);
			
			\node[red](20) [circle,draw,inner sep=0pt,minimum size=12pt] at (5.1,-0.75) {};
			\node(21) [circle,draw,inner sep=0pt,minimum size=3.5pt] at (5.5,-0.75) {};
			\node(24) [circle,draw,inner sep=0pt,minimum size=3.5pt] at (5.8,-0.85) {};
			\node(25) [circle,draw,inner sep=0pt,minimum size=3.5pt] at (5.8,-0.65) {};
			\node(26) [circle,draw,inner sep=0pt,minimum size=3.5pt] at (6.1,-0.75) {};
			\draw (21)--(24);
			\draw (21)--(25);
			\draw (26)--(24);
			\draw (26)--(25);
			
			\draw (0.5,-1.2)--(6.5,-1.2)--(6.5,-0.25)--(0.5,-0.25)--(0.5,-1.2);
		\end{tikzpicture}
	\end{subfigure}
	\hfill
	\begin{subfigure}{0.49\textwidth}
		\begin{tikzpicture}
			\node at (-1,6.5) {\textbf{b}};
			\begin{axis}[
				width=\linewidth, % Scale the plot to \linewidth
				%grid=major, % Display a grid
				%grid style={dashed,gray!30}, % Set the style
				xlabel= Training step, % Set the labels
				ylabel= Cost,
				xmin=0,
				xmax=1000,
				ymin=0.25,
				ymax=1.05,
				legend columns=2, 
				legend style={at={(0.51,0.25)},anchor=north},
				%legend style={at={(0.7,0.7)},anchor=north}, % Put the legend below the plot
				%x tick label style={rotate=90,anchor=east} % Display labels sideways
				]
				\addplot+[very thick,color=red,mark=None]  table[x=step,y=Cwomtraininground,col sep=comma] {Csigmaxflip.csv}; 
				\addlegendentry{\begin{tikzpicture}[yscale=0.1,xscale=0.05, baseline]
						\node(1) [circle,draw,inner sep=0pt,minimum size=3.5pt] at (-1,0.15) {};
						\node(4) [circle,draw,inner sep=0pt,minimum size=3.5pt] at (0,0) {};
						\node(5) [circle,draw,inner sep=0pt,minimum size=3.5pt] at (0,0.3) {};
						\node(6) [circle,draw,inner sep=0pt,minimum size=3.5pt] at (1,0.15) {};
						\draw (1)--(4);
						\draw (1)--(5);
						\draw (6)--(4);
						\draw (6)--(5);
					\end{tikzpicture} training}
				\addplot+[very thick,color=blue,mark=None]   table[x=step,y=C1mtraininground,col sep=comma] {Csigmaxflip.csv}; 
				\addlegendentry{\begin{tikzpicture}[yscale=0.1,xscale=0.05, baseline]
						\node(1) [circle,draw,inner sep=0pt,minimum size=3.5pt] at (-1,0) {};
						\node(2) [circle,draw,inner sep=0pt,minimum size=3.5pt] at (-1,0.3) {};
						\node(4) [circle,draw,inner sep=0pt,minimum size=3.5pt] at (0,0) {};
						\node(5) [circle,draw,inner sep=0pt,minimum size=3.5pt] at (0,0.3) {};
						\node(6) [circle,draw,inner sep=0pt,minimum size=3.5pt] at (1,0) {};
						\node(7) [circle,draw,inner sep=0pt,minimum size=3.5pt] at (1,0.3) {};
						\draw (1)--(4);
						\draw (2)--(4);
						\draw (1)--(5);
						\draw (2)--(5);
						\draw (6)--(4);
						\draw (7)--(4);
						\draw (6)--(5);
						\draw (7)--(5);
						\draw (7) -- (1.3,0.3);
						\draw[out=0,in=0] (1.3,0.3) to (1.3,0.5);
						\draw (1.3,0.5) to (-1.3,0.5);
						\draw[out=180,in=180] (-1.3,0.5) to (-1.3,0.3);
						\draw (-1.3,0.3) -- (2);
					\end{tikzpicture} training }
				\addplot+[very thick,color=red,mark=None,dashed]   table[x=step,y=Cwomtrainingpredictionround,col sep=comma] {Csigmaxflip.csv};
				\addlegendentry{\begin{tikzpicture}[yscale=0.1,xscale=0.05, baseline]
						\node(1) [circle,draw,inner sep=0pt,minimum size=3.5pt] at (-1,0.15) {};
						\node(4) [circle,draw,inner sep=0pt,minimum size=3.5pt] at (0,0) {};
						\node(5) [circle,draw,inner sep=0pt,minimum size=3.5pt] at (0,0.3) {};
						\node(6) [circle,draw,inner sep=0pt,minimum size=3.5pt] at (1,0.15) {};
						\draw (1)--(4);
						\draw (1)--(5);
						\draw (6)--(4);
						\draw (6)--(5);
					\end{tikzpicture} prediction}
				\addplot+[very thick,color=blue,mark=None,dashed]   table[x=step,y=C1mtrainingpredictionround,col sep=comma] {Csigmaxflip.csv}; 
				\addlegendentry{\begin{tikzpicture}[yscale=0.1,xscale=0.05, baseline]
						\node(1) [circle,draw,inner sep=0pt,minimum size=3.5pt] at (-1,0) {};
						\node(2) [circle,draw,inner sep=0pt,minimum size=3.5pt] at (-1,0.3) {};
						\node(4) [circle,draw,inner sep=0pt,minimum size=3.5pt] at (0,0) {};
						\node(5) [circle,draw,inner sep=0pt,minimum size=3.5pt] at (0,0.3) {};
						\node(6) [circle,draw,inner sep=0pt,minimum size=3.5pt] at (1,0) {};
						\node(7) [circle,draw,inner sep=0pt,minimum size=3.5pt] at (1,0.3) {};
						\draw (1)--(4);
						\draw (2)--(4);
						\draw (1)--(5);
						\draw (2)--(5);
						\draw (6)--(4);
						\draw (7)--(4);
						\draw (6)--(5);
						\draw (7)--(5);
						\draw (7) -- (1.3,0.3);
						\draw[out=0,in=0] (1.3,0.3) to (1.3,0.5);
						\draw (1.3,0.5) to (-1.3,0.5);
						\draw[out=180,in=180] (-1.3,0.5) to (-1.3,0.3);
						\draw (-1.3,0.3) -- (2);
					\end{tikzpicture} prediction }
				\addplot [domain=0:1000, samples=10, color=gray,]{0.5};
				\addplot [domain=0:1000, samples=10, color=gray,]{1};
				%\legend{training without memory, training with 1 memory qubit, testing without memory, testing with 1 memory qubit}
			\end{axis}
		\end{tikzpicture}
	\end{subfigure}
	\caption[]{\textbf{\(\sigma^x\) flip.}  Panel \textbf{(a)} shows the \(y,z\)-plane of the Bloch sphere. A \begin{tikzpicture}[yscale=0.6,xscale=0.4, baseline]
			\node(1) [circle,draw,inner sep=0pt,minimum size=3.5pt] at (-1,0) {};
			\node(2) [circle,draw,inner sep=0pt,minimum size=3.5pt] at (-1,0.3) {};
			\node(4) [circle,draw,inner sep=0pt,minimum size=3.5pt] at (0,0) {};
			\node(5) [circle,draw,inner sep=0pt,minimum size=3.5pt] at (0,0.3) {};
			\node(6) [circle,draw,inner sep=0pt,minimum size=3.5pt] at (1,0) {};
			\node(7) [circle,draw,inner sep=0pt,minimum size=3.5pt] at (1,0.3) {};
			\draw (1)--(4);
			\draw (2)--(4);
			\draw (1)--(5);
			\draw (2)--(5);
			\draw (6)--(4);
			\draw (7)--(4);
			\draw (6)--(5);
			\draw (7)--(5);
			\draw (7) -- (1.3,0.3);
			\draw[out=0,in=0] (1.3,0.3) to (1.3,0.5);
			\draw (1.3,0.5) to (-1.3,0.5);
			\draw[out=180,in=180] (-1.3,0.5) to (-1.3,0.3);
			\draw (-1.3,0.3) -- (2);
		\end{tikzpicture} QRNN with one memory qubit (blue) and a \begin{tikzpicture}[yscale=0.6,xscale=0.4, baseline]
			\node(1) [circle,draw,inner sep=0pt,minimum size=3.5pt] at (-1,0.15) {};
			\node(4) [circle,draw,inner sep=0pt,minimum size=3.5pt] at (0,0) {};
			\node(5) [circle,draw,inner sep=0pt,minimum size=3.5pt] at (0,0.3) {};
			\node(6) [circle,draw,inner sep=0pt,minimum size=3.5pt] at (1,0.15) {};
			\draw (1)--(4);
			\draw (1)--(5);
			\draw (6)--(4);
			\draw (6)--(5);
		\end{tikzpicture} feed-forward QNN (red) are trained over \(1000\) training rounds with a learningrate of \(\eta\epsilon=0.05\) to learn the time evolution, where the state goes from \(\ket{0}\) (black star 1) over \(\ket{\phi^-}\) (black star 2) to \(\ket{1}\) (black star 3), over \(\ket{\phi^-}\) (black star 2) back to \(\ket{0}\) (black star 3) and so on. To train the networks, we used \(N=10\) training pairs and predicted the time evolution for \(M=20\) points. The blue circles show the prediction for the QRNN, and the red the prediction of the QNN. In the latter case, the evolution goes from 1 to 2 to 1 to 3 to 1 to 2 to 1 to 3 to 1.... Panel \textbf{b} shows the cost in dependence on the training step while training. The drawn-through lines show the cost of the training, and the dashed lines are the ones for the training and prediction.}
	\label{fig:sigmaxflip} 
\end{figure}

A similar example, but with the time evolution going more smoothly from point \(\ket{0}\) to \(\ket{1}\) and back, is the time evolution of \(\ket{\psi (t_0)}=\ket{0}\) under \(H(t)=\frac{\pi}{2}f(t)(1-\sigma^x)\) where 
\[f(t)=\begin{cases}
			\begin{aligned}
				&1 &&,t\in [2k,2k+1)\\
				&-1 &&,t\in [2k+1,2k)
			\end{aligned}
		\end{cases},\ k\in \mathbb{N}\]
is the step function. This is shown in Figure \ref{fig:timeevolution} for \(\tau=0.1\) (for comparison, we used the same Hamiltonian before but with \(\tau=0.5\)).
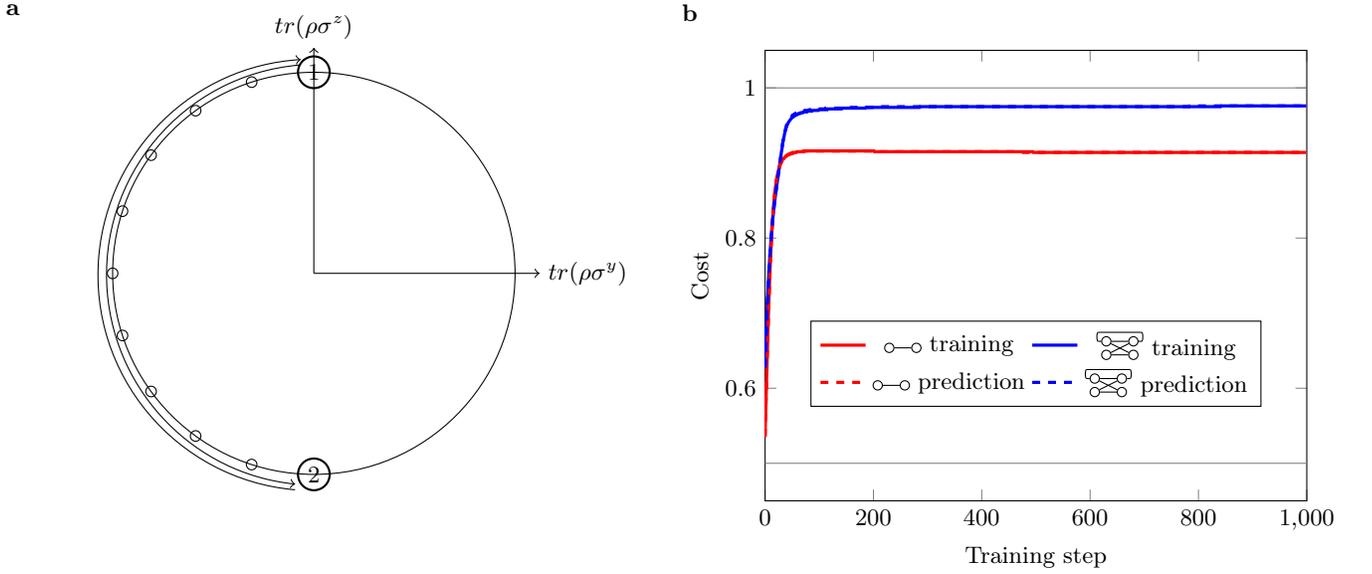
\begin{figure}[htbp]
	\begin{subfigure}{0.49\textwidth}
		\begin{tikzpicture}
			\node at (-1,6.5) {\textbf{a}};
			\draw[] (3,3) circle (2.675);
			\draw[->] (3,3)--(3,6) node[above]{\(tr(\rho \sigma^z)\)};
			\draw[->] (3,3)--(6,3) node[right]{\(tr(\rho \sigma^y)\)};
			\draw[white] (-1,-1)--(6,-1);
			\node[thick](0) [circle,draw,inner sep=0pt,minimum size=12pt] at (3,5.675) {\(1\)};
			\node[thick](10) [circle,draw,inner sep=0pt,minimum size=12pt] at (3,0.325) {\(2\)};
			\node(1) [circle,draw,inner sep=0pt,minimum size=4pt] at (2.173,5.544) {};
			\node(2) [circle,draw,inner sep=0pt,minimum size=4pt] at (1.428,5.164) {};
			\node(3) [circle,draw,inner sep=0pt,minimum size=4pt] at (0.836,4.572) {};
			\node(4) [circle,draw,inner sep=0pt,minimum size=4pt] at (0.456,3.827) {};
			\node(5) [circle,draw,inner sep=0pt,minimum size=4pt] at (0.325,3) {};
			\node(6) [circle,draw,inner sep=0pt,minimum size=4pt] at (0.456,2.173) {};
			\node(7) [circle,draw,inner sep=0pt,minimum size=4pt] at (0.836,1.428) {};
			\node(8) [circle,draw,inner sep=0pt,minimum size=4pt] at (1.428,0.836) {};
			\node(9) [circle,draw,inner sep=0pt,minimum size=4pt] at (2.173,0.456) {};
			\draw[->] (2.8,5.77) arc (95:264:2.8);
			\draw[->] (2.75,0.12) arc (265:93.5:2.87);
		\end{tikzpicture}
	\end{subfigure}
	\hfill
	\begin{subfigure}{0.49\textwidth}
		\begin{tikzpicture}
			\node at (-1,6.5) {\textbf{b}};
			\begin{axis}[
				width=\linewidth, % Scale the plot to \linewidth
				%grid=major, % Display a grid
				%grid style={dashed,gray!30}, % Set the style
				xlabel= Training step, % Set the labels
				ylabel= Cost,
				xmin=0,
				xmax=1000,
				ymin=0.45,
				ymax=1.05,
				legend columns=2, 
				legend style={at={(0.5,0.4)},anchor=north},
				%legend style={at={(0.7,0.6)},anchor=north}, % Put the legend below the plot
				%x tick label style={rotate=90,anchor=east} % Display labels sideways
				]
				\addplot[very thick,color=red,mark=None]  table[x=step,y=Cwomtraininground,col sep=comma] {timeevolution1m.csv}; 
				\addlegendentry{\begin{tikzpicture}[scale=0.05, baseline]
						\node(1) [circle,draw,inner sep=0pt,minimum size=3.5pt] at (-1,0.15) {};
						\node(4) [circle,draw,inner sep=0pt,minimum size=3.5pt] at (0,0.15) {};
						\draw (1)--(4);
					\end{tikzpicture} training} 
				\addplot[very thick,color=blue,mark=None]   table[x=step,y=C1mtraininground,col sep=comma] {timeevolution1m.csv};
				\addlegendentry{\begin{tikzpicture}[xscale=0.05,yscale=0.1, baseline]
						\node(1) [circle,draw,inner sep=0pt,minimum size=3.5pt] at (-1,0) {};
						\node(2) [circle,draw,inner sep=0pt,minimum size=3.5 pt] at (-1,0.3) {};
						\node(4) [circle,draw,inner sep=0pt,minimum size=3.5pt] at (0,0) {};
						\node(5) [circle,draw,inner sep=0pt,minimum size=3.5pt] at (0,0.3) {};
						\draw (1)--(4);
						\draw (2)--(4);
						\draw (1)--(5);
						\draw (2)--(5);
						\draw (5) -- (0.3,0.3);
						\draw[out=0,in=0] (0.3,0.3) to (0.3,0.5);
						\draw (0.3,0.5) to (-1.3,0.5);
						\draw[out=180,in=180] (-1.3,0.5) to (-1.3,0.3);
						\draw (-1.3,0.3) -- (2);
					\end{tikzpicture} training}
				\addplot[very thick,color=red,mark=None,dashed]  table[x=step,y=Cwomtrainingpredictionround,col sep=comma] {timeevolution1m.csv}; 
				\addlegendentry{\begin{tikzpicture}[scale=0.05, baseline]
						\node(1) [circle,draw,inner sep=0pt,minimum size=3.5pt] at (-1,0.15) {};
						\node(4) [circle,draw,inner sep=0pt,minimum size=3.5pt] at (0,0.15) {};
						\draw (1)--(4);
					\end{tikzpicture} prediction} 
				\addplot[very thick,color=blue,mark=None,dashed]   table[x=step,y=C1mtrainingpredictionround,col sep=comma] {timeevolution1m.csv};  
				\addlegendentry{\begin{tikzpicture}[xscale=0.05,yscale=0.1, baseline]
						\node(1) [circle,draw,inner sep=0pt,minimum size=3.5pt] at (-1,0) {};
						\node(2) [circle,draw,inner sep=0pt,minimum size=3.5 pt] at (-1,0.3) {};
						\node(4) [circle,draw,inner sep=0pt,minimum size=3.5pt] at (0,0) {};
						\node(5) [circle,draw,inner sep=0pt,minimum size=3.5pt] at (0,0.3) {};
						\draw (1)--(4);
						\draw (2)--(4);
						\draw (1)--(5);
						\draw (2)--(5);
						\draw (5) -- (0.3,0.3);
						\draw[out=0,in=0] (0.3,0.3) to (0.3,0.5);
						\draw (0.3,0.5) to (-1.3,0.5);
						\draw[out=180,in=180] (-1.3,0.5) to (-1.3,0.3);
						\draw (-1.3,0.3) -- (2);
					\end{tikzpicture} prediction}
				%\addplot[very thick,color=red,mark=None,dotted]   table[x=step,y=Cwomtestround,col sep=comma] {timeevolution1m.csv}; 
				%\addplot[very thick,color=blue,mark=None,dotted]   table[x=step,y=C1mtestround,col sep=comma] {timeevolution1m.csv}; 
				%\legend{training without memory, training with 1 memory qubit, prediction without memory, prediction with 1 memory qubit, testing without memory, testing with 1 memory qubit}
				\addplot [domain=0:1000, samples=10, color=gray,]{0.5};
				\addplot [domain=0:1000, samples=10, color=gray,]{1};
			\end{axis}
		\end{tikzpicture}
	\end{subfigure}
	\caption[]{\textbf{Time evolution of a state with time-dependent Hamiltonian.} Panel \textbf{a} shows the time evolution of \(\ket{0}\) governed by the Hamiltonian \(H(t)=\frac{\pi}{2}f(t)(1-\sigma^x)\) where \(f\) is the step function in the \(y,z\)-plane of the Bloch sphere. Panel \textbf{b} shows the local cost for pure output in dependence on the training step for the learning and prediction of that time evolution. The colour indicates which QNN is trained (blue for a
		\begin{tikzpicture}[xscale=0.4,yscale=0.6, baseline]
			\node(1) [circle,draw,inner sep=0pt,minimum size=3.5pt] at (-1,0) {};
			\node(2) [circle,draw,inner sep=0pt,minimum size=3.5 pt] at (-1,0.3) {};
			\node(4) [circle,draw,inner sep=0pt,minimum size=3.5pt] at (0,0) {};
			\node(5) [circle,draw,inner sep=0pt,minimum size=3.5pt] at (0,0.3) {};
			\draw (1)--(4);
			\draw (2)--(4);
			\draw (1)--(5);
			\draw (2)--(5);
			\draw (5) -- (0.3,0.3);
			\draw[out=0,in=0] (0.3,0.3) to (0.3,0.5);
			\draw (0.3,0.5) to (-1.3,0.5);
			\draw[out=180,in=180] (-1.3,0.5) to (-1.3,0.3);
			\draw (-1.3,0.3) -- (2);
		\end{tikzpicture} QRNN with one memory qubit, red for a \begin{tikzpicture}[yscale=0.6,xscale=0.4, baseline]
		\node(1) [circle,draw,inner sep=0pt,minimum size=3.5pt] at (-1,0.15) {};
		\node(4) [circle,draw,inner sep=0pt,minimum size=3.5pt] at (0,0.15) {};
		\draw (1)--(4);
	\end{tikzpicture} feed-forward QNN). \(N=100\) training states and a learningrate of \(\epsilon \eta=0.06\) are used. The prediction was tested on the next \(M=200\) states in the time evolution. The drawn-through lines show the training cost, and the dashed lines the training and prediction cost (which are directly on top of each other). }
	\label{fig:timeevolution} 
\end{figure}
As can be seen, the QRNN with one memory qubit reaches a cost value of approximately \(0.976\) and the feed-forward QNN one of approximately \(0.931\), so the QRNN does not learn the time evolution as perfectly as the delay by one channel but still learns and predicts well. The feed-forward QNN does this worse but not as bad as when learning the SWAP channel.
\subsubsection{High- and low-pass filters}
In interferometry, information is typically encoded in a relative phase of a superposition state. Let us introduce a short-hand notation for the simplest representatives of such states via
\begin{equation}
    |\phi \rangle = \frac{|\!\:\!\!\downarrow \rangle + e^{i \phi} |\!\:\!\!\uparrow \rangle}{\sqrt{2}}.
\end{equation}

{To avoid being confused by the meaning of $|0\ra$, in this subsection we computational basis states for a single qubit as $|\!\:\!\!\downarrow\ra$ and $|\!\:\!\!\uparrow \ra$ instead of $|0\ra$ and $|1\ra$.} Let us concentrate on the product interferometric inputs and labels. That is, the training data $S$ is composed of $N$ sequences, and $j$th sequence consists of $T_j$ input-label pairs
\begin{equation}
    S=\left\{ \left\{|\phi^j_t \rangle, |\psi^j_t\rangle \right\}_{t=1}^{T_j} \right\}_{j=1}^{N}.
\end{equation}
Imagine some quantity is being measured by both a cheap transportable and a big expensive interferometer. The latter is used as a source of labels, and the task is to design a post-processing quantum filter for the former.

We assume that the cheap interferometer suffers from frequency noise and that a good guess about the correct output can be inferred by looking at the sequence of inputs
\begin{equation} \label{function_noise}
    \psi^j_t \approx f\left( \left\{ \phi^j_t \right\}, t \right)
\end{equation}
where $f$ is a function to be learned. We assume that $\phi^j_t$ are randomly distributed. Such problems often arise in the design of classical filters, e.g. high-pass filters to combat drifts or low-pass filters to combat cracking in music recordings.

{For all numerical examples, we sample each $\phi^j_t$ from the uniform distribution on $[-\pi, \pi]$.}

For classical data $\left\{\left( \phi_{j}^{(k)}, \psi_{j}^{(k)}  \right)_{j=1}^{T_k} \right\}_{k=1}^{N}$ we can expect $f$ to be learned perfectly. This follows from the universality of RNNs. However, this is not the case for the quantum data set
\begin{equation}\label{dat:filter}
    S = \{ S_\alpha \} = \left\{|\!\:\!\!\downarrow\rangle, \left(|\phi^\alpha_t \rangle, |\psi^\alpha_t\rangle \right)_{t=1}^{T_\alpha} \right\}_{\alpha=1}^{N}
\end{equation} 
even if $f$ is affine. Indeed, it is impossible to realise neither (perfect) quantum multiplication (Proposition~\ref{thm:Qmulti}), nor quantum addition Corollary~\ref{thm:NoAdd}. The impossibility of the former directly follows from the definition of a quantum channel (see e.g.~\cite{wolf2012quantum}), while the impossibility of the latter is connected to the No Cloning Theorem (Corollary~\ref{thm:NoCloning} of the more general No Broadcasting Theorem~\ref{thm:NoBroadcast}).

\begin{proposition}[No perfect quantum multiplication]\label{thm:Qmulti}
Consider an open set $I \subseteq \mathbb{R}$ containing $0, \pi/2, \pi, 3 \pi/2$. There is no quantum channel $\varepsilon_a$ such that for a given number $a\notin \{0,1,-1\}$ 
\begin{equation} \label{def:Qmulti}
    \forall \theta \in I \ \ve_a\left(|\theta\rangle \langle \theta |\right) = 
    |a\cdot\theta\rangle \langle a\cdot \theta |
\end{equation}
with \( |\theta \rangle \equiv \frac{|\downarrow \rangle + e^{i \theta} |\uparrow \rangle}{\sqrt{2}}.\)
\end{proposition}
\begin{proof}
    The Operation~\ref{def:Qmulti} is non-linear and, thus, not a quantum channel. 
    
    Indeed, as 
    \be
    2|\theta \ra \la \theta | = \id_2 + e^{i\theta} |\!\:\!\! \uparrow \ra \la \downarrow\!\:\!\! | + e^{-i\theta} |\!\:\!\! \uparrow \ra \la \downarrow\!\:\!\!|,
    \ee
    we can write 
    \be
    |\theta \ra \la \theta | = 
    \alpha |0 \ra \la 0| + \beta|\pi/2 \ra \la \pi/2 | + \gamma|\pi \ra \la \pi |,
    \ee
    where coefficients $\alpha, \beta, \gamma$ satisfy
    \be
        \begin{cases}
            1 = \alpha + \beta + \gamma\\
            e^{i \theta} = \alpha + i \beta - \gamma \quad \Leftrightarrow \alpha = \frac{1 +\cos(\theta) - \sin(\theta)}{2}, \ \beta = \sin(\theta), \ \gamma=\frac{1 - \cos(\theta) - \sin(\theta)}{2}.\\
            e^{-i \theta} = \alpha - i \beta - \gamma\\
        \end{cases}
    \ee
    Let us assume that a perfect quantum multiplier $\varepsilon_a$ exists. Due to the linearity of quantum channels,
    \be \label{eq:nonlin}
        |a \theta \ra \la a \theta | =  
        \frac{1 +\cos(\theta) - \sin(\theta)}{2}|0 \ra \la 0| + \sin(\theta)|a\pi/2 \ra \la a\pi/2 | + \frac{1 - \cos(\theta) - \sin(\theta)}{2}|a\pi \ra \la a\pi |.
    \ee
    Rewriting Equation~\ref{eq:nonlin} in computational basis, we get
    \be \label{eq:nonlin_phase}
    e^{i a \theta} = \frac{1 +\cos(\theta) - \sin(\theta)}{2} + \sin(\theta) e^{i a \pi/2} 
                     + \frac{1 - \cos(\theta) - \sin(\theta)}{2} e^{i a \pi}.
    \ee
    For $\theta = 3 \pi/2$, this yields $1 - e^{i a \pi/2} + e^{i a \pi} + e^{i 3 a \pi/2}=0$. If we multiply this equation by $1 + e^{i a \pi/2}$, we obtain
    \be
        1 - e^{i2 a \pi} = 0 \ \Leftrightarrow \ a \in \mathbb{Z}.
    \ee
    Integer $a$ in Equation~\ref{eq:nonlin_phase} produce
    \be
        e^{i a \theta} = \frac{1 +\cos(\theta) - \sin(\theta)}{2} + \sin(\theta) i^a
                     + \frac{1 - \cos(\theta) - \sin(\theta)}{2} (-1)^a.
    \ee
    For even non-zero $a$, right-hand side is real. There exists $\theta$ such that the left-hand side is non-real, leading to the contradiction.
    
    For odd $a$ Equation~\ref{eq:nonlin_phase} yields
    \be
        e^{i a \theta} = e^{i (-1)^\frac{a-1}{2} \theta}.
    \ee
    As this has to hold for any $\theta \in I$, we can exclude $a \notin \{0, 1, -1\}$.
    
    Finally, for $a \notin \{0, 1, -1\}$ the channels $\varepsilon_a$ exist
    \be
        \varepsilon_0 (\cdot) = \tr(\cdot) |0\ra\la0|, \ \varepsilon_1 (\cdot) = \id (\cdot) \id, \ \varepsilon_{-1} (\cdot) = \sigma_x (\cdot) \sigma_x.
    \ee
\end{proof}

\begin{theorem}[No broadcasting, \cite{NoBroadcast}]\label{thm:NoBroadcast}
    Let $\rho_s\in \T(\h), \; s \in \{0, 1\}$ and $\tilde{\rho}_s \in \T(\h \otimes \h)$ be quantum states such that $\tr_{1}(\tilde{\rho}_s) = \tr_{2}(\tilde{\rho}_s) = \rho_s$.
    A quantum channel $\rho_s \rightarrow \tilde{\rho}_s$ exists if and only if $\rho_0$ and $\rho_1$ commute.
\end{theorem}
\begin{corollary}[No cloning, ~\cite{NoCloning,NoCloning1}]
\label{thm:NoCloning}
    There is no unitary operator $U$ acting on $\h \otimes \h$ and such that for any $| \phi \ra$ and some $| e \ra \in \h$ and real $\alpha(\phi, e)$
    \be
    U | \phi \ra \otimes | e \ra = e^{i\alpha(\phi, e)} 
    | \phi \ra \otimes | \phi \ra.
    \ee
\end{corollary}

\begin{corollary}[No perfect quantum adders] \label{thm:NoAdd}
There is no CPTP $\ve$ such that for any $| x \ra, | y \rangle \in \h$, with different values of continuous parameters $x,y$ corresponding to different quantum states, $\ve\left(|x\rangle\langle x| \otimes |y\rangle \langle y |\right) = \rho_{x, x+y}$ such that $\tr_{1}(\rho_{x, x+y}) = |x + y\rangle \langle x + y|$ and $\tr_{2}(\rho_{x, x+y}) = |x\rangle \langle x|$.
\end{corollary}
\begin{proof}
    If such $\ve$ had existed, $x$ could have been broadcasted by choosing $y=0$. By encoding into a tensor product of a sufficient number of states in the parameterised subspace, this would lead to a possibility of broadcasting an arbitrary state---a contradiction to the No Broadcasting Theorem~\ref{thm:NoBroadcast}.
\end{proof}

Nevertheless, there are no fundamental reasons why the task can not be solved approximately. First, it is often sufficient that the state $| x \ra$ is copied only for some values of $x$, e.g. a finite subset of integers. This is not forbidden by Theorem~\ref{thm:NoBroadcast}, and indeed devices that utilise such limited cloning exist. For example, a machine can use only orthogonal states in some known basis for information processing. As the basis is known, states in such a finite set can be measured without perturbation. Second, there is a bound on how well approximate cloning can be performed; notably, the error vanishes in the classical limit.
\begin{theorem}[Approximate $N$ to $M$ cloning,~\cite{Cloning_MtoN}]\label{thm:CloneApprox}
    An approximate cloning map is a completely positive, unital map $T: \B\left(\h^{\otimes M}\right)\rightarrow \B\left(\h_+^{\otimes N}\right)$ where $\h_+^{\otimes M}$ is a space spanned by $|\phi\ra^{\otimes M}, \; |\phi\ra \in \h$. It is a map on observables, and its dual approximately clones $N$- to $M$-particle states.
    
    For any approximate cloning map of $M$ to $N<M$ observables, the error in one-particle expectation values
    \be
    \Delta(T) = \sup_{\psi, \id \geq a  \in\B(\h),  k}
    \left| \la \psi ^{\otimes N}, T\left(\id^{\otimes (k-1)}\otimes a \otimes \id^{\otimes (M-k)} \right)  \psi^{\otimes N}\ra - \la \psi | a | \psi\ra \right|
    \ee
    is bounded from below as 
    \be
    \Delta \geq \frac{d-1}{d}\left|1 - \frac{N}{N+d}\frac{M+d}{M} \right|,
    \ee
    where $d= \dim(\h)$. The bound is tight.
\end{theorem}
For learning approximate quantum adders, see also~\cite{QaddApprox}.

This dramatically complicates the design of quantum filters and manifests itself in QRNNs typically not reaching fidelity 1 for Data~\ref{dat:filter}. Nevertheless, our examples show that a QRNN is a capable approach to filter design. Let us study the performance of the QRNNs by evaluating their performance on various synthetic data sets with different choices of $f$ in equation~\ref{function_noise}.
    
\paragraph{High-pass filters for drift mitigation} \leavevmode\\
    Let us start with a linear drift where parameters are constant in every sequence
    \begin{equation} \label{dat:drift}
        \psi^j_t = \phi^j_t - v \cdot t
    \end{equation}
    This task requires knowledge about the step number and the correct ordering of inputs. Let us estimate how well a feed-forward network agnostic to this knowledge can perform.
    \begin{claim}\label{Caim_FF_Drift}
   Consider drifts $\psi^j_t = \phi^j_t - f(t)$ for a non-constant function $f(\cdot)$. Let the inputs $\phi^j_t$ be sampled from a distribution that does not depend on $t$. As feed-forward QNNs do not have access to the time step $t$, they cannot correct drifts beyond a random guess of the value of $f$. For $f(t) = v \cdot t$, this yields a guess equivalent to outputting a totally mixed state.
    \end{claim}
    The rigorous versions of the Claim~\ref{Caim_FF_Drift} are contained in both Lemma~\ref{lem:FFforDrift} and Lemma~\ref{lem:DriftProb}; the bound in~\ref{lem:FFforDrift} is slightly stronger and the proof is less complicated, while assumptions in Lemma~\ref{lem:DriftProb} are more general and clearer.
    
    As the network sees only interferometric labels during training, for simplicity, we assume in Lemma~\ref{lem:FFforDrift} that it learns to output a pure state. We will be aided in our pursuits by the following lemma:
    \begin{lem}[Cost function of $\epsilon$-close states] \label{lem:epsilon-close_cost}
        Consider two pairs of states $\rho, \rho^\prime$ and $\sigma, \sigma^\prime$ such that $\tr \left|
        \rho - \sigma  \right|, \ \tr \left| \rho^\prime - \sigma^\prime  \right|  \leq \epsilon$ and a quantum channel $\ve$. 
        Then $\tr[\rho^\prime \ve (\rho)] \leq \tr[\sigma^\prime \ve (\sigma)] + 2 \epsilon$. 
    \end{lem}
    \begin{remark}
    The parameter $\epsilon$ in Lemma~\ref{lem:FFforDrift} tends to fall with both $N$ and $T$ if some assumptions about the distribution of phases $\{\phi^j_t\}$ are satisfied. See Example~\ref{ex:drift_sequences} and Lemma~\ref{lem:DriftProb} for the details.
    \end{remark}
    \begin{proof}
    \be
        \tr[\rho^\prime \ve (\rho)] &=& 
        \tr\left[\sigma^\prime \ve (\sigma)\right] + \tr\left[(\rho^\prime - \sigma^\prime) \ve (\rho )\right] + \tr\left[\sigma^\prime \ve (\rho - \sigma)\right]         \nn\\
        &\leq& \tr\left[\sigma^\prime \ve (\sigma)\right] +  \tr\left|\rho^\prime - \sigma^\prime \right| + \tr\left|\ve (\rho - \sigma) \right|
    \ee
    Quantum channels are contractive~\cite[Theorem 9.2]{Nielsen_Chuang_QI_2010}, i.e. $\tr \left| \ve(\rho) - \ve(\sigma) \right| \leq \tr \left| \rho - \sigma \right|$. Combining contractivity with the assumptions of the lemma, we obtain the desired bound.
    \end{proof}
    
    \begin{lem}[Feed-forward networks can not correct drifts]
    \label{lem:FFforDrift}
    Consider a feed-forward QNN trained on sequences $\left\{\left(|\phi^j_1 \ra, |\phi^j_1 \ra\right), \left( |\phi^j_2 \ra,  |\phi^j_2 - v\ra \right), \left(|\phi^j_3 \ra, |\phi^j_3 - v \cdot 2 \ra \right), \dots, \left(|\phi^j_{T} \ra, |\phi^j_{T} - v \cdot (T-1) \ra \right) \right\}_{j=1}^N$  to output pure states, where \( |\phi \rangle \equiv \frac{|\downarrow \rangle + e^{i \phi} |\uparrow \rangle}{\sqrt{2}}\). Split the phases $\{ \phi^j_t \}$ into disjoint classes $C^i$ such that $\phi^j_t, \phi^k_s \in C^i\ \Rightarrow \ \tr\left|| \phi^j_t \ra \la \phi^j_t | -  |\phi^k_s \ra \la \phi^k_s | \right| \leq  \epsilon$. Denote by $n^i_t$ the number of $j$s for which $\phi^j_t \in C^i$. If there exists such splitting where $ \forall t \ \frac{n^i_t}{|C_i|} \leq \frac{1}{T} + \epsilon$, then the achievable average fidelity of the QNN during training is bounded from above by $\frac{1}{2} + \frac{1}{T \left| 1 - e^{-iv} \right|} + \epsilon \left(2  + \frac{1}{\left| 1 - e^{-iv} \right|}\right)$ and from below by $\frac{1}{2}$.
    \end{lem}
    \begin{proof}
      For inputs $\{|\phi^j_t \ra\}$, consider outputs $\left\{a^*_{\phi^j_t} |\!\:\!\!\downarrow \ra + b^*_{\phi^j_t} |\!\:\!\!\uparrow \ra \right\}$, $\left|a_{\phi^j_t}\right|^2 + \left|b_{\phi^j_t}\right|^2 =1$ that maximize the average fidelity $\bar{F}$ on the full data set
    \be
        \bar{F} = \frac{1}{NT}\sum_{j=1}^N \sum_{t=1}^T \left| \left(a_{\phi^j_t}\la \downarrow\!\:\!\!| + b_{\phi^j_t}\la \uparrow\!\:\!\!| \right) \cdot \left( \frac{|\!\:\!\!\downarrow\ra + e^{i \left(\phi^j_t- v\cdot t\right)}|\!\:\!\!\uparrow \ra}{\sqrt{2}} \right) \right|^2.
    \ee  
  We can rewrite $\bar{F}$ as a convex combination of average fidelities on classes $C^j$ weighted with $\frac{|C^j|}{NT}$
    \be
        \bar{F} &=& \sum_i \frac{|C^i|}{NT} \bar{F}_i,
        \\
        \bar{F}_i &=& \frac{1}{|C^i|}\sum_{\phi^j_t \in C^i}
        \left| \left(a_{\phi^j_t}\la \downarrow\!\:\!\!| + b_{\phi^j_t}\la \uparrow\!\:\!\!| \right) \cdot \left( \frac{|\!\:\!\!\downarrow\ra + e^{i \left(\phi^j_t- v\cdot t \right)}|\!\:\!\!\uparrow \ra}{\sqrt{2}} \right) \right|^2
    \ee
    Let us show that $\bar{F}_i \leq \frac{1}{2} + \frac{1}{T \left| 1 - e^{-iv} \right|} 2 \epsilon$ for every class $C^i$. Define $p^i_t = \frac{n^i_t}{|C^i|}$, then
    \be
        \bar{F}_i = \sum_{t=1}^T p^i_t \left| \left(a_{\phi^j_t}\la \downarrow\!\:\!\!| + b_{\phi^j_t}\la \uparrow\!\:\!\!| \right) \cdot \left( \frac{|\!\:\!\!\downarrow\ra + e^{i \left(\phi^j_t- v\cdot t \right)}|\!\:\!\!\uparrow \ra}{
        \sqrt{2}} \right) \right|^2.
    \ee
    Denote by $|\phi^{(j)} \ra$ an arbitrary state that belongs to $C^j$. A QNN implements a channel. We can invoke Lemma~\ref{lem:epsilon-close_cost}
    \be
        \left| \left(a_{\phi^j_t}\la \downarrow\!\:\!\!| + b_{\phi^j_t}\la \uparrow\!\:\!\!| \right) \cdot \left( \frac{|\!\:\!\!\downarrow\ra + e^{i \left(\phi^j_t- v\cdot t \right)}|\!\:\!\!\uparrow \ra}{
        \sqrt{2}} \right) \right|^2_{\phi^j_t \in C^i} = \nn\\
        \left| \left(a_{\phi^{(i)}}\la \downarrow\!\:\!\!| + b_{\phi^{(i)}}\la \uparrow\!\:\!\!| \right) \cdot \left( \frac{|\!\:\!\!\downarrow\ra + e^{i \left(\phi^{(i)}- v\cdot t \right)}|\!\:\!\!\uparrow \ra}{
        \sqrt{2}} \right) \right|^2 + 2 \epsilon.
    \ee
    By opening the brackets, taking into account $\sum_{t} p^j_t=1$ and introducing phases $\theta^{(j)}: \ a_{\phi^{(j)}}^* \cdot b_{\phi^{(j)}} \equiv \left|a_{\phi^{(j)}}b_{\phi^{(j)}}\right|e^{i \theta^{(j)}}$, we get
    \be
        \bar{F}_{j} = \frac{1}{2} +  \max_{|ab|:|a|^2 + |b|^2 =1} |ab| \max_{\theta^{(j)}} \sum_{t} p^j_t \cos\left(\phi^{(j)} + \theta^j -v \cdot t\right) + 2 \epsilon \nn\\ 
        \prescript{{\textrm{Cauchy's inequality}}}{}\leq 
        \frac{1}{2} + \max_{\theta^{(j)}} \sum_{t} \frac{p^j_t}{4}\left( e^{i\left(\phi^{(j)} + \theta^j - v\cdot t\right)} + c.c. \right) 3 \epsilon.
    \ee
    Unsurprisingly, the maximum is achieved by an interferometric output. By the assumtion of the lemma, $p^j_t \leq \frac{1}{T} + \epsilon$. By summing up a geometrical progression and taking an absolute value,
    \be
        \bar{F}_j \leq \frac{1}{2} + \frac{1 + T \epsilon}{2T} \left| \frac{1-e^{-ivT}}{1 - e^{-iv}} \right| + 2 \epsilon \leq  \frac{1}{2} + \frac{1}{T \left| 1 - e^{-iv} \right|} + \epsilon \left(2  + \frac{1}{\left| 1 - e^{-iv} \right|}\right).
    \ee
    As $\bar{F}$ is the convex combination of $\bar{F}_j$,
    \be
        \bar{F} \leq \max_j \bar{F}_j,
    \ee
    yielding the desired bound from above.
    
    The bound from below on the average fidelity is achieved by a network that outputs the totally mixed state $\mathbbm{1}/2$ or a Haar-random pure state.
    \end{proof}
    
     The conditions of the Lemma~\ref{lem:FFforDrift} can seem unwieldy. Let us consider a couple of examples where conditions in Lemma~\ref{lem:FFforDrift} are or are not satisfied for small $\epsilon$.
    
    \begin{example}
        Consider a data set consisting of a single sequence
        \be 
        \left\{\left(|\phi \ra, |\phi \ra\right), \left( |\phi \ra,  |\phi - v\ra \right), \left(|\phi \ra, |\phi - v \cdot 2 \ra \right), \dots, \left(|\phi \ra, |\phi - v \cdot (T-1) \ra \right) \right\}.
        \ee
        The conditions of Lemma~\ref{lem:FFforDrift} are satisfied with $\epsilon = 0$. 
    \end{example}   
    
  \begin{example}\label{ex:drift_sequences}
        Consider a data set
        \be
        \left\{\left(|\phi^j_1 \ra, |\phi^j_1 \ra\right), \left( |\phi^j_2 \ra,  |\phi^j_2 - v\ra \right), \left(|\phi^j_3 \ra, |\phi^j_3 - v \cdot 2 \ra \right), \dots, \left(|\phi^j_{T} \ra, |\phi^j_{T} - v \cdot (T-1) \ra \right) \right\}_{j=1}^N
        \ee
        where for each $t$ and $j$ $\phi^j_t$s are drawn from the same distribution. Fix some decomposition into classes, then the distributions of probabilities $n^i_t/|C_i|$ are the same for every $t$. Thus, $n^i_t/|C_i|$ converge to $1/T$ as $N\to\infty$ with probability $1$. Hence, the method of proof of Lemma~\ref{lem:FFforDrift} can be applied in this situation.
        
        Does the assignment to different classes work even for a single long enough sequence? See Lemma~\ref{lem:DriftProb} for proof that it is indeed the case, at least if the distribution of phases is uniform.
        %The probability that the conditions of Lemma~\ref{lem:FFforDrift} are satisfied with $\epsilon < \delta$ tend to 1 as the numer of seqences $N \rightarrow \infty$ for any $\delta$.
    \end{example} 
    
    \begin{example}
        Consider a data set consisting of a single sequence
        \be 
        \left\{\left(|\phi \ra, |\phi \ra\right), \left( |\phi + \pi \ra,  |\phi - v\ra \right), \left(|\phi+ \pi \ra, |\phi - v \cdot 2 \ra \right), \dots, \left(|\phi+\pi\left(1 + (-1)^T \right) \ra, |\phi - v \cdot (T-1) \ra \right) \right\}.
        \ee
        The Lemma~\ref{lem:FFforDrift} fails to yield any useful bound. If a splitting $n^i_t$ is independent of $t$, a non-empty class contains phases in even and odd positions in a sequence, thus $\min(\epsilon) = 1$. A feed-forward network can distinguish between even and odd positions in the sequence. Nevertheless, the network can not distinguish between inputs on positions of the same parity. A small modification of the proof consisting of splitting phases on only odd or even positions into classes would yield a useful bound from above. This bound can be obtained by substituting $T \rightarrow \frac{T}{2}$ in the statement of the Lemma~\ref{lem:FFforDrift}.
    \end{example}
    The Lemma~\ref{lem:DriftProb} extends the results to non-pure outputs and clarifies the role of splittings $n^i_t$ in Lemma~\ref{lem:FFforDrift} at the cost of the slightly weaker bound and a more complicated proof.
    
\begin{lem}\label{lem:DriftProb}
	Consider a feed-forward QNN trained on a sequence $\left(|\phi_t \ra, |\phi_t -v(t-1) \ra\right)_{t=1}^T$ with each phase $\phi_t$ being drawn i.i.d. from the uniform measure on the circle and \( |\phi \rangle \equiv \frac{|\downarrow \rangle + e^{i \phi} |\uparrow \rangle}{\sqrt{2}}\). Then for any $\eta>0$, the network's optimal cost function is upper bounded by
	\begin{equation}
	C_{QNN} \le \frac{1}{2} + \frac{c}{\sqrt{\eta}|1-e^{iv}|}\frac{1}{\sqrt[3]{T}}.
	\end{equation}
	with probability greater or equal to $1-\eta$ and for $c>0$ independent of $T$ and $\eta$. Moreover, $C_{QNN}\ge 1/2$ with certainty.
\end{lem}
\begin{proof}
	Let $0<\eps \le 1/2$ and let $\eps' = \floor{1/\eps}^{-1}$ which by definition satisfies $1/\eps' \in\mathbbm{N}$. Note $\eps\le \eps'\le 2\eps$.
	Consider now the partition of $[0,2\pi)$ into the  intervals of length $2\pi \eps'$ given by $I^{\eps}_k= [(k-1)2\pi \eps',k2\pi \eps')$ for $k= 1,\cdots,1/\eps'\in\mathbbm{N}$. Note that by choosing $\eps'$ as above, this partition fully covers $[0,2\pi)$ without a remainder. For each $t=1,\cdots, T$, the $\phi_t$ are taken uniformly at random out the interval $[0,2\pi)$ and independent of each other. For $k=1,\cdots,1/\eps'$ denote by $M_k$ the set of $t=1,\cdots T$ for which $\phi_t\in I^{\eps}_k$ and their cardinality $m_k = |M_k|$. Let us now consider the average output state 
	\begin{equation*}
	\sigma_k = \frac{1}{m_k}\sum_{t \in M_k} \kb{\phi_t - v(t-1)} = \1/2 + \frac{1}{m_k}\sum_{t\in M_k} \left(e^{i(\phi_t -v(t-1))} \kbb{\!\:\!\!\uparrow}{\downarrow\!\:\!\!} + e^{-i(\phi_t -v(t-1))} \kbb{\!\:\!\!\downarrow}{\uparrow\!\:\!\!}\right)
	\end{equation*}
	Using  that for every $t\in M_k$ we have $|e^{i(k-1)2\pi\eps'}-e^{i\phi_t}| \le |1-e^{i 2\pi\eps'}|\le e^{2\pi}\eps'$, we see that $\sigma_k$ is close to
	\begin{equation*}
	\widetilde\sigma_k = \1/2 + \frac{1}{m_k}\sum_{t\in M_k} \left(e^{i((k-1)\eps -v(t-1))} \kbb{\!\:\!\!\uparrow}{\downarrow\!\:\!\!} + e^{-i((k-1)\eps-v(t-1))} \kbb{\!\:\!\!\downarrow}{\uparrow\!\:\!\!}\right)
	\end{equation*}
	as their trace distance can be bounded by 
	\begin{equation*}
	\frac{1}{2}\|\sigma_k-\widetilde\sigma_k\|_1 \le \left|\frac{1}{m_k}\sum_{t\in M_k} e^{i(\phi_t -v(t-1))} -e^{i(k-1)\eps} \frac{1}{m_k}\sum_{t\in M_k} e^{iv(t-1)}\right| \le e^{2\pi}\eps'.
	\end{equation*}
	 The probability that given $t=1,\cdots,T$ satisfies $t\in M_k$  is given by
	$
	p_k = \eps'.
	$
	Note that the probability that the set cardinality of $M_k$ is equal to $n\in\mathbbm{N}$ is given by $\Pr(m_k=n) = p_k^n(1-p_k)^{T-k}{T \choose n}$ and that the conditional probability of $M_k$ being equal to a particular set $Y\subset\{1,\cdots,T\}$ with cardinality $n$ is given by $\Pr\left(M_k= Y| \,|M_k|=n\right) = {T \choose n}^{-1}$. Denoting $$\mu = \frac{1}{T} \sum_{t=1}^T e^{iv(t-1)} = \frac{1-e^{ivT}}{T(1-e^{iv})}$$ and using this and Chebyshev's inequality we find
	\begin{align}
	\nonumber &\Pr\left(\left|\frac{1}{m_k}\sum_{t\in M_k} e^{iv(t-1)}  - \mu\right|\ge \delta \right) \le \frac{\bE\left[\left|\frac{1}{m_k}\sum_{t\in M_k} e^{iv(t-1)} -\mu\right|^2\right]}{\delta^2}
	\\&=\nonumber \sum_{n=1}^T \frac{1}{n^2 \delta^2}\Pr(m_k=n) \Pr\left(M_k= Y| \,|M_k|=n\right)\sum_{\substack{Y\subset\{1,\dots,T\}\\|Y|=n}} \sum_{t,s\in Y} (e^{iv(t-1)}  - \mu)(e^{-iv(s-1)}  - \mu^*) \\&=\nonumber \sum_{n=1}^T \frac{1}{n^2 \delta^2} p_k^n(1-p_k)^{T-n} \sum_{t,s=1}^T\sum_{\substack{Y\subset\{1,\dots,T\}\\|Y|=n, t,s\in Y}}  (e^{iv(t-1)}  - \mu)(e^{-iv(s-1)}  - \mu^*)
	\\&=\nonumber\sum_{n=1}^T \frac{1}{n^2 \delta^2} p_k^n(1-p_k)^{T-n} \left(\sum_{t=1}^T \sum_{\substack{Y\subset\{1,\dots,T\}\\|Y|=n, t\in Y}}|e^{iv(t-1)}  - \mu|^2 + \sum_{t\neq s\in\{1,\cdots T\}}^T\sum_{\substack{Y\subset\{1,\dots,T\}\\|Y|=n, t,s\in Y}}(e^{iv(t-1)}  - \mu)(e^{-iv(s-1)}  - \mu^*)\right)
	\\&=\nonumber\sum_{n=1}^T \frac{1}{n^2 \delta^2} p_k^n(1-p_k)^{T-n} \left(\sum_{t=1}^T {T-1 \choose n-1} |e^{iv(t-1)}  - \mu|^2 + \sum_{t\neq s\in\{1,\cdots T\}}^T{T-2 \choose n-2} (e^{iv(t-1)}  - \mu)(e^{-iv(s-1)}  - \mu^*)\right)
	\\&\le \sum_{n=1}^T \frac{1}{n^2 \delta^2} p_k^n(1-p_k)^{T-n} 2T{T-1 \choose n-1}  = \sum_{n=1}^T \frac{2}{n \delta^2} p_k^n(1-p_k)^{T-n} {T \choose n} \label{bound:Chebushev}
	\end{align}
	Here, we have used for the inequality in the last line that
	\begin{equation*}
	\sum_{t\neq s\in\{1,\cdots T\}}^T (e^{iv(t-1)}  - \mu)(e^{-iv(s-1)}  - \mu^*) = \sum_{t=1}^T(e^{iv(t-1)}  - \mu)(\mu^* -e^{-iv(t-1)} ) =  -\sum_{t=1}^T |e^{iv(t-1)}  - \mu|^2 \le 0
	\end{equation*}
	We can use $\int_{-\infty}^{0} dx \left(e^{x}\right)^n = \frac{1}{n}$ and $\left(p_k e^x + (1-p_k)\right)^T = \sum_{n=1}^T p_k e^{nx} (1-p_k)^{T-n} {T \choose n}$ to rewrite
	\begin{equation*}
	     \sum_{n=1}^T \frac{p_k^n}{n} (1-p_k)^{T-n} {T \choose n} = \int_{-\infty}^{0} dx \left[\left(p_k e^x + (1-p_k)\right)^T - (1-p_k)^T  \right].
	\end{equation*}
	After reminding ourselves that $a^T - b^T = (a-b)\sum_{n=0}^{T-1} a^n b^{T-n-1}$, we obtain
	\begin{equation*}
	    \sum_{n=1}^T \frac{p_k^n}{n} (1-p_k)^{T-n} {T \choose n} = p \int_{-\infty}^{0}dx e^x \sum_{n=0}^{T-1} \left(p_k e^x + (1-p_k)\right)^n (1-p_k)^{T-n-1}.
	\end{equation*}
	Evaluating the integrals
	\begin{equation*}
	    \int_{-\infty}^{0}dx e^x \left(p_k e^x + (1-p_k)\right)^n
	    =_{x \rightarrow p_k e^x} \frac{1}{p_k}\int_{0}^{p_k}dy (x + 1-p_k)^n
	    =_{x \rightarrow x+1-p_k} \frac{1}{p_k}\int_{1-p_k}^{1}dy x^n = \frac{1 - (1-p_k)^{n+1}}{p_k(n+1)},
	\end{equation*}
	we get
	\begin{equation*}
	    \sum_{n=1}^T \frac{p_k^n}{n} (1-p_k)^{T-n} {T \choose n} = \sum_{n=0}^{T-1} \frac{(1-p_k)^{T-n-1} - (1-p_k)^T}{n+1}.
	\end{equation*}
	This, combined with \eqref{bound:Chebushev} yields the bound
	\be
	    \Pr\left(\left|\frac{1}{m_k}\sum_{t\in M_k} e^{iv(t-1)}  - \mu\right|\ge \delta \right) \le
	    \frac{2(1-p_k)^T}{\delta^2} \sum_{n=1}^T \frac{(1-p_k)^{-n}-1}{n} \le  \frac{2(1-p_k)^T}{\delta^2} \sum_{n=1}^T\frac{(1-p_k)^{-n}}{n}
	\ee
	Let now $s=1,\cdots, T$. From the above, we find
	\begin{align*}
	&\Pr\left(\left|\frac{1}{m_k}\sum_{t\in M_k} e^{iv(t-1)}  - \mu\right|\ge \delta \right) \le \frac{2(1-p_k)^T}{\delta^2} \sum_{n=1}^s\frac{(1-p_k)^{-n}}{n} + \frac{2(1-p_k)^T}{\delta^2} \sum_{n=s+1}^T\frac{(1-p_k)^{-n}}{n} \\&\le \frac{2s(1-p_k)^{T-s}}{\delta^2}+ \frac{2}{\delta^2 s} \sum_{n=0}^{T-s-1}(1-p_k)^{n}  \le \frac{2s(1-p_k)^{T-s}}{\delta^2}+ \frac{2}{\delta^2 s p_k},
	\end{align*}
	where for the last inequality, we have upper bounded the sum with the corresponding geometric series.
	Picking now for $0 <\eta\le 1$ fixed $\delta=\sqrt{ \frac{2}{\eta}\left(s(1-p_k)^{T-s}+ \frac{1}{s p_k}\right)}$,  this gives that with probability smaller or equal $\eta$ we have
	\begin{align*}
	 \left|\frac{1}{m_k}\sum_{t\in M_k} e^{iv(t-1)} -\mu\right| \ge \sqrt{ \frac{2}{\eta}\left(s(1-p_k)^{T-s}+ \frac{1}{s p_k}\right)}.
	\end{align*}
	Hence, with probability greater or equal to $1-\eta$ 
	\begin{align}
	\label{eq:PhaseSumBound}
	 \left|\frac{1}{m_k}\sum_{t\in M_k} e^{iv(t-1)} \right| \le |\mu| + \sqrt{ \frac{2}{\eta}\left(s(1-p_k)^{T-s}+ \frac{1}{s p_k}\right)}.
	\end{align}
	Moreover, note that as above for $t\in M_k$ we have $\frac{1}{2}\|\kb{\phi_t} - \kb{(k-1)\eps'}\|_1 \le e^{2\pi}\eps'$ and that since any $\E$ CPTP is a contraction also $\frac{1}{2}\|\E(\kb{\phi_t}) - \E(\kb{(k-1)\eps'})\|_1 \le e^{2\pi}\eps'$.
	Using this, we see that the optimal fidelity of the network is bounded by
	\begin{align*}
	C_{QNN} &\le  \max_{\E \text{ CPTP}}\frac{1}{T}\sum_{t=1}^T \langle \phi_t- v(t-1)|\E(\kb{\phi_t})|\phi_t-v(t-1)\rangle \\&= \max_{\E \text{ CPTP}}\frac{1}{T}\sum_{k=1}^{1/\eps'}\sum_{t\in M_k} \langle \phi_t- v(t-1)|\E(\kb{\phi_t})|\phi_t-v(t-1)\rangle 
	\\&\le \frac{1}{T}\sum_{k=1}^{1/\eps'}\max_{\E \text{ CPTP}}\sum_{t\in M_k} \langle \phi_t- v(t-1)|\E(\kb{\phi_t})|\phi_t-v(t-1)\rangle \\&\le \frac{1}{T}\sum_{k=1}^{1/\eps'}\max_{\rho \text{ state}}\sum_{t\in M_k} \langle \phi_t- v(t-1)|\rho|\phi_t-v(t-1)\rangle + e^{2\pi}\eps' \\&= \frac{1}{T}\sum_{k=1}^{1/\eps'}m_k\max_{\rho \text{ state}}\tr\left(\rho\sigma_k\right) + e^{2\pi}\eps' \le \frac{1}{T}\sum_{k=1}^{1/\eps'}m_k\max_{\rho \text{ state}}\tr\left(\rho\widetilde\sigma_k\right) + 2e^{2\pi}\eps'\\& = \frac{1}{T}\sum_{k=1}^{1/\eps'}m_k \frac{1}{2}\left(1+	\left|\frac{1}{m_k}\sum_{t\in M_k} e^{iv(t-1))}\right| \right) +2e^{2\pi}\eps'
	\end{align*}
	Combing this with \eqref{eq:PhaseSumBound} we get that with probability greater or equal to $1-\eta$
	\begin{align*}
	  C_{QNN} &\le  \frac{1}{2} (1+|\mu|) +\frac{1}{2}\frac{1}{T}\sum_{k=1}^{1/\eps'}m_k\sqrt{ \frac{2}{\eta}\left(s(1-p_k)^{T-s}+ \frac{1}{s p_k}\right)} +2e^{2\pi}\eps' \\& \le\frac{1}{2}+ \frac{1}{T|1-e^{iv}|} +\frac{1}{2}\frac{1}{T}\sum_{k=1}^{1/\eps'}m_k\sqrt{ \frac{2}{\eta}\left(s(1-p_k)^{T-s}+ \frac{1}{s p_k}\right)} +2e^{2\pi}\eps' \\&\le \frac{1}{2}+ \frac{1}{T|1-e^{iv}|} +\frac{1}{2}\frac{1}{T}\sum_{k=1}^{1/\eps'}m_k\sqrt{ \frac{2}{\eta}\left(se^{-(T-s)p_k}+ \frac{1}{s p_k}\right)} +2e^{2\pi}\eps',
	\end{align*}
	where in the last line we have used the elementary inequality $\ln(1-p_k) \le -p_k$. 
 Noting that $p_k = \eps'$ and $\frac{1}{T}\sum_{k=1}^{1/\eps'}m_k=1$ and picking $s =T/2$ gives
	\begin{align*}
	   C_{QNN} &\le\frac{1}{2}+ \frac{1}{T|1-e^{iv}|} +\frac{1}{2}\sqrt{ \frac{2}{\eta}\left(\frac{Te^{-T\eps'/2}}{2}+ \frac{2}{T \eps'}\right)} +2e^{2\pi}\eps'.
	\end{align*}
	Picking now $\eps = \frac{2}{\sqrt[3]{T}}$, remembering $\eps \le \eps'\le 2\eps$ and noting $Te^{-T\eps'/2}/2\le Te^{-T^{2/3}}/2\le 1/(2T^{2/3}) \le 2/(T \eps')\le 1/T^{2/3}$ for all $T\ge 1$ gives
	\begin{align*}
	   C_{QNN} &\le\frac{1}{2}+ \frac{1}{T|1-e^{iv}|} +\frac{1}{2}\sqrt{ \frac{4}{\eta} \frac{1}{T^{2/3}}} +\frac{4e^{2\pi}}{\sqrt[3]{T}} \le \frac{1}{2}+ \frac{(4e^{2\pi}+1)}{\sqrt{\eta}|1-e^{iv}|}\frac{1}{\sqrt[3]{T}},
	\end{align*}
	which is the desired upper bound.
	
	For the lower bound on $C_{QNN}$, consider a QNN with the constant output being the completely mixed state $\1/2$.
\end{proof}
    
    With memory consisting of $m$ qubits, QRNNs can be trained to cancel drifts on sequences of length $2^m$ exactly (more precisely--up to the numerical precision). The circuit of such a QRNN can be seen in Figure~\ref{fig:Drift_exact}. The memory stores a binary representation in the computational basis of the position in a sequence $t$. At each time step, the QRNN applies $e^{-ivt\frac{\sigma_z}{2}}$ and adds 1 to the stored value.
    \begin{figure}[htbp]
        \centering
        \includegraphics[scale=1]{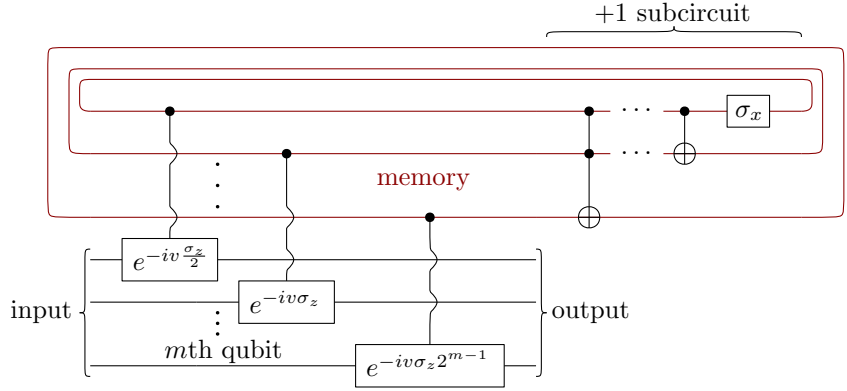}
        \caption{A circuit that cancels a linear drift with speed $v$ exactly. The memory is initialised as $\ket{\downarrow}^{\otimes m}$. Here we use the standard notation for controlled rotations, multi-controlled CNOTs and Pauli matrices.}
        \label{fig:Drift_exact}
    \end{figure}
    
    We expect the network's performance will deteriorate as test sequences become longer. Indeed, even a slight misestimation of $v$ results in fidelity approaching $0.5$ for long enough sequences.
    
    We can see all of the effects mentioned above: the limit on the achievable fidelity for feed-forward networks, QRNNs superseding this limit---increasingly so with the memory size, and performance decaying with test sequence length in Figure~\ref{fig:LinDrift}.
    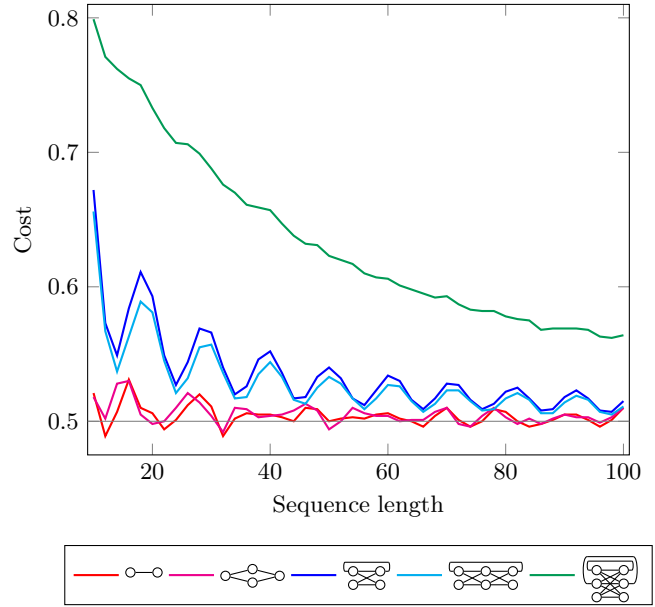
\begin{figure}[htbp]
	\centering
	\begin{tikzpicture}
		\begin{axis}[
			width=0.49\linewidth, % Scale the plot to \linewidth
			%grid=major, % Display a grid
			%grid style={dashed,gray!30}, % Set the style
			xlabel= Sequence length, % Set the labels
			ylabel= Cost,
			xmin=9,
			xmax=101,
			ymin=0.475,
			ymax=0.81,
			legend columns=5, 
			legend style={at={(0.5,-0.2)},anchor=north}, % Put the legend below the plot
			%x tick label style={rotate=90,anchor=east} % Display labels sideways
			]
			\addplot[thick,color=red,mark=None]  table[x=x,y=i1m0round,col sep=comma] {fidelity_LinearDrift.csv}; 
			\addlegendentry{\begin{tikzpicture}[yscale=0.1,xscale=0.05, baseline]
					\node(1) [circle,draw,inner sep=0pt,minimum size=3.5pt] at (-1,0.15) {};
					\node(4) [circle,draw,inner sep=0pt,minimum size=3.5pt] at (0,0.15) {};
					\draw (1)--(4);
			\end{tikzpicture}}
			\addplot[thick,color=magenta,mark=None]  table[x=x,y=i1m0h2round,col sep=comma] {fidelity_LinearDrift.csv}; 
			\addlegendentry{\begin{tikzpicture}[yscale=0.1,xscale=0.05, baseline]
					\node(1) [circle,draw,inner sep=0pt,minimum size=3.5pt] at (-1,0.15) {};
					\node(4) [circle,draw,inner sep=0pt,minimum size=3.5pt] at (0,0) {};
					\node(5) [circle,draw,inner sep=0pt,minimum size=3.5pt] at (0,0.3) {};
					\node(6) [circle,draw,inner sep=0pt,minimum size=3.5pt] at (1,0.15) {};
					\draw (1)--(4);
					\draw (1)--(5);
					\draw (6)--(4);
					\draw (6)--(5);
			\end{tikzpicture}}
			\addplot[thick,color=blue,mark=None]  table[x=x,y=i1m1round,col sep=comma] {fidelity_LinearDrift.csv}; 
			\addlegendentry{\begin{tikzpicture}[yscale=0.1,xscale=0.05, baseline]
					\node(1) [circle,draw,inner sep=0pt,minimum size=3.5pt] at (-1,0) {};
					\node(2) [circle,draw,inner sep=0pt,minimum size=3.5 pt] at (-1,0.3) {};
					\node(4) [circle,draw,inner sep=0pt,minimum size=3.5pt] at (0,0) {};
					\node(5) [circle,draw,inner sep=0pt,minimum size=3.5pt] at (0,0.3) {};
					\draw (1)--(4);
					\draw (2)--(4);
					\draw (1)--(5);
					\draw (2)--(5);
					\draw (5) -- (0.3,0.3);
					\draw[out=0,in=0] (0.3,0.3) to (0.3,0.5);
					\draw (0.3,0.5) to (-1.3,0.5);
					\draw[out=180,in=180] (-1.3,0.5) to (-1.3,0.3);
					\draw (-1.3,0.3) -- (2);
			\end{tikzpicture}}
			\addplot[ thick,color=cyan,mark=None]  table[x=x,y=i1m1h2round,col sep=comma] {fidelity_LinearDrift.csv}; 
			\addlegendentry{\begin{tikzpicture}[yscale=0.1,xscale=0.05, baseline]
					\node(1) [circle,draw,inner sep=0pt,minimum size=3.5pt] at (-1,0) {};
					\node(2) [circle,draw,inner sep=0pt,minimum size=3.5pt] at (-1,0.3) {};
					\node(4) [circle,draw,inner sep=0pt,minimum size=3.5pt] at (0,0) {};
					\node(5) [circle,draw,inner sep=0pt,minimum size=3.5pt] at (0,0.3) {};
					\node(6) [circle,draw,inner sep=0pt,minimum size=3.5pt] at (1,0) {};
					\node(7) [circle,draw,inner sep=0pt,minimum size=3.5pt] at (1,0.3) {};
					\draw (1)--(4);
					\draw (2)--(4);
					\draw (1)--(5);
					\draw (2)--(5);
					\draw (6)--(4);
					\draw (7)--(4);
					\draw (6)--(5);
					\draw (7)--(5);
					\draw (7) -- (1.3,0.3);
					\draw[out=0,in=0] (1.3,0.3) to (1.3,0.5);
					\draw (1.3,0.5) to (-1.3,0.5);
					\draw[out=180,in=180] (-1.3,0.5) to (-1.3,0.3);
					\draw (-1.3,0.3) -- (2);
			\end{tikzpicture} }
			\addplot[thick,color=ForestGreen,mark=None]  table[x=x,y=i1m2round,col sep=comma] {fidelity_LinearDrift.csv}; 
			\addlegendentry{\begin{tikzpicture}[yscale=0.1,xscale=0.05, baseline]
					\node(1) [circle,draw,inner sep=0pt,minimum size=3.5pt] at (-1,-0.15) {};
					\node(2) [circle,draw,inner sep=0pt,minimum size=3.5pt] at (-1,0.15) {};
					\node(3) [circle,draw,inner sep=0pt,minimum size=3.5pt] at (-1,0.45) {};
					\node(4) [circle,draw,inner sep=0pt,minimum size=3.5pt] at (0,-0.15) {};
					\node(5) [circle,draw,inner sep=0pt,minimum size=3.5pt] at (0,0.15) {};
					\node(6) [circle,draw,inner sep=0pt,minimum size=3.5pt] at (0,0.45) {};
					\draw (1)--(4);
					\draw (2)--(4);
					\draw (3)--(4);
					\draw (1)--(5);
					\draw (2)--(5);
					\draw (3)--(5);
					\draw (1)--(6);
					\draw (2)--(6);
					\draw (3)--(6);
					\draw (6) -- (0.3,0.45);
					\draw[out=0,in=0] (0.3,0.45) to (0.3,0.65);
					\draw (0.3,0.65) to (-1.3,0.65);
					\draw[out=180,in=180] (-1.3,0.65) to (-1.3,0.45);
					\draw (-1.3,0.45) -- (3);
					\draw (5) -- (0.35,0.15);
					\draw[out=0,in=0] (0.35,0.15) to (0.35,0.75);
					\draw (0.35,0.75) to (-1.35,0.75);
					\draw[out=180,in=180] (-1.35,0.75) to (-1.35,0.15);
					\draw (-1.35,0.15) -- (2);
			\end{tikzpicture}}
			\addplot [domain=9:101, samples=10, color=gray,]{0.5};
		\end{axis}
	\end{tikzpicture}
	 \caption[]{Linear drift mitigation for different network topologies. As expected, feed-forward networks can not be trained to fidelities above 0.5, yet memory access improves performance. The performance increases with the memory size. {Here the drift velocity $v =0.6$.}}
	\label{fig:LinDrift}
\end{figure}

    Similar reasoning can be applied to non-linear drifts
    \begin{equation}
        \psi^j_t = \phi^j_t - g(t),
    \end{equation}
    where $g(t)$ is some function. Results for $g(t) = v \cdot \sqrt{t}$ can be seen in Figure~\ref{fig:SqrtDrift} and are qualitatively similar to the case of linear drifts (Figure~\ref{fig:LinDrift}).
    	
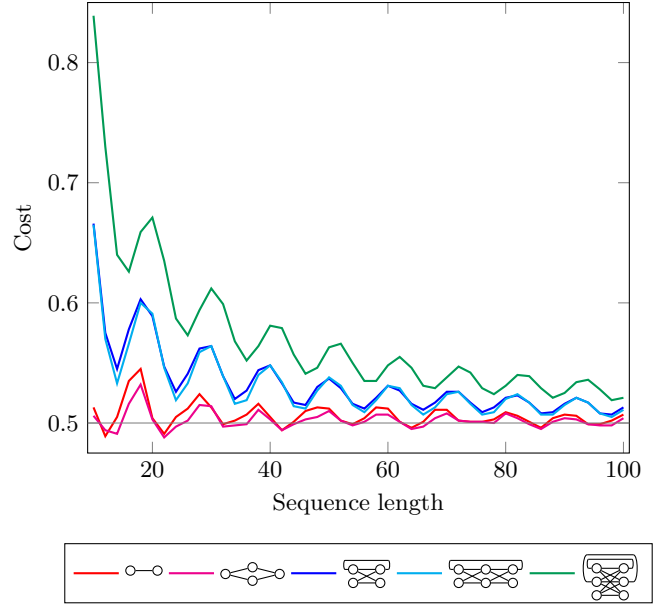
\begin{figure}[htbp]
	\centering
	\begin{tikzpicture}
		\begin{axis}[
			width=0.49\linewidth, % Scale the plot to \linewidth
			%grid=major, % Display a grid
			%grid style={dashed,gray!30}, % Set the style
			xlabel= Sequence length, % Set the labels
			ylabel= Cost,
			xmin=9,
			xmax=101,
			ymin=0.475,
			ymax=0.85,
			legend columns=5, 
			legend style={at={(0.5,-0.2)},anchor=north}, % Put the legend below the plot
			%x tick label style={rotate=90,anchor=east} % Display labels sideways
			]
			\addplot[thick,color=red,mark=None]  table[x=x,y=i1m0round,col sep=comma] {fidelity_SqrtDrift.csv}; 
			\addlegendentry{\begin{tikzpicture}[yscale=0.1,xscale=0.05, baseline]
					\node(1) [circle,draw,inner sep=0pt,minimum size=3.5pt] at (-1,0.15) {};
					\node(4) [circle,draw,inner sep=0pt,minimum size=3.5pt] at (0,0.15) {};
					\draw (1)--(4);
			\end{tikzpicture}}
			\addplot[thick,color=magenta,mark=None]  table[x=x,y=i1m0h2round,col sep=comma] {fidelity_SqrtDrift.csv}; 
			\addlegendentry{\begin{tikzpicture}[yscale=0.1,xscale=0.05, baseline]
					\node(1) [circle,draw,inner sep=0pt,minimum size=3.5pt] at (-1,0.15) {};
					\node(4) [circle,draw,inner sep=0pt,minimum size=3.5pt] at (0,0) {};
					\node(5) [circle,draw,inner sep=0pt,minimum size=3.5pt] at (0,0.3) {};
					\node(6) [circle,draw,inner sep=0pt,minimum size=3.5pt] at (1,0.15) {};
					\draw (1)--(4);
					\draw (1)--(5);
					\draw (6)--(4);
					\draw (6)--(5);
			\end{tikzpicture}}
			\addplot[thick,color=blue,mark=None]  table[x=x,y=i1m1round,col sep=comma] {fidelity_SqrtDrift.csv}; 
			\addlegendentry{\begin{tikzpicture}[yscale=0.1,xscale=0.05, baseline]
					\node(1) [circle,draw,inner sep=0pt,minimum size=3.5pt] at (-1,0) {};
					\node(2) [circle,draw,inner sep=0pt,minimum size=3.5 pt] at (-1,0.3) {};
					\node(4) [circle,draw,inner sep=0pt,minimum size=3.5pt] at (0,0) {};
					\node(5) [circle,draw,inner sep=0pt,minimum size=3.5pt] at (0,0.3) {};
					\draw (1)--(4);
					\draw (2)--(4);
					\draw (1)--(5);
					\draw (2)--(5);
					\draw (5) -- (0.3,0.3);
					\draw[out=0,in=0] (0.3,0.3) to (0.3,0.5);
					\draw (0.3,0.5) to (-1.3,0.5);
					\draw[out=180,in=180] (-1.3,0.5) to (-1.3,0.3);
					\draw (-1.3,0.3) -- (2);
			\end{tikzpicture}}
			\addplot[ thick,color=cyan,mark=None]  table[x=x,y=i1m1h2round,col sep=comma] {fidelity_SqrtDrift.csv}; 
			\addlegendentry{\begin{tikzpicture}[yscale=0.1,xscale=0.05, baseline]
					\node(1) [circle,draw,inner sep=0pt,minimum size=3.5pt] at (-1,0) {};
					\node(2) [circle,draw,inner sep=0pt,minimum size=3.5pt] at (-1,0.3) {};
					\node(4) [circle,draw,inner sep=0pt,minimum size=3.5pt] at (0,0) {};
					\node(5) [circle,draw,inner sep=0pt,minimum size=3.5pt] at (0,0.3) {};
					\node(6) [circle,draw,inner sep=0pt,minimum size=3.5pt] at (1,0) {};
					\node(7) [circle,draw,inner sep=0pt,minimum size=3.5pt] at (1,0.3) {};
					\draw (1)--(4);
					\draw (2)--(4);
					\draw (1)--(5);
					\draw (2)--(5);
					\draw (6)--(4);
					\draw (7)--(4);
					\draw (6)--(5);
					\draw (7)--(5);
					\draw (7) -- (1.3,0.3);
					\draw[out=0,in=0] (1.3,0.3) to (1.3,0.5);
					\draw (1.3,0.5) to (-1.3,0.5);
					\draw[out=180,in=180] (-1.3,0.5) to (-1.3,0.3);
					\draw (-1.3,0.3) -- (2);
			\end{tikzpicture} }
			\addplot[thick,color=ForestGreen,mark=None]  table[x=x,y=i1m2round,col sep=comma] {fidelity_SqrtDrift.csv}; 
			\addlegendentry{\begin{tikzpicture}[yscale=0.1,xscale=0.05, baseline]
					\node(1) [circle,draw,inner sep=0pt,minimum size=3.5pt] at (-1,-0.15) {};
					\node(2) [circle,draw,inner sep=0pt,minimum size=3.5pt] at (-1,0.15) {};
					\node(3) [circle,draw,inner sep=0pt,minimum size=3.5pt] at (-1,0.45) {};
					\node(4) [circle,draw,inner sep=0pt,minimum size=3.5pt] at (0,-0.15) {};
					\node(5) [circle,draw,inner sep=0pt,minimum size=3.5pt] at (0,0.15) {};
					\node(6) [circle,draw,inner sep=0pt,minimum size=3.5pt] at (0,0.45) {};
					\draw (1)--(4);
					\draw (2)--(4);
					\draw (3)--(4);
					\draw (1)--(5);
					\draw (2)--(5);
					\draw (3)--(5);
					\draw (1)--(6);
					\draw (2)--(6);
					\draw (3)--(6);
					\draw (6) -- (0.3,0.45);
					\draw[out=0,in=0] (0.3,0.45) to (0.3,0.65);
					\draw (0.3,0.65) to (-1.3,0.65);
					\draw[out=180,in=180] (-1.3,0.65) to (-1.3,0.45);
					\draw (-1.3,0.45) -- (3);
					\draw (5) -- (0.35,0.15);
					\draw[out=0,in=0] (0.35,0.15) to (0.35,0.75);
					\draw (0.35,0.75) to (-1.35,0.75);
					\draw[out=180,in=180] (-1.35,0.75) to (-1.35,0.15);
					\draw (-1.35,0.15) -- (2);
			\end{tikzpicture}}
			\addplot [domain=9:101, samples=10, color=gray,]{0.5};
		\end{axis}
	\end{tikzpicture}
	\caption[]{Non-linear drift mitigation for different network topologies, $\psi^j_t = \phi^j_t - 0.6 \cdot\sqrt{t}$. }
	\label{fig:SqrtDrift}
\end{figure}

\paragraph{Low-pass filters for smoothing}\leavevmode\\
High-frequency noise can be mitigated by an average running method such as exponential smoothing
    \be \label{eq:HiFreq}
        \psi^j_0 &=& \phi^j_0 \nn\\
        \psi^j_t &=& \alpha \cdot \phi^j_{t} 
                 + (1-\alpha) \cdot \psi^j_{t-1}, 
        \quad 0 < \alpha < 1.
    \ee
This task explicitly requires knowledge of previous states. As such, a feed-forward network does seem ill-equipped for the job. Nevertheless, as there are no quantum adders and multipliers, we do not expect even QRNNs with large memory to solve exponential smoothing with fidelity one. Our numerical results suggest a gap in the performance of recurrent and feed-forward QNNs; see Figure~\ref{fig:HiFreq}.
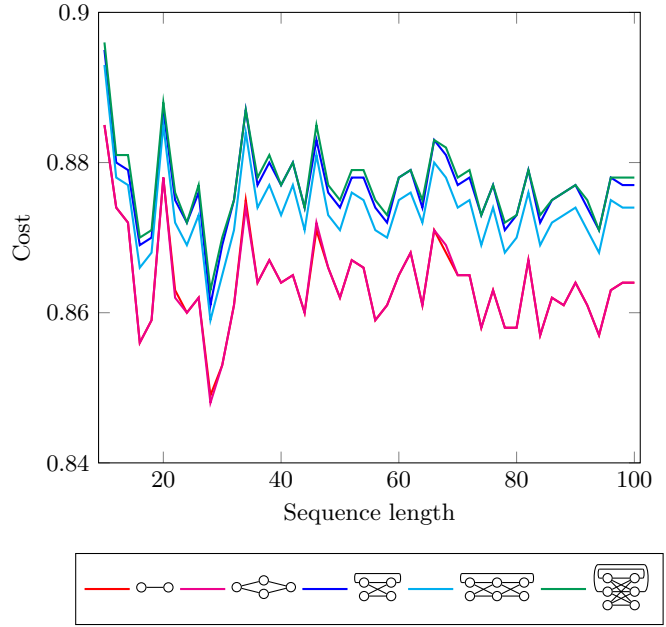
\begin{figure}[htbp]
	\centering
	\begin{tikzpicture}
		\begin{axis}[
			width=0.49\linewidth, % Scale the plot to \linewidth
			%grid=major, % Display a grid
			%grid style={dashed,gray!30}, % Set the style
			xlabel= Sequence length, % Set the labels
			ylabel= Cost,
			xmin=9,
			xmax=101,
			ymin=0.84,
			ymax=0.9,
			legend columns=5, 
			legend style={at={(0.5,-0.2)},anchor=north}, % Put the legend below the plot
			%x tick label style={rotate=90,anchor=east} % Display labels sideways
			]
			\addplot[thick,color=red,mark=None]  table[x=x,y=i1m0round,col sep=comma] {fidelity_LowpassDrift.csv}; 
			\addlegendentry{\begin{tikzpicture}[yscale=0.1,xscale=0.05]
					\node(1) [circle,draw,inner sep=0pt,minimum size=3.5pt] at (-1,0.15) {};
					\node(4) [circle,draw,inner sep=0pt,minimum size=3.5pt] at (0,0.15) {};
					\draw (1)--(4);
			\end{tikzpicture}}
			\addplot[thick,color=magenta,mark=None]  table[x=x,y=i1m0h2round,col sep=comma] {fidelity_LowpassDrift.csv}; 
			\addlegendentry{\begin{tikzpicture}[yscale=0.1,xscale=0.05]
					\node(1) [circle,draw,inner sep=0pt,minimum size=3.5pt] at (-1,0.15) {};
					\node(4) [circle,draw,inner sep=0pt,minimum size=3.5pt] at (0,0) {};
					\node(5) [circle,draw,inner sep=0pt,minimum size=3.5pt] at (0,0.3) {};
					\node(6) [circle,draw,inner sep=0pt,minimum size=3.5pt] at (1,0.15) {};
					\draw (1)--(4);
					\draw (1)--(5);
					\draw (6)--(4);
					\draw (6)--(5);
			\end{tikzpicture}}
			\addplot[thick,color=blue,mark=None]  table[x=x,y=i1m1round,col sep=comma] {fidelity_LowpassDrift.csv}; 
			\addlegendentry{\begin{tikzpicture}[yscale=0.1,xscale=0.05]
					\node(1) [circle,draw,inner sep=0pt,minimum size=3.5pt] at (-1,0) {};
					\node(2) [circle,draw,inner sep=0pt,minimum size=3.5 pt] at (-1,0.3) {};
					\node(4) [circle,draw,inner sep=0pt,minimum size=3.5pt] at (0,0) {};
					\node(5) [circle,draw,inner sep=0pt,minimum size=3.5pt] at (0,0.3) {};
					\draw (1)--(4);
					\draw (2)--(4);
					\draw (1)--(5);
					\draw (2)--(5);
					\draw (5) -- (0.3,0.3);
					\draw[out=0,in=0] (0.3,0.3) to (0.3,0.5);
					\draw (0.3,0.5) to (-1.3,0.5);
					\draw[out=180,in=180] (-1.3,0.5) to (-1.3,0.3);
					\draw (-1.3,0.3) -- (2);
			\end{tikzpicture}}
			\addplot[ thick,color=cyan,mark=None]  table[x=x,y=i1m1h2round,col sep=comma] {fidelity_LowpassDrift.csv}; 
			\addlegendentry{\begin{tikzpicture}[yscale=0.1,xscale=0.05]
					\node(1) [circle,draw,inner sep=0pt,minimum size=3.5pt] at (-1,0) {};
					\node(2) [circle,draw,inner sep=0pt,minimum size=3.5pt] at (-1,0.3) {};
					\node(4) [circle,draw,inner sep=0pt,minimum size=3.5pt] at (0,0) {};
					\node(5) [circle,draw,inner sep=0pt,minimum size=3.5pt] at (0,0.3) {};
					\node(6) [circle,draw,inner sep=0pt,minimum size=3.5pt] at (1,0) {};
					\node(7) [circle,draw,inner sep=0pt,minimum size=3.5pt] at (1,0.3) {};
					\draw (1)--(4);
					\draw (2)--(4);
					\draw (1)--(5);
					\draw (2)--(5);
					\draw (6)--(4);
					\draw (7)--(4);
					\draw (6)--(5);
					\draw (7)--(5);
					\draw (7) -- (1.3,0.3);
					\draw[out=0,in=0] (1.3,0.3) to (1.3,0.5);
					\draw (1.3,0.5) to (-1.3,0.5);
					\draw[out=180,in=180] (-1.3,0.5) to (-1.3,0.3);
					\draw (-1.3,0.3) -- (2);
			\end{tikzpicture} }
			\addplot[thick,color=ForestGreen,mark=None]  table[x=x,y=i1m2round,col sep=comma] {fidelity_LowpassDrift.csv}; 
			\addlegendentry{\begin{tikzpicture}[yscale=0.1,xscale=0.05]
					\node(1) [circle,draw,inner sep=0pt,minimum size=3.5pt] at (-1,-0.15) {};
					\node(2) [circle,draw,inner sep=0pt,minimum size=3.5pt] at (-1,0.15) {};
					\node(3) [circle,draw,inner sep=0pt,minimum size=3.5pt] at (-1,0.45) {};
					\node(4) [circle,draw,inner sep=0pt,minimum size=3.5pt] at (0,-0.15) {};
					\node(5) [circle,draw,inner sep=0pt,minimum size=3.5pt] at (0,0.15) {};
					\node(6) [circle,draw,inner sep=0pt,minimum size=3.5pt] at (0,0.45) {};
					\draw (1)--(4);
					\draw (2)--(4);
					\draw (3)--(4);
					\draw (1)--(5);
					\draw (2)--(5);
					\draw (3)--(5);
					\draw (1)--(6);
					\draw (2)--(6);
					\draw (3)--(6);
					\draw (6) -- (0.3,0.45);
					\draw[out=0,in=0] (0.3,0.45) to (0.3,0.65);
					\draw (0.3,0.65) to (-1.3,0.65);
					\draw[out=180,in=180] (-1.3,0.65) to (-1.3,0.45);
					\draw (-1.3,0.45) -- (3);
					\draw (5) -- (0.35,0.15);
					\draw[out=0,in=0] (0.35,0.15) to (0.35,0.75);
					\draw (0.35,0.75) to (-1.35,0.75);
					\draw[out=180,in=180] (-1.35,0.75) to (-1.35,0.15);
					\draw (-1.35,0.15) -- (2);
			\end{tikzpicture}}
			\addplot [domain=9:101, samples=10, color=gray,]{0.5};
		\end{axis}
	\end{tikzpicture}
	\caption[]{High-frequency noise mitigation for different network topologies, smoothing parameter $\alpha = 0.4$.}
	\label{fig:HiFreq}
\end{figure}
It would be nice to have an analytic estimate of the performance boost that the memory provides for the exponential smoothing. Unfortunately, we do not know how to derive such an estimate yet. Note that the numerical results in Figure~\ref{fig:HiFreq} fall within our expectations. Indeed, as feed-forward networks do not have access to the previously seen phases, a pretty good strategy is to output
\be \label{eq:LowPassMulti}
    \left. \psi^j_t \right|_{FF} &=& \alpha \cdot \phi^j_{t} 
                 + (1-\alpha) \cdot \overline{\psi^{j}_{t-1}}\nn\\
                 &=& \alpha \cdot \left( \phi^j_{t} +  \sum_{\tau < t} (1-\alpha)^{t - \tau} \overline{\phi^{j}_{\tau}} \right) = |\text{for distributions such that }\overline{\phi^j_{\tau}} =0 | \nn\\
                 &=&   \alpha \cdot \phi^j_{t}
\ee
where $\overline{\phi}$ denotes an expectation value of $\phi$. The expected fidelity that this strategy would have produced if it was realisable is
\be
    F_{\overline{\phi}} &=& \int_{-\pi}^{\pi} \frac{d \phi^j_{1}}{2 \pi} \dots \int_{-\pi}^{\pi} \frac{d \phi^j_{t}}{2 \pi}  \left| \left\la \alpha \cdot \phi^j_{t} \left| \alpha \cdot \left( \phi^j_{t} +  \sum_{\tau < t} (1-\alpha)^{t - \tau} \phi^j_{\tau} \right) \right. \right\ra \right|^2 \nn\\
           &=& \frac{1}{4} \cdot \int_{-\pi}^{\pi} \frac{d \phi^j_{1}}{2 \pi} \dots \int_{-\pi}^{\pi} \frac{d \phi^j_{t}}{2 \pi} 
           \left|1 + e^{\alpha \cdot \sum_{\tau < t} (1-\alpha)^{t - \tau} \phi^j_{\tau}} \right|^2 \nn\\
           &=& \frac{1}{2} + \frac{1}{2} \cdot \prod_\tau \frac{\sin\left((1-\alpha)^\tau \cdot \alpha \cdot \pi \right)}{(1-\alpha)^\tau\cdot \alpha \cdot \pi}
\ee
For the parameters of our data set and test sequence length $100$, $F_{\overline{\phi}}$ is close to its minimum and is$\sim0.93$.
The map~\ref{eq:LowPassMulti} is not a channel, see Proposition~\ref{thm:Qmulti}. We train feed-forward QNNs with the same topology as in Figure~\ref{fig:HiFreq} and observe that each network trains to the fidelity $F_{\cdot \alpha}~\sim0.925$. We can bound from below the attainable fidelity of feed-forward QNNs $F_{FF} \geq F_{\overline{\phi}} \cdot F_{\cdot \alpha} \sim 0.86$. Intriguingly, this bound, up to fluctuations, corresponds very well with the average performance of feed-forward QNNs observed in Figure~\ref{fig:HiFreq}. We hope this argument can be used as a basis for a rigorous proof.

On the other hand, Recurrent networks have access to information about previously encountered phases leading to increased performance. However, this knowledge is very limited, as some unclonable information was encoded into the previous outputs. Moreover, as a trained Q(R)NN is a quantum channel, some of this knowledge might not be utilisable. Together with the decent performance of feed-forward QNNs for exponential smoothing, these reasons lead to detectable, but in comparison to drift mitigation, modest performance gap.

Deeper networks have higher expressive power but may require more data or optimisation rounds to train. This is especially true if the extra expressive power is not needed for the task. We have used the same data and number of optimisation steps for every network we trained. One may attribute the slightly lower performance of the deeper network to the unneeded complexity.

\paragraph{Bandwidth filters}\leavevmode\\
Combining high- and low-frequency noise mitigation, we can train a filter on a data set that requires a band window
    \be
        \psi^j_0 &=& \phi^j_0 \nn\\
        \psi^j_t &=& \alpha \cdot \phi^j_{t} 
                 + (1-\alpha) \cdot \psi^j_{t-1} - g^j(t), 
        \quad 0 < \alpha < 1.
    \ee
    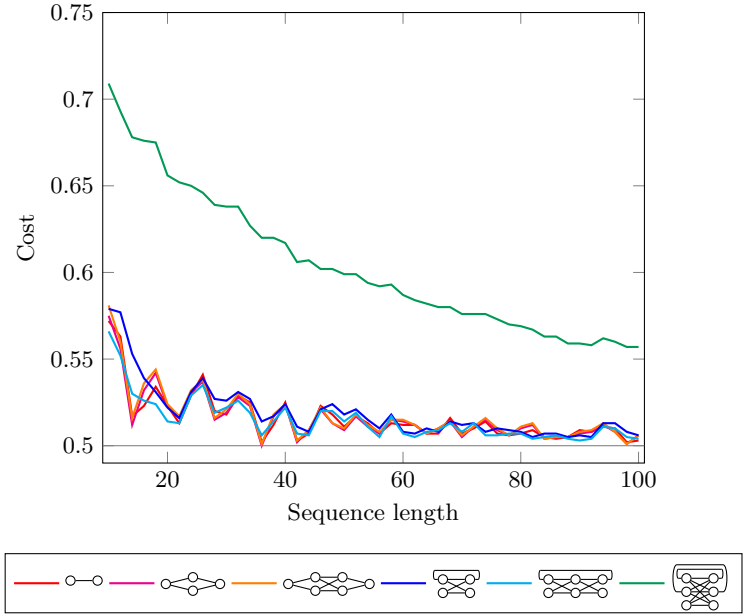
\begin{figure}[htbp]
	\centering
	\begin{tikzpicture}
		\begin{axis}[
			width=0.49\linewidth, % Scale the plot to \linewidth
			%grid=major, % Display a grid
			%grid style={dashed,gray!30}, % Set the style
			xlabel= Sequence length, % Set the labels
			ylabel= Cost,
			xmin=9,
			xmax=101,
			ymin=0.49,
			ymax=0.75,
			legend columns=6, 
			legend style={at={(0.5,-0.2)},anchor=north}, % Put the legend below the plot
			%x tick label style={rotate=90,anchor=east} % Display labels sideways
			]
			\addplot[thick,color=red,mark=None]  table[x=x,y=i1m0round,col sep=comma] {fidelity_DriftAndSmooth.csv}; 
			\addlegendentry{\begin{tikzpicture}[yscale=0.1,xscale=0.05, baseline]
					\node(1) [circle,draw,inner sep=0pt,minimum size=3.5pt] at (-1,0.15) {};
					\node(4) [circle,draw,inner sep=0pt,minimum size=3.5pt] at (0,0.15) {};
					\draw (1)--(4);
			\end{tikzpicture}}
			\addplot[thick,color=magenta,mark=None]  table[x=x,y=i1m0h2round,col sep=comma] {fidelity_DriftAndSmooth.csv}; 
			\addlegendentry{\begin{tikzpicture}[yscale=0.1,xscale=0.05, baseline]
					\node(1) [circle,draw,inner sep=0pt,minimum size=3.5pt] at (-1,0.15) {};
					\node(4) [circle,draw,inner sep=0pt,minimum size=3.5pt] at (0,0) {};
					\node(5) [circle,draw,inner sep=0pt,minimum size=3.5pt] at (0,0.3) {};
					\node(6) [circle,draw,inner sep=0pt,minimum size=3.5pt] at (1,0.15) {};
					\draw (1)--(4);
					\draw (1)--(5);
					\draw (6)--(4);
					\draw (6)--(5);
			\end{tikzpicture}}
			\addplot[thick,color=orange,mark=None]  table[x=x,y=i1m0h2h2round,col sep=comma] {fidelity_DriftAndSmooth.csv}; 
			\addlegendentry{\begin{tikzpicture}[yscale=0.1,xscale=0.05, baseline]
					\node(1) [circle,draw,inner sep=0pt,minimum size=3.5pt] at (-1,0.15) {};
					\node(2) [circle,draw,inner sep=0pt,minimum size=3.5pt] at (0,0) {};
					\node(3) [circle,draw,inner sep=0pt,minimum size=3.5pt] at (0,0.3) {};
					\node(4) [circle,draw,inner sep=0pt,minimum size=3.5pt] at (1,0) {};
					\node(5) [circle,draw,inner sep=0pt,minimum size=3.5pt] at (1,0.3) {};
					\node(6) [circle,draw,inner sep=0pt,minimum size=3.5pt] at (2,0.15) {};
					\draw (1)--(2);
					\draw (1)--(3);
					\draw (2)--(4);
					\draw (2)--(5);
					\draw (3)--(4);
					\draw (3)--(5);
					\draw (6)--(4);
					\draw (6)--(5);
			\end{tikzpicture}}
			\addplot[thick,color=blue,mark=None]  table[x=x,y=i1m1round,col sep=comma] {fidelity_DriftAndSmooth.csv}; 
			\addlegendentry{\begin{tikzpicture}[yscale=0.1,xscale=0.05, baseline]
					\node(1) [circle,draw,inner sep=0pt,minimum size=3.5pt] at (-1,0) {};
					\node(2) [circle,draw,inner sep=0pt,minimum size=3.5 pt] at (-1,0.3) {};
					\node(4) [circle,draw,inner sep=0pt,minimum size=3.5pt] at (0,0) {};
					\node(5) [circle,draw,inner sep=0pt,minimum size=3.5pt] at (0,0.3) {};
					\draw (1)--(4);
					\draw (2)--(4);
					\draw (1)--(5);
					\draw (2)--(5);
					\draw (5) -- (0.3,0.3);
					\draw[out=0,in=0] (0.3,0.3) to (0.3,0.5);
					\draw (0.3,0.5) to (-1.3,0.5);
					\draw[out=180,in=180] (-1.3,0.5) to (-1.3,0.3);
					\draw (-1.3,0.3) -- (2);
			\end{tikzpicture}}
			\addplot[ thick,color=cyan,mark=None]  table[x=x,y=i1m1h2round,col sep=comma] {fidelity_DriftAndSmooth.csv}; 
			\addlegendentry{\begin{tikzpicture}[yscale=0.1,xscale=0.05, baseline]
					\node(1) [circle,draw,inner sep=0pt,minimum size=3.5pt] at (-1,0) {};
					\node(2) [circle,draw,inner sep=0pt,minimum size=3.5pt] at (-1,0.3) {};
					\node(4) [circle,draw,inner sep=0pt,minimum size=3.5pt] at (0,0) {};
					\node(5) [circle,draw,inner sep=0pt,minimum size=3.5pt] at (0,0.3) {};
					\node(6) [circle,draw,inner sep=0pt,minimum size=3.5pt] at (1,0) {};
					\node(7) [circle,draw,inner sep=0pt,minimum size=3.5pt] at (1,0.3) {};
					\draw (1)--(4);
					\draw (2)--(4);
					\draw (1)--(5);
					\draw (2)--(5);
					\draw (6)--(4);
					\draw (7)--(4);
					\draw (6)--(5);
					\draw (7)--(5);
					\draw (7) -- (1.3,0.3);
					\draw[out=0,in=0] (1.3,0.3) to (1.3,0.5);
					\draw (1.3,0.5) to (-1.3,0.5);
					\draw[out=180,in=180] (-1.3,0.5) to (-1.3,0.3);
					\draw (-1.3,0.3) -- (2);
			\end{tikzpicture} }
			\addplot[thick,color=ForestGreen,mark=None]  table[x=x,y=i1m2round,col sep=comma] {fidelity_DriftAndSmooth.csv}; 
			\addlegendentry{\begin{tikzpicture}[yscale=0.1,xscale=0.05, baseline]
					\node(1) [circle,draw,inner sep=0pt,minimum size=3.5pt] at (-1,-0.15) {};
					\node(2) [circle,draw,inner sep=0pt,minimum size=3.5pt] at (-1,0.15) {};
					\node(3) [circle,draw,inner sep=0pt,minimum size=3.5pt] at (-1,0.45) {};
					\node(4) [circle,draw,inner sep=0pt,minimum size=3.5pt] at (0,-0.15) {};
					\node(5) [circle,draw,inner sep=0pt,minimum size=3.5pt] at (0,0.15) {};
					\node(6) [circle,draw,inner sep=0pt,minimum size=3.5pt] at (0,0.45) {};
					\draw (1)--(4);
					\draw (2)--(4);
					\draw (3)--(4);
					\draw (1)--(5);
					\draw (2)--(5);
					\draw (3)--(5);
					\draw (1)--(6);
					\draw (2)--(6);
					\draw (3)--(6);
					\draw (6) -- (0.3,0.45);
					\draw[out=0,in=0] (0.3,0.45) to (0.3,0.65);
					\draw (0.3,0.65) to (-1.3,0.65);
					\draw[out=180,in=180] (-1.3,0.65) to (-1.3,0.45);
					\draw (-1.3,0.45) -- (3);
					\draw (5) -- (0.35,0.15);
					\draw[out=0,in=0] (0.35,0.15) to (0.35,0.75);
					\draw (0.35,0.75) to (-1.35,0.75);
					\draw[out=180,in=180] (-1.35,0.75) to (-1.35,0.15);
					\draw (-1.35,0.15) -- (2);
			\end{tikzpicture}}
			\addplot [domain=9:101, samples=10, color=gray,]{0.5};
		\end{axis}
	\end{tikzpicture}
	\caption[]{Low and high frequency noise mitigation, drift velocity $v=0.6$ and smoothing parameter $\alpha=0.4$.}
	\label{fig:BandFilter}
\end{figure}
You can see the results for the performance of the trained filter in Figure~\ref{fig:BandFilter}, which is a larger version of Figure~\ref{fig:resultsmainpaper} (c) in the main text. We see that this task has a threshold behaviour regarding the memory; namely, at least two qubits of memory are needed to outperform feed-forward QNNs significantly.
\subsection{Local Cost with mixed output}
For a proof of concept of the algorithm given for mixed output states, Figure \ref{fig:mixedSWAP1m} shows the learning procedure for mixed states for the SWAP channel with each \(N=20\) training and test pairs and a learning rate of \(\epsilon  \eta=0.05\). As the optimisation here is done with the averaged Hilbert-Schmidt norm and the averaged fidelity is physically more relevant, both are shown in dependence on the training step. The plot in panel (b) looks similar to the one with pure output, only that the starting fidelities and the fidelity the feed-forward QNN reaches are higher as the fidelity of two random mixed states is higher than the one of two random pure states if the Haar measure is used as probability measure (for mixed states in combination with purification).
\begin{figure}[htbp]
	\begin{subfigure}{0.49\textwidth}
		\begin{tikzpicture}
			\node at (-1,6.5) {\textbf{a}};
			\begin{axis}[
				width=\linewidth, % Scale the plot to \linewidth
				xlabel= Training step, % Set the labels
				ylabel= Averaged Hilbert-Schmidt norm,
				xmin=0,
				xmax=1500,
				ymin=-0.025,
				ymax=0.55,
				legend columns=2, 
				legend style={at={(0.5,-0.2)},anchor=north},
				]
				\addplot[very thick,color=red,mark=None]  table[x=step,y=hilbertwomround,col sep=comma] {mixedSWAP1m.csv}; 
				\addlegendentry{\begin{tikzpicture}[scale=0.05]
						\node(1) [circle,draw,inner sep=0pt,minimum size=3.5pt] at (-1,0.15) {};
						\node(4) [circle,draw,inner sep=0pt,minimum size=3.5pt] at (0,0.15) {};
						\draw (1)--(4);
					\end{tikzpicture} training}
				\addplot[very thick,color=blue,mark=None]   table[x=step,y=hilbert1mround,col sep=comma] {mixedSWAP1m.csv}; 
				\addlegendentry{\begin{tikzpicture}[xscale=0.05,yscale=0.1]
						\node(1) [circle,draw,inner sep=0pt,minimum size=3.5pt] at (-1,0) {};
						\node(2) [circle,draw,inner sep=0pt,minimum size=3.5 pt] at (-1,0.3) {};
						\node(4) [circle,draw,inner sep=0pt,minimum size=3.5pt] at (0,0) {};
						\node(5) [circle,draw,inner sep=0pt,minimum size=3.5pt] at (0,0.3) {};
						\draw (1)--(4);
						\draw (2)--(4);
						\draw (1)--(5);
						\draw (2)--(5);
						\draw (5) -- (0.3,0.3);
						\draw[out=0,in=0] (0.3,0.3) to (0.3,0.5);
						\draw (0.3,0.5) to (-1.3,0.5);
						\draw[out=180,in=180] (-1.3,0.5) to (-1.3,0.3);
						\draw (-1.3,0.3) -- (2);
					\end{tikzpicture} training}
				\addplot[very thick,color=red,mark=None,dashed]   table[x=step,y=hilbertwomtestround,col sep=comma] {mixedSWAP1m.csv}; 
				\addlegendentry{\begin{tikzpicture}[scale=0.05]
						\node(1) [circle,draw,inner sep=0pt,minimum size=3.5pt] at (-1,0.15) {};
						\node(4) [circle,draw,inner sep=0pt,minimum size=3.5pt] at (0,0.15) {};
						\draw (1)--(4);
					\end{tikzpicture} testing} 
				\addplot[very thick,color=blue,mark=None,dashed]   table[x=step,y=hilbert1mtestround,col sep=comma] {mixedSWAP1m.csv}; 
				\addlegendentry{\begin{tikzpicture}[xscale=0.05,yscale=0.1]
						\node(1) [circle,draw,inner sep=0pt,minimum size=3.5pt] at (-1,0) {};
						\node(2) [circle,draw,inner sep=0pt,minimum size=3.5 pt] at (-1,0.3) {};
						\node(4) [circle,draw,inner sep=0pt,minimum size=3.5pt] at (0,0) {};
						\node(5) [circle,draw,inner sep=0pt,minimum size=3.5pt] at (0,0.3) {};
						\draw (1)--(4);
						\draw (2)--(4);
						\draw (1)--(5);
						\draw (2)--(5);
						\draw (5) -- (0.3,0.3);
						\draw[out=0,in=0] (0.3,0.3) to (0.3,0.5);
						\draw (0.3,0.5) to (-1.3,0.5);
						\draw[out=180,in=180] (-1.3,0.5) to (-1.3,0.3);
						\draw (-1.3,0.3) -- (2);
					\end{tikzpicture} testing}
			\end{axis}
		\end{tikzpicture}
	\end{subfigure}
	\hfill
	\begin{subfigure}{0.49\textwidth}
		\begin{tikzpicture}
			\node at (-1,6.5) {\textbf{b}};
			\begin{axis}[
				width=\linewidth, % Scale the plot to \linewidth
				xlabel= Training step, % Set the labels
				ylabel=Averaged fidelity,
				xmin=0,
				xmax=1500,
				ymin=0.82,
				ymax=1.025,
				legend columns=2, 
				legend style={at={(0.5,-0.2)},anchor=north},
				]
				\addplot[very thick,color=red,mark=None]  table[x=step,y=fidelitywomround,col sep=comma] {mixedSWAP1m.csv}; 
				\addlegendentry{\begin{tikzpicture}[scale=0.05]
						\node(1) [circle,draw,inner sep=0pt,minimum size=3.5pt] at (-1,0.15) {};
						\node(4) [circle,draw,inner sep=0pt,minimum size=3.5pt] at (0,0.15) {};
						\draw (1)--(4);
					\end{tikzpicture} training}
				\addplot[very thick,color=blue,mark=None]   table[x=step,y=fidelity1mround,col sep=comma] {mixedSWAP1m.csv}; 
				\addlegendentry{\begin{tikzpicture}[xscale=0.05,yscale=0.1]
						\node(1) [circle,draw,inner sep=0pt,minimum size=3.5pt] at (-1,0) {};
						\node(2) [circle,draw,inner sep=0pt,minimum size=3.5 pt] at (-1,0.3) {};
						\node(4) [circle,draw,inner sep=0pt,minimum size=3.5pt] at (0,0) {};
						\node(5) [circle,draw,inner sep=0pt,minimum size=3.5pt] at (0,0.3) {};
						\draw (1)--(4);
						\draw (2)--(4);
						\draw (1)--(5);
						\draw (2)--(5);
						\draw (5) -- (0.3,0.3);
						\draw[out=0,in=0] (0.3,0.3) to (0.3,0.5);
						\draw (0.3,0.5) to (-1.3,0.5);
						\draw[out=180,in=180] (-1.3,0.5) to (-1.3,0.3);
						\draw (-1.3,0.3) -- (2);
					\end{tikzpicture} training}
				\addplot[very thick,color=red,mark=None,dashed]   table[x=step,y=fidelitywomtestround,col sep=comma] {mixedSWAP1m.csv}; 
				\addlegendentry{\begin{tikzpicture}[scale=0.05]
						\node(1) [circle,draw,inner sep=0pt,minimum size=3.5pt] at (-1,0.15) {};
						\node(4) [circle,draw,inner sep=0pt,minimum size=3.5pt] at (0,0.15) {};
						\draw (1)--(4);
					\end{tikzpicture} testing} 
				\addplot[very thick,color=blue,mark=None,dashed]   table[x=step,y=fidelity1mtestround,col sep=comma] {mixedSWAP1m.csv}; 
				\addlegendentry{\begin{tikzpicture}[xscale=0.05,yscale=0.1]
						\node(1) [circle,draw,inner sep=0pt,minimum size=3.5pt] at (-1,0) {};
						\node(2) [circle,draw,inner sep=0pt,minimum size=3.5 pt] at (-1,0.3) {};
						\node(4) [circle,draw,inner sep=0pt,minimum size=3.5pt] at (0,0) {};
						\node(5) [circle,draw,inner sep=0pt,minimum size=3.5pt] at (0,0.3) {};
						\draw (1)--(4);
						\draw (2)--(4);
						\draw (1)--(5);
						\draw (2)--(5);
						\draw (5) -- (0.3,0.3);
						\draw[out=0,in=0] (0.3,0.3) to (0.3,0.5);
						\draw (0.3,0.5) to (-1.3,0.5);
						\draw[out=180,in=180] (-1.3,0.5) to (-1.3,0.3);
						\draw (-1.3,0.3) -- (2);
					\end{tikzpicture} testing}
			\end{axis}
		\end{tikzpicture}
	\end{subfigure}
\caption[]{\textbf{Comparison of QRNN and QNN for the delay by one channel with mixed in- and outputs.} A 
	\begin{tikzpicture}[yscale=0.6,xscale=0.4, baseline]
		\node(1) [circle,draw,inner sep=0pt,minimum size=3.5pt] at (-1,0) {};
		\node(2) [circle,draw,inner sep=0pt,minimum size=3.5 pt] at (-1,0.3) {};
		\node(4) [circle,draw,inner sep=0pt,minimum size=3.5pt] at (0,0) {};
		\node(5) [circle,draw,inner sep=0pt,minimum size=3.5pt] at (0,0.3) {};
		\draw (1)--(4);
		\draw (2)--(4);
		\draw (1)--(5);
		\draw (2)--(5);
		\draw (5) -- (0.3,0.3);
		\draw[out=0,in=0] (0.3,0.3) to (0.3,0.5);
		\draw (0.3,0.5) to (-1.3,0.5);
		\draw[out=180,in=180] (-1.3,0.5) to (-1.3,0.3);
		\draw (-1.3,0.3) -- (2);
	\end{tikzpicture} QRNN with one memory qubit (blue) and a \begin{tikzpicture}[yscale=0.6,xscale=0.4, baseline]
		\node(1) [circle,draw,inner sep=0pt,minimum size=3.5pt] at (-1,0.15) {};
		\node(4) [circle,draw,inner sep=0pt,minimum size=3.5pt] at (0,0.15) {};
		\draw (1)--(4);
	\end{tikzpicture} feed-forward QNN (red) are trained with learningrate \(\epsilon  \eta=0.05\) to learn the delay by one channel with \(N=20\) training pairs; the network is then tested on a test set of the same size. The drawn-through lines show the training, and the dashed lines the testing. Panel \textbf{(a)} shows the Cost used to train the networks, which is the averaged Hilbert-Schmidt norm, in dependence on the training step.   Panel \textbf{(b)} shows the physically more relevant averaged fidelity in dependence of the training step.}
\label{fig:mixedSWAP1m} 
\end{figure}
\subsection{Global Cost with pure output}
For a proof of principle of the algorithm provided for the global cost, Figure \ref{fig:mixedSWAP1m} shows the learning procedure for the SWAP channel with each \(N=20\) training and test pairs and a learningrate of \(\epsilon  \eta=0.005\) with the algorithm for maximising the global cost. As the global cost is smaller and before the training at about \(10^{-7}\), for clarity, the cost's logarithm is shown in a second plot. Except for the different scaling of the cost from zero to one and the steepness, the plots look similar to the one for the local cost.
\begin{figure}[htbp]
	\begin{subfigure}{0.49\textwidth}
		\begin{tikzpicture}
			\node at (-1,6.5) {\textbf{a}};
			\begin{axis}[
				width=\linewidth, % Scale the plot to \linewidth
				xlabel= Training step, % Set the labels
				ylabel= Cost,
				xmin=0,
				xmax=1000,
				ymin=-0.05,
				ymax=1.05,
				legend columns=2, 
				legend style={at={(0.5,-0.2)},anchor=north},
				]
				\addplot[very thick,color=red,mark=None]  table[x=step,y=Cwomround,col sep=comma] {bigSWAP1m.csv};
				\addlegendentry{\begin{tikzpicture}[scale=0.05, baseline]
						\node(1) [circle,draw,inner sep=0pt,minimum size=3.5pt] at (-1,0.15) {};
						\node(4) [circle,draw,inner sep=0pt,minimum size=3.5pt] at (0,0.15) {};
						\draw (1)--(4);
					\end{tikzpicture} training} 
				\addplot[very thick,color=blue,mark=None]   table[x=step,y=C1mround,col sep=comma] {bigSWAP1m.csv};
				\addlegendentry{\begin{tikzpicture}[xscale=0.05,yscale=0.1, baseline]
						\node(1) [circle,draw,inner sep=0pt,minimum size=3.5pt] at (-1,0) {};
						\node(2) [circle,draw,inner sep=0pt,minimum size=3.5 pt] at (-1,0.3) {};
						\node(4) [circle,draw,inner sep=0pt,minimum size=3.5pt] at (0,0) {};
						\node(5) [circle,draw,inner sep=0pt,minimum size=3.5pt] at (0,0.3) {};
						\draw (1)--(4);
						\draw (2)--(4);
						\draw (1)--(5);
						\draw (2)--(5);
						\draw (5) -- (0.3,0.3);
						\draw[out=0,in=0] (0.3,0.3) to (0.3,0.5);
						\draw (0.3,0.5) to (-1.3,0.5);
						\draw[out=180,in=180] (-1.3,0.5) to (-1.3,0.3);
						\draw (-1.3,0.3) -- (2);
					\end{tikzpicture} training} 
				\addplot[very thick,color=red,mark=None,dashed]   table[x=step,y=Cwomtestround,col sep=comma] {bigSWAP1m.csv}; 
				\addlegendentry{\begin{tikzpicture}[scale=0.05, baseline]
						\node(1) [circle,draw,inner sep=0pt,minimum size=3.5pt] at (-1,0.15) {};
						\node(4) [circle,draw,inner sep=0pt,minimum size=3.5pt] at (0,0.15) {};
						\draw (1)--(4);
					\end{tikzpicture} testing} 
				\addplot[very thick,color=blue,mark=None,dashed]   table[x=step,y=C1mtestround,col sep=comma] {bigSWAP1m.csv};
				\addlegendentry{\begin{tikzpicture}[xscale=0.05,yscale=0.1, baseline]
						\node(1) [circle,draw,inner sep=0pt,minimum size=3.5pt] at (-1,0) {};
						\node(2) [circle,draw,inner sep=0pt,minimum size=3.5 pt] at (-1,0.3) {};
						\node(4) [circle,draw,inner sep=0pt,minimum size=3.5pt] at (0,0) {};
						\node(5) [circle,draw,inner sep=0pt,minimum size=3.5pt] at (0,0.3) {};
						\draw (1)--(4);
						\draw (2)--(4);
						\draw (1)--(5);
						\draw (2)--(5);
						\draw (5) -- (0.3,0.3);
						\draw[out=0,in=0] (0.3,0.3) to (0.3,0.5);
						\draw (0.3,0.5) to (-1.3,0.5);
						\draw[out=180,in=180] (-1.3,0.5) to (-1.3,0.3);
						\draw (-1.3,0.3) -- (2);
					\end{tikzpicture} testing}
				\addplot [domain=0:1000, samples=10, color=gray,]{1};
			\end{axis}
		\end{tikzpicture}
	\end{subfigure}
	\hfill
	\begin{subfigure}{0.49\textwidth}
		\begin{tikzpicture}
			\node at (-1,6.5) {\textbf{b}};
			\begin{axis}[
				width=\linewidth, % Scale the plot to \linewidth
				xlabel= Training step, % Set the labels
				ylabel= ln(Cost),
				xmin=0,
				xmax=1000,
				ymin=-16,
				ymax=0.5,
				legend columns=2, 
				legend style={at={(0.5,-0.2)},anchor=north},
				]
				\addplot[very thick,color=red,mark=None]  table[x=step,y=logCwomround,col sep=comma] {bigSWAP1m.csv}; 
				\addlegendentry{\begin{tikzpicture}[scale=0.05, baseline]
						\node(1) [circle,draw,inner sep=0pt,minimum size=3.5pt] at (-1,0.15) {};
						\node(4) [circle,draw,inner sep=0pt,minimum size=3.5pt] at (0,0.15) {};
						\draw (1)--(4);
					\end{tikzpicture} training}
				\addplot[very thick,color=blue,mark=None]   table[x=step,y=logC1mround,col sep=comma] {bigSWAP1m.csv};
				\addlegendentry{\begin{tikzpicture}[xscale=0.05,yscale=0.1, baseline]
						\node(1) [circle,draw,inner sep=0pt,minimum size=3.5pt] at (-1,0) {};
						\node(2) [circle,draw,inner sep=0pt,minimum size=3.5 pt] at (-1,0.3) {};
						\node(4) [circle,draw,inner sep=0pt,minimum size=3.5pt] at (0,0) {};
						\node(5) [circle,draw,inner sep=0pt,minimum size=3.5pt] at (0,0.3) {};
						\draw (1)--(4);
						\draw (2)--(4);
						\draw (1)--(5);
						\draw (2)--(5);
						\draw (5) -- (0.3,0.3);
						\draw[out=0,in=0] (0.3,0.3) to (0.3,0.5);
						\draw (0.3,0.5) to (-1.3,0.5);
						\draw[out=180,in=180] (-1.3,0.5) to (-1.3,0.3);
						\draw (-1.3,0.3) -- (2);
					\end{tikzpicture} training} 
				\addplot[very thick,color=red,mark=None,dashed]   table[x=step,y=logCwomtestround,col sep=comma] {bigSWAP1m.csv}; 
				\addlegendentry{\begin{tikzpicture}[scale=0.05, baseline]
						\node(1) [circle,draw,inner sep=0pt,minimum size=3.5pt] at (-1,0.15) {};
						\node(4) [circle,draw,inner sep=0pt,minimum size=3.5pt] at (0,0.15) {};
						\draw (1)--(4);
					\end{tikzpicture} testing} 
				\addplot[very thick,color=blue,mark=None,dashed]   table[x=step,y=logC1mtestround,col sep=comma] {bigSWAP1m.csv}; 
				\addlegendentry{\begin{tikzpicture}[xscale=0.05,yscale=0.1, baseline]
						\node(1) [circle,draw,inner sep=0pt,minimum size=3.5pt] at (-1,0) {};
						\node(2) [circle,draw,inner sep=0pt,minimum size=3.5 pt] at (-1,0.3) {};
						\node(4) [circle,draw,inner sep=0pt,minimum size=3.5pt] at (0,0) {};
						\node(5) [circle,draw,inner sep=0pt,minimum size=3.5pt] at (0,0.3) {};
						\draw (1)--(4);
						\draw (2)--(4);
						\draw (1)--(5);
						\draw (2)--(5);
						\draw (5) -- (0.3,0.3);
						\draw[out=0,in=0] (0.3,0.3) to (0.3,0.5);
						\draw (0.3,0.5) to (-1.3,0.5);
						\draw[out=180,in=180] (-1.3,0.5) to (-1.3,0.3);
						\draw (-1.3,0.3) -- (2);
					\end{tikzpicture} testing}
				\addplot [domain=0:1000, samples=10, color=gray,]{0};
				\addplot [domain=0:1000, samples=10, color=gray,]{-13.863};
			\end{axis}
		\end{tikzpicture}
	\end{subfigure}
	\caption[]{\textbf{Comparison of QRNN and QNN for the delay by one channel with global cost.} A 
		\begin{tikzpicture}[yscale=0.6,xscale=0.4, baseline]
			\node(1) [circle,draw,inner sep=0pt,minimum size=3.5pt] at (-1,0) {};
			\node(2) [circle,draw,inner sep=0pt,minimum size=3.5 pt] at (-1,0.3) {};
			\node(4) [circle,draw,inner sep=0pt,minimum size=3.5pt] at (0,0) {};
			\node(5) [circle,draw,inner sep=0pt,minimum size=3.5pt] at (0,0.3) {};
			\draw (1)--(4);
			\draw (2)--(4);
			\draw (1)--(5);
			\draw (2)--(5);
			\draw (5) -- (0.3,0.3);
			\draw[out=0,in=0] (0.3,0.3) to (0.3,0.5);
			\draw (0.3,0.5) to (-1.3,0.5);
			\draw[out=180,in=180] (-1.3,0.5) to (-1.3,0.3);
			\draw (-1.3,0.3) -- (2);
		\end{tikzpicture} QRNN with one memory qubit (blue) and a \begin{tikzpicture}[yscale=0.6,xscale=0.4, baseline]
		\node(1) [circle,draw,inner sep=0pt,minimum size=3.5pt] at (-1,0.15) {};
		\node(4) [circle,draw,inner sep=0pt,minimum size=3.5pt] at (0,0.15) {};
		\draw (1)--(4);
	\end{tikzpicture} feed-forward QNN (red) are trained with learningrate \(\epsilon  \eta=0.005\) to learn the delay by one channel with \(N=20\) training pairs; the network is then tested on a test set of the same size. The drawn-through lines show the training, and the dashed lines the testing. Panel \textbf{(a)} shows the cost in dependence on the training step.   Panel \textbf{(b)} shows the logarithm of that cost as the training of the feed-forward QNN is not visible otherwise.}
	\label{fig:globalSWAP} 
\end{figure}
\section{Quantum algorithm for quantum training of the QRNN}
\label{quantumalgorithm}
The quantum algorithm is mainly the same as for feed-forward QNNs as presented in \cite{Beer2021training}. In the QNN, parameterised unitaries are used with \(N_p\) parameters in total, where all parameters are stored in a vector \(\Vec{p_t}\in \mathbb{R}^{N_p}\) that is trained using classical gradient descent. This vector is initialised as \(\Vec{p_0}\) and then updated according to \(\Vec{p_{t+1}}=\Vec{p_{t}}+\Vec{\mathrm{d}p}\) where \(\Vec{\mathrm{d}p}=\eta\nabla C\left(\Vec{p_t}\right)\) with learningrate \(\eta\). The gradient is estimated by
\[\left(\nabla C\left(\Vec{p_t}\right)\right)_k=\frac{C\left(\Vec{p_t}+\epsilon\Vec{e_k}\right)-C\left(\Vec{p_t}-\epsilon\Vec{e_k}\right)}{2\epsilon}\]
with \(\left(\Vec{e_k}\right)^j=\delta_k^j, \ k,j=1,...,N_p\) where the cost is estimated on a quantum computer using the SWAP trick as described in \cite{beerTrainingDeepQuantum2020} in supplementary note 4 for pure outputs.\\
Here, we have two possibilities; as already mentioned in the main text, we can either use the global or the local cost to train the QRNN. For the global cost, we first would have to evaluate the whole QRNN and compare the global output of the QRNN to the global training output. Using the global cost has the advantage of learning entanglement between different outputs. As opposed to this, for the local cost, we would start with the first input and memory state, make one run of the QRNN, and do the SWAP trick on the output and the corresponding output in the training set. We would then be left with a state in the memory register on which we then - together with the following input - use the QRNN a second time and so on until the last run is reached. The major advantage of this is that fewer qubits are needed because the in- and output qubit registers can be used again.

\newpage
\bibliographystyle{naturemag}
\bibliography{literaturnotes}
\end{document}